\documentclass{article}
\usepackage{macros}
\usepackage{mathpazo}
\usepackage{pdfpages}
\usepackage{bm}
\usepackage{braket}
\usepackage{physics}
\usepackage{cleveref}
\usepackage{amsmath,amssymb}

\usepackage{thmtools}
\usepackage{thm-restate}

\newcommand{\etrunc}{\epsilon_{\mathrm{trunc}}}
\newcommand{\edisc}{\epsilon_{\mathrm{disc}}}

\newcommand{\eamp}{\epsilon_{\mathrm{amp}}}
\newcommand{\xpath}{\widehat{X}}

\newcommand{\vertiii}[1]{{\left\vert\kern-0.25ex\left\vert\kern-0.25ex\left\vert #1 
    \right\vert\kern-0.25ex\right\vert\kern-0.25ex\right\vert}}

\DeclareMathOperator\arctanh{arctanh}

\title{Quantum Speedups for Derivative Pricing\\ Beyond Black-Scholes} 
\author[1]{Dylan Herman\protect\thanks{Email: \texttt{dylan.a.herman@jpmorgan.com}}}
\author[1]{Yue Sun}
\author[3]{Jin-Peng Liu}
\author[1]{Marco Pistoia}
\author[2]{\\Charlie Che}
\author[1]{Rob Otter}
\author[1]{Shouvanik Chakrabarti\protect\thanks{Email: \texttt{shouvanik.chakrabarti@jpmchase.com}}}
\author[3]{Aram W.~Harrow\protect\thanks{Email: \texttt{aram@mit.edu}}}

\affil[1]{Global Technology Applied Research, JPMorganChase, New York, NY 10001, USA}
\affil[2]{Quantitative Trading \& Research, JPMorganChase, New York, NY 10017, USA}
\affil[3]{Center for Theoretical Physics -- a Leinweber Institute, MIT, Cambridge, MA 02139, USA}
\date{}

\usepackage[style=alphabetic, maxbibnames=99, backref=true, backrefstyle=three]{biblatex}
\bibliography{bibliography}

\begin{document}
\maketitle

\begin{abstract}
This paper explores advancements in quantum algorithms for derivative pricing of exotics, a computational pipeline of fundamental importance in quantitative finance. For such cases, the classical Monte Carlo integration procedure provides the state-of-the-art provable, asymptotic performance: polynomial in problem dimension and quadratic in inverse-precision. While quantum algorithms are known to offer quadratic speedups over classical Monte Carlo methods, end-to-end speedups have been proven only in the simplified setting over the Black-Scholes geometric Brownian motion (GBM) model. 
This paper extends existing frameworks to demonstrate novel quadratic speedups for more practical models, such as the Cox–Ingersoll–Ross (CIR) model and a variant of Heston's stochastic volatility model, utilizing a characteristic of the underlying SDEs which we term ``fast-forwardability.''
Additionally, for general models that do not possess the fast-forwardable property, we introduce a quantum Milstein sampler, based on a novel quantum algorithm for sampling L\'evy areas, which enables quantum multi-level Monte Carlo to achieve quadratic speedups for multi-dimensional stochastic processes exhibiting certain types of correlations.
 
We also present an improved analysis of numerical integration for derivative pricing. This leads to substantial reductions in the resource requirements for pricing GBM and CIR models. Furthermore, we investigate the potential for additional reductions using arithmetic-free quantum procedures.   
Finally, we critique quantum partial differential equation (PDE) solvers as a method for derivative pricing based on amplitude estimation, identifying theoretical barriers that obstruct achieving a quantum speedup through this approach. Our findings significantly advance the understanding of quantum algorithms in derivative pricing, addressing key challenges and open questions in the field.
\end{abstract}

\clearpage

\tableofcontents

\clearpage

\section{Introduction}

A derivative contract~\cite{hull1993options} is an asset whose value is derived from the price of underlying assets such as stocks and currencies. Specifically, the value is obtained by evaluating a pre-specified \emph{payoff function} of the prices of the underlying assets over the duration of the contract. Derivative contracts have many applications in the financial industry, including hedging risk, speculation, and the exploitation of arbitrage opportunities. The current global derivatives market is estimated to be in the excess of trillions of dollars~\cite{isda2022derivatives}. The accurate valuation of these contracts is therefore of crucial importance and accounts for one of the primary computational challenges in financial engineering.

\emph{Derivative Pricing} aims to determine the value of entering a derivative contract, taking into account the uncertainty of the underlying asset prices. Uncertain assets prices are typically modelled by stochastic differential equations (SDEs) and it is well established~\cite{glasserman2004monte} that the task of derivative pricing can be reduced to computing the expectation of the payoff function over future realizations of the SDE. As such, algorithms for derivative pricing fall into two categories: when the number of underlying assets in the contract is small and their distribution is known in closed form, the expectations can be computed by explicit quadrature methods or solving a PDE (via spectral methods). These methods have the advantage of a poly-logarithmic dependence on the inverse precision and are therefore strongly polynomial time when the number of assets is a constant. Classical methods in this category, such as Fast-Fourier Transform Pricing \cite{carr1999option}, are particularly efficient for calibration tasks, which usually require pricing a variety of frequently-traded, path-independent options.

However, when the number of assets is large, quadrature and PDE methods suffer from the so-called ``curse of dimensionality" and cannot be used. Alternately, when a single-asset derivative has a complicated path-dependent payoff (i.e. an exotic), the additional state variables required to account for the path dependence may effectively re-introduce the curse-of-dimensionality.
In both cases, Monte Carlo methods, specifically Monte Carlo integration (MCI) \cite{glasserman2004monte}, are the method of choice for pricing such contracts. These methods have a worse inverse polynomial dependence on the precision and are therefore only weakly polynomial time. However, they have the advantage of having no explicit dimension dependence, with their runtime instead scaling with the variance of the estimators used for the random payoff. Since the variance usually scales only polynomially with the number of underlying assets, Monte Carlo methods do not suffer from the curse of dimensionality.

Monte Carlo methods exhibit polynomial scaling in both the dimension and inverse-precision. However in practice, it is still computationally challenging to implement these methods, especially at the scale of millions of contracts that must be priced at large financial institutions. A particular challenge is the \emph{super-linear} scaling of the number of required Monte Carlo simulation paths with the inverse precision. This unfavorable scaling means that some complex derivatives may require millions of paths to be simulated for an accurate valuation. The pricing of exotic derivatives is therefore one of the main computational challenges and bottlenecks in the financial industry. Due to current constraints, large portfolios of derivatives are usually repriced via batch simulations overnight, instead of the more accurate but computationally intensive live repricing with the latest market data.
This affects the ability to perform accurate and timely risk management in reaction to market changes.
Quantum algorithms have been previously shown to improve Monte Carlo algorithms in several settings~\cite{montanaro2015quantum} and it is therefore natural to study their applicability in the domain of derivative pricing. 

Quantum algorithms for Monte-Carlo-based derivative pricing (which we will simple refer to as quantum derivative pricing) have been studied before \cite{Herman_2023}: Stamatopoulos et al.~\cite{stamatopoulos2020option} described an algorithm to price path-independent options over an underlying geometric Brownian motion (GBM) evolution with quadratic quantum speedup in terms of the target precision. Chakrabarti et al.~\cite{chakrabarti2021threshold} extended this method to price certain path-dependent derivatives over geometric Brownian motion and produce estimates for the quantum hardware resources required to obtain an advantage in this regime. 

The main algorithmic tool behind the speedup in derivative pricing is the quantum Monte Carlo integration (QMCI) algorithm \cite{montanaro2015quantum}, which quadratically speeds up its classical counterpart, MCI. Note that the algorithm itself only provides a quadratic sampling advantage, in terms of queries to the derivative pricing model, in the form of an SDE. Generally, a resource analysis must be performed on a case-by-case basis to determine if there is an end-to-end quadratic speedup for pricing over a given model. This includes the quantum resources for simulating the SDE.

\subsection{Motivations}
\label{sec:motivation}
Unfortunately, the GBM model for asset prices is often considered too simplistic due to its assumptions of constant volatility and drift, which do not capture the complexities observed in real financial markets (e.g., volatility clustering and leverage effects)~\cite{rubinstein1985nonparametric,melino1991pricing}. While GBM is widely used for its mathematical tractability and is a foundational model in the Black-Scholes option pricing framework~\cite{black1973pricing,merton1973theory}, it is not typically employed for modeling quantities like volatilities, which are better captured by more sophisticated models such as stochastic volatility models (e.g., the Heston model~\cite{heston1993closed-form}) or GARCH models~\cite{bollerslev1986generalized}. %
These models account for the dynamical nature of volatility and provide a more accurate representation of market behavior.

Beyond being one of the most well-known financial models, the (multi-asset) Black-Scholes model (which we use interchangeably with GBM) has a particular property that makes Monte Carlo derivative pricing over it quite amenable to end-to-end resource analysis. Specifically, the GBM stochastic process (or SDE) possesses a notable property, which, although not explicitly highlighted in previous research, simplifies both its analysis and implementation compared to general SDEs.
We refer to this property as \emph{independent fast-forwardability}. 
Specifically, despite being a continuous-time process, GBM allows for the sampling of $T$ points in time, separated by a time increment $\Delta$ (referred to as monitoring points), by merely sampling $T$ independent standard Gaussian random variables.
The quantum counterpart of this property is that a quantum state encoding the distribution of GBM paths (termed a \emph{qsample}) can be constructed using a tensor product of $T$ Gaussian states. 
It appears that, among stochastic processes that are relevant for derivative pricing, only standard and geometric Brownian motions exhibit this unique characteristic of independent fast-forwardability, as no other processes have been identified to share this property in the current literature.

As mentioned earlier, GBM is well-suited for modeling equities with a fixed volatility. The Cox-Ingersoll-Ross (CIR) process \cite{cox1985theory} is an important process in finance that is more  natural for modeling interest rates and volatilities. As we will show, the CIR process satisfies a more general definition of fast-forwardability,  but is not independently fast-forwardable. However, the scheme for efficiently simulating the CIR process (the fast-forwarding scheme) is recursive, which results in a more complicated discretization error analysis than GBM.  This is a result of only having the transition density in closed-form as opposed to the stochastic process itself, like with GBM. Hence, most of the techniques used to provide a resource analysis for quantum MCI applied to GBM are no longer sufficient for CIR, indicating that it is unclear if there is an end-to-end quadratic speedup in this setting.

Even though CIR and GBM processes are quite different, it turns out that the characteristic of fast-forwardability is still what appears to enable efficient pricing with either classical or quantum MCI. Unfortunately, there are not many SDEs that appear to even satisfy our more general definition of fast-forwardability. In such settings, one must perform a time-discretization to approximately simulate the process \cite{platen1999introduction}. The lowest-order time-discretization scheme for SDEs is called the Euler-Maruyama (EM) method and has a (strong) convergence error of $\mathcal{O}(\sqrt{h})$ for step-size $h$. This scheme was used by \cite{Wang_2024} to perform quantum derivative pricing over Heston's Stochastic Volatility model \cite{heston1993closed-form}. The Heston model can, in some sense, be viewed as a combination of the GBM and CIR models, enabling equity processes with random volatilities.

It is well-known in the quantitative finance community that, at least in theory, the additional overhead from applying time-discretization schemes increases the sampling complexity of the ``vanilla'' classical MCI.
Giles \cite{giles2008multilevel} proposed a modified Monte Carlo method, called multi-level Monte Carlo (MLMC), that can retain the typically inverse-precision squared complexity of MCI. The MLMC approach of Giles was combined with QMCI by An et al. \cite{An2021quantumaccelerated}, who also presented sufficient conditions for retaining a quadratic speedup over classical MLMC. Hence, one needs to use quantum MLMC to be competitive with the best classical Monte Carlo methods, when in the approximate SDE simulation regime. Furthermore, \cite{An2021quantumaccelerated} showed in order to achieve a quadratically improved Monte Carlo convergence with quantum, it is sufficient to utilize a scheme that has a faster convergence than Euler-Maruyama. This was the only sufficient condition presented,  hence it is unclear whether \cite{Wang_2024} provides a speedup for pricing over the Heston model. Complicating the comparison further, it turns out that classical MLMC retains an inverse-squared precision complexity when using only the EM scheme.

As shown in \cite{An2021quantumaccelerated}, quantum MLMC can achieve a quadratically-faster convergence rate for globally-Lipschitz payoffs when the discretization scheme has an error convergence of $\mathcal{O}\left(h\right)$ for step-size $h$, i.e. strong-order one. This rate is provided by the so-called \emph{Milstein scheme} \cite{platen1999introduction}. Note that the Milstein scheme is not significantly more complicated to implement than EM in the case of completely uncoupled processes. However, for coupled processes,  such as the multi-asset Heston model, the Milstein scheme requires sampling the difference of double stochastic integrals called a \emph{L\'evy area}. Sampling L\'evy areas  can in general be just as expensive as MCI \cite{foster2023brownian}, which can introduce significant multiplicative factors of inverse-precision into the overall complexity. Hence, given the current literature, it is unclear if quantum can achieve an end-to-end quadratic speedup in the setting of approximate simulation of SDEs with correlations.

Lastly, one of the main limitations of current quantum derivative pricing for path-dependent derivatives is that both the SDE simulation cost and total number of quantum registers need to grow at least linearly with the number of monitoring points. However, it is well-known that depending on the state we want to prepare, there are state preparation techniques that use $\text{poly}\log(N)$ qubits in total \cite{harrow2009quantum, costa2021optimal}, where $N$ is the dimension of the state. For (It\^o) SDEs, we know that the marginal distribution of the solution satisfies the Fokker-Planck equation \cite{chewi2023optimization}. Hence, quantum PDE solvers \cite{An2022_theoryqode} could be used to prepare a state encoding the marginals. It is then natural to ask whether such methods can improve the resources for quantum MCI-based derivative pricing.

\subsection{Contributions}

\label{sec:contributions}

In this work, we significantly extend the existing frameworks for quantum derivative pricing by addressing the points mentioned in  Section~\ref{sec:motivation}, answering a variety of questions left open by prior work. As a result, we provide an end-to-end resource analysis of QMCI for more realistic models, as well as, improved and novel techniques for analyzing such algorithms. Our focus will mostly be on the \emph{end-to-end asymptotic complexity} of derivative pricing algorithms. Specifically, until this point, it was unclear whether quantum algorithms attained an end-to-end quadratic speedup for models besides the simple geometric Brownian motion, i.e. Black-Scholes model. 

An interesting consequence of most of our analysis is a theoretical analysis of classical Monte Carlo-based derivative pricing. Most of the techniques that we use to analyze the various sources of error also apply classically. Hence, we also make significant contributions to the classical derivative pricing literature. It seems that there has not been a lot of work on the bit-complexity of Monte Carlo pricing, for understandable reasons, which we analyze.

Lastly, most of our techniques are not even specific to mean estimation for derivative pricing. The sources of error that we analyze, i.e. discretization, truncation and distribution loading, are present wherever quantum MCI is applied. Hence, a lot of the improvements that we introduce, such as new subroutines for distribution loading and reducing the qubit count required to suppress the discretization error can also be applied beyond derivative pricing.

We summarize our contributions as follows.
\subsubsection*{Contribution 1: Extended Framework for Analyzing Quantum Derivative Pricing}

We present a general and highly-detailed framework for analyzing the error propagation and resources for quantum derivative pricing algorithms applied to arbitrary models  (Section~\ref{sec:framework_for_analyze}). This analysis reveals the various ways in which the different sources of error, such as truncation, discretization and distribution loading errors, depend on one another and must be scaled appropriately. We base our framework on an extension of the reparameterization approach of \cite{chakrabarti2021threshold} but with more general pricing problems in mind. 

Additionally, our extended framework also leads to some resource improvements. In prior work, such as \cite{chakrabarti2021threshold}, the estimated number of (qu)bits for arithmetic grew at least linearly with the dimension of the problem. Given that Monte Carlo pricing is performed regularly on classical computers with significantly fewer bits, this estimate seemed to be a limitation of the existing analysis and overly pessimistic. This high bit complexity was removed in the analysis of \cite{chen2023quantum} under  the assumption that the probability mass over a region could be estimated efficiently, enabling Grover-Rudolph-style loading \cite{grover2002creatingsuperpositionscorrespondefficiently}. However, as discussed in 
\cite{chakrabarti2021threshold} there are issues with applying Grover-Rudolph style loading to quantum derivative pricing.
Due to a more refined analysis of the complexity of numerical integration, we show (Section~\ref{sec:reduced_qubit_estimates}) that only \emph{logarithmic} in dimension number of (qu)bits is required for arithmetic for GBM and CIR using a simpler distribution loading scheme. Additionally, we explain how this technique could be applied more generally. This results in a substantial reduction in the number of required qubits for high-dimensional pricing problems. It is possible that one could extend our distribution loading techniques, based on, say, the scheme we present for loading with known characteristic function, to also approximate regions of probability mass like in \cite{chen2023quantum} but likely with a higher gate complexity.

\subsubsection*{Contribution 2: New Subroutines for Distribution Loading}

A necessary condition for ensuring that a quadratic sample complexity advantage with QMCI can be translated into an end-to-end speedup is an efficient procedure for encoding classical processes into quantum states (i.e. qsamples). To this end, we present (Section~\ref{sec:subroutines_for_distribution_loading}) a refined resource analysis of loading probability distributions for derivative pricing onto a quantum computer. The routines are for loading one-dimensional distributions that form the basic components of the reparameterization approach of \cite{chakrabarti2021threshold}, which we term \emph{primitives}.  Specifically, we showcase the use of the arithmetic-free approach of McArdle et al. \cite{mcardle2022quantum}. While \cite{Wang_2024, prakash2024quantum}, proposed utilizing this framework for Gaussian loading, we extend it to other probability distributions. 

We also develop quantum algorithms for sampling $\chi^2$ distributions, the integral of a CIR process, and two-dimensional L\'evy areas, which are all common primitives in derivative pricing.
These routines are crucial for demonstrating speedups for CIR and Heston. The L\'evy area sampler is based on a procedure proposed by Gaines and Lyons \cite{gaines1994random}, which we  analyze and produce a version that performs qsampling.

\subsubsection*{Contribution 3: Speedups with Quantum MCI beyond Geometric Brownian Motion}

One of our key contributions is demonstrating that quantum does achieve an end-to-end speedup for models beyond the Black-Scholes GBM. First, we perform a resource analysis for quantum derivative pricing over the CIR model (Section~\ref{sec:cir_fast_forwardable}, Main Theorem \ref{thm:cir-main}), showing concretely that quantum derivative pricing provides an end-to-end quadratic speedup for this setting. This pushes the boundary of models that were thought to be tractable to analyze. The major challenge in this analysis is dealing with the recursive nature of the scheme for fast-forwarding CIR.

In addition to CIR, we show that it is still possible to achieve a quadratic speedup for pricing over the Heston model (Section~\ref{sec:heston-fast-forwardable}, Main Theorem \ref{thm:heston-main}) using only the vanilla version of quantum MCI, avoiding L\'evy areas and MLMC. It turns out that a version of the multi-asset Heston model, specifically one with only correlations between the asset processes (not cross-correlations between the variance processes or variance processes being correlated to  multiple equity processes), is fast-forwardable. We show that the Heston model can be efficiently loaded onto a quantum computer and compute the required number of bits for ensuring low discretization error.

The fact that the Heston price process is geometric and coupled to a CIR introduces additional analysis difficulties that are neither present in the resource analysis for CIR nor GBM. Unlike the increments of GBM, the increments of Heston involve random variables that are only subexponential, i.e. tails fall at least as fast as an exponential distribution. This means that the constants and parameters in the process play a significantly role in obtaining useful bounds on the truncation error of the algorithm, and thus must be tracked carefully. Still, we present a non-trivial parameter-regime where truncation is possible. Hence, we present a regime under which the multi-asset Heston model can be efficiently priced end-to-end on a digital (classical or quantum) device, which might be of independent interest.

Additionally, due to singularities in the derivatives of the fast-forwarding scheme for the Heston model, we need to provide lower-endpoint truncation bounds, i.e. remove a neighborhood around the origin. The process that needs to be truncated is the integral over the CIR process. This analysis relies on a hitting-time estimate provided in the proof of the Feller condition for CIR. Fortunately, it turns out  that this truncation does not introduce any additional conditions on the model parameters, but presents an additional technical challenge.

A critical subroutine for loading the Heston model is the, previously-mentioned, algorithm for quantum sampling from  time-integrals of the CIR process. This procedure can be viewed as a qsampling version of the Broadie-Kaya technique \cite{broadie2006exact} for exact Heston simulation. This scheme is classically considered to be too computationally intensive to use in practice \cite{van2010efficient}. It appears that a lot of the classical literature has benchmarked time-discretization schemes (usually Euler-Maruyama) for single-asset Heston models. However, as our goal is to study the asymptotic computational advantages of quantum algorithms for derivative pricing, EM does not seem to be sufficient, as it has not be theoretically proven to provide a quadratic speedup when used with quantum MLMC. As mentioned earlier, for multi-dimensional processes, like even the single-asset Heston model, quantum MLMC seems to need to utilize higher-order schemes to maintain its quadratic speedup over classical MLMC. Furthermore, if multiple Heston asset processes are correlated, then we need to sample from multi-dimensional L\'evy areas, which can remove the quantum speedup \cite{dickinson2007optimal}. 

We prove an end-to-end asymptotic advantage of vanilla quantum MCI when using the Broadie-Kaya scheme, which also works when any pair of assets is correlated. Given the scale of the current quantum devices, it is not tractable to benchmark the various simulation methods that appear to work well in practice yet not in theory. It is possible though that such schemes could work well for quantum MCI as well. One consequence of our analysis is an end-to-end bound on the computational complexity of classical MCI applied to the Heston model.

\subsubsection*{Contribution 4: Speedup with Quantum MLMC for Correlated Processes}
\label{sec:speedupwithcorrelated}

We re-iterate that quantum MLMC has only been theoretically shown to provide a quadratic speedup for globally-Lipschitz payoffs and SDE simulation with the Milstein scheme.
While the authors of \cite{an2022efficient} only show that the Milstein scheme is \emph{sufficient} for a quantum speedup, there is some intuition as to why it is in fact \emph{necessary}. 

MLMC requires sampling paths at varying levels of granularity in discrete time (see Section \ref{sec:approximate_sde_simulation} for more details). For MLMC to work, the variance at each ``level''  needs to fall faster than the the rate at which the cost of simulating the path grows.  The strong convergence of the scheme effectively controls how the variance decays with each level.
The variance appears in the overall complexity due to its appearance in the sampling complexity, which quantum quadratically reduces.  Hence, one needs to use a scheme that converges quadratically faster to compensate. Alternatively, one could reduce the cost of path simulation. However, it seems that SDE approximation schemes are very iterative and hence not amenable to a quantum speedup.

Utilizing our quantum algorithm for two-dimensional L\'evy area sampling, leading to a quantum Milstein sampler, we show that quantum multi-level Monte Carlo can achieve an end-to-end quadratic speedup when the model has correlations between only two-processes at a time (Section~\ref{sec:multi_level_monte_carlo}, Main Theorem \ref{thm:speed_mlmc}), which we term bipartite correlations. This enables the result of \cite{An2021quantumaccelerated} to provide an end-to-end quadratic speedup beyond the previously-shown uncorrelated setting. Hence, before our work, it was unclear whether an end-to-end quantum speedup would be possible in the approximate SDE-simulation regime and in the presence of correlations. 

Also, adding to the works of \cite{giles2008multilevel, An2021quantumaccelerated}, we discuss how the various sources of error that are accounted for in our framework presented in Contribution 1 impact the analysis of (Q)MLMC. We show that these errors can be efficiently suppressed and do not impact the guarantees of MLMC.

\subsubsection*{Contribution 5: Barriers to Sublinear Simulation for Quantum Pricing of Path-dependent Derivatives}

Regardless of how efficiently a single-time point of the SDE can be simulated, if the derivative payoff depends on the SDE at $T$ points (commonly referred to as monitoring points), then we must use $\Omega(T)$ time and space classically. The current frameworks for quantum derivative pricing, effectively attempt to coherently reproduce the classical simulation methods, and hence require $\Omega(T)$ gates and qubits.

We investigate (Section~\ref{sec:quantum_pde_solvers}) whether quantum PDE solvers could be used to reduce the qubit dependence on the number of monitoring points for path-dependent derivatives to be  polylog in $T$. Unfortunately, it turns out that this approach, at least with current techniques, is not compatible with quantum MCI. We demonstrate various no-go results and significant bottleknecks in using quantum PDE solvers for state preparation in derivative pricing, with significant focus on solving the Fokker-Planck (FP) equation. The FP is a natural consideration, since it is the PDE for the marginals of an It\^o SDE.

The approach of Prakash et al. \cite{prakash2024quantum} provided an alternative route to sublinear simulation for  the particular case of GBM. However, in all of cases that they present the quantum speedup in inverse-precision from QMLC is destroyed. Additionally, it seems that most of the barriers that we present also extend to quantum walk techniques for simulating symmetric Markov chains in sublinear time \cite{apers2019quantumfastforwardingmarkovchains, Gily_n_2019}. Hence, it still remains an open question if it is possible to simulate $T$ points of an SDE with resources that are sublinear in $T$ and ensuring compatibility with QMCI.

\subsection{Related Work}
High-dimensional, exotic derivatives, those consisting of many underlyings and/or many monitoring points, are typically priced, classically, using the randomized Monte Carlo integration method \cite{glasserman2004monte}. The standard error after $N$ samples is well-known to be $\mathcal{O}(\sigma/\sqrt{N})$, where $\sigma$ is the standard deviation of the payoff process, i.e. the random variable whose mean we want to estimate. It turns out that there are classical deterministic algorithms based on low-discrepancy sequences, called quasi-Monte Carlo methods \cite{glasserman2004monte}, which achieve an error scaling of $\mathcal{O}\left(\text{poly}(\log^d(N) )\sigma'/N\right)$ where $\sigma'$ measures the variation of the $d$-dimensional integrand \cite{bansal2025quasi}. For reasons not fully understood, it is possible in some high-dimensional financial applications for quasi-Monte Carlo to converge quadratically faster than randomized Monte Carlo in practice \cite{acworth2011comparison, tezuka2005necessity}, i.e. without the dimension scaling dominating. Hence, quasi-Monte Carlo can be practically competitive with quantum MCI in some settings.

Quantum computation enables an asymptotic quadratic reduction in the number of quantum samples required to estimate the mean of a Bernoulli random variable  via amplitude estimation \cite{brassard2002quantum}. By amplitude-encoding truncated sums, this algorithm enables one to speedup mean estimation generically and hence Monte Carlo integration \cite{montanaro2015quantum, hamoudi2021, kothari2022, blanchet2024quadraticspeedupinfinitevariance}. Quantum Monte Carlo integration has an error scaling of $\widetilde{\mathcal{O}}\left(\sigma/N\right)$ in terms of queries to a quantum sampling oracle (qsampler). This \emph{generic} and \emph{provable}, dimension-independent quadratic reduction \emph{in sampling complexity} has obviously attracted a lot of attention from the quantitative finance community. Given that this is only a ``black-box'' advantage, one needs to perform a complete end-to-end resource analysis to see if it translates into a ``white-box'' or end-to-end speedup.

The area of derivative pricing with quantum algorithms started with the works of \cite{woerner2019quantum, rebentrost2018quantum, stamatopoulos2020option}. The focus was on applying quantum MCI to price various path-independent options. This was also shown to be possible with the framework provided by the quantum singular value transform \cite{stamatopoulos2024derivative}. The reparameterization method of \cite{chakrabarti2021threshold} enabled quantum algorithms to efficiently price path-dependent options over geometric Brownian motion, with additional analysis provided by \cite{chen2023quantum}. A quantum version of  multi-level Monte Carlo integration was proposed by \cite{An2021quantumaccelerated}, where it was shown that a sufficient condition for quantum algorithms to achieve a quadratic speedup for globally-Lipschitz payoffs is that one uses SDE simulation schemes that converge at least quadratically faster than the Euler-Maruyama scheme. Wang et al. \cite{Wang_2024} performed resource estimation of the Euler-Maruyama scheme combined with vanilla quantum MCI for pricing derivatives over Heston's stochastic volatility. Extending the work of Bouland et al. \cite{bouland2023quantumspectralmethodsimulating}, Prakash et al. \cite{prakash2024quantum} proposed utilizing Fourier expansions of stochastic processes to reduce the dependence on the number of monitoring points for pricing certain Asian options. 

Along with mean-estimation, quantum computation provides an exponential reduction in dimension for solving a quantum version of the standard (sparse) linear systems problem \cite{harrow2009quantum, costa2022optimal}. This has also led to an interest in speeding up partial-differential-equation (PDE) based pricing methods, which can usually be reduced to linear system solving. Specifically, certain derivative pricing tasks can be specified as the solution to a compact PDE, which is classically challenging to solve in large dimensions. Quantum algorithms can potentially accelerate the task of preparing a state encoding the solution to the PDE \cite{miyamoto2021pricing}, which may enable extracting certain quantities of interest. Additionally, \cite{Kubo_2021,fontanela2021quantum,Alghassi2022variationalquantum, kubo2022pricingmultiassetderivativesvariational} proposed a variational version of this approach. To extract the price, one typically needs to run QMCI anyways. On a similar note, \cite{Gonzalez_Conde_2023} used Hamiltonian simulation to simulate the Black-Scholes model.

Beyond just pricing, quantum algorithms have also been applied to risk estimation \cite{woerner2019quantum, egger2019creditriskanalysisusing}. Additionally, the quantum gradient estimation algorithm \cite{jordan2005fast, gilyen2019optimizing} was applied to computing partial derivatives of the price of a financial derivative, called ``Greeks'' \cite{stamatopoulos2022towards}, which enable hedging.  Cherrat et al. \cite{raj2023quantum} developed a quantum version of the Deep hedging framework of \cite{bühler2018deephedging}.

More on the classical front, there have been various works investigating the potential efficient simulation \cite{glasserman2004monte, broadie2006exact, andersen2007efficient, lord2010comparison, smith2007almost, van2010efficient} and truncation of the Heston model \cite{andersen2007moment, kellerressel2008momentexplosionslongtermbehavior, jacquier2012largedeviationsextendedheston}.  The existing results on truncation focus on the truncation of the log-return for single-asset Heston models. In most practical scenarios, the Heston model is priced via SDE time-discretization schemes \cite{lord2010comparison}. However, Broadie and Kaya \cite{broadie2006exact} provided an exact simulation method for the single-asset Heston model, conditioned on being able to sample from the integral of a CIR process. They left out the complete computational complexity to do so, along with handing errors. In the classical literature, this scheme has been mostly considered to be impractical \cite{van2010efficient} and various alternatives based on discretization schemes inspired by the Broadie-Kaya scheme are used instead \cite{andersen2007efficient, van2010efficient}. 
Furthermore, the characteristic function of the single-asset Heston model is known in closed-form \cite{albrecher2007little}, which enables one to apply the Fast-Fourier Transform pricing method of \cite{carr1999option} in the path-independent, single-asset setting.

The various issues with truncating the Heston model and more generally the class of affine stochastic volatility models, via the lens of moment-explosions, has also been explored extensively \cite{kellerressel2008momentexplosionslongtermbehavior}. Specifically, Ricatti equations for the moments have been used to identify the points of finite-time explosions \cite{andersen2007moment}, where it was revealed that the second moment can be infinite for the Heston model. Jacquier and Mijatovic \cite{jacquier2012largedeviationsextendedheston} analyzed  the large-deviations behavior of the Heston model.

Additionally, the sampling of L\'evy areas has been investigated through multiple lenses \cite{gaines1994random, dickinson2007optimal, foster2020numerical, foster2023convergence}. For example, \cite{gaines1994random} numerically benchmarked a method for sampling two-dimensional L\'evy areas,  \cite{foster2023brownian} utilized stochastic series expansions, such as the Karhunen-Lo\'eve expansion \cite{alexanderian2015brief}, and \cite{jelinčič2023generativemodellinglevyarea} investigated even using deep generative modeling. Sampling iterated stochastic integrals is classically believed to be a computationally challenging problem, with a complexity that is at least on the order of Monte Carlo integration in general \cite{dickinson2007optimal}.

\subsection{Outlook}

In this work, we provide an extended analysis of quantum algorithms for derivative pricing. Specifically, we have highlighted the importance of identifying end-to-end asymptotic speedups, showing that they do not follow trivially from the well-known black-box QMCI speedup. Currently, it seems each end-to-end speedup for derivative pricing needs to be analyzed on a case-by-case basis. However, our framework at least provides a sketch for which components need to be analyzed and which characteristics of the models lend themselves to quantum speedups. Additionally, most of the analysis that we perform: distribution loading cost/error bounding, discretization error bounding, and truncation error bounding are present in typical applications of quantum mean estimation. Hence the techniques that we present extend beyond pricing financial derivatives and may be of broader interest, particularly to general continuous integration problems. 

In the rest of this section, we discuss additional open questions.

\paragraph{On the practical value of quadratic speedups} Recent results in quantum computing~\cite{Babbush_2021} have provided evidence that end-to-end quadratic speedups may be challenging to realize in practice on fault-tolerant quantum hardware due to the large constant overheads of quantum error-correction. Despite this practical consideration, the identification and analysis of robust, end-to-end quadratic speedups for problems of interest remains of fundamental importance in the larger goal of realizing a quantum advantage for this problem. This is because, most candidates for large polynomial speedups that \emph{can} survive error-correction overheads are built by composing new algorithmic routines on top of existing quadratic speedups~\cite{Dalzell_2023, chakrabarti2025generalizedshortpathalgorithms, Schmidhuber_2025, Buhrman_2025}. A complete analysis for the underlying quadraic speedup in the regime of interest is therefore often a pre-requisite to such a faster algorithm.

\paragraph{Applicability of fast-forwardable SDEs} The notion of fast-forwardability seems to be key to obtaining speedups with the vanilla version of QMCI. Unfortunately, this characteristic is usually the result of knowing certain distributions in closed-form, which is unlikely in general. 
Additionally, while the fast-forwarding schemes are asymptotically efficient, they may not be practical. Hence, it is important to retain speedups when using approximate SDE discretization schemes and MLMC as well.

\paragraph{Quantum MLMC beyond 2D-correlations} While our speedup for quantum MLMC applied to processes with bipartite correlations is a significant improvement, the ability to sample only two-dimensional L\'evy areas does not enable an asymptotic quantum speedup for all financially relevant models. For example, our quantum MLMC algorithm does not apply to the multi-asset Heston model in its most general form, i.e. any two equity or variance processes can be correlated. For a single asset,  the Heston model involves a geometrically-evolving price process coupled to a CIR process, which determines the volatility of the price.  With regards to the Heston model specifically, only allowing bipartite correlations  would imply that there is only a speedup with quantum MLMC when no equity processes are correlated, i.e. since each equity process already has one variance process correlated with it. Interestingly still, for Heston, the model that we can fast-forward and accelerate with vanilla QMCI is more general that what is currently possible with quantum MLMC.

\paragraph{Generic quantum speedups for derivative pricing} One unfortunate consequence of the above is that, if it is in fact necessary for quantum MLMC to use the Milstein scheme and all financially-relevant models cannot be fast-forwarded, then there may be no end-to-end  quadratic quantum advantage for Monte-Carlo-based derivative pricing in general, which is in opposition to common belief. This is because it appears to be significantly challenging, and may not even be possible, to efficiently sample from the Milstein scheme (classically or quantumly) in general. This is due to known barriers associated with multi-dimensional L\'evy area sampling, when only using Brownian increments \cite{dickinson2007optimal, foster2023brownian, foster2020numerical}. To emphasize this difficulty, the classical literature on multi-dimensional L\'evy area sampling has investigated even using deep generative models \cite{jelinčič2023generativemodellinglevyarea}. Hence, in contrast to common belief, it is still an open question if quantum algorithms provide a \emph{generic end-to-end} quadratic speedup for the derivative pricing task. 

\paragraph{Faster L\'evy area sampling} Luckily, there appear to be no known unconditional lower bounds against L\'evy area sampling. Specifically, the existing results apply when  the algorithm is only provided access to the Brownian increments. Note that our L\'evy area sampler enables one to sample from certain marginals of the L\'evy area distribution, putting one outside the regime in which the existing lower bounds apply. It seems that new algorithmic techniques, classical or quantum, that take advantage of more than just the Brownian increments would be required for more efficient L\'evy area sampling. 

\paragraph{Opportunities for further resource reduction in QMCI} Additionally, there is a still an opportunity for quantum algorithms to provide asymptotic reductions in the resources for simulation in terms of the number of monitored points $T$ for a path-dependent derivative. Given that $T$ is not generally considered an asymptotic parameter, it would be ideal to achieve this sublinear simulation while preserving the asymptotic QMCI advantage in $\sigma/\epsilon$. Hence, the purpose of sublinear simulation would be to reduce the amount of qubits and/or gates for arithmetic, which are generally consider precious resources. As mention in Contribution 5, it seems  the most reasonable candidates: quantum PDEs solvers and quantum walks are not sufficient, at least with current techniques. It appears that new quantum algorithmic techniques will be required for sublinear SDE simulation.

\subsection{Organization}
Here we present the organization of 
 the rest of the main text. Section~\ref{sec:background_derivative_pric} discusses background on stochastic differential equations and derivative pricing. Additionally, it introduces most of the notation that we will use. Section~\ref{sec:algo_for_derive} reviews quantum algorithms for derivative pricing. Section~\ref{sec:framework_for_analyze} presents our new framework for analyzing quantum derivative pricing algorithms. Section~\ref{sec:subroutines_for_distribution_loading} presents the new routines for distribution loading. Section~\ref{sec:pricing_deriv_fast_forwardable} discusses the analysis of the fast-forwardable versions of CIR and Heston. Section~\ref{sec:multi_level_monte_carlo}  combines our quantum L\'evy area sampler with  quantum MLMC to achieve a speedup in the case of approximate SDE simulation in the presence of correlations. Lastly, Section~\ref{sec:quantum_pde_solvers} shows how to perform state preparation with quantum PDE solvers and discusses the various limitations.
All proofs not present in the main text are left to the appendices.

\section{Background : Derivatives and Stochastic Models}

\label{sec:background_derivative_pric}

Due to computational considerations, financial models for derivative pricing are usually Markovian. The Markov Chains underlying financial models are typically continuous in time. These models are $d$-dimensional stochastic processes $\vec{X} : \mathbb{R}_+ \rightarrow \mathbb{R}^d$ that can always be expressed as an It\^o Stochastic Differential Equation (SDE) \cite{Brunick_2013}:
\begin{equation} \label{eqn:ito-process}
    d \vec{X}(t) = \vec{\mu}(\vec{X}(t), t) dt + \boldsymbol{\sigma} (\vec{X}(t), t)d \vec{W}(t),
\end{equation}
where $\vec{W}(t) : \mathbb{R}_+ \rightarrow \mathbb{R}^{d'}$ is a $d'$-dimensional Wiener process with correlation matrix $\mathbf{C} \in \mathbb{R}^{d' \times d'}$ defined by the relation $dW^{(j)}(t) \cdot dW^{(k)}(t) = C_{jk}dt$, $j, k \in \{1, \dots, d'\}$. Note that for two infinitesimal Brownian increments $dW^{(1)}(t), dW^{(2)}(t)$, we use $dW^{(1)}(t) \cdot  dW^{(2)}(t)$ to denote their quadratic variation. The term $\vec{\mu}(\vec{X}(t), t) \in \mathbb{R}^d$ is called the \emph{drift} and $\boldsymbol{\sigma}(\vec{X}(t), t) \in \mathbb{R}^d \times \mathbb{R}^{d'}$ is called the \emph{volatility}.
In the discretely-monitored setting, which is what we consider, the price of the derivative depends on the process $\vec{X}(t)$ at only a discrete set of time points, called \emph{monitoring points}. Thus for our purposes, we can restrict to the following induced $T$-length, discrete-time process $\widehat{X} : [T] \rightarrow \mathbb{R}^d$, which we call a \emph{path process}:
\begin{align}
   \widehat{X} := (\vec{X}({0}), \vec{X}({\Delta}), \cdots, \vec{X}((T-1)\Delta)), 
\end{align}
where $\vec{X}(0)$ is a deterministic initial condition and $[T] := \{0, \dots, T- 1\}$. We call a realization of the stochastic process $\widehat{X}$ a \emph{path}. For simplicity, we consider uniformly-spaced time points, controlled by $\Delta$, where the value of $\Delta$ will be clear from context when necessary and is usually part of the derivative contract specifications. In reality, each monitored point $\widehat{X}(t)$ will be an $\epsilon$ approximation of the continuous-time process at that time point. In a later section, we will discuss the efficiency of approximating a path process.

For $t \in [T]$, the \emph{path increment} is a random variable $\vec{I}(t) \in \mathbb{R}^{d''}$ that for a, potentially time-dependent, deterministic \emph{transition function} $g_t :  \mathbb{R}^d \times \mathbb{R}^{d''}  \rightarrow \mathbb{R}^d$  moves a path process forward in time : \begin{align*}
    \xpath(t) = g_t(\xpath(t-1), \vec{I}(t-1)).
\end{align*} 
We do allow for $\vec{I}(t-1)$ to depend on $\xpath({t-1})$. The SDEs of the form Equation \eqref{eqn:ito-process} are continuous-time Markov processes, and thus the path process is a discrete-time Markov process. Also, most processes we consider are time-homogeneous, hence $g_t$ is fixed in time. Thus in most cases we consider SDEs of the form
\begin{equation} 
\label{eqn:ito-process-time-homog}
    d \vec{X}(t) = \vec{\mu}(\vec{X}(t)) dt + \boldsymbol{\sigma} (\vec{X}(t))d \vec{W}(t).
\end{equation}

In most financial use cases, and particularly throughout this work, the components of $\vec{X}(t)$ can be categorized into a price component, denoted $\vec{S}(t)$, and a volatility component, denoted $\vec{V}(t)$. This captures the general class of assets following stochastic-volatility models. Specifically, we will consider the following models, where we start by stating the one-dimensional or single-asset versions.

\begin{restatable}[Geometric Brownian motion model]{defin}{singleGBM}
The price of an asset $S(t)$ at time $t \in \mathbb{R}_+$ follows a geometric Brownian motion (GBM) with constant volatility $\sigma > 0$ and drift $\mu > 0$ if
\begin{align}
dS(t) = S(t)\mu dt + \sigma S(t) dW(t).
\end{align}  
\end{restatable}

\begin{restatable}[Cox--Ingersoll--Ross model]{defin}{CIR}
\label{defn:cir-model}
We say that $V(t)$ follows a Cox-Ingersoll-Ross (CIR) process with reverting-mean $\theta > 0$, mean-reversion rate $\kappa$ and volatility $\sigma > 0$ if
\begin{align}
dV(t) = \kappa(\theta- V(t)) dt+ \sigma\sqrt{V(t)}dW(t).
\end{align}  
\end{restatable}

\begin{restatable}[Single-Asset Heston model]{defin}{hestmodel}
\label{defn:heston-single}
The price of an asset $S(t)$ follows the Heston model with volatility-process $V(t)$ if
\begin{align}
&dS(t) = S(t)\mu dt + S(t)\sqrt{V(t)} dW^{(S)}(t)\\
&dV(t) =  \kappa(\theta - V(t))dt + \sigma \sqrt{V(t)}dW^{(V)}(t)\\
&dW^{(S)}(t) \cdot dW^{(V)}(t) = \rho dt \nonumber
\end{align}
where $\mu > 0$ is the drift, $\sigma > 0$ is the volatility of volatility (vol-of-vol), $\theta > 0$ is the reverting-mean of the volatility and $\kappa$ the rate of mean reversion.  Lastly, $\rho \in [-1, 1]$ denotes the correlation between the processes.
\end{restatable}

One can observe that the Heston model is effectively a combination of the GBM and CIR models, and is thus at least as hard to analyze as its component processes. The GBM and Heston models have natural generalizations to the case of multiple, correlated assets. For reasons that will come back to later, it is not clear how to efficiently (i.e. in $\mathcal{O}\left(\text{poly}(T, \Delta, d, \log(1/\epsilon))\right)$ time) sample from a correlated CIR model. The following are natural multi-asset generalizations of the GBM and Heston models

\begin{restatable}[multi-asset geometric Brownian motion model]{defin}{multiassetgbm}
\label{defn:multiassetgbm}
The prices of $d$ assets $\vec{S}(t$) follow a multi-asset geometric Brownian motion (GBM) with constant volatility $\vec{\sigma} \in \mathbb{R}_{+}^{d}$ and drift $\vec{\mu} \in \mathbb{R}_{+}^{d}$ if
\begin{align}
&d \vec{X}(t) = (\vec{\mu} \circ \vec{X}(t)) dt + (\vec{\sigma} \circ  \vec{X}(t)) d \vec{W}(t)\\
& dW_{t, j} \cdot dW_{t, k} = C_{j,k} dt \nonumber, 
\end{align}  
where $\mathbf{C} \in \mathbb{R}^{d\times d}$ is a correlation matrix and $\circ$ denotes the Hadamard product.
\end{restatable}

We will consider a specific, yet realistic, version of the Heston model that only contains correlations between asset processes.

\begin{restatable}[multi-asset Heston model (asset-asset correlations)]{defin}{multiassethest}
\label{defn:multi-asset-hest-asset-only}
The prices of $d$ assets $\vec{S}(t)$ follow a multi-asset Heston model with stochastic volatility $\vec{V}(t)$ and asset-asset correlations only, if
\begin{align}
\label{eqn:heston-increment}
&\frac{d\vec{S}(t)}{\vec{S}_t} = \vec{\mu}dt + \frac{\vec{\rho}}{\vec{\sigma}}\circ(d\vec{V}(t)  - \vec{\kappa}\circ\vec{\theta} dt) + \left(\frac{\vec{\kappa}\circ\vec{\rho}}{\vec{\sigma}} - \frac{1}{2}\cdot\vec{1}\right)\circ \vec{V}(t) + \sqrt{1-\vec{\rho}\circ\vec{\rho}} \circ \sqrt{\vec{V}(t)}\circ d\vec{W}^{(S)}(t) \nonumber\\
&d\vec{V}(t) =  \vec{\kappa}\circ(\vec{\theta} - \vec{V}(t))dt + (\vec{\sigma}\circ \sqrt{\vec{V}(t)})\circ d\vec{W}^{(V)}(t)\\
&dW^{(S)}_{j} \cdot dW^{(S)}_{k} = C_{j,k} dt \nonumber
\end{align}
 where $\mathbf{C} \in \mathbb{R}^{d\times d}$ is the asset-asset correlation matrix.
Also, $\vec{\rho} \in \mathbb{R}^{d}$ is the vector of correlations between an asset process and its corresponding volatility process. Note that the volatility processes are decoupled between assets. The vector $\vec{1}$ is the all-ones vector.
\end{restatable}
The presentation above may appear somewhat non-standard. However, we choose to display the Heston model this way to later emphasize the fast-forwardability and efficient simulation.

As mentioned in the introduction, a derivative is a financial contract that pays a value dependent on some other financial source, called the underlying \cite{follmer2011stochastic}. We present the following mathematical model of a derivative that we use throughout the paper. For a path process $\widehat{X}$ we will denote $\widehat{X}_{:t}$ all points up to and including the $t$-th random variable.

\begin{definition}[Financial Derivative Model]
\label{def:fin_deriv_model}
Consider a sequence of $T-1$ functions $f_i : \mathbb{R}^{i+1} \rightarrow \mathbb{R}$. A financial derivative model with maturity $T$ is a process $(f_1(\widehat{X}_{:1}), f_2(\widehat{X}_{:2}), \dots, f_{T-1}(\widehat{X}))$ formed from a tuple $(\widehat{X}, f_1, \dots, f_{T-1})$, where $\widehat{X}$ is a $T$-length path process of some SDE of the form \eqref{eqn:ito-process}.
\end{definition}

We will call 
$(f_1(\widehat{X}_{:1}), f_2(\widehat{X}_{:2}), \dots, f_{T-1}(\widehat{X}))$ the \emph{payoff process}. Note that we use $f_{i}(X_{:i})$ to denote the \emph{cumulative cash flows}, i.e. sum of payoffs up to time $i$. There is typically a difference between these values and the payoff, but for simplicity we do not make a distinction. The fair price at the current time point (indicated by $t=0$) is then the expectation of the (discounted) cumulative sum of all future cash flows until maturity, i.e. $\mathbb{E}[f_{T-1}(\widehat{X}) | \widehat{X}(0)= \vec{x}_0]$ given some initial condition $\vec{x}_0$. Hence we will simply drop the subscript in $f$ as we will solely focus on the cumulative payoff at maturity and sometimes simply call the random variable $f(\widehat{X})$ the payoff process.

As an example, a European option with strike $K$ has $f_i = 0$ for all $i < T-1$ and $f_{T-1} = \max(\widehat{X}(T-1) - K, 0)$. A derivative is called \emph{path-dependent} if $f$ is a function of more than just $\widehat{X}(T-1)$.
The goal of derivative pricing is to determine the fair price of a financial derivative model. For simplicity, we consider the following definition of derivative price:

\begin{definition}[Derivative Price]
\label{def:deriv_price}
The  price of a derivative model with deterministic initial condition $\vec{x}_0$ is 
\begin{align}\mathbb{E}[f(\widehat{X}) | \widehat{X}(0) = \vec{x}_0] = \int_{(\mathbb{R}^{d})^{\times T}} f(\widehat{x})p(\widehat{x})d\widehat{x},
\end{align}
where $p$ is the joint density of the path process.
\end{definition}
One technicality, is that the underlying path process $\widehat{X}$ needs to have its drift adjusted so that it becomes a martingale, and the payoff includes a discount factor. Since these are not computationally intensive operations, for simplicity, we do not include these in our definition.

A classical approach for computing the price is \emph{Monte Carlo integration} (MCI), where one samples from the measure $\mu(\Omega) = \int_{\Omega}p(\widehat{x})d\widehat{x}$ and computes the empirical average of $ f(\widehat{X})$.

We will typically assume that the payoff $f$ is a piecewise linear function, which is satisfied by common derivatives. For example, an \emph{Asian option} with underlying assets $\widehat{X}$ and strike price $K$ has the payoff:
\begin{align*}
    f(\widehat{X}) = \max\left(\frac{1}{T}\sum_{t=0}^{T-1} \widehat{X}(t) - K, 0\right),
\end{align*}
which is of course piecewise linear. 

The term \emph{exotic} is typically used to refer to derivatives that are path-dependent and/or depend on non-standard underlyings. This is in contrast to \emph{vanillas}, which are frequently traded derivatives, e.g. European options, that are path-independent and depend on simple underlyings, like equities.

\section{Background : Quantum Algorithms For Derivative Pricing}

\label{sec:algo_for_derive}

In the following subsections we review the details of quantum derivative pricing (Section~\ref{sec:quantum_derivative_pricing_intro}), approximate SDE simulation (Section~\ref{sec:approximate_sde_simulation}), which will be relevant for Section~\ref{sec:multi_level_monte_carlo}, and the quantum eigenvalue transformation state preparation technique (Section~\ref{eqn:qet_state_prep_prelims}), relevant for Section~\ref{sec:subroutines_for_distribution_loading}.

\subsection{Quantum Derivative Pricing}
\label{sec:quantum_derivative_pricing_intro}

The widely-used classical Monte Carlo integration approach to derivative pricing can estimate the price to additive error $\epsilon$ using at most $\mathcal{O}\left(\frac{\text{Var}(f(\widehat{X}))}{\epsilon^2}\right)$ samples of  $f(\widehat{X})$  \cite{glasserman2004monte}. However, quantum MCI reduces the number of quantum samples (Definition~\ref{defn:qsample-exact}) to $\widetilde{\mathcal{O}}\left(\frac{\sqrt{\text{Var}(f(\widehat{X}))}}{\epsilon}\right)$ \cite{montanaro2015quantum}, where the polylog factors can be removed in the white-box setting \cite{kothari2022}. One can also handle the infinite variance setting by using alternative truncation bounds \cite{blanchet2024quadraticspeedupinfinitevariance}. In this section, we review the components of the quantum algorithm for achieving this black-box quadratic speedup, along with state-preparation procedures that are useful in the white-box case.

\subsubsection{Approximating Continuous Price by a Discrete Sum}

In a digital implementation, we cannot actually sample from a continuous distribution and instead algorithms implicitly sample from a discrete distribution over an underlying grid in a manner such that the corresponding averages closely approximate that over the continuous distribution. In a classical algorithm, this detail is usually ignored due to the very high bit precision available. However, since quantum accessible bits are a valuable resource for early quantum computers, we will have to carefully estimate the desired number of bits needed to obtain an acceptable error. In addition, the domain of the continuous probability distribution must be truncated to agree with the convex hull of the finite discrete grid used. In a classical implementation, we often do not need to truncate explicitly in advance and an adaptive sampling algorithm can be used that in principle allows the diameter of the set of sampled points to grow indefinitely. In the quantum setting however, 
we identify the finite quadrature grid with the basis of a Hilbert space, and encode the corresponding probabilities in the amplitudes of a quantum state in that space. The number of quantum bits required is logarithmic in the total number of grid points. For this reason, a quantum algorithm must specify the domain of truncation of $\widehat{X}$ in advance and account for the error arising from this. Suppose this domain of truncation is $[-R, R]^{dT}$, i.e. consider the process formed by $\widehat{X}$ multiplied component-wise by the indicator of this hypercube. Suppose we then choose uniformly-spaced grid points $\mathcal{M} \subset [-R, R]^{dT}$, which then considers only a finite number of values for $\widehat{X}\circ \mathbb{1}_{[-R, R]^d}$. Then the final estimate we (ideally) encode in a quantum state is given by
\begin{align}
\label{eqn:discretized-sum}
    \sum_{\vec{x} \in \mathcal{M}} f(\vec{x})p(\vec{x})\frac{(2R)^{Td}}{\lvert \mathcal{M} \rvert},
\end{align}
which is then approximated by quantum amplitude estimation. The truncation error is analyzed using the tails of the continuous density $p$, and the chosen quadrature rule allows for the computation of discretization error. The truncation is roughly on the order of the standard deviation $\sigma$ of $\widehat{X}$, and the quantum algorithm can estimate the above sum bounded by $B$ to error $\epsilon$ using $\widetilde{\mathcal{O}}(B/\epsilon)$ quantum samples, which roughly translates into $\widetilde{\mathcal{O}}(\sigma/\epsilon)$. This is described in more detail in a later section.

\subsubsection{Estimating Discretized Sums by Quantum Monte Carlo Integration}

The primary algorithmic technique in quantum algorithms for Monte Carlo is Quantum Amplitude Estimation (QAE)~\cite{brassard2002quantum}, (Theorem~\ref{thm:quantum-counting}). %
The input to QAE is a unitary $U_p$ such that $U_p\ket{0} = \sqrt{\tilde{p}}\ket{0} + \sqrt{1 - \tilde{p}}\ket{1}$, where $\tilde{p} \in [0, 1]$ is a scaled (by a factor $P$) version of the value we want to estimate $p$. For derivative pricing, $p$ is an approximation to the sum in Equation \eqref{eqn:discretized-sum}. The value $P$ also depends on the region of truncation, which should roughly correspond to the standard deviation. The quantum amplitude estimation subroutine then allows $\tilde{p}$ to be estimated to additive error $\epsilon$ (consequently $p$ to error $P\epsilon$) with probability at least $1 - \delta$ using $O\left( \frac{1}{\epsilon}\log\left( \frac{1}{\delta} \right) \right)$ applications of $U_p$. In general, $U_p$ can only be implemented to some additive error, and may itself be implemented with some failure probability, with the use of some ancillary registers. To capture this, we use the notion of an \emph{amplitude encoder}.

\begin{definition}[Amplitude Encoder]
A unitary $U_p$ acting on $k$ qubits is called an $(k,\epsilon,P,p_0)$ amplitude encoder if
\begin{align}
    U_p|0^k\rangle = \sqrt{\widetilde{p}}\ket{0^m}\ket{\phi_0} + \sqrt{1 - \widetilde{p}}\ket{\bot}\ket{\phi_1}
\end{align}
where $\braket*{0^m}{\bot} = 0$,  $\widetilde{p} \in [0,1]$, and there exist a known constant $P$ such that $\lvert \frac{p}{P} - \widetilde{p} \rvert \le \epsilon$.
\end{definition}
Amplitude Estimation provides the following guarantee.%
\begin{theorem}[Quantum Amplitude Estimation~\cite{brassard2002quantum}]
    \label{thm:quantum-counting}
    Let $U_p$ be a $(k,\epsilon',P,p_0)$ amplitude encoder for $p$. There is a quantum algorithm that uses $\widetilde{\mathcal{O}}\left(\frac{1}{\epsilon}\log\left(\frac{1}{\delta}\right)\right)$ applications of $U_p$, $\widetilde{\mathcal{O}}\left(\frac{1}{\epsilon}\log\left(\frac{1}{\delta}\right)\right)$ extra gates, and $\widetilde{\mathcal{O}}\left(\log\left(\epsilon^{-1}\delta^{-1}\right)\right)$ additional ancillary registers, and returns an estimate $\tilde{p}$ so that with probability at least $1 - \delta$, $|p - \tilde{p}| \le P(\epsilon + \epsilon')$.
\end{theorem}
Hence an amplitude encoder naturally loads a Bernoulli distribution, i.e. measure the first $m$-qubit register and the two outcomes considered are  $\{0^m \}$ and $\{ s \in \{0, 1\}^m/\{0^m\}\}$, and QAE can be used to obtain an estimate of the mean. When considering input distributions that are not restricted to being Bernoulli, the algorithm is referred to as quantum Monte Carlo integration (QMCI). The reduction to QAE comes from properly truncating and converting  the input unitaries into an amplitude encoder. One procedure for preparing an appropriate amplitude encoder for QMCI was given by Montanaro~\cite{montanaro2015quantum}. The procedure assumes access to a unitary $U_{\mathrm{path}}$ preparing a \emph{qsample} from the process whose mean we want to estimate. Namely, we assume access to a unitary such that for some $k \in \mathbb{N}$,
\begin{align}
\label{eqn:path_unitary}
    U_{\mathrm{path}}|0^k\rangle = \sum_{\vec{x} \in \mathcal{M}} \sqrt{\widetilde{\mathbb{P}}[\xpath = \vec{x}]} \ket{\vec{x}},
\end{align}
where $\vec{x}$ ranges over the support of the distribution and $\ket{\vec{x}}$ corresponds to a binary encoding of $\vec{x}$ in a computational basis state. For our setting, one can consider the measure here $\widetilde{\mathbb{P}}$ to be an approximation to the discretized and renormalized path distribution restricted to the grid $\mathcal{M}$, denoted $\mathbb{P}$. Specifically, suppose we have total variation distance of at most, say, $\frac{\epsilon}{2}$. An exact encoding of $\mathbb{P}$ is  called a discrete qsample (Definition~\ref{defn:qsample-exact})
 
An appropriate amplitude encoder can then be obtained by composing $U_{\textup{path}}$ with a \emph{payoff oracle} $U_f$, such that
\begin{align}
    U_{f}\ket{\vec{x}}\ket{0} = \ket{\vec{x}}\left( \sqrt{\widetilde{f}(\vec{x})} \ket{0} + \sqrt{1 - \widetilde{f}(\vec{x})} \ket{1} \right),
\end{align}
where $\widetilde{f} \in [0,1]$ and $\sup_{\vec{x} \in \mathcal{M}} |f(\vec{x})/M - \widetilde{f}(\vec{x})| \le \frac{\epsilon}{2}$, with $M$ some constant greater than $0$. The quantity we seek to estimate (Equation \eqref{eqn:discretized-sum}) using QMCI is an expectation of the form $\sum_{\vec{x} \in \mathcal{M}} {\mathbb{P}[\xpath = \vec{x}] f(\vec{x})}$. It is evident that applying $U_{\mathrm{path}}$ followed by $U_{f}$, will yield a $(k,\epsilon,M,f_0)$ amplitude encoding for $p = \sum_{\vec{x} \in \mathcal{M}} {\mathbb{P}[\xpath = \vec{x}] f(\vec{x})}$, and hence can be passed to the QAE subroutine. Given that we only one require one query to $U_{\mathrm{path}}$ and $ U_{f}$ to construct the amplitude encoder, the complexity in terms of queries to these oracles is given by distributional properties and the complexity of QAE.

Note that the estimate we obtain for the derivative price using QAE, unlike with classical MCI, is biased. However, we can arbitrary suppress the bias \cite{Cornelissen_2023}. Specifically, the current guarantee is that we obtain an $\epsilon$ estimate of the price with high probability.

\subsubsection{Preparing Discrete Sums for Quantum Derivative Pricing}
\label{eqn:discrete-sum-prep}

The end-to-end resources required for QMCI can thus be reduced to an analysis of the resources required to implement $U_{\mathrm{path}}$ and $U_{f}$ for the appropriate applications. Given the guarantees of QAE, we would like these procedures to have at most polylog dependence on the inverse error. One can construct  $U_{\mathrm{path}}$ by coherently simulating a classical random walk. Specifically, suppose that we want $U_{\mathrm{path}}$ to prepare a quantum state whose measurement distribution in the computational basis matches the probability distribution of a truncation and discretization of the path process $\xpath$. Suppose we use $m$-bits to encode a component of the $d$-dimensional $\xpath(t-1)$ and $m''$ bits to encode a component of the $d''$-dimensional increment $\vec{I}_t$. We will need access to unitaries $U_{\mathrm{incr}_t}$  and  $U_{\mathrm{jump}_t}$ that perform 
\begin{align}
\label{eqn:randomwalk_operators_incr}
&U_{\mathrm{incr}_t}|\vec{x}_{:t-1}\rangle|\vec{i}_{:t-1}\rangle|0^{md(T-t+1)}\rangle|0^{m''d''(T-t+1)}\rangle \nonumber \\&= \sum_{\vec{i}_t \in \mathcal{A}} \sqrt{\mathbb{Q}_t(\vec{I}_{t} = \vec{i}_t ~|~\xpath(t-1)=\vec{x}_{t-1})}|\vec{x}_{:t-1}\rangle|\vec{i}_{:t}\rangle |0^{md(T-t+1)}\rangle|0^{m''d''(T-t)}\rangle
\end{align}
\begin{align}
\label{eqn:randomwalk_operators_jump}
&U_{\mathrm{jump}_t}|\vec{x}_{:t-1}\rangle|\vec{i}_{:t}\rangle|0^{md(T-t+1)}\rangle|0^{m''d''(T-t)}\rangle = |\vec{x}_{:t-1}\rangle|g_t(\vec{x}_{t-1},\vec{i}_t)\rangle|\vec{i}_{:t}\rangle|0^{md(T-t)}\rangle|0^{m''d''(T-t)}\rangle,
\end{align}
where we are using the notation introduced in Section~\ref{sec:background_derivative_pric}, and $\mathbb{Q}_t$ denotes the (truncated and discretized) measure of $\vec{I}_t$ with support $\mathcal{A}$. Also, we can view $\vec{x} \in \mathbb{R}^{Td}$ as a linearization of a realization of the path process, i.e. $\vec{x}$ is a linearization of $(\vec{x}_0, \dots, \vec{x}_{T-1}), \vec{x}_t \in \mathbb{R}^d$. Then $\vec{x}_{:t-1} \in \mathbb{R}^{(t-1)d}$ encodes all vectors up to and including the $(t-1)$-th. We can consider $\vec{i} \in \mathbb{R}^{(T-1)d'}$ similarly.

In most cases $\vec{I}_t$, via a coordinate transformation, can be expressed as product of $d$, simple one-dimensional distributions.  The implementation of $g_t$ introduces an additional arithmetic cost. We let an alternating sequence of the above unitaries form $U_{\mathrm{path}}$ producing the state in \eqref{eqn:path_unitary}. In $\eqref{eqn:path_unitary}$, $\vec{x}$, by unitarity, encodes the entire path history including the increments.

The quantum loading of a continuous distribution via truncation and discretization can be expressed by what we call a discrete qsample:
\begin{definition}[Discrete Qsample]
\label{defn:qsample-exact}
Let $p \colon (\mathbb{R}^d)^{\times T}  \to \R^{+}$ be the density of some probability distribution. A discrete qsample of $p$ over the rectangular region $ \mathcal{S} = \bigtimes_{i=1}^{d} [a_i,b_i]$ is the following $n \times d$-qubit quantum state:%
\begin{align}
    \frac{1}{\sqrt{\int_{\CS} p(x)\,dx}} \sum_{j_1 = 0}^{2^n-1}\dots\sum_{j_d = 0}^{2^n-1} \sqrt{I(j_1,\dots,j_d)}|j_1,\dots,j_d\rangle,
\end{align}
where $I(j_1,\dots,j_d) = \int_{a_1 + j_1\frac{b_1-a_1}{2^n}}^{a_1 + (j_1+1)\frac{b_1-a_1}{2^n}} \dots \int_{a_d + j_d\frac{b_d-a_d}{2^n}}^{a_d + (j_d+1)\frac{b_d-a_d}{2^n}} p(x) \,dx$. 
\end{definition}

Note that after applying arithmetic operations, the state no longer encodes a superposition over equally spaced values, and so is not strictly a Qsample in the sense of Definition~\ref{defn:qsample-exact}. Instead, the arithmetic operations can be understood as composing an additional operation with the original payoff to be integrated, such that the integral of this new payoff over simple qsamples (prepared via the $U_{\mathrm{jump}_t}$) matches that of the original payoff over the path distribution. In the continuous case, there is no distinction between these settings. However, in the quantum setting where an explicit discrete state is maintained, the division of labor between the generation of qsamples and arithmetic operations affects both the error analysis as well as the resource requirements of the algorithm.

The payoff oracle $U_{f}$ can be accomplished by using coherent arithmetic (see Section~\ref{sec:coherent-arith} for a review) to evaluate the payoff function $f$ and perform $|\vec{x}\rangle|0\rangle \rightarrow |\vec{x}\rangle|f(\vec{x})\rangle$. The amplitude encoding of $f(x)$ can be performed by computing $\arcsin(\sqrt{\widehat{f}(\vec{x})})$ to $m$-bits into an ancillary register, where $\widehat{f}(\vec{x})$ is a scaled version of $f(\vec{x})$ to fit into $[0, 1]$. Then, we apply a bank of  $m$ singly-controlled $R_{\mathsf{Y}}$ rotations, one for each bit. In our analysis we will assume that $f$ is  piecewise linear.

\subsubsection{Quantum Derivative Pricing with Approximate SDE Simulation}
\label{sec:approximate_sde_simulation}

For general SDEs, the distribution of the path increment and the transition function may not be known in closed-form. Hence, one has to resort to time-discretization schemes that only approximately track the process and typically have a complexity that is inverse-polynomial in the desired error.  Specifically, these approximation schemes must be applied when the continuous-versions of the $\mathbb{Q}_t$ from \eqref{eqn:randomwalk_operators_incr}, for exactly sampling the increments of the process,  are not known in closed form. In this case, we replace the increments of the true process with those used by the approximation scheme.   However, as mentioned in the introduction, the use of approximation schemes can significantly increase the asymptotic complexity of classical and quantum Monte Carlo integration.

In this section, we review the time-discretization methods, specifically It\^o-Taylor schemes, for approximate SDE simulation along with the quantum-version of the multi-level MCI algorithm for retaining a $\widetilde{\mathcal{O}}(1/\epsilon)$ complexity in the presence of time discretization. The way we implement these schemes in quantum is via coherent encodings, like in Section~\ref{eqn:discrete-sum-prep}.

Following \cite{glasserman2004monte, platen1999introduction}, we recall well-known schemes for performing time discretization of SDEs. We can express Equation \eqref{eqn:ito-process-time-homog} component-wise in the following way:
\begin{align*}
&dX_i = \mu_i(\vec{X}) dt + \sum_{j=1}^{m}\sigma_{ij}(\vec{X})dW_j.
\end{align*}

The approximation scheme produces a path process $\widetilde{X}$ that approximates the true path process $\widehat{X}$ in some sense. The scheme is usually parameterized by a scale $h$, which corresponds to the step size and hence relates to the error of the scheme. 

Suppose we want to simulate a $T$-length path process $\widehat{X}$, where  $\widehat{X}(t)$ and $\widehat{X}(t+1)$ differ in time by $\Delta$. We further divide $T\Delta$ into steps of size $h$ and simulate points in between $\widehat{X}(t)$ and $\widehat{X}(t+1)$. The approximate path process $\widetilde{X}$ will still be a $T$-length discrete-time process, retaining only the output of the approximation scheme at points that are integer multiples of $\Delta$. Hence, when denoting $\widetilde{X}$ we discard all intermediate points. Of course, quantumly, due to the need to retain reversability, we must also continue to encode the intermediate points in the qsample. However, the intermediate points can be pushed into ancillas.

We say that a scheme has a \emph{weak convergence} of order $r$ if for any polynomial $f : \mathbb{R}^d \rightarrow \mathbb{R}$
\begin{align*}
    \sup_{0 \leq k < T}\lvert\mathbb{E}[f(\widehat{X}(k))] - \mathbb{E}[f(\widetilde{X}(k))]\rvert = \mathcal{O}\left(h^{r}\right).
\end{align*}
We say that a scheme has a
\emph{strong convergence} of order $r$ if 
\begin{align*}
    \sup_{0 \leq k < T}\mathbb{E}\lVert \widehat{X}(k) - \widetilde{X}(k)\rVert_{2}^2 = \mathcal{O}\left( h^{2r}\right).
\end{align*}
Clearly Jensen's inequality leads to that strong convergence implies weak convergence of at least the same order. Hence, it suffices to look at the strong convergence order only. We will want $h = \mathcal{O}(\epsilon^{1/r})$ to ensure an error of $\mathcal{O}(\epsilon)$, which results in taking $\mathcal{O}\left(\text{poly}(1/\epsilon)\right)$ steps. This introduces a multiplicative factor, when estimating the (Q)MCI complexity in terms of ``basic'' samples used to simulate the process.

Similar to Taylor approximations, there is a hierarchy of discretization methods. The lowest-order discretization method is the \emph{Euler-Maruyama} scheme defined as 
\begin{align}
\label{eqn:EM-scheme}
    \widetilde{X}_i(t+1) &= \widetilde{X}_i(t) + \mu_i(\widetilde{X}(t))h + \sum_{k=1}^{m}\sigma_{ik}(\widetilde{X}(t))\Delta W_k,
\end{align}
and has a strong convergence of $\frac{1}{2}$. The next method in the hierarchy is the \emph{Milstein} scheme

\begin{align*}
 \widetilde{X}_i(t+h) &= \widetilde{X}_i(t) + \mu_i(\widetilde{X}(t))h + \sum_{k=1}^{m}\sigma_{ik}(\widetilde{X}(t))\Delta W_k + \frac{1}{2}\sum_{k, j=1}^m \mathcal{L}^j\sigma_{ik}(\widetilde{X}(t))(\Delta W_{j}\Delta W_{k}  + A_{(j,k)}),
\end{align*}
where \begin{align*}
    \mathcal{L}^{k} := \sum_{i=1}^{d}\sigma_{ik}\frac{\partial}{\partial x_i}, k= 1, \dots, m,
\end{align*}
and $A_{(j,k)}$ is defined as
\begin{align*}
   A_{(j, k)} := \int_{t}^{t+h}\int_{t}^{u} dW_j(s) dW_k(u)  - \int_{t}^{t+h}\int_{t}^{u} dW_k(s) dW_j(u).
\end{align*}
The terms  $A_{(j,k)}$ are referred to as  \emph{L\'evy areas}.
The Milstein scheme has a strong convergence of order $1$. We say that an SDE satisfies the \emph{commutativity condition} \cite{glasserman2004monte} if for all $i, j, k \leq m$: 
\begin{align}
\label{eqn:commutivity}
    \mathcal{L}^{k}\sigma_{ij} = \mathcal{L}^{j}\sigma_{ik}.
\end{align}
If an SDE satisfies the commutativity condition for a triple $(i, j, k)$, then the asymmetry of L\'evy areas implies that the coefficient of $A_{(j,k)}, j < k$ in the approximation is zero for the $i$-component. Hence, we do not need to sample the corresponding L\'evy area.

Giles proposed multi-level Monte Carlo (MLMC) \cite{giles2008multilevel}, which is able to retain an overall sampling complexity of $\mathcal{O}(1/\epsilon^2)$ when using discretization schemes. The ability for MLMC to retain the usual convergence of MCI depends on the convergence properties of the approximate payoff process $f(\widehat{{X}})$ and not just those of $\widehat{{X}}$. In the general case, the orders of convergence are estimated by empirical investigation. However, as shown by \cite{giles2008multilevel, an2022efficient}, if the payoff $f$ is globally Lipchitz, then the convergence properties of $f$ only depend on those of $\widehat{{X}}$. Just to be precise, we say that a function $f:\mathbb{R}^d \rightarrow \mathbb{R}$ is \emph{globally Lipschitz} if there exists a constant $L$ such that
\begin{align*}
    \lvert f(\vec{x}) - f(\vec{y}) \rvert \leq L \lVert \vec{x} - \vec{y}\rVert_2, \forall \vec{x}, \vec{y} \in \mathbb{R}^{d}.
\end{align*}

We present a result combining the guarantees of classical \cite{giles2008multilevel} and quantum \cite{An2021quantumaccelerated} MLMC, where classical corresponds to $\delta =1$ and quantum to $\delta =2$.
\begin{theorem}[MLMC \cite{giles2008multilevel} \& \cite{An2021quantumaccelerated}]
\label{thm:gen_quantum_mlmc}
Let $P$ denote a random variable, and let $P_l (l=0,1,\dots,L)$ denote a sequence of random variables such that $P_l$ approximates $P$ at level $l$.
Let $\hat{Y}_l$ denote the unbiased estimator for $P_l - P_{l-1}$ constructed from $N_l$ samples of $P_l - P_{l-1}$, where we define $P_{-1} \equiv 0$.
Let $V_l$ and $C_l$ be the variance and computational complexity of $\hat{Y}_l$ respectively.
If there exists positive constants $\alpha$, $\beta$, $\gamma$, $\delta$ and $B_0, \sigma_0$ such that
\begin{align}
    &\abs{\mathbb{E}\left[ P_l - P \right]} \le B_0 h_l^\alpha, \\
    &V[\hat{Y}_l] \le \sigma_0 N_l^{-\delta} h_l^\beta, \\
    &C_l = \mathcal{O}\left(N_l h_l^{-\gamma}\right),
\end{align}
where $h_l = \Theta(M^{-l})$ and $M > 1$ is an integer,
then for any $\epsilon < 1$, there is an algorithm that estimates $\mathbb{E}[P]$ up to additive-error $\epsilon$ with a computational complexity bounded by
\begin{equation}
\begin{cases}
\widetilde{\mathcal{O}}\left( \left(\frac{\sigma_0}{\epsilon}\right)^{\frac{2}{\delta}} + 
\left(\frac{\sigma_0}{\epsilon}\right)^{\frac{\gamma}{\alpha}} \right), & \gamma \leq \frac{\beta}{\delta},  \\
\widetilde{\mathcal{O}}\left(  \left(\frac{\sigma_0}{\epsilon}\right)^{\frac{\gamma}{\alpha}-\frac{1}{\delta}\left(2 - \frac{\beta}{\alpha}\right)}  \right), & \gamma > \frac{\beta}{\delta}.
\end{cases}
\end{equation}
\end{theorem}
For the case of derivative pricing, we take $P_l$ to be some discretization of $f(\widehat{X})$, where $l$ is related to the step size. Intuitively, $\beta$ represents the convergence of the discretization scheme, $\gamma$ reflects the cost of the path generation, and $\delta$ gives the convergence of the sampling scheme.
When $\frac{\beta}{\delta} > \gamma$, or equivalently $\frac{\beta}{\gamma} > \delta$, discretization converges faster than sampling, therefore the sampling cost dominates and hence the quadratic speedup from quantum MCI is fully recovered.
On the other hand, if $\frac{\beta}{\delta} < \gamma$, or equivalently $\frac{\beta}{\gamma} < \delta$, discretization converges slower than sampling, therefore the path generation cost dominates.
Since quantum MCI only helps with the cost reduction in sampling, the amount of speedup one would get from quantum MCI in MLMC would be less than quadratic.

If $f$ is globally-Lipschitz continuous, it can be shown that if $r$ is the strong-order convergence of the scheme, then we can take $\gamma=1$, $\alpha = r$, and $\beta = 2r$ \cite[Proposition 2]{An2021quantumaccelerated}.
Classically, in the globally-Lipschitz setting, an $\epsilon$-approximation for $\mathbb{E}[f(\widehat{X})]$ can be obtained using standard Monte Carlo with a computational complexity of $\mathcal{O}\left(\left(\frac{\sigma_0}{\epsilon}\right)^{2+\frac{1}{\alpha}}\right)$.
Using MLMC, it has been shown that for $r \ge 1/2$, such complexity can be improved to~\cite[Theorem 3.1]{giles2008multilevel}:
\begin{align}
\widetilde{\mathcal{O}}\left(\left(\frac{\sigma_0}{\epsilon}\right)^{2}\right), &\qquad \beta \geq 1, \\
    \mathcal{O}\left(\left(\frac{\sigma_0}{\epsilon}\right)^{2+\frac{1-\beta}{\alpha}}\right), &\qquad 0 < \beta < 1.
\end{align}
Applying the Euler-Maruyama scheme to the case where $f$ is Lipschitz continuous, we have $\alpha = \frac{1}{2}$ and $\beta = 1$, therefore the computational complexities of the standard Monte Carlo and MLMC are $\mathcal{O}\left(\left(\frac{\sigma_0}{\epsilon}\right)^{4}\right)$ and $\widetilde{\mathcal{O}}\left(\left(\frac{\sigma_0}{\epsilon}\right)^{2}\right)$ respectively.
Giles et al.~\cite{giles2009analysing} also derived $\beta$ under the Euler-Maruyama scheme for some commonly seen functions $f$ in option pricing that do not satisfy the global Lipschitz condition. 
Specifically, when $f$ is the payoff for a lookback option, $\beta = 1 - \eta$, and when $f$ is the payoff for a digital or barrier option, $\beta = \frac{1}{2} - \eta$, where $\eta$ is an arbitrarily small positive number in both cases.
Nevertheless, the corresponding weak error rate $\alpha$ in these cases is still an open question.

We also have the following analogous guarantee for quantum MLMC in the globally-Lipschitzness setting.
\begin{theorem}[Theorem 3 in \cite{An2021quantumaccelerated} Adapted]
\label{thm:globally-lip-mlmc}
Consider the payoff process $f(\vec{X})$ such that $f$ is globally Lipschitz. If one approximates $\vec{X}$ with the Milstein scheme, then there is 
 a quantum algorithm  that estimates $\mathbb{E}[f(\vec{X})]$ to additive error $\epsilon$ with $0.99$ probability and $\widetilde{\mathcal{O}}(\sigma_0/\epsilon)$ queries.
\end{theorem}
Here, one sees a drastic contrast with the classical case. Quantum appears to require the Milstein scheme to retain the speedup over classical MLMC using Euler-Maruyama. While this is actually only presented as a sufficient condition, as mentioned in Section~\ref{sec:speedupwithcorrelated} (Contribution 4), there is also intuition for why this could be a necessary condition.

\subsection{Quantum Eigenvalue Transformation for State Preparation}
\label{eqn:qet_state_prep_prelims}
In this section, we review the construction of McArdle et al. \cite{mcardle2022quantum}, which provides a general procedure for amplitude encoding one-dimensional functions given access to a high-precision polynomial approximation.

Recall that a  $(\alpha, m, \epsilon)$ block-encoding \cite{Gily_n_2019} $U$ of a linear operator $H$ on $n$-qubits is a  unitary satisfying the following spectral norm error bound
\begin{align*}
\lVert \alpha(\langle 0|^{\otimes m} \otimes I_{n}) U( |0\rangle^{\otimes m} \otimes I_{n}) - H\rVert_{2} \leq \epsilon,
\end{align*}
where we will assume $H$ is Hermitian. The quantum eigenvalue transformation (QET) \cite{Gily_n_2019, Low_2024} can apply a degree $d$ polynomial, $p : [-1, 1] \rightarrow \mathbb{C}$ to $H$ using $\frac{d}{2}$ calls to the block-encoding $U$ and its inverse. We take the QET procedure as a black-box, which also must transform the polynomial $p$ into a set of rotation angles using additional work.

The goal of \cite{mcardle2022quantum} was to prepare the following $n$-qubit state ($N = 2^n)$:
\begin{align}
\label{eqn:mcardle_state}
|\Psi_{f} \rangle = \frac{1}{\mathcal{Z}_f} \sum_{k=0}^{N-1}f\left(a + k\frac{b-a}{N}\right)|k\rangle,
\end{align}
corresponding to a function $f: [a, b] \rightarrow \mathbb{C}$. We assume for simplicity that $\lVert f \rVert_{\infty} \leq 1$. This can  easily be seen to be an encoding of $f$ uniformly sampled on its domain. The idea was to utilize as few coherent-arithmetic operations as possible.

The proposed algorithm starts with a $(1, 1, 0)$ block-encoding $U_{\sin}$ satisfying
\begin{align}
\label{eqn:sin_block_encoding}
    (\langle0|\otimes I_n)U_{\sin}(|0\rangle \otimes I_n)  = \sum_{k=0}^{N-1} \sin(k/2^n)|k\rangle\langle k|,
\end{align}
and a polynomial $p_{\epsilon}$ satisfying
\begin{align*}
    \lVert p_{\epsilon}(y) - f(a + \frac{2(b-a)}{\pi}\arcsin(y))\rVert_{\infty} \leq \epsilon,
\end{align*}
where the degree of $p_{\epsilon}$ is $\mathcal{O}\left(\text{poly}\log(1/\epsilon)\right)$. While this may not always be possible, it is the only scenario where such a procedure is asymptotically efficient. Specifically, if $f$ has such as family of approximating polynomials then the composed $ f(a + (b-a)\arcsin(y))$ does as well. This is due to the following result:

\begin{lemma}[Theorem 2 \cite{mcardle2022quantum}]
\label{lem:arcsin_cost_shift}
Let $b > 0$ and $f(x_0 + x) = \sum_{k=0}^{\infty}a_kx^k$ for every $x \in (-b, b)$ and suppose that $\sum_{k=0}^{\infty} \lvert a_k \rvert b^k \leq B$. Then $g(y) := f(x_0 + \frac{2b-a}{\pi}\arcsin(y)) = \sum_{k=0}^{\infty}c_ky^k$ is such that $\sum_{k=0}^{\infty} \lvert c_k \rvert \leq B$, thus for $\nu \in (0, 1]$ and $T \geq \ln(B/\delta)/\nu$ we have for all $y \in [-1 +\nu, 1-\nu]$:
\begin{align*}
    \lvert g(y) - \sum_{k=0}^{T-1}c_ky^k\rvert \leq \delta.
\end{align*}
\end{lemma}

Then we use QET to apply the transformation $p_{\epsilon}(y)$ to $U_{\sin}$ using at most $\mathcal{O}\left(\text{poly}\log(1/\epsilon)\right)$ calls to $U_{\sin}$ and $U_{\sin}^{\dagger}$. The result is a $(1,3, 0)$ block-encoding $U_{p_{\epsilon}}$ satisfying:

\begin{align}
\label{eqn:transformed-block}
(\langle000|\otimes I_n)U_{p_{\epsilon}}(|000\rangle \otimes H^{\otimes n}) = \sum_{k=0}^{N-1} \frac{p_{\epsilon}(\sin(k/N))}{\sqrt{N}}|k\rangle.
\end{align}
We then make use of exact amplitude amplification, which we recall below.
\begin{theorem}[\cite{mcardle2022quantum} Theorem 4 informally restated] 
\label{thm:amplitude_amp}
If the unitary $U$ and orthogonal projector $\Pi$ satisfy
\begin{align*}
    \Pi U|0\rangle = a|\psi\rangle,
\end{align*}
for state $|\psi\rangle$. Then using two queries to $U$, one to its inverse and single qubit rotations, we can construct a unitaries $W$, $U'$ satisfying
\begin{align*}
    W^kU'|0\rangle|\bar{0}\rangle = |0\rangle|\psi\rangle,
\end{align*}
where $k = \mathcal{O}\left(1/a\right)$.
\end{theorem}

The exact amplitude amplification procedure transforms \eqref{eqn:transformed-block} into
\begin{align*}
    |\widetilde{\Psi}_f\rangle = \frac{1}{\mathcal{Z}_p}\sum_{k=0}^{N-1} p_{\epsilon}(\sin(k/N))|k\rangle,
\end{align*}
which by construction can be seen to be an approximation to the state $|\Psi_f\rangle$. 

The cost of the amplification step is inversely-proportional to a quantity called the $L_2$-filling fraction $\mathcal{F}_{f,\epsilon}$:
\begin{align*}
\mathcal{F}_{f,\epsilon} := \sqrt{\frac{\sum_{k=0}^{N-1}[p_{\epsilon}(\sin(k/N))]^2}{N}},
\end{align*}
which effectively measures how sub-normalized the Riemann sum approximation to $\int_{[a, b]} \lvert f(x)\rvert^2 dx$ is. We restate the guarantee for the overall procedure provided by \cite{mcardle2022quantum}.
\begin{theorem}[\cite{mcardle2022quantum} Theorem 1  restated]
Given a degree $d_{\delta}$ polynomial $p$ supported on $[0, \sin(1)]$ that is guaranteed to satisfy
\begin{align*}
 \lVert p(y) - f(a + (b-a)\arcsin(y))\rVert_{\infty} = \mathcal{O}\left(\delta \mathcal{F}\right),
\end{align*}
then we can prepare the $n$-qubit quantum state $|\widetilde{\Psi_f}\rangle$, which is at most $\epsilon$ in trace distance from $|\Psi_f\rangle$ (Equation \eqref{eqn:mcardle_state}). This procedure uses at most \begin{align*}
    \mathcal{O}\left(\frac{n d_{\epsilon/\mathcal{F}}}{\mathcal{F}_{f,\epsilon}}\right)
\end{align*}
basic quantum gates.
\end{theorem}

Given that this framework enables amplitude-encoding functions, it is also useful for loading probability distributions.

\section{Framework for Analyzing Quantum Derivative Pricing Algorithms}
\label{sec:framework_for_analyze}

In this section we present our detailed framework for analyzing quantum derivative pricing algorithms. This presents a scheme that could be used for future analysis of QMCI applied to other models not considered in this work. Starting with Section~\ref{sec:efficient_of_sim}, we present the notion of fast-forwardable SDEs, which are the most amenable to end-to-end pricing speedups via vanilla (Q)MCI. In Section~\ref{sec:error_analysis_for_quant_deriv_pricing}, we detail how to analyze the various errors that occur in derivative pricing. Lastly, Section~\ref{sec:reduced_qubit_estimates}, highlights our improved numerical integration analysis that leads to significant qubit-count reductions, i.e. going from linear in $T\times d$ per dimension to logarithmic.

\subsection{Categorizing Simulation Efficiency of SDEs}
\label{sec:efficient_of_sim}
We present a new categorization of SDEs based on the computational efficiency of generating a corresponding path process, which can be captured by a notion we call \emph{fast forwardability}.

\begin{restatable}[Fast-forwardable SDE]{defin}{ffsde}
\label{def:fastforward}
The SDE for $\vec{X}(t)$ in Equation \eqref{eqn:ito-process-time-homog} is called fast forwardable (FF) if for any path process $\xpath$ and  $\forall t \in [T]$, we can sample $g_t(\vec{I}(t),  \xpath(t-1))$ (to within $\epsilon$ total variation distance)
with a computational cost of $\mathcal{O}\left(\text{poly}(d, \log(\Delta/\epsilon)\right))$\footnote{We assume the dimension of $\vec{I}(t), d''$ satisfies $d'' = \mathcal{O}(\textup{poly}(d))$.}.
\end{restatable}

We term a simulation procedure that realizes the computational cost presented in Definition~\ref{def:fastforward}  a \emph{fast-forwarding scheme} for the SDE. The multi-asset GBM,  CIR, and multi-asset Heston model with only asset-asset correlations are all examples of fast-forwardable stochastic processes. 

The quantum analog to the above is that a qsample of the path distribution, with an $\ell_2$ measurement distribution in the computational basis that is  at most $\epsilon$ in total variation distance from the desired path distribution, can be prepared with the same gate complexity as the classical sampling complexity. However, fast-forwardability of an SDE does not immediately imply an $\mathcal{O}\left(\text{poly}(d, \log(\Delta/\epsilon))\right)$ quantum sampling complexity. One must show that the fast-forwarding scheme for the stochastic process can be decomposed into distributions that are easy to load quantumly, i.e. well-known one-dimensional distributions. In some cases, like multi-asset Heston, these one-dimensional distributions may not have a closed-form pdf and require a more complicated loading procedure.  Regardless, the main purpose of introducing this notion is that fast-forwardable processes are the ones that are most amenable to both vanilla classical and quantum MCI.

There is an additional subclass of fast-forwardable processes, which have very simple simulation schemes.
\begin{restatable}[Independently Fast-forwardable]{defin}{indpffsde}
\label{def:ff-indep}
The SDE for $\vec{X}(t)$ in Equation \eqref{eqn:ito-process-time-homog} is called independently fast forwardable if $\forall t \in [T]$, $ \xpath_{t} - \xpath_{t-1}$ or $\ln\xpath_{t} - \ln\xpath_{t-1}$ are random variables that are independent of the path $(\xpath(0), \cdots, \xpath(t-1))$.
\end{restatable}

Independent FF is meant to capture the special cases of processes with increments $\vec{I}(t)$ that are independent of the current time point, i.e. $\xpath(t)$. These are significantly easier to analyze and usually have closed-form solutions. As an example, the multi-asset GBM path process has the following closed form solution:
\begin{align*}
    \vec{S}(t) = \vec{S}(0)\exp\left(\vec{\mu}t + \vec{\sigma} \circ \vec{Z}\right), \vec{Z}_t \sim \mathcal{N}(0, \mathbf{C}),
\end{align*}
and thus the update of the path process is $g(\xpath({t-1}), \vec{Z}_t) = \xpath({t-1})\exp\left(\vec{\mu} + \vec{\sigma}\circ\sqrt{\Delta}\vec{Z}_t\right)$, where $\vec{Z}_t \sim \mathcal{N}(0, \mathbf{C})$ independently $\forall t \in [T]$. The $t$-th increment here is thus $\vec{I}(t) = \vec{Z}_t$ and so independent of $\xpath(t-1)$. Hence, since one can sample from a $d$-dimensional rotated Gaussian in $\mathcal{O}\left(\text{poly}(d, \log(1/\epsilon)\right))$ time, the conditions of Definition~\ref{def:ff-indep} are satisfied. In contrast, the CIR process does not have a known closed-form solution for the SDE, and it does not satisfy Definition~\ref{def:ff-indep}. 

There are also financially relevant models that are not known to be fast-forwardable, in the sense of Definition \ref{def:fastforward}, such as the multi-asset Heston model with arbitrary correlations.  When FF is absent, we must resort to the approximation techniques presented in Section~\ref{sec:approximate_sde_simulation}.

\subsection{Error Analysis for Quantum Derivative Pricing}
\label{sec:error_analysis_for_quant_deriv_pricing}

In this section we outline the framework we shall use for error analysis, and to correspondingly determine the parameters of the underlying algorithms and their resource requirements as a function of the target approximation error. There are many sources of error that can interact in non-trivial ways.  Since there does not appear to be any discussion as detailed as this one in the existing literature, we consider the framework itself to be a contribution. Recall that the goal is to estimate a multi-dimensional integral.

Quantum derivative pricing algorithms (reviewed in more detail in Section~\ref{sec:quantum_derivative_pricing_intro}) can roughly be summarized in two steps. We use the notation introduced in Section~\ref{sec:background_derivative_pric}.
\begin{enumerate}
    \item Load quantum samples corresponding to a truncated and discretized $\widehat{X}$ by constructing a unitary:
        \begin{align}
        U_{\mathrm{path}}|0^k\rangle = \sum_{\vec{x} \in \mathcal{M}} \sqrt{\widetilde{\mathbb{P}}[\xpath = \vec{x}]} \ket{\vec{x}},
    \end{align}
    where $\vec{x}$ ranges over the support of the distribution and $\ket{\vec{x}}$ corresponds to a binary encoding of $\vec{x}$ in a computational basis state. As mentioned earlier, for our setting, one can consider the measure here $\widetilde{\mathbb{P}}$ to be an approximation to the discretized and renormalized path distribution, $\mathbb{P}$, restricted to a grid $\mathcal{M}$.  This is done by composing the unitaries $U_{\mathrm{incr}_t}$ and $U_{\mathrm{jump}_t}$ defined in  \eqref{eqn:randomwalk_operators_incr} and \eqref{eqn:randomwalk_operators_jump}, respectively.
    \item Apply the unitary
    \begin{align*}
            U_{f}\ket{\vec{x}}\ket{0} = \ket{\vec{x}}\left( \sqrt{\widetilde{f}(\vec{x})} \ket{0} + \sqrt{1 - \widetilde{f}(\vec{x})} \ket{1} \right),
    \end{align*}
    to amplitude encode the payoff $f$, and input the composed unitary to amplitude estimation to compute the expectation. The function $\widetilde{f}$ is some rescaled approximation to $f$.
\end{enumerate}
Note that $(U_{\text{jump}_t})_{t=0}^{T-1}$ effectively corresponds to an arithmetic function  that takes random variables following the distribution loaded by $(U_{\text{incr}_t})_{t=0}^{T-1}$. Using well-known techniques for quantum arithmetic (Section~\ref{sec:coherent-arith}), we can implement such an operation efficiently. Now $U_{\text{incr}_t}$ itself may be composed of unitaries that load simpler distributions and perform arithmetic operations. At the lowest level, we call these  distributions \emph{primitives}. The primitives we consider are well-known one-dimensional distributions, e.g. the standard Gaussian or central $\chi^2$. 

Note that we can always reduce the simulation to a common class of distributions, i.e. via SDE approximation schemes (Section~\ref{sec:approximate_sde_simulation})) we can always reduce to Gaussians and iterated stochastic integrals. However, the ``primitives'' may not be one-dimensional in general, for example the stochastic integrals can be correlated if they involve overlapping Brownian motions. Still, all of the cases considered in this paper utilize one-dimensional primitives.

Using the notion of distributional primitives, we can reformulate the coherent random walk generated by  $U_{\mathrm{incr}_t}$ and $U_{\mathrm{jump}_t}$ in the following way that is more amenable to analysis (and correspond to what one would actually do in practice). Note that the approach that we propose using quantum PDE solvers (Section~\ref{sec:quantum_pde_solvers}) slightly falls outside of this framework, however, since our main contribution there will be to show why they do not work, we do not cater to that case. Let each unitary in $(U_{\text{prim}_1}, \dots U_{\text{prim}_m})$  prepare a discrete qsample corresponding to a one-dimensional distribution when applied to the all-zeroes state. We also include a separate list $(U^{(c)}_{\text{prim}_1}, \dots U^{(c)}_{\text{prim}_k})$  of unitaries that include loaders that can operate on an ancillary register, as to encode conditional distributions. For our purposes, these will typically be local unitaries, i.e. encoding a one-dimensional distribution conditioned on at most a constant many other random variables.

The procedure, which was called the \emph{reparametrization approach} in \cite{chakrabarti2021threshold}, is as follows:
\begin{enumerate}
    \item \textbf{Resource Randomness:} Use as many parallel calls to unitaries from $(U_{\text{prim}_1}, \dots U_{\text{prim}_m})$ and subsequently (potentially some sequential, if there is overlap in the conditioning) calls to unitaries from  $(U^{(c)}_{\text{prim}_1}, \dots U^{(c)}_{\text{prim}_m'})$ as needed to prepare the ``resource randomness'', i.e. $\Omega(T)$ calls and $\Omega(T)$ quantum registers. The $T$ components of resource randomness correspond to random variables with a joint measure that is the pushforward under some map, $g'$, that produces $(\vec{I}_1, \dots, \vec{I}_{T-1})$. When composed with  the functions $(g_1, \dots, g_{T-1})$, this gives a transformation $g$ that takes the primitives to a qsample of the entire path $(\widehat{X}(1), \dots, \widehat{X}({T-1}))$. This is performed by a unitary $U_{g}$ encoding the arithmetic function $g$.
    \item \textbf{Payoff Encoding:} Apply $U_{f}$ to amplitude encode the payoff $f$, and  use QAE to compute the expectation.
\end{enumerate}
The primitive distributions will approximately be qsamples in the sense of Definition \ref{defn:qsample-exact}. For the ease of analysis, we will always assume the payoff $f$ is piecewise linear. Here, we will use $d$ to refer to the dimension of the  ``reparameterized'' integration problem. For example if there are $k$ assets (and $k$ volatilities) encoded in $\xpath$ over $T$ time steps, and a single stochastic transition for an asset/volatility component requires $r$ sources of randomness (primitives) then we use $d = r \times k \times T$ primitives in total\footnote{Note that we overload some notation for $d$, as this $d$ is not the same as the dimension of $\xpath(t)$. The dimension of $\hat{X}(t)$ is $k$ here.}.

This procedure can easily be seen to be approximating the integral:
\begin{align}
\label{eqn:analysis-integral}
    \int_{\mathbb{R}^{d}} f\circ g(\vec{x})p_{\text{prim}}(\vec{x})d\vec{x},
\end{align}
where $p_{\text{prim}}(\vec{x})$ is a product of $d$ primitive densities of the form $p(x_i)$ or $p(x_i| x_{S})$, for some $S \subset [d]$ with $\lvert S \rvert$ a constant. The sources of error in the approximation are as follows (note 1-4 are also present classically, with 5 having a classical analog via MCI):
\begin{enumerate}
    \item \textbf{Truncation Error:} The transition densities are truncated to have bounded support. The bounds on the support must be chosen so that the error in the integral introduced by neglecting the region outside of the pre-fixed support is of the same order as the target error in our integral. We denote this component of the error by $\etrunc$. The truncation error is a function of the size of the interval to which the distribution is truncated. We always truncate the distribution to an $\ell_\infty$ ball of size length $2R$, ie. $[-R,R]^d$. The truncation error has a further dependence on the function that is integrated over the loaded distribution. This leads to estimating 
    \begin{align}
    \label{eqn:trunc_integ}
    \int_{[-R, R]^d} f\circ g(\vec{x})p_{\text{prim}}(\vec{x})d\vec{x},    
    \end{align} which differs from Equation \eqref{eqn:analysis-integral} by at most $\etrunc$.
    \item \textbf{Discretization Error:} The true expectation is approximated by a discrete sum that approximates the integral in the sense of typical quadrature rules. An important benefit of the reparameterization approach is that the normalization of the primitives involved ensures that upper bounds on the density or smoothness of the path distribution do not enter the computation. The discretization error $\epsilon_{\mathrm{disc}}$ depends on the number of qubits used to represent the support of the loaded qsamples. Specifically, discretization leads to a sum of the form 
    \begin{align}
        \sum_{\vec{x}\in\mathcal{M}} f\circ g(\vec{x}) p_{\text{prim}}(\vec{x})\left(\frac{2R}{N}\right)^d,
    \end{align} with error at most $\edisc$ from Equation \eqref{eqn:trunc_integ}. In this paper, we will consider the left-endpoint rule for Riemann summation (Lemma~\ref{lem:left_rule_error}) over a grid $\mathcal{M} \subset [-R, R]^{d}$.
    \item \textbf{Arithmetic Error:} The transformation of primitive distributions, via $g$, as well as the evaluation of the payoff function $f$, for integration is performed using coherent quantum arithmetic. Performing fixed point arithmetic on registers of a given size leads to an accumulating arithmetic error, that leads to a corresponding error $\epsilon_{\mathrm{arith}}$ in the final integral. 
    \item \textbf{Distribution Error}: In reality, we can only implement each $U_{\text{prim}_i}$ up to some error. The error $\epsilon_{\text{dist}}$ corresponds to the total variation distance (TVD) between the measurement distributions (in the computational basis) of the desired discrete qsample and the one prepared by $U_{\text{prim}_i}$ applied to the all zeroes state (or uniform over an ancillary register for the conditional distribution case). For our purposes, we will only require the states to be close in the natural, computational basis. Note this is  weaker than the usual trace-distance or fidelity quantum-state metrics. Since the discretized distributions will not correspond to the square-amplitudes of a normalized quantum state, there will be an additional sub-normalization error. This is related to $\epsilon_{\text{distr}}$.
    \item \textbf{Amplitude Estimation Error:} The final component of the error, denoted by $\epsilon_{\mathrm{amp}}$ is due to the discrete-sum being estimated by amplitude estimation, and follows from well known bounds on the efficiency of quantum amplitude estimation subroutines. Specifically, we make use of an approximate amplitude encoder that prepares
    \begin{align}
    \label{eqn:amp_encoder}
        \sqrt{\sum_{\vec{x}\in\mathcal{M}} \widetilde{f}\circ \widetilde{g}(\vec{x}) \frac{\widetilde{p}_{\text{prim}}(\vec{x})}{\widetilde{\mathcal{Z}}^2}\left(\frac{2R}{N}\right)^d}|0\rangle|\psi\rangle + |\perp\rangle,
    \end{align}
    where tilde denote approximations due to arithmetic and distribution error. We can only recover the probability of observing the $|0\rangle$ state up to some additive $\eamp$. The distribution $\widetilde{p}_{\text{prim}}$ will be a product of one-dimensional distributions so $\widetilde{\mathcal{Z}}$ will be the product of the normalization constants.
\end{enumerate}

We now discuss a bit about how these errors depend on each other and how they are analyzed. Using the left-endpoint rule $\mathcal{R}_N$, the discretization error (Lemma~\ref{lem:left_rule_error}) is of the form
\begin{align}
\label{eqn:num_int_error}
    \lvert \mathcal{R}_N -   \int_{[-R, R]^d} f\circ g(\vec{x})p_{\text{prim}}(\vec{x})d\vec{x} \rvert \leq  \left[\sum_{\vec{x} \in \mathcal{M}}\sup_{\vec{y}\in\mathcal{C}_{\vec{x}}} \lVert \nabla ([f\circ g] \cdot p_{\text{prim}})(\vec{y}))\rVert_{\infty}\left(\frac{2R}{N}\right)^d\right] \frac{2dR}{N},
\end{align}
where $\mathcal{C}_{\vec{x}} := \vec{x} + [0, 2R/N]^{d}$. It is apparent from the above expression that we will need to analyze the derivatives of $f, g,$ and $p_{\text{prim}}$ to determine $\edisc$. The value for $\log_2(N)$ corresponds to the number of (qu)bits used for performing arithmetic, where there are $d\log_2(N)$ in total. It is relatively easy to show that in most cases $N = \mathcal{O}\left(d\log(R/\epsilon_{\text{disc}})\right)$. However, this seems to be an overly pessimistic bound as classically and in practice the number of bits needed for arithmetic (per dimension) does not grow with $d$. As discussed in Section~\ref{sec:reduced_qubit_estimates}, we will show that a more refined analysis can lead to the more reasonable $N = \mathcal{O}\left(\log(dR/\epsilon_{\text{disc}})\right)$.

When analyzing $\edisc$, the truncation error enters the picture because we want to ensure that $\log(N)$ is at most $\log$ is all of the sources of error. Since for GBM and Heston, $g$ involves an exponential, we want to ensure that the primitives can be truncated to $R = \mathcal{O}\left(\text{poly}\log(1/\etrunc)\right)$. So $N$ will be a function of $\etrunc$ along with $\edisc$, and thus the analysis always starts with truncation before discretization.

The distribution error can be computed as follows:
\begin{align}
&\lvert \sum_{\vec{x}\in\mathcal{M}} f\circ g(\vec{x}) \frac{\widetilde{p}_{\text{prim}}(\vec{x})}{\widetilde{\mathcal{Z}}^2}\left(\frac{2R}{N}\right)^d - \sum_{\vec{x}\in\mathcal{M}} f\circ g(\vec{x}) p_{\text{prim}}(\vec{x})\left(\frac{2R}{N}\right)^d\rvert \nonumber\\ 
\label{eqn:discrete-error-bound}
&\leq \max_{\vec{x} \in [-R, R]^d} \lvert f\circ g(\vec{x}) \rvert\sum_{\vec{x} \in \mathcal{X}} \lvert \frac{\widetilde{p}_{\text{prim}}(\vec{x})}{\widetilde{\mathcal{Z}}^2}\left(\frac{2R}{N}\right)^d - p_{\text{prim}}(\vec{x})\left(\frac{2R}{N}\right)^d \rvert,
\end{align}
where the dependence on the truncation error is clear from the first factor. The second factor is like a total-variation distance for discretized distributions. By our assumptions $p_{\textup{prim}}$ will be of the form
\begin{align*}
    p_{\textup{prim}}(\vec{x}) = \prod_{j=1}^{d_1}p_{\text{prim}_{i_j}}(x_j)\prod_{k=1}^{d_2}p_{\text{prim}_{i_k}}(x_j | x_{S_k}),
\end{align*}
with $i_j \in \{1, \dots, m\}$, $i_k \in \{1, \dots, m'\}$,  and $d = d_1 + d_2$. 

We can hence approximate each component of $p_{\textup{prim}}$ separately, assuming an approximation error on the $p_{\text{prim}_{i_k}}(x_j | x_{S_k})$ that is uniform in the value of $x_{S_k}$. The analysis would then become the same for the conditional and unconditional given these uniform bounds. Hence, without loss of generality, we from now on assume $d_2 = 0$.
We will have the guarantee that $U_{\text{prim}_i}$ prepares a quantum state that is at most $\epsilon_{\text{distr}}$ away in computational-basis TVD from the state:
\begin{align}
    |p_{\text{prim}_i}\rangle = \frac{1}{\mathcal{Z}_i}\sum_{{x}_k \in \mathcal{M}} \sqrt{p_{\text{prim}_i}({x}_k)}\cdot\left(\frac{2R}{N}\right)|x_k\rangle,
\end{align}
and that
$U^{(c)}_{\text{prim}_i}$ prepares a quantum state when applied to $|0^a\rangle|x_{S_i}\rangle$ that is at most $\epsilon_{\text{distr}}$ away in computational-basis TVD, uniform in $x_{S_i}$, from the state:
\begin{align}
    |p_{\text{prim}_i}(x_{S_i})\rangle|x_{S_i}\rangle = \frac{1}{\mathcal{Z}_i(x_{S_i})}\sum_{{x}_k \in \mathcal{M}} \sqrt{p_{\text{prim}_i}({x}_k | x_{S_i})}\cdot\left(\frac{2R}{N}\right)|x_k\rangle|x_{S_i}\rangle.
\end{align}

This gives that 
\begin{align}
&\sum_{\vec{x} \in \mathcal{M}}\lvert \frac{\widetilde{p}_{\text{prim}}(\vec{x})}{\widetilde{\mathcal{Z}}^2}\left(\frac{2R}{N}\right)^d - \frac{p_{\text{prim}}(\vec{x})}{{\mathcal{Z}}^2}\left(\frac{2R}{N}\right)^d \rvert 
=\mathcal{O}\left(d\epsilon_{\text{distr}}\right),
\end{align}
which is one component of the distribution error.

There is an additional source of error in distribution loading that comes from the discretized probability distribution not being normalized, i.e. the subnormalization error. Continuing with our notation, each one-dimensional primitive is truncated to at most $[-R, R]$. Suppose that for the $j$-th primitive $\lvert \mathbb{P}_j([-R, R]) - 1 \rvert =\mathcal{O}(\epsilon_{\text{trunc}})$, which will be assured by the chosen $R$ for truncation error already. We can then use Theorem~\ref{thm:prob_state_prep_guarantee} so that the same $N$ that ensures $\epsilon_{\text{distr}}$ is small enough  implies that $\lvert \mathcal{Z}_{i_j}^2 - \mathbb{P}_{i_j}(x_j \in [-R, R]) \rvert = \mathcal{O}(\epsilon_{\text{distr}})$. We also assume that we can have $ \lvert \mathcal{Z}_{i_j}^2(x_{S_{i_j}}) - \mathbb{P}_{i_j}(x_j \in [-R, R] | x_{S_{i_j}}) \rvert = \mathcal{O}(\epsilon_{\text{distr}})$, uniformly in $x_{S_{i_j}}$.

Hence we get a bound on distance between the loaded state and the unnormalized, ideal state:
\begin{align*}
&\lvert \sum_{\vec{x}\in\mathcal{M}} f\circ g(\vec{x}) \left[\prod_{j=1}^{d_1}{p_{\text{prim}_{i_j}}(x_j)}\prod_{k=1}^{d_2}{p_{\text{prim}_{i_k}}(x_k | x_{S_{i_k}})}-  \prod_{j=1}^{d_1}\frac{p_{\text{prim}_{i_j}}(x_j)}{\mathcal{Z}_{i_j}^2}\prod_{k=1}^{d_2}\frac{p_{\text{prim}_{i_k}}(x_k | x_{S_{i_k}})}{\mathcal{Z}_{i_k}^2(x_{S_{i_k}})}\right]\left(\frac{2R}{N}\right)^d \rvert \\ 
&\leq \max_{\vec{x} \in [-R, R]^d} \lvert f\circ g(\vec{x})\rvert \lvert 1 - \prod_{j=1}^{d_1}\mathcal{Z}_{i_j}^2\prod_{k=1}^{d_2}\min_{x_{S_{i_k}}} \mathcal{Z}^2_{i_k}(x_{S_{i_k}}) \rvert \\
&\leq \max_{\vec{x} \in [-R, R]^d} \lvert f\circ g(\vec{x})\rvert \lvert 1 - (1 \pm (\epsilon_{\text{trunc}}+ \epsilon_{\text{distr}}))^{d}) \rvert \\
&= \mathcal{O}(d\max_{\vec{x} \in [-R, R]^d} \lvert f\circ g(\vec{x})\rvert (\epsilon_{\text{trunc}}+ \epsilon_{\text{distr}})).
\end{align*}
The triangle inequality gives that the two sources of error discussed above bound \eqref{eqn:discrete-error-bound}. 
Thus we need  to scale
\begin{align}
\label{eqn:distribution_loading_error_bound}
\epsilon_{\text{distr/trunc}} \rightarrow \frac{\epsilon_{\text{distr/trunc}}}{d \max_{\vec{x} \in [-R, R]^d} \lvert f\circ g(\vec{x})\rvert}.
\end{align}
We will require that the dependence of the distribution loading procedure is poly-logarithmic in the inverse error, which is assured by the procedures that we use. This will imply that when $g$ or $f$ are exponential in $\vec{x}$ and if $R=\mathcal{O}\left(\text{poly}\log(1/\epsilon_{\text{trunc}})\right)$, we will get $\mathcal{O}\left(\text{poly}\log(d/\epsilon_{\text{trunc}}\epsilon_{\text{distr}})\right)$ complexity for the distribution loading.

The arithmetic error is in general easy to handle because of prior work on circuit constructions for coherent arithmetic \cite{häner2018optimizingquantumcircuitsarithmetic}. However, in some cases, like CIR, we will need to deal with recursions that can cause the arithmetic error to propagate. 
In general, the arithmetic error can bounded by:
\begin{align*}
&\lvert \sum_{\vec{x}\in\mathcal{M}} \widetilde{f}\circ \widetilde{g}(\vec{x}) \widetilde{p}_{\text{prim}}(\vec{x})\left(\frac{2R}{N}\right)^d - \sum_{\vec{x}\in\mathcal{M}} f\circ g(\vec{x}) \widetilde{p}_{\text{prim}}(\vec{x})\left(\frac{2R}{N}\right)^d\rvert \\
&\leq(1+\epsilon_{\text{distr}})\max_{\vec{x} \in \mathcal{M}}\lvert \widetilde{f}\circ \widetilde{g}(\vec{x}) - f\circ g(\vec{x})\rvert =:(1+\epsilon_{\text{distr}}) \epsilon_{\text{arith}}.
\end{align*}
Since $f$ and $g$ will be computed via quantum arithmetic, the error can be bounded using standard results (Section~\ref{sec:coherent-arith}). The gate complexities are $\mathcal{O}(\text{poly}(\log(1/\epsilon_{\text{arith}}), d\log(N)))$.

If all of the above mentioned assumptions are met (which we will show is possible for the applications we consider),  the triangle inequality, and equating all $\epsilon$'s, show that we can prepare a state encoding equation \eqref{eqn:amp_encoder} such that
\begin{align}
\label{eqn:overall_bound}
    \lvert \sum_{\vec{x}\in\mathcal{M}} \widetilde{f}\circ \tilde{g}(\vec{x}) \frac{\widetilde{p}_{\text{prim}}(\vec{x})}{\widetilde{\mathcal{Z}}^2}\left(\frac{2R}{N}\right)^d -     \int_{\mathbb{R}^{d}} f\circ g(\vec{x})p_{\text{prim}}(\vec{x})d\vec{x} \rvert = \mathcal{O}(\epsilon),
\end{align}
with $\mathcal{O}\left(\text{poly}(d, \log(1/\epsilon))\right)$ one- and two-qubit gates and either $\mathcal{O}\left(d\text{poly}\log(d/\epsilon))\right)$ or $\mathcal{O}\left(d^2\text{poly}\log(d/\epsilon))\right)$ qubits in total.

The only unaccounted for source of error  $\eamp$ comes from  amplitude estimation. It is well known (and recalled in Section~\ref{sec:quantum_derivative_pricing_intro}) that $\mathcal{O}\left(B/\epsilon\right)$ calls to the amplitude encoder in \eqref{eqn:amp_encoder}, suffices to make $\eamp < \epsilon$. Note that $B$ is a bound on the square-amplitude we are estimating and can be obtained from the guarantee in \eqref{eqn:overall_bound} and the region of truncation. The region of truncation should ideally be on the order of the standard deviation, as mentioned earlier. This will lead to a 
\begin{align*}
\widetilde{\mathcal{O}}\left(\text{poly}(d, \log(1/\epsilon))\cdot\sqrt{\text{Var}(f(\widehat{X}))}/\epsilon\right)
\end{align*}
gate-complexity quantum algorithm for pricing a derivative satisfying all of the above conditions using $\mathcal{O}\left(d\text{poly}\log(d/\epsilon))\right)$ or $\mathcal{O}\left(d^2\text{poly}\log(d/\epsilon))\right)$ qubits. If such a complexity is attained, then we have an end-to-end quadratic speedup in $\sqrt{\text{Var}(f(\widehat{X}))}/\epsilon$ over classical MCI. Ideally, the factor that comes from the cost to produce (q)samples should be on the same order for classical and quantum.
While the triangle inequality shows that at least all sources of error add, the above shows that there is still an apparent coupling between all of the sources error. 

Our main focus will be on truncation, discretization and distribution error, since these are the ones that determine the feasibility of an end-to-end asymptotic speedup.  To meet the above guarantees on the space and time complexity of quantum derivative pricing we only need to show poly-logarithmic dependence on the inverse of these sources of errors. Additionally, if we want $\mathcal{O}\left(d\text{poly}\log(d/\epsilon))\right)$ space we need a more careful discretization analysis. 

\subsection{Reduced Qubit Estimates by Improved Numerical Quadrature Analysis}
\label{sec:reduced_qubit_estimates}

As shown earlier, the error for the multi-dimensional left-endpoint rule is the following
\begin{align}
\label{eqn:num_int_error_1}
    \lvert \mathcal{R}_N -   \int_{[-R, R]^d} f\circ g(\vec{x})p_{\text{prim}}(\vec{x})d\vec{x} \rvert \leq  \left[\sum_{\vec{x} \in \mathcal{M}}\sup_{\vec{y}\in\mathcal{C}_{\vec{x}}} \lVert \nabla ([f\circ g] \cdot p_{\text{prim}})(\vec{y}))\rVert_{\infty}\left(\frac{2R}{N}\right)^d\right] \frac{2dR}{N}.
\end{align}

Of course, we could upper bound the right-hand side by 
\begin{align*}
\sup_{\vec{x} \in [-R, R]^{d}} \lVert \nabla ([f\circ g] \cdot p_{\text{prim}})(\vec{x})\rVert_{\infty} \frac{d(2R)^d}{N},
\end{align*}
which may look more familiar (specifically its one-dimensional variant). However, this leads to something that appears to contradict the well-known success of Monte Carlo integration for high-dimensional financial problems \cite{glasserman2004monte}. Specifically, the above bound leads to the number of bits for arithmetic going as $\log(N) = \Omega\left(d\right)$, when $R > 1$. Modern classical computers of course work with a fixed number of bits for arithmetic and can handle problems with $d$ significantly above that bit count. So, we must have been too loose, or a different error metric is more appropriate. 

Let us  look back at the error term in \eqref{eqn:num_int_error_1} but isolating one specific factor
\begin{align}
\label{eqn:sup_term}
\left[\sum_{\vec{x} \in \mathcal{M}} \sup_{\vec{y}\in\mathcal{C}_{\vec{x}}} \lVert \nabla ([f\circ g] \cdot p_{\text{prim}})(\vec{y}))\rVert_{\infty}\left(\frac{2R}{N}\right)^d\right].
\end{align}

If this sum is  $\mathcal{O}(\text{poly}(d))$, then we can get an absolute $\edisc$ with $\log(N)$ only scaling as  $\mathcal{O}(\log(d))$. Alternatively, if one can show that for each cell $\mathcal{C}_{\vec{x}}$
\begin{align}
\label{eqn:derive_bound}
\sup_{\vec{y}\in\mathcal{C}_{\vec{x}}} \lVert \nabla ([f\circ g] \cdot p_{\text{prim}})(\vec{y}))\rVert_{\infty} = \mathcal{O}\left(\inf_{\vec{y} \in \mathcal{C}_{\vec{y}}} f\circ g(\vec{y})p_{\text{prim}}(\vec{y})\right)
\end{align}
then from
\begin{align}
\label{eqn:rel_err_bound}
\left[\sum_{\vec{x} \in \mathcal{M}} \inf_{\vec{y}\in\mathcal{C}_{\vec{x}}} ([f\circ g] \cdot p_{\text{prim}})(\vec{y}))\left(\frac{2R}{N}\right)^d\right] \leq \int_{[-R, R]^d} f\circ g(\vec{x})p_{\text{prim}}(\vec{x})d\vec{x},
\end{align}
we would have a relative-error of  $\edisc$ where $N$ only has to grow linearly with $d$. The above two cases show the importance of carefully analyzing \eqref{eqn:sup_term} to obtain more realistic resource estimates.

In the case of $g$ or $f$ involving exponential functions, for example GBM and Heston, we will typically only get the improvement  for the relative-error case. Note that \eqref{eqn:derive_bound} can be shown for the case of GBM. This would then lead to a bound of the form \eqref{eqn:rel_err_bound} asymptotically.

To highlight the significant improvements in space, we use the above  trick to obtain the following relative discretization error guarantee for pricing on multi-asset GBM:

\begin{restatable}[GBM Relative Discretization Error]{theorem}{gbmRelativeError}
\label{thm:gaussian-relative}
Let $f$ be a piecewise linear payoff for a derivative over a multi-asset GBM with $d$ assets monitored over $T$ time steps. Suppose that the maximum slope is $\mathcal{O}(1)$. If the number of bits per standard Gaussian is $\log_2(N)=\mathcal{O}\left(\log\left((1+\sigma)dT/\epsilon_{\text{disc}}\right)\right)$, then the error  in the price over a multi-asset GBM is $\mathcal{O}(\mathbb{E}[g(\vec{x})]\epsilon)$. 
\end{restatable}

The above (proven in Appendix \ref{subsec:proof-thm-gauss-rela}) is presented here to illustrate the kind of discretization analysis we will strive to achieve for CIR and Heston. In addition, GBM has already been analyzed in prior work, so the above is an  improved resource estimate. 

For CIR we will utilize the absolute-error approach, and it is slightly more complicated than the GBM relative error. Unfortunately, for the Heston model we are unable to achieve a relative-error estimate. This is a result of one of densities involved, specifically the integral over CIR, not being known in closed form, i.e. we are unable to get Equation~\eqref{eqn:derive_bound}.

Still, this emphasizes the issues that could come about due to being too loose with the discretization analysis. We were also unable to find a mention of such issues in the numerical quadrature literature. This is potentially because most cases restrict to $R=1$, where this issue is not present \cite{novak2015resultscomplexitynumericalintegration}.

\section{Subroutines for Primitive Distribution Loading}
\label{sec:subroutines_for_distribution_loading}

As mentioned in Section~\ref{sec:error_analysis_for_quant_deriv_pricing}, we analyze a reparameterized formulation of the derivative pricing integral that reduces the task to evaluating the expectation of function of a random variable with a product distribution. The components of the product distribution are called primitives. In this section, we extend the analysis of McArdle et al. \cite{mcardle2022quantum} and Grover black-box state preparation to be catered towards loading one-dimensional probability distributions. This culminates in Theorems~\ref{thm:prob_state_prep_guarantee} and~\ref{thm:prob_state_prep_guarantee_black_box} in Section~\ref{sec:gen_prim_loading}. In the subsequent subSection~\ref{sec:resources_for_loading_common}, we apply the results of Section~\ref{sec:gen_prim_loading} to obtain new algorithms for loading $\chi^2$, L\'evy areas,  and the integral of a CIR process. Note that we also show how to load a distribution when only the characteristic function is known in closed form. These routines are crucial for enabling the speedups presented in Section~\ref{sec:pricing_deriv_fast_forwardable} and~\ref{sec:multi_level_monte_carlo}.

\subsection{General Routines for Primitive Distribution Loading}
\label{sec:gen_prim_loading}

In Section~\ref{eqn:qet_state_prep_prelims}, we reviewed the low-coherent-arithmetic procedure for encoding a function  $f : [a, b] \rightarrow \mathbb{C}$ onto the amplitudes of a quantum state:
\begin{align*}
|\Psi_{f} \rangle = \frac{1}{\mathcal{Z}_f} \sum_{k=0}^{N-1}f\left(a + k\frac{b-a}{N}\right)|k\rangle.
\end{align*}
However, for generating discrete qsamples, we will want to encode the square-root of some probability density. We start by presenting a modified analysis of the framework by \cite{mcardle2022quantum} that is more suited for distribution loading tasks. Recall that the dominating cost of  this procedure comes from the $L_2$-filling fraction (Section~\ref{eqn:qet_state_prep_prelims}). For the case of probability densities, our main result will remove the explicit dependence on this quantity. The dependence falls to the derivatives of the densities. However, this ends up impacting the qubit count and not the gate count.

As mentioned in Section~\ref{sec:error_analysis_for_quant_deriv_pricing}, we will construct a discrete qsample for the path process $\xpath$ by composing well-known, one-dimensional distributions that we call \emph{primitives}. If $p$ denotes the continuous density of the primitive distribution, then we seek to load an $n$-qubit state
\begin{align}
\label{eqn:one-d-disc-q-samp}
|\widehat{\Psi}_p\rangle = \frac{1}{\mathcal{Z}_p} \sum_{k=0}^{N -1} \sqrt{\frac{p(a+ k \frac{b-a}{N})(b-a)}{N}}|k\rangle,
\end{align}
and it should be apparent that the square of the amplitude (ignoring the normalization $\mathcal{Z}_p$) is a rectangular approximation to the probability mass in $[a + k\frac{b-a}{N}, a + (k+1)\frac{b-a}{N}]$. Also recall that $N = 2^n$, so this loads a, renormalized, $N$-point rectangular approximation to $p$ over $[a, b]$. We have also implicitly assumed that $\lVert p \rVert_{\infty} \leq 1,$ which is valid for all primitives that we apply the QET loading procedure to.  A consequence of our results will be that the state we prepare has a computational-basis measurement distribution that is actually $\mathcal{O}\left(\epsilon\right)$ in TVD from that of the unnormalized state:

\begin{align}
\label{eqn:unnorm-one-d-disc-q-samp}
|\Psi_p\rangle =  \sum_{k=0}^{N -1} \sqrt{\frac{p(a+ k \frac{b-a}{N})(b-a)}{N}}|k\rangle.
\end{align}
We do not require bounds on stronger quantum state metrics.

We have the following guarantee (proven in Appendix \ref{lem:proofLemQET}), which is a 
version of \cite[Theorem 1]{mcardle2022quantum} catered to the case of loading probability densities and our chosen metric. %

\begin{restatable}[Polynomial Approximation State Preparation]{lemma}{polyApproxStatePrep}
\label{lem:poly_approx_qet}
Consider $N$ uniform grid points, $\{x_k\}_{k=0}^{N-1}$, over $[a, b]$, and let
 $\mathcal{Z}_p^2 := \frac{b-a}{N}\sum_{k=0}^{N-1}p(x_k)$. Let $p : [a, b] \rightarrow \mathbb{C}$ and $\lVert p \rVert_{\infty} \leq 1$. Suppose $\epsilon < \frac{1}{3}$. If there exists a degree $d_{\delta}$ degree polynomial approximating 
\begin{align}
\label{eqn:shft_sqrt_p}
    \sqrt{p(a + \frac{2(b-a)}{\pi}\arcsin(x))}
\end{align} uniformly to $ \delta := \epsilon \cdot \frac{\mathcal{Z}_p^4}{(b-a)^2}$ error on $[0, \sin(1)]$, then, we can prepare a $\log_2(N)$-qubit quantum state $|\widetilde{\Psi}\rangle$ with a measurement distribution in the computational basis that is at most $\epsilon$ in TVD  from that of the state in \eqref{eqn:one-d-disc-q-samp}
using 
\begin{align*}
\mathcal{O}\left(\frac{\log(N)d_{\delta}}{\mathcal{Z}_p}\right)
\end{align*}
one- and two-qubit gates.
\end{restatable}

The above only works for primitives with compact support. In the analysis, we will need to truncate the support of the primitives. In all cases, the chosen values of $a$ and $b$ will be $\mathcal{O}\left(\text{poly}\log(1/\epsilon_{\text{trunc}})\right)$. In addition, we show that all primitives we consider have $\mathcal{O}\left(\text{poly}\log(1/\delta)\right)$ degree uniform approximations. The last piece to consider is bounding $\frac{1}{\mathcal{Z}_p}$, which comes from amplification amplification and is the inverse of the $L_2$-filling ratio. We have the following bound on this cost for a $p$ corresponding to a pdf. 

\begin{theorem}[Arithmetic Free Loading for Truncated Probability distributions]
\label{thm:prob_state_prep_guarantee}
Suppose $p: \mathbb{R} \rightarrow [0, 1]$ is a probability density. Consider the pdf truncated to the domain $[a, b]$, where $a, b$ are chosen such that $\mathbb{P}([a, b]^{c}) < \frac{1}{2}$. In addition, suppose that $p$ has for any $\delta < 1$ a degree $d_{\delta}$, $\delta$-uniform-error polynomial approximation on $[a, b]$. Lastly, take $N = \Omega\left((b-a)\max_{[a,b]}\lvert p' \rvert/\epsilon\right).$ Then we can construct a state with a measurement distribution in the computational basis that is at most $\epsilon$ in TVD from that of  \eqref{eqn:one-d-disc-q-samp} using \begin{align*}
    \mathcal{O}\left(\log(N)d_{\epsilon/(b-a)^2}\right).
\end{align*}
one- and two-qubit gates and satisfying 
\begin{align*}
    \lvert \mathbb{P}([a, b]) - \mathcal{Z}_p^2 \rvert = \mathcal{O}\left(\epsilon\right).
\end{align*}

Hence the output state will be $\mathcal{O}\left(\epsilon\right)$ ``TVD'' \footnote{The quotations are to emphasize that this is w.r.t. the unnormalized $\ell_2$ measurement distribution of an unnormalized quantum state. Hence we compute the TV metric between normalized and unnormalized discrete sequences.} from the unnormalized state in Equation \eqref{eqn:unnorm-one-d-disc-q-samp}.
\end{theorem}
\begin{proof} 

If $N \geq 3(b-a)\max_{x\in[a,b]}\lvert p'\rvert$, then
\begin{align*}
\frac{1}{\mathcal{Z}_p^2} &\leq \frac{1}{\mathbb{P}([a,b]) - \frac{(b-a)\max_{x\in [a,b]}\lvert p' \rvert}{N}}  \\
&\leq  \frac{1}{1- \mathbb{P}([a,b]^c) - \frac{(b-a)\max_{x\in [a,b]}\lvert p' \rvert}{N}}\\
&\leq  \frac{1}{\frac{1}{2} - \frac{(b-a)\max_{x\in [a,b]}\lvert p' \rvert}{N}} \\
&=\mathcal{O}(1).
\end{align*}
which follows from the left-endpoint quadrature error, i.e. $p'$ denotes the first derivative of $p$.

From $\mathcal{Z}_p^2 = \mathcal{O}(1)$ and Lemma~\ref{lem:poly_approx_qet}, we want the error in the polynomial approximation to be $\delta < \frac{\epsilon}{(b-a)^2}$.

This gives a complexity of 
\begin{align*}
    \mathcal{O}\left(\log(N)d_{\epsilon/(b-a)^2}\right).
\end{align*}
Obviously we want $N = \Omega\left((b-a)\max_{[a, b]}\lvert p'\rvert/\epsilon\right)$.
\end{proof}

In all the cases we consider $p'$ will be bounded by problem dependent parameters and $\mathcal{O}\left(\text{poly}\left(1/\epsilon_{\text{trunc}}\right)\right)$. So based on an earlier discussion, this implies we can load all primitives to the necessary error $\epsilon$ with $\mathcal{O}\left(\text{poly}\log(1/\epsilon)\right)$
one- and two-qubit gates.  This suffices for efficient distribution loading, as per our discussion in Section~\ref{sec:error_analysis_for_quant_deriv_pricing}.

In some cases, it will be challenging to make use of the QET framework for state preparation. This is because we either have poor polynomial approximations or need to perform a different polynomial approximation conditioned on some register. The later results in the need to perform coherent QET angle finding, which is expensive \cite{Gily_n_2019}. In these cases, we opt to utilize standard Grover black-box state preparation \cite{grover2000synthesis, Sanders_2019}. The following combines the estimate on the filling fraction from the previous theorem with the guarantees of black-box state prep (Lemma~\ref{lem:black-box-state-prep}). This procedure utilizes a $\delta$-accurate binary oracle $O_p$ for $p$, which performs $|x\rangle|0\rangle \rightarrow |x\rangle|\widetilde{p}(x)\rangle$ such that $\sup_{x \in [a, b]}\lvert \widetilde{p}(x) - p(x) \rvert < \delta$. 

\begin{theorem}[Black-box Loading for Truncated Probability distributions]
\label{thm:prob_state_prep_guarantee_black_box}
Suppose $p: \mathbb{R} \rightarrow \mathbb{R}$ is a probability density. Consider the pdf truncated to the domain $[a, b]$ where $a, b$ are chosen such that $\mathbb{P}([a, b]^{c}) < \frac{1}{2}$. Suppose we have a $\delta$-accurate binary oracle $O_p$ for ${p}$ and that we know  $\Lambda \geq \max_{[a, b]} {\lvert \widetilde{p} \rvert}$. Lastly, take $N = \Omega \left( 3(b-a)\max_{x\in[a,b]}p'/\epsilon\right)$ for the input accuracy. 

If $\delta = \mathcal{O}\left(\frac{\epsilon^2}{\Lambda^3(b-a)^4}\right)$,  then we can construct a state with a measurement distribution in the computational basis that is at most $\epsilon$ in TVD from that of  \eqref{eqn:one-d-disc-q-samp} using \begin{align*}
    \mathcal{O}\left(1\right)
\end{align*}
queries to $O_p$ and 
\begin{align*}
    \mathcal{O}\left(\log^3(\Lambda/\delta)\right)
\end{align*}
 one- and two-qubit gates. The state also satisfies
\begin{align*}
    \lvert \mathbb{P}([a, b]) - \mathcal{Z}_p^2 \rvert = \mathcal{O}\left(\epsilon\right).
\end{align*}

Hence the output state will be $\mathcal{O}\left(\epsilon\right)$ ``TVD'' from the unormalized state in Equation \eqref{eqn:unnorm-one-d-disc-q-samp}.
\end{theorem}
The proof follows from Lemma~\ref{lem:black-box-state-prep} combined with the same proof of the previous theorem.

When comparing QET-based state loading (Theorem~\ref{thm:prob_state_prep_guarantee}) and black-box state preparation (Theorem~\ref{thm:prob_state_prep_guarantee_black_box}), one will note that in reality there is only a saving of at most a factor of $\mathcal{O}\left(\log(1/\epsilon)\right)$. While this comparison excludes the arithmetic cost of implementing  the oracle $O_f$, one can usually still use a degree $d$-polynomial approximation for $f$, which costs $\mathcal{O}\left(\log^2(N)\cdot d\right)$ to implement with arithmetic, where $N = \Omega(1/\epsilon)$. It is expected that $d = \Omega(\log(1/\epsilon))$, and hence a saving of only a $\log(1/\epsilon)$ factor when using QET over black-box. Even if the oracle for $f$ can implemented more efficiently, black-box state preparation has the additional overhead of needing to implement $\sin^{-1}$ coherently, which uses $\Omega\left(\log^3(N)\right)$ one- and two-qubit gates. The ability to avoid explicitly computing $\sin^{-1}$ appears to be the main benefit of QET  state prep loading. Specifically, the $\sin^{-1}$ is built into the QET transformation \eqref{eqn:shft_sqrt_p}, and the $\ell_1$ norm of the coefficients of the Taylor series, the cost for QET, for $\sin^{-1}$ is $\mathcal{O}(1)$.

While this additional saving could be beneficial in practice, unfortunately, QET-based state prep is not as versatile as black-box. In many cases, especially in derivative pricing, we want to load conditional distributions. In this case, QET will need to perform a coherent computation of the signal processing angles, which requires arithmetic, and likely cancels-out the advantage over the black-box approach. Still, in some cases, the conditioning operation can be expressed as a simple affine transformation of the coordinates, in which case we can still use QET.

\subsection{Resources For Loading Common Primitives}
\label{sec:resources_for_loading_common}
In light of Theorems~\ref{thm:prob_state_prep_guarantee} and~\ref{thm:prob_state_prep_guarantee_black_box}, the only quantities we need to determine for a given density are bounds on the first derivative, the scaling of $a, b$, and the existence an efficient polynomial approximation if we want to use the arithmetic-free approach. Unfortunately, we will be unable to completely remove coherent arithmetic from the loading procedures for all primitives we consider, which showcases a limitation of  the QET-loading framework. 

We will consider four kinds of primitive distributions that will appear frequently in our applications, along with financial models more broadly: standard Gaussian, $\chi^2$, integral of CIR, and two-dimensional L\'evy areas. Lastly, the case of integral of CIR and L\'evy areas showcase the ability to perform distribution loading when only the characteristic function is known in closed form.

\subsubsection{Loading with Closed-Form Density}

\paragraph{Gaussian Loading}
The first is the standard Gaussian, which was already handled by \cite{mcardle2022quantum}. However, we present a slightly modified version using our new framework. For the Gaussian, we have the following result on the polynomial degree.
\begin{lemma}[Adapted from Corollary 1 \cite{mcardle2022quantum}]
Let $\delta, b> 0$, then there is a degree $d = \mathcal{O}\left(b\log(1/\epsilon)\right)$ degree polynomial $P(y)$ such that for every $y \in [0, \sin(1)]$ we have that
\begin{align*}
    \lvert \exp\left(-\frac{(-b + \frac{4b}{\pi}\arcsin(y))^2}{4}\right) - P(y)\rvert \leq \delta.
\end{align*}
\end{lemma}

Hence using Theorem~\ref{thm:prob_state_prep_guarantee}
\begin{corollary}[Standard Gaussian Loading]
\label{cor:standard_gauss_qsvt}
Let $p$ be the standard Gaussian density, and $N = \Omega\left(b/\epsilon\right)$. Then we can load an $\epsilon$ in $\ell_2$-distance approximation to Equation~\ref{eqn:unnorm-one-d-disc-q-samp} for $a = -b$ using $\mathcal{O}\left(b\log(N)\log(b/\epsilon)\right)$ one- and two-qubit gates.
\end{corollary}

\paragraph{Central \& Non-central $\chi^2$ Loading}
The next primitive is the central $\chi^2_{r}$ distribution, with $r$ degrees of freedom, which has the density
\begin{align*}
    p_{r}(x) = \frac{x^{r/2 -1}e^{-x/2}}{2^{r/2}\Gamma(r/2)},
\end{align*}
with $ x \in [0, \infty)$ and $\lVert p_{r} \rVert_{\infty} \leq 1$ for $r \geq 2$.

\begin{lemma}
\label{lem:central_chi_sqr_poly}
Let $\delta, b >0$, $r \geq 2$, then there is a degree $d = \mathcal{O}\left(b+ r+ \ln(1/\epsilon)\right)$ degree polynomial $P(y)$ such that for every $y \in [0, \sin(1)]$ we have that
\begin{align*}
    \lvert \sqrt{p_r(\frac{2b}{\pi}\arcsin(y))} - P(y)\rvert \leq \delta.
\end{align*}
\end{lemma}
\begin{proof}
For a power series $h$, let $\vertiii{h}_1$  be the sum of the absolute values of $h$'s coefficients. Note that on its entire domain
    \begin{align*}
    \sqrt{p_{r}(x)} = \frac{x^{\frac{r/2-1}{2}}}{\sqrt{\Gamma(r/2)}2^{p/4}}\sum_{k=0}^{\infty}\frac{(-x/4)^{k}}{k!}.
    \end{align*}

We can bound the absolute sum of the coefficients of the power series for $h(x)$:

\begin{align}
\vertiii{h(x)}_{1}&=\vertiii{\sqrt{p_{r}\left( b x\right)}}_{1}\nonumber\\
&\leq \frac{b^{r/4-1/2}}{\sqrt{\Gamma(r/2)}2^{r/4}}\sum_{k=0}^{\infty}\frac{(\frac{\pi b}{2})^{k}}{k!}\nonumber\\
&\leq\frac{b^{\frac{r}{4}-1/2}e^{\frac{b}{2}}}{\sqrt{\Gamma(r/2)}2^{\frac{r}{4}}}\\
&\leq (b/r)^{p}e^{b+r},
\end{align}
where we have used that for $r \geq 2$, $\Gamma(r/2) \geq (\frac{r}{2e})^{r/2}$. The degree follows from Lemma~\ref{lem:arcsin_cost_shift}.
\end{proof}

Hence using Corollary~\ref{thm:prob_state_prep_guarantee}.
\begin{corollary}[Central $\chi^2$ Loading]
\label{cor:central_chi_square_qsvt}
If $p$ is the $\chi_r^2$ distribution with $r \geq 2$ degrees of freedom, and $N = \Omega\left(br/\epsilon\right)$. Then we can load an $\epsilon$ in $\ell_2$-distance approximation to Equation~\ref{eqn:unnorm-one-d-disc-q-samp} with $a =0$ using $\mathcal{O}\left((b+r)\log(N)\log(b/\epsilon)\right)$ one- and two-qubit gates.
\end{corollary}
\begin{proof}
    Note that for $\chi^2$ for $r\geq 2$, $\lVert p' \rVert_{\infty} = \mathcal{O}(r)$. Then we apply Lemma~\ref{lem:central_chi_sqr_poly} and Theorem~\ref{thm:prob_state_prep_guarantee}.
\end{proof}

The non-central $\chi^2$ pdf with $r$ degrees of freedom and non-centrality parameter $\lambda$ is
\begin{align*}
    p(x) = \frac{1}{2}e^{-(x+\lambda)/2}\left(\frac{x}{\lambda}\right)^{r/4-1/2}I_{r/2-1}\left(\sqrt{\lambda x}\right).
\end{align*}

Using additional coherent arithmetic operations, we can prepare a non-central $\chi^2$-square distribution $\chi^2_r(\lambda)$, using the following simple observation:
\begin{align}
\label{eqn:non-central-chi-map}
    \widetilde{X} = Y + (\lambda + Z)^2
\end{align}
is distributed like $\chi_r(\lambda)$ when $Y$ is distributed like $\chi_r^2$ and $Z$ like $\mathcal{N}(0, 1)$. This will be useful for preparing a discrete qsample corresponding to the CIR process (Section~\ref{sec:cir_fast_forwardable}). There, the truncated support of the non-central $\chi^2$ will be expressed in terms of the truncated supports of $Y$ and $Z$.

\subsubsection{Loading with Closed-Form Characteristic Function}

Here, we consider loading pdfs when only knowing the characteristic function in closed-form. It turns out that both of these happen to be conditional distributions, hence making the black-box state preparation more suitable. To be able to load using only the characteristic function, we show that it suffices to load a truncated version of the inverse Fourier transformation. Specifically, for a pdf $f$, we have
\begin{align}
    f(x) = \frac{1}{\pi}\int_{0}^{\infty}\text{Re}[e^{-iax}\Phi(a)da],
\end{align}
where $\Phi$ is the characteristic function of $f$.

\paragraph{Two-dimensional L\'evy Area Loading}

\label{sec:char_func_loading}
Consider Brownian increments $W_1$ and $W_2$ over a unit time step. Let $r^2 := (W_1)^2 + (W_2)^2$, which is distributed as an Exponential random variable with rate $\frac{1}{2}$.
The L\'evy area  of $W_1$ and $W_2$ is the random variable:
\begin{align*}
   A_{(j, k)} := \int_{t}^{t+h}\int_{t}^{u} dW_j(s) dW_k(u)  - \int_{t}^{t+h}\int_{t}^{u} dW_k(s) dW_j(u).
\end{align*}

For a non-unit time step $h$, we simply have that $A_{(j,k)} \rightarrow hA_{(j,k)}$. The conditional characteristic function of the L\'evy area \cite{levy1951wiener, gaines1994random} with respect to two independent Brownian motions $W_1, W_2$ is 
\begin{align}
\label{eqn:cond_char}
    \Phi(x, r) = \mathbb{E}[e^{ixA} | W_1^2 + W_2^2 = r^2]
    = \frac{x}{\sinh(x)}e^{\frac{r^2}{2}-\frac{r^2x}{2\tanh(x)}}.
\end{align}

Thus Fourier transforming and using polar coordinates, the conditional pdf of the L\'evy area is
\begin{align}
\label{eq:conditional-levy-polar}
    f(A | W_1, W_2) = \frac{1}{\pi}\int_{0}^{\infty}\frac{x}{\sinh(x)}e^{\frac{r^2}{2}-\frac{r^2x}{2\tanh(x)}}\cos(Ax) dx.
\end{align}

\begin{lemma}
\label{lem:levy-area-poly}
    The conditional pdf of the two-dimensional L\'evy  can be uniformly approximated to $\epsilon$ additive error using an ${\mathcal{O}}(\log^2(1/\epsilon))$ degree polynomial.
\end{lemma}
\begin{proof}

It is simple to show  that 
\begin{align}
\lvert \frac{1}{\pi}\int_{m}^{\infty}\frac{x}{\sinh(x)}e^{\frac{r^2}{2}-\frac{r^2x}{2\tanh(x)}}\cos(Ax) dx\rvert \leq \lvert \frac{1}{\pi}\int_{m}^{\infty}e^{\frac{r^2}{2}-\frac{r^2x}{2}} dx\rvert 
\end{align}
so $m = \mathcal{O}(\log(1/\epsilon))$ suffices for $\epsilon$ truncation error with $r \geq 1$.  For $r \in (0, 1)$, we can use $\lvert x/\sinh(x) \rvert \leq \frac{2}{1-e^{-2}}e^{-x/e}$ for $x \geq 1               $. This gives an $m = \mathcal{O}(\log(1/\epsilon))$ for small $r$ as well.

By Lemma~\ref{lem:poly_approx_lem} we have that $\frac{x}{\tanh{x}}$ can be approximated to $\epsilon$ additive error using an $\mathcal{O}\left(\log(m/\epsilon)\right)$ degree polynomial. This implies an $\epsilon$ relative error approximation for $e^{-\frac{r^2x}{2\tanh(x)}}$. We can additive-error approximate $e^{-x}$ with an $\mathcal{O}\left(\log(m/\epsilon)\right)$ degree polynomial. Thus we have an $\epsilon$ additive error approximation of $e^{-\frac{r^2x}{2\tanh(x)}}$ with a $\mathcal{O}\left(\log^2(m/\epsilon)\right)$ degree polynomial.

Thus, considering all components of the integrand, except $\cos(Ax)$ expressed as polynomial approximations, we can compute:
\begin{align}
    \int_{0}^{m} x^{k}\cos(Ax) dx.
\end{align}

Recall the lower incomplete Gamma function:
\begin{align*}
    \Gamma_{m}(s + 1) :=  \int_{0}^m te^{-t} dt,
\end{align*}
which can be expressed as the the following power series valid for $m \in \mathbb{C}$ and $s$ not a non-positive integer:
\begin{align*}
\Gamma_{m}(s + 1) = \frac{1}{\Gamma(s)}\sum_{k=0}\frac{(-m)^k}{k!(s+k)},
\end{align*}
which has a truncation index of $\mathcal{O}\left(\log(m/\epsilon)\right)$.

Thus the total degree of the polynomial is $\mathcal{O}\left(\log^2(1/\epsilon)\right)$.
\end{proof}

\begin{corollary}[Two-dimensional L\'evy area Loading]
\label{cor:levy-area-loading}
Suppose that $W_1^2 + W_2^2 = r^2$, and that we want to load the discretized pdf over a grid over $[a, b]$ using $N$ grid points. Suppose that $N = \Omega\left(\frac{(b-a)^5}{\epsilon^2}\right)$.
Then we can prepare a quantum state with a measurement distribution in the computational basis that is $\epsilon$ TVD approximation to that of Equation~\ref{eqn:unnorm-one-d-disc-q-samp} with $p$ set to be the conditional pdf of the L\'evy  area \eqref{eq:conditional-levy-polar}, using at most
\begin{align*}
    \mathcal{O}\left(\log^{2}\left(N\right)\log^{2}\left(N/(b-a)\right) + \log^{3}\left(N\right)\right)
\end{align*}
one- and two-qubit gates and $\mathcal{O}(\log(N))$ qubits in total..
\end{corollary}
\begin{proof}

By dominated convergence, we have that 
\begin{align*}
 \lvert p'\rvert  &= \lvert\frac{1}{\pi}\int_{0}^{\infty} \frac{x^2}{\sinh(x)}e^{\frac{r^2}{2} - \frac{r^2x}{2\tanh{x}}}\sin(xA) dx\rvert \\
 &\leq 1 +  \int_{1}^{\infty} e^{\frac{r^2}{2} -\frac{r^2x}{2}} dx \\
 &=\mathcal{O}(1),
\end{align*}
and similarly $\lVert p \rVert_{\infty} = \mathcal{O}(1)$.

We will be using fixed-point arithmetic with $n$ higher-order bits before the decimal and $m$ lower-order bits for after the decimal. The arithmetic will be done on a uniform grid of size $N$ and so $n + m  = \mathcal{O}(\log_2(N))$.

We want at least $N = \Omega\left((b-a)\max_{[a, b]} p'/\epsilon\right)$ by  Theorem~\ref{thm:prob_state_prep_guarantee_black_box}, so $N = \Omega\left((b-a)/\epsilon\right)$ suffices.  %

From Lemma~\ref{lem:levy-area-poly} and Section~\ref{sec:coherent-arith}, the cost to apply a polynomial via coherent arithmetic to $\epsilon'$ error is $\mathcal{O}\left(\log^2(N)\log^2(1/\epsilon')\right)$ gates. Thus we want $\epsilon' = \mathcal{O}\left(\frac{b-a}{N}\right)$ to maintain the same level of precision. Also from Theorem~\ref{thm:prob_state_prep_guarantee_black_box}, we want
\begin{align*}
    \frac{b-a}{N} < \frac{\epsilon^2}{\Lambda^3(b-a)^4},
\end{align*}

so $N = \Omega\left(\frac{(b-a)^5}{\epsilon^2}\right)$ , where we take $\Lambda = \mathcal{O}(1)$.
\end{proof}

\paragraph{Integral of CIR Loading}
The CIR process, which we recall is the following SDE:
\begin{align*}
    dV_t = \kappa(\theta - V_t)dt + \sigma\sqrt{V_t}dW_t,
\end{align*}
belongs to the family of squared-Bessel processes \cite{broadie2006exact}. We would like to be able to sample from $\int_{t}^{t+2\Delta}V(s) ds$, specifically conditioned on the two endpoints $V(t), V(t+2\Delta)$. The characteristic  function was computed in closed form by Broadie and Kaya \cite{broadie2006exact}:

\begin{align}
\label{eqn:int_cir_char_func}
\mathbb{E}[\exp(ia\int_{t}^{t+2\Delta} V(s) ds) | V(t), V({t+2\Delta})]&=\frac{\gamma_{a}\sinh(\kappa \Delta)}{\kappa \sinh(\gamma_{a}\Delta)}\exp\left(\frac{V
(t) + V({t+2\Delta})}{\sigma^2}\cdot \left(\frac{\kappa}{\tanh(\kappa\Delta)} - \frac{\gamma_{a}}{\tanh(\gamma_{a}\Delta)}\right)\right)\nonumber\\
&\cdot \frac{I_{\xi}\left(\frac{\sqrt{V({t})V({t+2\Delta})}}{\sigma^2}\frac{2\gamma_{a}}{\sinh(\gamma_{a}\Delta)}\right)}{I_{\xi}\left(\frac{\sqrt{V({t})V({t+2\Delta})}}{\sigma^2}\frac{2\kappa}{\sinh(\kappa\Delta)}\right)},
\end{align}
where $\gamma_a = \sqrt{\kappa^2 - 2\sigma^2ia}$. 

Conditioned on the two end-points, we load the inverse Fourier transform
of the above pdf via polynomial approximation and black-box state preparation. The following bound is proven in Appendix \ref{subsec:proof-lem-int-cir-poly}.
\begin{restatable}[]{lemma}{polyApproxIntCir}
\label{lem:poly-approx-int-cir}
The pdf of $\int_{t}^{t+2\Delta}V(s) ds$ conditioned on the endpoints $V(t), V(t+2\Delta)$ can be uniformly approximated to $\epsilon$ additive error using an ${\mathcal{O}}\left((\max_t V(t))^2\log^2(\max_t V(t)/ \epsilon)\right)$ degree polynomial.
\end{restatable}

Similar to the L\'evy  case, we have a state preparation guarantee using the black-box approach and polynomial approximation. The proof is very similar.
\begin{corollary}[Integral of CIR Loading]
\label{cor:int_of_cir_loading_cost}
 Suppose $\max_t V(t)$ is bounded over time, and that we want to load the discretized pdf over $[a, b]$ using $N$ grid points. Assume that $N = \Omega\left(e^{\max_t V(t)}(b-a)/\epsilon\right)$. Then,  we can prepare a quantum state with a measurement distribution in the computational basis that is $\epsilon$ TVD to that of Equation~\ref{eqn:unnorm-one-d-disc-q-samp} with $p$ set to be the condition pdf of the L\'evy  area \eqref{eq:conditional-levy-polar}, using at most
\begin{align*}
&\mathcal{O}\left((\max_t V(t))^4\left[\log^{2}\left(N\right)\log^{2}\left(\max_t V(t)N/(b-a)\right) + \log^{3}\left(N\right)\right]\right)
\end{align*}
one- and two-qubit gates and $\mathcal{O}(\log(N))$ qubits in total.
\end{corollary}
\begin{proof}
Let $p(x | V_{t}, V_{t+2\Delta})$ denote the pdf. It will be apparent from the proof of Lemma~\ref{lem:poly-approx-int-cir} and dominated convergence that $\lVert p'\rVert_{\infty} = \mathcal{O}(e^{\max_t V(t)})$ and $\lVert p\rVert_{\infty} = \mathcal{O}(e^{\max_t V(t)})$. The rest of the proof is basically the same as Theorem~\ref{cor:levy-area-loading}.  

\end{proof}

\section{Quantum Speedups for Derivative Pricing over CIR and Heston}
\label{sec:speedup_heston_cir}

In this section, we present our first results on end-to-end speedups for derivative pricing beyond the Black-Scholes model. These speedups make use of the fast-forwardability (Definition~\ref{def:fastforward}) of the underlying models. First, (Section~\ref{sec:cir_fast_forwardable}) we show that vanilla QMCI can provide an end-to-end quadratic speedup for pricing path-dependent derivatives with Lipschitz payoffs over the Cox-Ingersoll-Ross (CIR) model. The main results are Theorem~\ref{thm:cir_qsample_resources}, which displays the resources for amplitude-encoding the discrete sum approximating the price, and Theorem~\ref{thm:sqrt_discretization}, analyzes the discretization error. Together, these results imply that the overhead from state preparation is only polynomial in the $\log(1/\epsilon)$ and $T$, hence the quadratic sampling reduction from QMCI is retained up to polylog factors in $1/\epsilon$.

In Section~\ref{sec:heston-fast-forwardable}, we show that vanilla QMCI also provides an end-to-end quadratic speedup for pricing over the multi-asset Heston model (Definition~\ref{defn:multi-asset-hest-asset-only}). Here, we also consider path-dependent, Lipchitz payoffs. The main point will be to show something similar to CIR, i.e. the speedup is retained when accounting for the overhead from state preparation and the errors can be controlled. This will result in Theorem~\ref{thm:heston_loading_resource} for the resource analysis of amplitude encoding the discrete sum and Theorem~\ref{thm:hest_discretization} for analyzing the discretization error. The overall asymptotic runtime for pricing CIR with quantum MCI is presented in Theorem \ref{thm:cir-main} and for Heston in Theorem \ref{thm:heston-main}.

\label{sec:pricing_deriv_fast_forwardable}

\subsection{Cox--Ingersoll--Ross Process}
\label{sec:cir_fast_forwardable}

In this section, we perform a complete resource analysis of quantum derivative pricing over  the CIR process, which we now recall.

\CIR*

This model is fast-forwardable in the more general sense of Definition~\ref{def:fastforward}, which enables efficient qsampling from the path process. This analysis will be asymptotic and the key asymptotic quantity of interest will be the number of monitoring points that the payoff depends on, $T$, which represents the dimension of the integration problem. In applications, $T$ is typically a bounded but large parameter, and hence (besides the $1/\epsilon$ dependence) ends up dominating the runtime.

Unfortunately, this analysis will not extend to multiple, coupled CIR models, as the process no longer becomes fast-forwardable. %
The main challenges in the error analysis will be obtaining a bit scaling per primitive distribution that scales like $\mathcal{O}\left(\log(T)\right)$. This is possible due to a careful use of the left-endpoint error bound (Lemma~\ref{lem:left_rule_error}). Additionally, unlike for GBM, the scheme for simulating the CIR is recursive, which complicates the analysis. 

We start by presenting the scheme for fast-forwarding the CIR process and then proceed to bound the resources for loading discrete qsamples and the discretization error. The main results of the section are that we can that we can amplitude encode the discrete sum approximating the integral with $\mathcal{O}\left(\text{poly}(T, \log(1/\epsilon))\right)$ one- and two-qubit gates (Theorem~\ref{thm:cir_qsample_resources}), and discretize the integration problem to roughly $\mathcal{O}\left(T\log(\frac{BT}{\epsilon_{\textup{disc}}\epsilon_{\textup{trunc}}})\right)$ bits in total for pricing a piecewise linear payoff with a maximum slope of $B$ (Theorem~\ref{eqn:cir_discretization_bound}).

\subsubsection{Quantum Fast-forwarding Scheme}
CIR falls into the category of square-root diffusion processes and is fast-forwardable according to Definition~\ref{def:fastforward}. It is known for CIR that given $V({t})$ then $V({t+\Delta})$ is distributed as a non-central $\chi^2$ \cite{glasserman2004monte}:
\begin{align*}
    V({t+\Delta})~|~V(t)  =_{d} c\left(Y({t})  + \left(\sqrt{\beta V({t})}+Z({t})\right)^2\right),
\end{align*}
where $Y(t) \sim \chi^{2}_{\eta-1}$ and $Z(t) \sim \mathcal{N}(0, 1)$
with
\begin{align*}
&\beta = \frac{4\kappa e^{-\kappa\Delta}}{\sigma^2(1-e^{-\kappa\Delta})},\\
& \eta =  \frac{4\theta\kappa}{\sigma^2},\\
& c = \frac{\sigma^2(1- e^{-\kappa\Delta})}{4\kappa},\\
&\gamma : = \sqrt{c\beta} = e^{-\kappa\Delta/2}.
\end{align*}
Also we define
\begin{align*}
\xi = \frac{\eta}{2} - 1,
\end{align*}
which we call the \emph{Feller gap}. This is due to the Feller condition which assures that almost-surely $V(t) > 0$ when $\xi > 0$ \cite{glasserman2004monte}.  However, to make the distribution loading process easier, we assume at least $\eta \geq 5$, which avoids singularities in $\chi^2_{\eta-1}$ and its derivative. Note  that this is stronger than the Feller condition which only requires $\eta > 2$. Since we only know the transition density for a single CIR in closed-form, it is unclear how to extend this to multiple correlated CIR processes. 

We can construct a $T$-length path process $\widehat{V}$ in the following way. We have a deterministic initial condition $\widehat{V}(0) = V(0) = v_0$. The increment process is formed by $I(t) = (Y(t), Z(t))$, where  $Y(t) \sim \chi^{2}_{\eta-1}$ and $Z(t) \sim \mathcal{N}(0, 1)$ are sampled independently. Lastly the time-homogeneous, transition function is 
\begin{align}
\label{eqn:cir_recursion}
    g((y, z), v) = c\left(y + \left(\sqrt{\beta v} + z\right)^2\right),
\end{align}
where $c$, $\beta$ and $\eta$ are all functions of the time-increment $\Delta$.

Since the path increment can clearly be implemented efficiently classically, i.e. following Definition~\ref{def:fastforward}, the process is fast-forwardable. However, $\widehat{V}(t) - \widehat{V}(t-1)$ is not independent of the history, and hence it does not following Definition~\ref{def:ff-indep}. This implies that the path must be constructed in a recursive manner and poses additional challenges not present in the GBM case. The recursive structure leads us to define:
\begin{align*}
    v_t = g_t(\vec{y}, \vec{z}) := g((y_{t-1}, z_{t-1}), g((y_{t-2}, z_{t-2}, g(\cdots, g((y_0, z_0), v_0))))).
\end{align*}
The dependence structure also makes it unclear how to generalize the analysis to the multi-dimensional setting, i.e. multiple, coupled CIRs.

With this notation, the derivative pricing  task associated with payoff $f : \mathbb{R}^T \rightarrow \mathbb{R}$ can be defined as computing
\begin{align*}
    \mathbb{E}[f(\widehat{V}) | \widehat{V}(0) = v_0] = \int_{\mathbb{R}^{2T}} f\left(g_{T-1}(\vec{y}, \vec{z}), \dots, g_1(\vec{y}, \vec{z}), v_0\right) p_{\chi}(\vec{y})p_{\mathcal{N}}(\vec{z})d\vec{y}d\vec{z},
\end{align*}
where
\begin{align*}
&p_{\mathcal{N}}(\vec{z}) = \prod_{t=}^{T-1}p_{\mathcal{N}(0,1)}(z_t)\\
&p_{\chi}(\vec{y}) = \prod_{t=0}^{T-1}p_{\chi_{\eta-1}}(y_t).
\end{align*}

To implement the fast-forwarding scheme quantumly, we will utilize the standard Gaussian and central $\chi^2_{\eta-1}$ as the primitive distributions for loading path increments for the CIR process. To load onto a digital device, we will need to truncate the densities for the standard Gaussian and chi-square to bounded intervals $[-a, a]$ and $[b_L ,b_U]$, respectively. As mentioned in Section~\ref{sec:error_analysis_for_quant_deriv_pricing}, this introduces an error in the price denoted $\epsilon_{\text{trunc}}$. Furthermore, we will need to replace the integral with a Riemann sum, which introduces another error $\epsilon_{\text{disc}}$. The various impacts that these errors have on each other was discussed in Section~\ref{sec:error_analysis_for_quant_deriv_pricing}.

The following theorem (proven in Appendix \ref{subsec:proof-cir-load}) provides the cost guarantees for preparing quantum samples corresponding to scheme for CIR described above. %

\begin{restatable}[CIR Discrete-Sum Loading]{theorem}{cirDiscreteSumLoad}
\label{thm:cir_qsample_resources}
Suppose we utilize $N$ grid points from $[-b, b]$  per primitive distribution and that the payoff costs $\mathcal{N}_f$ gates to evaluate. If $N = \Omega\left(BT^3b^3r/\epsilon_{\text{distr}}\right)$,  then we can amplitude encode the discretized and truncated price of a derivative  with a $B$-Lipschitz $f$ payoff over a CIR  to $\epsilon_{\text{distr}}$ error using
\begin{align*}
\mathcal{O}\left(T^2 \log^2(N)\log(T) + T(b+r)\log(N)\log(BTb/\epsilon) + \log^3(BTbN) + \mathcal{N}_f \right)
\end{align*}
one- and two-qubit gates. The state is over $\mathcal{O}(T\log(N))$ qubits and  uses $\mathcal{O}(T\log(BTN))$ qubits in total to prepare.
\end{restatable}

The value for $N$ will be provided by the discretization error,  bounded in  the next section.

\subsubsection{Discretization Analysis}

The following computes the asymptotics of the required number of (qu)bits for piecewise-linear payoffs.
\begin{restatable}[CIR Discretization Error]{theorem}{cirDiscreteError}
\label{thm:sqrt_discretization}
 If $f$ is piecewise linear and $B$-Lipschitz, then it suffices to choose
$$\mathcal{O}\left(\log(\frac{\eta BT}{\epsilon_{\textup{disc}}\epsilon_{\textup{trunc}}})\right)$$
bits per distribution  to achieve an absolute discretization error of $\epsilon_{\textup{disc}}$.
\end{restatable}

We present a sketch of the proof to highlight the main ideas, leaving the full details to Appendix \ref{subsec:proof-thm-cir-discr}. The main technical challenges are handling the recursive nature of Equation \eqref{eqn:cir_recursion} and ensuring that the number of required bits remains logarithmic in all parameters. 

According to Lemma~\ref{lem:sqrt_diff_truncation}, if we truncate each $\chi^2_r$ and Gaussian to $b_{{U}} = \mathcal{O}\left(r + \log(BT/\epsilon_{\text{trunc}})\right)$ and $a =  \mathcal{O}\left(\sqrt{\log(BT/\epsilon_{\text{trunc}})}\right)$, respectively, then
\begin{align*}
    \lvert \int_{\mathbb{R}^{2T}}g(\vec{y}, \vec{z}, x_0)p_{\chi}(\vec{y})p_{\mathcal{N}}(\vec{z})d\vec{y}d\vec{z} - \int_{[b_L, b_{{U}}]^{T} \times [-a, a]^{T}}g(\vec{y}, \vec{z}, x_0)p_{\chi}(\vec{y})p_{\mathcal{N}}(\vec{z})d\vec{y}d\vec{z}\rvert < \epsilon_{\text{trunc}}.
\end{align*}
For reasons apparent later, we will also want to leave out a small region around zero of $(0, b_{{L}}) = (0, \frac{\epsilon_{\text{trunc}}}{BT^2r})$, which only introduces an additional $\mathcal{O}\left(\epsilon_{\text{trunc}}\right)$ error.

As shown in Section~\ref{sec:error_analysis_for_quant_deriv_pricing} and more formally in Lemma~\ref{lem:left_rule_error} the error for a left-endpoint Riemann sum $\mathcal{R}_N$ is
\begin{align}
\label{eqn:cir_discretization_bound}
&\lvert \mathcal{R}_N -\int_{[b_{{L}}, b_{{U}}]^{T} \times [-a, a]^{T}}g(\vec{y}, \vec{z}, x_0)p_{\chi}(\vec{y})p_{\mathcal{N}}(\vec{z})d\vec{y}d\vec{z}\rvert \nonumber\\ &\leq  \left[\sum_{(\vec{y}, \vec{z}) \in \mathcal{M}} \sup_{(\vec{y}', \vec{z}') \in C_{\vec{x}}}\lVert \nabla [g\cdot p_{\mathcal{N}}\cdot p_{\chi_r}](\vec{y}', \vec{z}')\rVert_{\infty}\delta^{2T}\right] \frac{T(b_{{U}} + 2a)}{N},
\end{align}
where we will pick the spacing $\delta$ uniformly in dimension. If we can show that the term in the square-brackets is at most $\mathcal{O}\left(\text{poly}(T, b_{\mathcal{U}})\right)$, then we will have that the number of bits is logarithmic in all parameters.

Let $\vec{c}(\vec{x})$ map $\vec{x}$ to some point in $\mathcal{C}_{\vec{x}}$ for all $\vec{x} \in \mathcal{M}$. It can be shown that the above is true if for some mapping $\vec{c}(\cdot)$:
\begin{align*}
    \sum_{(\vec{y}, \vec{z})  \in c(\mathcal{M})}p_{\chi}(\vec{y})p_{\mathcal{N}}(\vec{z})\delta^{2T} = \mathcal{O}\left(\text{poly}(T, b_{{U}})\right)
\end{align*}
for sufficiently small $\delta$. The points in $\vec{c}(\mathcal{M}) := \{ \vec{c}(\vec{x}) : \vec{x} \in \mathcal{M}\}$ will correspond to the points that attain the supremums shown above. The challenge will be showing that we can shift the outputs of $\vec{c}(\cdot)$ to the infimums in each cell instead. If in the above sum $\vec{c}(\mathcal{M})$ contained grid points corresponding to the infimum of the terms in the sum above, then the overall sum would be bounded by the integral, which is one.

To do this, we need to lower bound the $\chi^2$ random variable on the grid, hence the reason to truncate an interval near zero. The recursion that appears from taking the gradient of $g$ have a nice uniform bound. Hence, the main relative change of the integrand under $\vec{c}(\cdot)$ will be to  the $\chi^2$. 
The relative change in the $\chi^2$ pdf is  bounded by (using Lemma~\ref{lem:sqrt_diff_truncation})
\begin{align*}
\lvert\left(\prod_{j=0}^{T-2} \frac{y_j}{x_j}\right)^{r/2-1}e^{-(\lVert \vec{y}\rVert_{1} - \lVert \vec{x}\rVert_{1}) /2}\rvert = \mathcal{O}(e^{T^3r^2\delta/\epsilon_{\text{trunc}}}),
\end{align*}
where we consider any $\vec{x} \in \mathcal{B}_{\infty}(\vec{y}, \delta)$. Thus the above is $\mathcal{O}\left(1\right)$ with $\delta = \mathcal{O}(\epsilon_{\text{trunc}}/r^2T^3)$, recall the number of bits is log in $\frac{R}{\delta}$, for $R = \max(b_U, a)$. We can then freely shift the evaluation point to be the infimum. The overall result then follows.

We emphasize the importance of only bounding the gradient over each cell as opposed to uniformly over the grid. If we uniformly bounded the gradient over the grid, then we would obtain that $N$ needs to scale at least linearly with $T$. Specifically, the right-hand side of Equation \eqref{eqn:cir_discretization_bound} would be $\mathcal{O}\left(\frac{(2R)^{2T}}{N}\right)$. The technique we have shown above  enables for absorbing the volume element in each cell into a bound that is polynomial in $T$ and $R$.

The above also showcases the way in which two sources of error can impact each other, i.e. discretization and truncation.

Given our analysis of discretization and distribution loading errors, combined with Section \ref{sec:framework_for_analyze}, the following theorem is self-evident.
\begin{theorem}[Quantum Speedup for Pricing Over CIR Model]
\label{thm:cir-main}
Suppose $(f, \widehat{V})$ is a Financial Derivative Model (Definition~\ref{def:fin_deriv_model}), where is $f$ be a piecewise linear, $B$-Lipschitz payoff function and $\widehat{V}$ is a CIR path-process  initialized at $v_0$ evolved over $T$ time steps. Then, there is a quantum algorithm using $\widetilde{\mathcal{O}}\left(\textup{poly}(T, B)/\epsilon\right)$ one- and two-qubit gates, $\mathcal{O}(\textup{poly}(T))$ qubits, and outputs an $\epsilon$-additive estimate of $\mathbb{E}[f(\widehat{V}) | \widehat{V}(0) = v_0]$ with constant probability.
\end{theorem}

\subsection{Heston's Stochastic Volatility}

\label{sec:heston-fast-forwardable}

In this section, we perform the error analysis for QMCI applied to the following fast-forwardable version of the multi-dimensional Heston model:
\multiassethest

The goal, like in the previous section, is to price a financial derivative with piecewise linear payoff $f$ depending on $T$ monitored prices of $d$ assets:
\begin{align}
\label{eqn:heston_price_path}
   \widehat{S} := (\vec{S}(0), \vec{S}({\Delta}), \dots, \vec{S}({(T-1)\Delta})), 
\end{align}
forming a path process. However, unlike earlier, we now have an additional process $V_k(t)$, $k \in [d]$ coupled to each component of the main price process $S_k(t)$. The stochastic process $\vec{V}(t)$ follows $d$ uncoupled CIR models and determines the instantaneous, random volatility of each of the geometrically-evolving prices.
Recall that conditioned on $V({t})$, $V({t+\Delta})$ is $\chi^2$ distributed and thus has only subexponential tails. Since the price evolves geometrically, its average at any time point is related to the mean of the exponential of a subexponential process. This implies that the Heston model has only polynomial tails. However, unlike GBM, the log price process is not subGaussian, i.e. does not have tails that fall at least as fast as those of a Gaussian.
This means that the parameters of the model and other constants can play a significant role in determining the ability to truncate and the existence of moments. This issue has been termed the ``moment explosion'' problem in quantitative finance literature \cite{kellerressel2008momentexplosionslongtermbehavior}.

The classical MCI and QMCI both, as usually presented, depend on the second moment of the random variable being integrated. Interestingly, there does exist an alternative analysis of classical MCI that applies for processes whose second-moment is infinite but $(1+\delta)$-th moment is finite. Specifically Ref.~\cite{blanchet2024quadraticspeedupinfinitevariance} analyzed quantum MCI applied to this setting. The analysis works by using the $(1+\delta)$-th moment to construct a truncated process whose second-moment exists. 

The truncated process mentioned above can be constructed when $C_{1+\delta} :=\mathbb{E}[S(t)^{1+\delta}]^{1/(1+\delta)} < \infty$. Then, if we truncate $S(t)$ to $\left(\frac{8C_{1+\delta}}{\epsilon}\right)^{1/\delta}$ we get a bounded random variable $\widetilde{S}(t)$ with $\mathbb{E}[\widetilde{S}(t)^2] \leq \left(\frac{8}{\epsilon}\right)^{1/\delta-1}C_{1+\delta}^{1/\delta}$, and $\lvert \mathbb{E}[ S(t) - \widetilde{S}(t)]\rvert = \mathcal{O}\left(\epsilon\right)$. This implies it suffices to truncate outside $\log S(t) \leq \frac{1}{\delta}\log(C_{1+\delta}/\epsilon)$. The multidimensional case introduces an additional $(Td)^{1/\delta}$ factor in the truncation bound. Thus if one could directly load qsamples from $S(t)$, then this would provide a sufficient truncation bound. This appears to be possible when the assets are uncorrelated and in the path-independent case. Specifically, the characteristic function of $S(t)$ is known in closed form \cite{albrecher2007little} and one could  use polynomial approximation to load the inverse Fourier transform of the characteristic function. However, the path-independent setting is exactly the case where we can classically utilize Fourier-transform-based pricing \cite{carr1999option}.

In the multi-dimensional, path-dependent setting it appears unclear how to make effective use of the above truncation bound. The path process in Equation \eqref{eqn:heston_price_path} for the model in Definition~\ref{defn:multi-asset-hest-asset-only} needs to be constructed from three kinds of increments:
\begin{align*}
& \widehat{{V}}({t+1}) = g_1(\widehat{{V}}({t}), {Z}, {Y})\\
&\widehat{{S}}({t+1}) = g_2(\widehat{{S}}({t}),  \widehat{{V}}({t+1}), \widehat{{V}}({t}), {X}, {W}),
\end{align*}
where $g_1$ is the previously-presented transition function for  the CIR process and $g_2$ solves Equation \eqref{eqn:heston-increment}. The random four-tuple $({Z}, {W}, {Y}, {X})$ form the independent increments of the coupled processes. The first two are standard Gaussians, the second is $\chi^2$ and the last is distributed as a functional of a CIR process (i.e. $\int_{t}^{t+1}V_sds$). As for the GBM and CIR cases, we need to determine how to truncate the increments. Unfortunately, it is not clear how to translate a truncation bound like $\log S_t \leq \frac{1}{\delta}\log(C_{1+\delta}/\epsilon)$ into bounds on the increments.

To highlight the issue mentioned in the previous paragraph, we present a simple example using a linear combination of i.i.d. standard Gaussians $Y = \sum_{k=1}^{d} a_k Z_k$, weighted by $\vec{a}$. We know that $Y$ has the same distribution as $Z' \sim \mathcal{N}(0, \lVert \vec{a} \rVert_{2}^2)$. Hence, through a change of variables one could convert $\int_{[-M, M]^{d}} h(Y) \mu(d\mathbf{Z})$ over measure $\mu$ into a one-dimensional integral $\int_{\alpha} h(Z') \nu(dZ')$ over measure $\nu$. It is then clear that if $\lvert Z' \rvert \geq \alpha$ for $\int_{Z'\geq \alpha} h(Z') \nu(dZ') = \mathcal{O}(\epsilon)$, then it suffices to take $\lvert \sum_{k=1}^{d} a_k Z_k\rvert \geq \alpha$. We would like to turn this into a statement about $(Z_1, \dots, Z_d)$ needing to lie outside some $\ell_p$ ball. However, since $Z_k$ can be negative, utilizing upper bounds can lead us to truncating too much of $(Z_1, \dots, Z_d)$. Hence, we no longer have the same $\mathcal{O}(\epsilon)$ guarantees on the multi-dimensional integral. This is even more of a problem when $\vec{a}$ can be random, like in the Heston model. As a result, it appears necessary to perform a direct truncation analysis on the multi-dimensional integral over the increments.

 The truncation result that we get actually has a few additional conditions on the model parameters $\kappa, \Delta, \sigma, \rho$.
Additionally, we will need to truncate a small ball of $\vec{X}$ around zero when doing the discretization analysis (see Lemma~\ref{lem:trunc_around_zero_hest}). This lower-endpoint truncation does not introduce any additional constraints on the parameters.

In the following subsection, we demonstrate how to efficiently prepare sufficient qsampling access to the model in Definition~\ref{defn:multi-asset-hest-asset-only}. Specifically, we show in Theorem~\ref{thm:heston_loading_resource} that this can be done with $\mathcal{O}\left(\text{poly}(T, d, \log(1/\epsilon))\right)$ one- and two-qubit gates. Then in the subsequent subsection, we determine the asymptotics of a sufficient region of truncation for $\mathcal{O}(\epsilon)$ error (providing a proof sketch), and the number of (qu)bits for suppressing the discretization error. Our main result on the discretization error, that follows from the truncation analysis, will be that it suffices to use
\begin{align*}
\mathcal{O}\left(T^4d^2\log^{3/2}(Td/\epsilon_{\text{disc}})\log(Td/\epsilon_{\text{disc}})\right)
\end{align*} (qu)bits for an $\epsilon_{\text{disc}}$  additive discretization error. We believe that this result is overly-pessimistic. The main challenge in improving to $\mathcal{O}(dT\log(dT/\epsilon_{\text{discr}}))$ dependence for relative error is that we do not have a closed form expression for the pdf of the integral of a CIR.

\subsubsection{Quantum Fast-forwarding Scheme}

The reformulation of the standard Heston model shown in Definition~\ref{defn:multi-asset-hest-asset-only} was proposed by \cite{andersen2007efficient} and enables efficient simulation. Specifically, it reveals that the above multi-dimensional model is in fact fast forwardable, conditioned on being able to sample from $\int_{0}^{t} \vec{V}(s) ds$.
Ref.~\cite{broadie2006exact} computed the conditional characteristic function of this quantity, which we can then use to approximately sample. 
For a single dimension, the characteristic function for  $\int_{t}^{t+2\Delta} V(u)dt$ conditioned on $V({t}), V({t+2\Delta})$ is known in closed form (Equation \eqref{eqn:int_cir_char_func}). In Section~\ref{sec:char_func_loading}, we showed how to load an approximate qsample from  $\int_{t}^{t+2\Delta} V(u)dt | V({t}), V({t+2\Delta})$. The rest of the simulation procedure relies on coherent arithmetic and the procedure for fast-forwarding CIR (Section~\ref{sec:cir_fast_forwardable}).

Recall that $\vec{V}$ is composed of $d$ independent CIR processes,
\begin{align*}
    dV_{k}(t) = \kappa_{k}(\theta_{k}- V_{k}(t)) dt+ \sigma_{k}\sqrt{V_{k}(t)}dW^{(V)}_{k}(t).
\end{align*}
We can utilize the approach of Section~\ref{sec:cir_fast_forwardable} to perform  simulation. Specifically, we have that
\begin{align}
    \widehat{V}_{k}({t+1}) | \widehat{V}_{k}({t})   =_{d} c_k\left(Y_{k}({t})  + \left(\sqrt{\beta_k V_{k}(t)} +Z_{k}({t})\right)^2\right),
\end{align}
where $\forall t, Y_{k}(t) \sim \chi^{2}_{\eta_k-1}$ and $Z_{k}(t) \sim \mathcal{N}(0, 1)$
with
\begin{align}
&\beta_k = \frac{4\kappa_k e^{-2\kappa_k\Delta}}{\sigma^2(1-e^{-2\kappa_k\Delta})},\\
& \eta_k =  \frac{4\theta_k\kappa_k}{\sigma_k^2},\\
& c_k = \frac{\sigma_k^2(1- e^{-2\kappa_k\Delta})}{4\kappa_k},\\
&\xi_k = \frac{\eta_k}{2} - 1.
\end{align}
We refer to $Y : [T+1] \rightarrow \mathbb{R}^{d}$ as a $\chi^2$ process and the $Z : [T+1] \rightarrow \mathbb{R}^d$ as a standard Gaussian process.

Let $\vec{X} : [T-1] \rightarrow \mathbb{R}^{d}$ be a discrete-time, stochastic process with components that are independent given the CIR path process $\widehat{V}_k$ and sampled according to
\begin{align*}
    X_k(t) | \widehat{V}_k =_d \int_{t}^{t+1} V_k(2\Delta s) ds,
\end{align*}
which has a pdf that is the inverse Fourier transform of Equation \eqref{eqn:int_cir_char_func}. 
For the discrete-time processes mentioned above, we will use the notation $\vec{X}_k \in \mathbb{R}^T$ to denote all increments for the  $k$-th asset.

Lastly, we define the discrete-time process $\vec{U} : [T] \rightarrow \mathbb{R}^d$:
\begin{align}
\label{eqn:log-increment-heston}
U_{k}(t) &= 2\Delta\mu_kt -2\kappa_k\theta_k\Delta\frac{\rho_k}{\sigma_k}\\
\label{eqn:v_part_of_u_incr}&+
\frac{\rho_k}{\sigma_k}(g_{t+1}(\vec{Y}_k, \vec{Z}_{k})  - g_{t}(\vec{Y}_{k}, \vec{Z}_{k}))\\ &+ \left(\frac{2\kappa_k\Delta\rho_k}{\sigma_k} -\Delta\right)X_{k}(t) \\ &+ \sqrt{2\Delta(1-\rho_k^2)X_{k}(t)}\cdot (\mathbf{A} \vec{W}({t}))_{k},
\end{align}
where $\vec{W}$ is another standard Gaussian process independent of all other processes, and $\mathbf{A} \in \mathbb{R}^{d\times d}$ is the Cholesky factor for the correlation matrix $\mathbf{C}$ from Definition~\ref{defn:multi-asset-hest-asset-only}.

Thus
\begin{align*}
    \widehat{S}_{k}(t) = \widehat{S}_{k}({0})\exp\left(\sum_{r=0}^{t-2}U_k(t)\right),
\end{align*}
forms a path process for the price process in Definition~\ref{defn:multi-asset-hest-asset-only}. It is then evident, using the distribution loading procedures presented in Section~\ref{sec:subroutines_for_distribution_loading}, that this presents a fast-forwarding scheme for the Heston model, when there are only asset-asset correlations.

The path increments are clearly the discrete-time processes $\vec{Y}, \vec{Z}, \vec{X}, \vec{W}$ and the transition function is 
\begin{align}
\label{eqn:heston-path-construction}
    h_{t}(\vec{y}, \vec{z}, \vec{x}, \vec{w}) :=h((\vec{y}, \vec{z}, \vec{x}, \vec{w}), \widehat{S}_k(t-1)) = \widehat{S}_k(t-1)\exp\left(\sum_{r=0}^{t}u((\vec{y}, \vec{z}, \vec{x}, \vec{w}), t)\right),
\end{align}
where $u(\cdot, t)$ computes~\ref{eqn:log-increment-heston} for the $t$-th step. The derivative pricing task with payoff $f$ can thus be phrased as computing
\begin{align}
\label{eqn:heston_deriv_price}
    \mathbb{E}[f(\widehat{S}) | \widehat{S}(0) = s_0, \widehat{V}(0) = v_0] = \int_{\mathbb{R}^{2Td} \times \mathbb{R}^{2(T-1)d}} f(h_{T-1}, \dots, s_0)p_{\mathcal{N}}(\vec{z})p_{\chi}(\vec{y})p_{X}(\vec{x} | \vec{z}, \vec{y})p_{\mathcal{N}}(\vec{w})d\vec{z}d\vec{y}d\vec{x}d\vec{w}.
\end{align}

We have the following guarantee for amplitude-encoding a discrete-sum approximating the above integral.
\begin{restatable}[Heston Discrete-Sum Loading]{theorem}{hestDiscrSumLoading}
\label{thm:heston_loading_resource}
Suppose we utilize $N$ grid points per primitive distribution to $[-b, b]$ and that the payoff costs $\mathcal{N}_f$ gates to evaluate. If $N = \Omega\left(Td\cdot BTdre^{dT^2b^{3/2}}/\epsilon_{\text{distr}}\right)$, then we can amplitude encode the discretized and truncated price of a derivative with a payoff $B$-Lipschitz $f$ over the Heston model (\ref{defn:multi-asset-hest-asset-only}) using
\begin{align*}
\mathcal{O}\left(\textup{poly}(T, d, b, \log(B)) + \mathcal{N}_f\right)
\end{align*}
one- and two-qubit gates. The state is over $\mathcal{O}
\left(dT\log(N)\right)$ qubits  and $\mathcal{O}(\textup{poly}(d, T, b, \log(B)))$ qubits are used in total. 
\end{restatable}
The actual polynomial scaling in terms of $T$ and $d$ can be found in the proof. Since everything is being done to additive error the polynomial dependence can be quite large. However, the point is that the overhead is only polynomial in $d$ and $T$.

\subsubsection{Discretization Error Analysis}

The next lemma determines sufficiently sized $\ell_{\infty}$ balls to truncate $\vec{X}, \vec{W}$, $\vec{Y}$, and $\vec{Z}$ to ensure a truncation error of $\epsilon_{\text{trunc}}$. The goal is for there to be a $\mathcal{O}\left(\text{poly}( T, d, \log(1/\epsilon))\right)$ dependence on the point of truncation, where the asymptotic quantities are $T, d $ and $1/\epsilon$. The reason for this is that, due to the exponential nature of the process, the number of bits used for discretization will end up depending on directly on the point of truncation. 

It turns out that we need some conditions on the parameters of the process to ensure that our upper bound is even finite. This comes from that fact that we are integrating the exponential of a subexponential process. The Heston model is known to experience so called ``moment explosion'' (mostly for moments greater than one) issues for certain parameter regimes, payoff functions (especially super-linear payoffs), and simulation times  \cite{andersen2007moment}.

\begin{restatable}[Heston Truncation Error]{lemma}{hestTruncLemma}
\label{lem:hest_truncation_error}
Consider a multi-dimensional Heston model in Definition~\ref{defn:multi-asset-hest-asset-only} evolving to a time $T$.
Suppose the model parameters (for each asset) satisfy the following conditions:
\begin{align*}
&\Delta \geq 1\\
&e^{-\kappa\Delta/2} + e^{-\kappa\Delta}< 1\\
& \kappa^2 > 2\kappa\sigma\rho - \rho^2\sigma^2 \geq 0\\
& \frac{1 + e
^{-2\kappa\Delta}}{4} + \frac{\rho\sigma}{\kappa} \frac{1 - e
^{-2\kappa\Delta}}{4} < \frac{1}{2(1+e^{-\kappa\Delta/2})}.
\end{align*}
Suppose the payoff is piecewise-linear with maximum slope $B$. Then to achieve an overall error of $\epsilon_{\text{trunc}}$ in Equation \eqref{eqn:heston_deriv_price}, it suffices to truncate $\vec{w} \in \mathbb{R}^{Td}, \vec{y},\vec{z} \in \mathbb{R}^{T},  \vec{x} \in \mathbb{R}^{T-1}$ to $\ell_{\infty}$ balls of size $\mathcal{O}\left((T^{3/2}\sqrt{d} + T^2)\log(Td/\epsilon_{\text{trunc}})\right)$.
\end{restatable}
We provide the following proof sketch, leaving the complete proof to Appendix \ref{subsec:proofOfLem61}. The goal will be to highlight the reasons for the parameter constraints.

Given that the payoff is piecewise linear with maximum slope $B$, we can bound the price by the sum of the means of the price at all time points. Then, we will need to scale the error down by two different scaling factors. The first is $Td$ accounting for the total points the payoff depends on. The second is $e^{\mu t} = \mathcal{O}(e^T)$, which comes from the drift term in \eqref{eqn:heston-increment}.

Hence we will focus on the truncation error for
\begin{align}
\label{eqn:heston_price_bound}
&\int s(t) p_{\mathcal{N}}(\vec{z})p_{\chi}(\vec{y})p_{X}(\vec{x} | \vec{z}, \vec{y})p_{\mathcal{N}}(\vec{w}) \nonumber\\&\propto \int e^{\frac{\rho}{\sigma}v_{t+1}+ \left(\frac{2\kappa\Delta\rho}{\sigma} -\Delta\right)\lVert\vec{x}\rVert_1 + \sum_t\sqrt{2\Delta(1-\rho^2)x_t}\cdot (\mathbf{A}_{k, \star}  \vec{w}_{t})} p(\vec{x} | v_{t+1}, v_0) p(v_{t+1} | v_0)p_{\mathcal{N}}(\vec{w}),
\end{align}
which follows from a simple change of variables. The above can be grouped as a triple integral over three kinds of variables: Gaussians $\vec{w}$, the CIR variables $\vec{v}$, and the integral over CIR $\vec{x}$. The truncation error can be computed via union bound on all groups. We first estimate truncation regions in terms of $\ell_2$ and $\ell_1$ norms, which give an overestimate of the amount of truncation allowed when phrased in terms of $\ell_{\infty}$.

We start by truncating the Gaussian integral over $\vec{w}$ outside of an $\alpha_2$-radius,  $\ell_2$-ball, i.e. $\lVert \vec{w}\rVert_2 \geq \alpha_2$. After some manipulation and using standard Gaussian concentration bounds, we get a bound that is a function of $\vec{x}$, namely
\begin{align*}
\alpha_2 = \mathcal{O}\left(\lVert \vec{x}\rVert_1 + \sqrt{d}T^{3/2}\log(Td/\epsilon)\right)
\end{align*}
suffices. We will then need to integrate out $\vec{v}$ and $\vec{x}$ to obtain this part of the union bound. 

Unfortunately, we only know the characteristic function in closed form for the pdf of  $X(t)$ conditioned on $\widehat{V}(t), \widehat{V}(t+1)$. However, the Chernoff bound enables us to still obtain the asymptotics of the tail, which is subexponential. The integrand will be exponential in $\lVert \vec{x}\rVert_1$, as is apparent from Equation \eqref{eqn:heston_price_bound}. This means our bound will only be finite for a certain range of parameters, specifically we require
\begin{align*}
&\kappa^2 > 2\kappa\sigma\rho - \rho^2\sigma^2 \geq 0.
\end{align*}
Since the integrand will then be bounded by a falling exponential, we can use this both for determining the truncation of $\vec{x}$ and integrating it out in other parts of the union bound. Specifically, if suffices to only keep $\lVert \vec{x} \rVert_1 \leq \alpha_3$, where 
\begin{align*}
    \alpha_3 = \mathcal{O}\left(T^2\ln(Td/\epsilon)\right)
\end{align*}
 for $\mathcal{O}(\epsilon)$ truncation error.

The last component of the union bound involves $\vec{v}$. We will need to express the truncation bound in terms of the CIR increments $\vec{y}$ and $\vec{z}$, which are standard Gaussian and central $\chi^2$ distributed, respectively. After undoing the change of variables mentioned above and marginalizing out $\vec{w}$ and $\vec{x}$, we obtain that the integral over $\vec{v}$ is proportional to
\begin{align*}
&\int_{ \lVert \vec{z}\rVert_{2}^2 + \lVert \vec{y}\rVert_1 \geq \alpha_1}p_{\chi}(\vec{y})p_{\mathcal{N}}(\vec{z})\exp(\frac{\kappa/\tanh(\Delta(t+1)\kappa) + \rho\sigma}{\sigma^2}g_{t+1}(\vec{z}, \vec{y}, v_0)) \\
&\leq \int_{ \mathcal{Y} \geq \alpha_1}p(\mathcal{Y})\exp(c(1+\gamma)\frac{\kappa/\tanh(\Delta(t+1)\kappa) + \rho\sigma}{\sigma^2}\mathcal{Y}),
\end{align*}
where the latter is a one-dimensional $\chi^2$ integral with $\mathcal{O}(T)$ degrees of feedom. Hence, the above is only finite when the scaling factor in the exponent is $ < 1/2$.  This leads to our next condition on the parameters
\begin{align*}
c\frac{\frac{\kappa}{\tanh(\kappa\Delta)} + \rho\sigma }{\sigma^2} &= \frac{1-e^{-2\kappa\Delta}}{4\kappa}\cdot \left(\frac{\kappa}{\tanh(\kappa\Delta)}+\rho\sigma\right) \\
&=  \frac{1 + e
^{-2\kappa\Delta}}{4} + \frac{\rho\sigma}{\kappa} \frac{1 - e
^{-2\kappa\Delta}}{4} < \frac{1}{2(1+\gamma)}.
\end{align*}
We can use known truncation bounds for $\chi^2$ to obtain that  $\alpha_1  = \mathcal{O}\left(T^2\ln(Td/\epsilon)\right)$ truncation suffices for $\vec{y}$ and $\vec{z}$. The result then follows by uniformly upper bounding $\alpha_1, \alpha_2, \alpha_3$ and that $\ell_1, \ell_2$ norms dominate the $\ell_\infty$ norm.

Given that the above bounds depend significantly on the model parameters, it is important to obtain good estimates of the constants in the exponentials appearing above. Since it is very easy for a loose bound (say an additional factor of two) on the constants to lead to a case where no range of model parameters will make our bounds finite. 

While we believe the above is still loose, it does lead to a non-vacuous range of model parameters allowing for truncation. This is discussed in more detail in Section~\ref{sec:heston_model_parameters_disc}. To the best of our knowledge, such a truncation analysis for the Heston model has not appeared in prior work.

The following theorem using Lemma~\ref{lem:hest_truncation_error} to determine how the number of (qu)bits needs to scale to obtain an $\epsilon_{\text{disc}}$ discretization error for numerical integration applied to~\ref{eqn:heston_deriv_price}.

\begin{restatable}[Heston Discretization Error]{theorem}{hesDiscreteError}
\label{thm:hest_discretization}
Under the assumptions of Lemmas~\ref{lem:hest_truncation_error} along with $\eta_k \geq 5$, it suffices to use a total of 
\begin{align*}
\mathcal{O}\left(T^4d^2\log^{3/2}(Td/\epsilon_{\text{trunc}})\log(Td/\epsilon_{\text{disc}})\right)
\end{align*} (qu)bits for an $\epsilon_{\text{disc}}$ discretization error. %
\end{restatable}

The proof (in Appendix \ref{subsec:hestDiscrProof}) follows from applying the left-endpoint rule (Lemma~\ref{lem:left_rule_error}) to Equation \eqref{eqn:heston-path-construction} over the region defined by Lemma~\ref{lem:hest_truncation_error}. In this setting, we are unable to get a bound on the number of bits per primitive distribution that scales only logarithmically in $Td$. Additionally, we need to leave a small region of $\vec{x}$ (Lemma~\ref{lem:trunc_around_zero_hest}) near zero to avoid singularities in the derivatives of \eqref{eqn:heston-path-construction}. The challenge in doing this comes from the lack of a closed-form expression for the pdf of $\vec{X}$. Instead, we utilize a hitting-time bound that appears in the proof of the Feller condition. The result, unfortunately, leads to an exponential blow-up in the asymptotics of the derivatives in  terms of $Td$. 

An immediate corollary of Theorems~\ref{thm:heston_loading_resource} and~\ref{thm:hest_discretization} are that the we can obtain an end-to-end quadratic quantum speedup for pricing derivatives over the version of the Heston model in Definition~\ref{defn:multi-asset-hest-asset-only}.
\begin{theorem}[Quantum Speedup for Pricing Over Heston Model]
\label{thm:heston-main}
Suppose $(f, \widehat{S})$ is a Financial Derivative Model (Definition~\ref{def:fin_deriv_model}), where is $f$ be a piecewise linear, $B$-Lipschitz payoff function, and $\widehat{S}$ is $d$-dimensional Heston Model (Definition \ref{defn:multi-asset-hest-asset-only}) asset path-processes $\widehat{S}$ initialized at $(\vec{S}(0), \vec{V}(0))$ evolved over $T$ time steps. Then, there is a quantum algorithm using $\widetilde{\mathcal{O}}\left(\textup{poly}(d, T, B)/\epsilon\right)$ one- and two-qubit gates, $\mathcal{O}(\textup{poly}(T))$ qubits, and outputs an $\epsilon$-additive estimate of $ \mathbb{E}[f(\widehat{S}) | \widehat{S}(0) = s_0, \widehat{V}(0) = v_0]$ with constant probability.
\end{theorem}
Note that $f$ can also be a function of $\widehat{V}$, which reduces to the CIR case, Section \ref{sec:cir_fast_forwardable}.

\subsubsection{Discussion on Model Parameter Constraints}
\label{sec:heston_model_parameters_disc}

In this section we investigate the implications of the constraints we put on the parameters on the finiteness of the Heston moments. A previous result from the literature shows the following conditions on these parameters for $\mathbb{E}[S(t)^{\omega}]$ to be finite, where $\omega \in \mathbb{R}_+$. To apply the below result to a path process, like Equation \eqref{eqn:heston_price_path}, consider setting $r = t\Delta $ for $t \in [T]$:
\begin{proposition}[Finite-time Moment Explosion Conditions {\cite[Proposition 3.1]{andersen2007moment}}] \label{prop:heston_finite_moments}
Consider a stochastic process $S({r})$ following the Heston SDEs given by \Cref{defn:heston-single}. Consider $\omega > 1$, then the $\omega$-th moment of $S({r})$, i.e. $\mathbb{E}[S^\omega({r})]$, is finite for all $r \in [0, \mathcal{T})$ and infinite for $r \ge \mathcal{T}$, where
\begin{align*}
    \mathcal{T} =
    \begin{cases}
        \infty, & b^2 - 4ac \ge 0, b < 0;\\
        \frac{1}{\gamma} \log(\frac{b + \gamma}{b - \gamma}), & b^2 - 4ac \ge 0, b > 0;\\
        \frac{1}{\beta} \left( \pi - 2 \arctan(\frac{c}{\beta}) \right), & b^2 - 4ac < 0;
    \end{cases}
\end{align*}
where $a= \frac{\sigma^2}{2}, b = \rho\sigma\omega - \kappa, c=\frac{\omega^2 - \omega}{2}$, $\beta := \sqrt{4ac - b^2}$, and $\gamma := -i \beta$. 
However, if $\omega = 1$, then $\mathcal{T} = \infty$.
\end{proposition}
We provide a proof of this result in the appendix with additional details. The above bounds the total time the path process can evolve for $T\Delta$ such that the $\omega$-th moments for all points in the process remain finite. We now contrast the above with the conditions Lemma~\ref{lem:hest_truncation_error} puts on the parameters $\kappa, \sigma, \rho, \Delta$.

We will assume $\rho \geq 0$ as this is the harder regime to have finite  moments.
We make a few observations regarding these conditions and how they connect to finite-time moment explosions. 
First note that, unlike the moment explosion result, the conditions in Lemma~\ref{lem:hest_truncation_error} are independent of $T\Delta$, the total time of the path, and only depend on the time between points in the path, $\Delta$. This implies that Lemma~\ref{lem:hest_truncation_error} is meant to show when we can truncate the increments forming an arbitrarily long path process with $\Delta$ increments. This leads us to compare with the necessary and sufficient conditions for $\mathcal{T} = \infty$ in Proposition~\ref{prop:heston_finite_moments}. 
The goal will be identify a regime where the $(1+\delta)$-th moment does not exist for some range of $\delta \in [0, 1]$, for all time but that our truncation analysis applies.

To compare our parameter conditions with Proposition~\ref{prop:heston_finite_moments}, we would like to put the conditions for finiteness of the moments, for $\mathcal{T} = \infty$, in terms of the quantity $\frac{\rho\sigma}{\kappa}$. To satisfy the second of our conditions in Lemma~\ref{lem:hest_truncation_error}, it suffices for $\kappa  \geq \frac{1}{\Delta}$. For the third condition to hold we need $\kappa^2 > 2\kappa\sigma\rho - \rho^2\sigma^2 \geq 0$ and so it suffices for $\frac{1}{2} >\frac{\rho\sigma}{\kappa}$.  If $\Delta \geq 1$, then  $\kappa  \leq 1$. Additionally, the fourth  condition above is always satisfied for any $\Delta \geq 3$ when $\frac{\rho\sigma}{\kappa} < \frac{1}{2}$.

With regards to Proposition~\ref{prop:heston_finite_moments}, our conditions put us in the regime of $b < 0,$ i.e. $\frac{\rho\sigma}{\kappa} < \frac{1}{2} \leq \frac{1}{1+\delta}$, as $\delta \leq 1$. The second condition for finiteness with $\mathcal{T}=\infty$,  $b^2 - 4ac \geq 0$, and for the $(1+\delta)$th moment is:
\begin{align*}
&(\rho\sigma(1+\delta) -\kappa)^2 - \sigma^2(\delta(1+\delta)) \geq 0,
\end{align*}
which is violated when $\rho \in \left(\frac{\kappa}{\sigma(1+\delta)} - \sqrt{\frac{\delta}{1+\delta}}, \frac{\kappa}{\sigma(1+\delta)} + \sqrt{\frac{\delta}{1+\delta}}\right)$. Our condition is equivalent to $\rho < \frac{\kappa}{2\sigma}$. Hence, for any $\kappa, \sigma$, there is a range of $\rho$ and time $\mathcal{T}$, where Proposition~\ref{prop:heston_finite_moments} shows that the second moment ($\delta = 1$) does not exist but Lemma \ref{lem:hest_truncation_error} applies.

Thus, this at least shows that for some parameter regimes our truncation analysis still works in the infinite variance case, like the Markov inequality approach of \cite{blanchet2024quadraticspeedupinfinitevariance}. In addition, our analysis directly produces truncation bounds on the increments $({Z}, {W}, {Y}, {X})$ of the Heston model.

\section{Speedups with Quantum MLMC for Correlated Processes}

\label{sec:multi_level_monte_carlo}

In this section, we tackle applying QMCI to models that lack fast-forwardability. In this case, the SDE evolution can only be approximated to error $\epsilon$ with a number of steps growing like $\mathcal{O}\left(\text{poly}(1/\epsilon)\right)$. We refer the reader to Section~\ref{sec:approximate_sde_simulation} for a review of approximate SDE simulation (via It\^o-Taylor schemes) and classical/quantum multi-level Monte Carlo (MLMC). Multi-level schemes provide the state-of-the-art asymptotic complexity in this regime, and can achieve a total sample complexity that is comparable to vanilla (Q)MCI. We will demonstrate that quantum MLMC can achieve an end-to-end quantum speedup over classical vanilla and multi-level Monte Carlo, when the SDE needs to be simulated using a It\^o-Taylor scheme and in the presence of correlations. The types of correlations that we consider are ``bipartite'' and can be reduced to the task of $(n,2)$-Milstein sampling (see  Section~\ref{sec:multi-dim-levy-area} for definition). In Section~\ref{sec:mlmc_error_analysis} we examine the inner workings of MLMC and, unlike previous studies, specifically elucidate the role played by the errors discussed in Section~\ref{sec:error_analysis_for_quant_deriv_pricing}.

We will demonstrate the speedup by showing that a subroutine (Section~\ref{sec:quantum_milstein_sampler}), which we coin the quantum Milstein sampler, is a sufficient proxy for the Milstein scheme. When combined with the results of An et al. \cite{An2021quantumaccelerated}, i.e. Theorem~\ref{thm:globally-lip-mlmc}, this will imply that we have an end-to-end quadratic speedup (Theorem~\ref{thm:speed_mlmc}) with quantum MLMC for pricing derivatives with Lipschitz payoffs over models with bipartite correlations.

Unfortunately, it is still unclear if quantum MCI can achieve a quadratic speedup for models with arbitrary correlations. This is because the existing methods for multi-dimensional L\'evy area sampling are insufficient. Specifically, these methods scale roughly as MCI, i.e $\gamma$ in Theorem~\ref{thm:gen_quantum_mlmc} would become too large if we need to embed an approximate L\'evy area sampler inside the Milstein scheme. Additionally, the existing methods are known to be optimal when one is given access to only the Brownian increments~\cite{dickinson2007optimal, foster2023brownian}. To emphasize the challenges with multi-dimensional L\'evy area sampling, the existing literature has started to move to use more practical methods like deep generative models \cite{jelinčič2023generativemodellinglevyarea}. Hence, new algorithmic techniques would be require to sample L\'evy areas more efficiently and obtain theoretical guarantees.

\subsection{Multi-dimensional L\'evy area Sampling}
\label{sec:multi-dim-levy-area}

One consequence of the results in \cite{An2021quantumaccelerated} is that quantum MLMC can achieve a quadratic speedup for globally Lipschitz payoffs, provided that a sufficiently high-order numerical scheme is employed.
Unlike classical MLMC, the quantum approach requires higher-order schemes, which become significantly harder to implement in the presence of correlations.
Specifically, it appears that a strong-order one scheme is generally required for quantum MLMC.
While this condition is currently only sufficient, as mentioned in Section~\ref{sec:speedupwithcorrelated} (Contribution 4), it also seems to be necessary, as a consequence of the quadratic variance reduction provided by quantum. 
Our contribution demonstrates that quantum MLMC retains an end-to-end speedup even when some correlations are present.

It is known  that the Euler-Maruyama scheme (Equation \eqref{eqn:EM-scheme}) attains the optimal order of strong convergence ($\frac{1}{2}$)  given access to only the increments of the multi-dimensional Wiener process $\vec{W}_{t}$~\cite{clark2005maximum, dickinson2007optimal}. To achieve a faster order of convergence, one must look to higher-order schemes~\cite[Theorem 10.6.3]{platen1999introduction}, which, in general, require sampling iterated stochastic integrals (in this case It\^o integrals):
\begin{align*}
    I_{\vec{j}}(h)=\int_{0}^{h}\cdots\int_{0}^{s_3}\int_{0}^{s_2}dW_{j_1}(s_1)dW_{j_2}(s_2)\cdots dW_{j_m}(s_m),
\end{align*}
where $W_{j_k}$ are the ``spatial'' components of the multi-dimensional Brownian motion  $\vec{W}$. To perform time-discretization schemes for SDEs, we desire being able to sample from the joint distribution of the increments and iterated stochastic integrals:
\begin{align*}
    (\Delta W_{1}, \dots \Delta W_{m}, \dots I_{(j_1,j_2)}, \dots, I_{(j_{m-1},j_{m})}, \dots, I_{\vec{j}}),
\end{align*}
which involves $2^m$ random variables. Like in the deterministic case, the (strong) It\^o-Taylor time-discretization schemes form a hierarchy of increasingly accurate, yet more computationally expensive, approximations~\cite[Theorem 5.5.1]{platen1999introduction}. Luckily, for classical and quantum MLMC and at least globally-Lipschitz payoffs, we only need to consider at most double integrals. However, unfortunately, these can still be quite challenging to sample from.

The Milstein scheme~\cite[Theorem 10.3.5]{platen1999introduction} with  strong order of convergence of $1$ only requires sampling Brownian increments and double stochastic integrals:
\begin{align*}
    I_{jk}(h) = \int_{0}^{h}\int_{0}^{u}dW_{j}(s)dW_{k}(u),
\end{align*}
whose differences are known as L\'evy areas:
\begin{align*}
   A_{(j, k)} := \int_{0}^{h}\int_{0}^{u} dW_j(s) dW_k(u)  - \int_{0}^{h}\int_{0}^{u} dW_k(s) dW_j(u) = I_{jk}(h) - I_{kj}(h).
\end{align*}
Hence, we only require estimating $m + \binom{m}{2}$ random variables. It is well-known that the Milstein scheme can avoid computing L\'evy areas if either (1) there are no correlations between processes or (2) the processes satisfy the commutativity condition (Section~\ref{sec:approximate_sde_simulation}). This commutativity condition appears to be hard to satisfy for general multi-dimensional models used in finance, such as the multi-asset Heston. 

Giles \cite{Giles_2014} proposed an antithetic sampling approach that can still avoid L\'evy areas. It appears unclear how to use this approach quantumly, so we do not discuss it any further. It is also possible to bypass this requirement on computing L\'evy areas if a different metric is used. For example, it is known that a strong-order one scheme in Wasserstein-2 metric  exists without sampling iterated integrals \cite{davie2014kmt}. Unfortunately, in our case, we will require something stronger, i.e. bounds on $L_2$ error. 

We will refer to the task of $(n, m)$-\emph{L\'evy  sampling} of an $n$-dimensional Brownian Motion as sampling from the joint distribution of the random variables associated with some set $\mathcal{S} \subseteq [n] \times [n]$:
\begin{align*}
 \mathcal{A}_{(n,m)} := \{ \Delta W_{1}, \dots \Delta W_{n}\} \cup \{ A_{(j,k)} : (j,k) \in \mathcal{S} \subset [n] \times [n], j < k, \lvert \mathcal{S} \rvert = m\},
\end{align*}
and where $dW_j \cdot dW_k = 0, \forall j, k \leq n$.  We will pay special attention to the $(n, 2)$ case, which is sometimes referred to as two-dimensional L\'evy area sampling.

It is known that the KL expansion \cite[Section 5.8]{platen1999introduction} is optimal for approximate $(n,m)$-L\'evy area sampling given only access to Brownian increments ($\Delta W_j$), with error scaling as $\mathcal{O}(N^{-1/2})$ \cite{dickinson2007optimal,foster2023brownian}, in terms of the number of Gaussian samples. The use of such approximations would cause the cost of simulation, $\gamma$, in Theorem~\ref{thm:gen_quantum_mlmc}, to be too large. This would then imply an even faster converging scheme would be required to retain the speedup. Furthermore, it seems that a convergence rate of $\mathcal{O}(N^{-1/2})$ in $L_2$ for L\'evy area approximation is not sufficient to retain the convergence rate of the Milstein scheme (Corollary 10.6.5 \cite{platen1999introduction}). 

For the moment, lets ignore the last comment and consider the cost of Milstein simulation using the KL scheme, with the goal of retaining the $\mathcal{O}(h^{-1})$ strong-convergence rate. The sample cost of simulation with step size $h$ goes as $\mathcal{O}(h^{-5})$, $\gamma = 5$, which is significantly worse than the cost of  Euler-Maruyama with $\mathcal{O}(h^{-1})$, $\gamma = 1$. The proof that Milstein is sufficient for QMCI to retain its quadratic speedup needs $\gamma$ to remain $\leq 1$. This implies it is unclear from the results of \cite{An2021quantumaccelerated} whether quantum can retain even an end-to-end speedup when L\'evy areas need to be computed.

It turns out that the pdf for $(n,2)$-L\'evy area sampling is known in closed form \cite{levy1951wiener}. This enabled Gaines and Lyons \cite{gaines1994random} to develop an efficient numerical procedure for sampling from
\begin{align*}
    (\Delta W_1, \Delta W_2, A_{(1,2)}),
\end{align*}
which is $(2,2)$-L\'evy area sampling. The authors did not perform an asymptotic analysis of the procedure. However, Corollary~\ref{cor:levy-area-loading} in Section~\ref{sec:char_func_loading} shows how to quantumly sample form $(2,2)$-L\'evy areas.
It should be apparent that the $(n,2)$ case is just the product of the densities of the $\binom{n}{2}$, $(2,2)$ cases.

If we are considering an $n$-dimensional model where any given process is correlated with at most one other process, then, as shown below, sampling from the $n$-dimensional Milstein scheme only requires $(n, 2)$ L\'evy area sampling. We refer to this task as $(n,2)$-\emph{Milstein Sampling} and say the process has \emph{bipartite correlations}.

It is important to mention that for MLMC, we actually require something weaker than ensuring that the approximate Milstein scheme retains its $\mathcal{O}(h)$ strong convergence rate. Hence, ensuring that we have an $L_2$ approximation to the L\'evy areas may not be necessary (as assumed in Corollary~10.6.5~\cite{platen1999introduction}). This fact will be apparent from the derivations of the results in Section~\ref{sec:mlmc_error_analysis} and is what enables us to get a speedup with the above approach for the $(n,2)$ case. Still, as presented earlier, the KL approach is too expensive to retain the speedup.

In the next section, we analyze the role of the various sources of error from Section~\ref{sec:error_analysis_for_quant_deriv_pricing} within the MLMC framework, which has not been done previously. Then, in  Section \ref{sec:quantum_milstein_sampler} , using the quantum L\'evy area sampler introduced in Section~\ref{sec:subroutines_for_distribution_loading},  we present an efficient ($\mathcal{O}\left(\text{poly}\log(1/\epsilon)\right)$) scheme for $(n,2)$-Milstein qsampling. Additionally, we analyze the impact of the distribution error from the quantum Milstein sampler. Together, these results  show that the error produced by the Milstein scheme with the approximate L\'evy area sampler does not destroy the quadratic speedup provided by quantum MLMC (obtained with the ideal Milstein).  When combined with quantum MLMC (Theorem~\ref{thm:globally-lip-mlmc}), our result shows that quantum computation does provide an end-to-end speedup for derivative pricing in the setting of approximate SDE simulation with some correlations.

\subsection{Error Analysis for Quantum MLMC}
\label{sec:mlmc_error_analysis}

While it was not addressed in \cite{An2021quantumaccelerated}, we need to account for the various sources of error discussed in Section~\ref{sec:error_analysis_for_quant_deriv_pricing} within the MLMC framework. This necessary for showing that our quantum Milstein sampler is compatible with quantum MLMC. To do this, we need to recall some of the components of MLMC, following Giles \cite{giles2008multilevel}, and extend the analysis to include other sources of error. 

Like usual, suppose our payoff $f$ depends on $T$, $d$-dimensional monitoring points separated by a time increment of $\Delta$, forming a path process $\widehat{X}$. However, to ensure the approximation error is sufficiently low, we will need to simulate $X(t)$ at points that are in-between monitoring points. Suppose we apply an  SDE discretization scheme to simulate the continuous-time process $X(t)$ with step-size $h_\ell = 2^{-\ell}\Delta$. This will result in approximating a $2^{\ell}T$-length, $d$-dimensional path process, which we denote by $\widehat{Y}$. The payoff $f$ will only depend on the points $\widehat{Y}(2^{\ell}r), r \in [T]$.  We define $\widetilde{Y}_{\ell}$ to be the approximation of $\widehat{Y}$ output by the discretization scheme with ``level'' $\ell$, which takes $h_\ell = 2^{-\ell}\Delta$.

 The MLMC framework considers an  estimator (obtained via $\Theta(L)$ (Q)MCI runs) for the quantity
\begin{align}
\label{eqn:mlmc_telescope}
    \mathbb{E}[f(\widetilde{Y}_{0})] + \sum_{\ell=1}^{L} \mathbb{E}[f(\widetilde{Y}_{\ell}) - f(\widetilde{Y}_{\ell-1})],
\end{align}
which  is exactly $\mathbb{E}[f(\widetilde{Y}_{L})]$. This of course requires estimating  $\mathbb{E}[f(\widetilde{Y}_{\ell}) - f(\widetilde{Y}_{\ell-1})]$ for various $\ell$. Specifically $\widetilde{Y}_{\ell}$ and $\widetilde{Y}_{\ell-1}$ are coupled by first sampling $\widetilde{Y}_{\ell}$, then averaging intermediate points along the discrete path in groups of size $2$ to obtain a sample of $\widetilde{Y}_{\ell-1}$. The purpose of forming a telescoping sum is to perform variance reduction.

We have the following result regarding the Milstein (Recall scheme from Section \ref{sec:approximate_sde_simulation}) convergence.
\begin{theorem}[Theorem 10.6.3 \cite{platen1999introduction}]
\label{thm:milstein_convergence}
Suppose for all $1 \leq i \leq d$, $1 \leq j \leq m$, $b^{0} = \mu_i(\vec{X}), b^{1} = \sigma_{ij}(\vec{X})$ and for $j \in \{0, 1\}$:
\begin{align*}
&b^{j}(\vec{X}),\\ 
&\mathcal{L}^{j_1}b^{j}(\vec{X}), j_1 \in \{1, \dots, m \}
\end{align*}
are Lipschitz continuous in $\ell_2$ norm. Then the output of the Milstein scheme $\widetilde{X}$ for the $d$-dimensional, $T$-length path process $\widehat{X}$ driven by $m$ Brownian motions satisfies
\begin{align*}
\sup_{0 \leq k < T}\mathbb{E}\lVert \widehat{X}(k) - \widetilde{X}(k)\rVert_{2}^2 = \mathcal{O}\left(\textup{poly}(d)h^{2}\right).
\end{align*}
\end{theorem}
So, if we use a scheme with strong convergence one, i.e. Milstein, then
\begin{align*}
    \sup_{0 \leq k < 2^{\ell}T}\mathbb{E}\lVert \widehat{Y}(k) - \widetilde{Y}_{\ell}(k)\rVert_{2}^2 = \mathcal{O}(\text{poly}(d)h_{\ell}^2)=\mathcal{O}(\text{poly}(d)\Delta^2 2^{-2\ell}).
\end{align*}
If the payoff $f$ is a globally-Lipschitz payoff with maximum slope $B$, then from Lemma~\ref{thm:milstein_convergence}
\begin{align*}
    \lvert \mathbb{E}[f(\widehat{Y})] -  \mathbb{E}[f(\widetilde{Y}_{L})]\rvert \leq B \sum_{t=0}^{T-1}\mathbb{E}[\lVert \widehat{Y}(t2^{L}) - \widetilde{Y}_{L}(t2^{L})\rVert_2] = \mathcal{O}\left(\text{poly}(d)BT\Delta 2^{-\ell}\right),
\end{align*}
and so with 
\begin{align*}
    L= \mathcal{O}\left(\log(dBT\Delta/\epsilon)\right)
\end{align*} levels we can suppress the error between the two means to $\mathcal{O}(\epsilon)$. The goal of quantum MLMC is to ensure that the overall cost to approximate \eqref{eqn:mlmc_telescope} has only a linear in $1/\epsilon$ dependence, up to polylog factors. The typical analysis of MLMC assumes that we exactly have the spatially-continuous processes $\widetilde{Y}_{\ell}$. However, in reality, we only have an approximation to a truncated and discretized version of $\widetilde{Y}_{\ell}$, which will denote by $\mathcal{Y}_\ell$. 

Using QMCI, the standard error for estimating $\mathbb{E}[ f(\mathcal{Y}_{\ell}) - f(\mathcal{Y}_{\ell-1})]$ falls as
\begin{align*}
\widetilde{\mathcal{O}}\left(\frac{\sqrt{\textup{Var}(f(\mathcal{Y}_{\ell}) - f(\mathcal{Y}_{\ell-1}))}}{N_{\ell}}\right),
\end{align*}
where $N_{\ell}$ is the number of  quantum samples (Section~\ref{sec:quantum_derivative_pricing_intro}). As shown in the appendix (Section~\ref{sec:mlmc-variance_proof}), the numerator satisfies
\begin{align}
\label{eqn:mlmc-variance}
&\textup{Var}(f(\mathcal{Y}_{\ell}) - f(\mathcal{Y}_{\ell-1})) = \mathcal{O}\left(B^2T\left(\sum_{t=0}^{T-1} \text{Err}(\ell, t) + \mathbb{E}[\lVert \widetilde{Y}_{\ell}(2^{\ell}t) - \widehat{Y}(2^{\ell}t)\rVert_2^2]\right)\right),
\end{align}
where $\text{Err}(\ell, t)$, explicitly shown in the appendix, is effectively the error in approximating the \emph{truncated} expectation of $\lVert \widetilde{Y}_{\ell}(t2^{\ell}) - \widetilde{Y}_{\ell-1}(t2^{\ell})\rVert_2^2$, i.e. with each spatial component lying in $[-R, R]$ with $R$ chosen to control the truncation error in estimating $\mathbb{E}[f(\widehat{X})]$ (Lemma~\ref{lem:milstein_loader_truncation}). Given that $\widetilde{Y}_{\ell}$ and  $\widetilde{Y}_{\ell-1}$ will be coupled, we can view this as estimating the expectation of $\lVert(i-g_{\ell})(\widetilde{Y}_{\ell})\rVert_2^2$, where $i$ is the identity function and $g_{\ell}$ averages intermediate points in pairs of two to produce the coupled $\widetilde{Y}_{\ell-1}$. The main deviation from the kinds of functions considered in Section~\ref{sec:error_analysis_for_quant_deriv_pricing} is the quadratic component, which does not end up causing any problems. Specifically, $g_{\ell}$ is clearly linear in the components of $\widetilde{Y}_{\ell}$.

Given that the truncation results in approximating an integral over a bounded domain, the only sources of error in $\text{Err}(\ell, t)$ are discretization and  distribution/renormalization.  The second term in \eqref{eqn:mlmc-variance} is just the Milstein strong convergence error. Hence if 
\begin{align}
\label{eqn:bounded_error}
    \text{Err}(\ell, t) =  \mathcal{O}\left(\text{poly}(d)\Delta^2 2^{-2\ell}\right)
\end{align}
then
\begin{align*}
    \textup{Var}(f(\mathcal{Y}_{\ell}) - f(\mathcal{Y}_{\ell-1})) = \mathcal{O}\left(\text{poly}(d)B^2T^2 \Delta^{2}2^{-2\ell}\right),
\end{align*}
which exactly the guarantee provided by the Milstein scheme (Theorem~\ref{thm:milstein_convergence}). We will discuss in the next subsection how it is easy to ensure Equation~\ref{eqn:bounded_error} using the techniques of Section~\ref{sec:error_analysis_for_quant_deriv_pricing}. As MLMC is concerned with retaining the sampling complexity of vanilla MCI, the main source of error to check will be the distribution/renormalization error, as discretization error only contributes to space overheads.

Note that each step of the Milstein scheme under bipartite correlations (explored more in the next subsection) for a given asset component makes $\mathcal{O}(1)$ calls to a standard Gaussian loader or two-dimensional L\'evy area loader, which each take $\mathcal{O}(\text{poly}\log(1/\epsilon))$ gates to prepare (Section~\ref{sec:subroutines_for_distribution_loading}). This cost is multiplied by a factor of $C_{\ell} = \mathcal{O}(d\cdot  2^{\ell}T)$ to account for the number of time steps and spatial components at the $\ell$-th level. Hence the total gate complexity (ignoring the cost to compute $f$) is
\begin{align*}
    \sum_{\ell=1}^{L} N_{\ell}C_{\ell} =  \sum_{\ell=1}^{L} \widetilde{\mathcal{O}}\left(\frac{\sqrt{\text{poly}(d)}{B}T\Delta 2^{-\ell}}{\epsilon}\right) \cdot \widetilde{\mathcal{O}}\left(d\cdot 2^{\ell}T\right) = \widetilde{\mathcal{O}}(\text{poly}(T, d)/\epsilon),
\end{align*}
to estimate $\mathbb{E}[f(\mathcal{X}_{L})]$. If we can also suppress $\epsilon_{\text{trunc}}$, $\epsilon_{\text{distr}}$, and $\epsilon_{\text{disc}}$ for the payoff $f$, then we will have an $\epsilon$-additive error estimate of the price.

In the next subsection, we will  discuss how we can satisfy Equation \eqref{eqn:bounded_error} and suppress $\epsilon_{\text{trunc}}$, $\epsilon_{\text{distr}}$, and $\epsilon_{\text{disc}}$. Thus this shows that our quantum Milstein sampler combined with quantum MLMC retains a quadratic speedup for globally-Lipschitz payoffs and bipartite correlated SDEs.

\subsection{Quantum Milstein Sampler and End-To-End Speedup}
\label{sec:quantum_milstein_sampler}
In this section we show that we can efficiently load a quantum state encoding  the $(d,2)$-Milstein scheme and that the errors discussed in the previous subsection can be efficiently suppressed.

Recall the Milstein scheme (Section~\ref{sec:approximate_sde_simulation})
\begin{align*}
 \widetilde{X}_i(t+h) &= \widetilde{X}_i(t) + \mu_i(\widetilde{X}(t))h + \sum_{k=1}^{d}\sigma_{ik}(\widetilde{X}(t))\Delta W_k + \sum_{k, j=1}^d \mathcal{L}^j\sigma_{ik}(\widetilde{X}(t))(\Delta W_{j}\Delta W_{k}  + A_{(j,k)}),
\end{align*}
where \begin{align*}
    \mathcal{L}^{k} := \frac{1}{2}\sum_{i=1}^{d}\sigma_{ik}\frac{\partial}{\partial x_i}, k= 1, \dots, d.
\end{align*}
The ``no correlation'' case of \cite{An2021quantumaccelerated} requires $\sigma_{ij} = 0$ if $i \neq j$. For the case of $(d,2)$-Milstein sampling, the scheme reduces to
\begin{align*}
\widetilde{X}_i(t+h)&= \widetilde{X}_i(t) + \mu_i(\widetilde{X}(t))h + \sigma_{ii}(\widetilde{X}(t))\Delta W_i  + \sigma_{ij}(\widetilde{X}(t))\Delta W_j  \\
 &+\mathcal{L}^i\sigma_{ii}(\widetilde{X}(t))(\Delta W_{i}\Delta W_{i}) + \mathcal{L}^j\sigma_{ij}(\widetilde{X}(t))(\Delta W_{j}\Delta W_{j}) \\
 &+[\mathcal{L}^i\sigma_{ij}(\widetilde{X}(t)  - \mathcal{L}^j\sigma_{ij}(\widetilde{X}(t))][\Delta W_i \Delta W_j + A_{(i,j)}],
\end{align*}
where $j$ is the unique index $\neq i$ such  that $\sigma_{ij} \neq 0$. It should be apparent that $\widetilde{X}_j$ has an analogous form.
Hence to perform a single transition for $(\widetilde{X}_i, \widetilde{X}_j)$, we only need to sample $(\Delta W_i, \Delta W_j, A_{(i,j)})$, i.e. $(2,2)$-L\'evy  sampling, where $W_i, W_j$ are independent (and can be made correlated via a linear transformation). It should also be clear that if the commutativity condition \eqref{eqn:commutivity} is satisfied then the coefficient in front of $A_{(i,j)}$ is zero. Lastly, note that the drift $\mu_i$ and diffusion $\vec{\sigma}_{i,\cdot}$ terms can still be functions of the entire $d$-dimensional process.

We can view this scheme as applying an arithmetic function $g$ to the random variables in the set $\mathcal{A}_{(d,2)}$.
Hence, by quantumly sampling from the Milstein scheme, we mean preparing a state that amplitude encodes (a truncated and discretized version of) the distribution of $g(\mathcal{A}_{(d,2)})$. This constructs a reparameterized integration problem, as discussed in Section~\ref{sec:error_analysis_for_quant_deriv_pricing}. Note that the number of random variables in $\mathcal{A}_{(d,2)}$ is at most $d + \lceil d/2 \rceil$. Specifically, for a payoff $f$ that is a function of $\widetilde{X}$, the price can be expressed as 
\begin{align}
\label{eqn:integral_2n_milstein_sampling}
    \int_{ (\mathbb{R}^{d} \times \mathbb{R}^{\lceil d/2 \rceil})^{\times 2^{\ell}T}}  f\circ g(\vec{z}, \vec{a}) d\vec{z}d\vec{a},
\end{align}
where $\vec{z} \in \mathbb{R}^{d}$ corresponds to the standard Gaussians and $\vec{a} \in \mathbb{R}^{\lceil d/2\rceil}$ to the L\'evy areas. The integrand for MLMC will really correspond to $f(\widetilde{Y}_{\ell}) - f(\widetilde{Y}_{\ell-1})$. However, it suffices to determine the truncation region needed to suppress the error for $f(\widetilde{Y}_{\ell})$.  As mentioned in Section~\ref{sec:error_analysis_for_quant_deriv_pricing}, we need to truncate the domain of $\vec{z}$ and $\vec{a}$ and bound the truncation error.  If  $\mu_i, \sigma_{ij}$ have at most linear growth, then we obtain the following result. 

\begin{restatable}[Milstein Truncation Error]{lemma}{milTruncErr}
\label{lem:milstein_loader_truncation}
Suppose that the payoff is piecewise-linear with maximum slope $B$. Additionally, suppose that $\mu$ and $\sigma$ have at most linear growth. If we truncate each L\'evy area and standard Gaussian to an $\ell_{\infty}$ ball with radius $\mathcal{O}\left(T\log(B\lVert \widehat{X}(0)\rVert_1Td/\epsilon)\right)$, then $\epsilon_{\text{trunc}}$ for Equation \eqref{eqn:integral_2n_milstein_sampling} is $\mathcal{O}(\epsilon)$. 
\end{restatable}

In the setting of approximate SDE simulation, we only have two types of primitives. The first is the standard Gaussian (which we can load using Corollary~\ref{cor:standard_gauss_qsvt}), and the second is the conditional L\'evy area loader (Corollary~\ref{cor:levy-area-loading}). Even though there is conditioning, the analysis in Section~\ref{sec:error_analysis_for_quant_deriv_pricing} still applies. Hence, for a $d$-dimensional process over $T$ time steps, we have $\Theta(Td2^{\ell})$ primitives. Thus from \eqref{eqn:distribution_loading_error_bound} we need to scale
\begin{align*}
    \epsilon_{\text{distr/trunc}} \rightarrow \mathcal{O}\left(\frac{\epsilon_{\text{distr/trunc}}}{Td \max_{\vec{x} \in [-R, R]^{\Theta(Td2^{\ell})}} \lvert f\circ g(\vec{x})\rvert}\right)
\end{align*}
to ensure that out loader provides an at most $\mathcal{O}(\epsilon)$ error from the ``true'' discrete-sum we would like to estimate.

From the proof of Lemma~\ref{lem:milstein_loader_truncation} in the appendix, it follows that
\begin{align*}
&\max_{\vec{x} \in [-R, R]^{\Theta(Td2^{\ell})}} \lvert f\circ g(\vec{x})\rvert \\ &\leq BT \lVert \widetilde{X}(0)\rVert_1\prod_{t=0}^{T-1}\left(\sum_{k=1}^{d}\alpha h + \beta \sqrt{h}R+ \beta \sqrt{h}\lvert w_{k',t}\rvert + \beta h dR^2 + \beta h dR^2 + \beta dR^2+\beta d h R\right) \\
&=\mathcal{O}\left(BT \lVert \widetilde{X}(0)\rVert_1 (dR^2)^T\right),
\end{align*}
where $R = \mathcal{O}\left(T\log(B\lVert \widehat{X}(0)\rVert_1Td/\epsilon_{\text{trunc}})\right)$, so
\begin{align*}
\max_{\vec{x} \in [-R, R]^{\Theta(Td2^{\ell})}} \lvert f\circ g(\vec{x})\rvert  =\mathcal{O}\left(BT \lVert \widetilde{X}(0)\rVert_1 (d[T\log(B\lVert \widehat{X}(0)\rVert_1Td/\epsilon_{\text{trunc}})]^2)^T\right).
\end{align*}
Hence due to the $\mathcal{O}\left(\text{poly}\log(1/\epsilon_{\text{distr}})\right)$ dependence for all of the loaders, this at most adds a $\mathcal{O}(\text{poly}(T))$ gate cost. This shows that the discrete-sum we actually load is $\mathcal{O}(\epsilon)$ to a discrete-sum over the exact Milstein path distribution.

For the $\text{Err}(\ell, t)$ term in Equation \eqref{eqn:bounded_error}, we do not need to worry about truncation error, as  the random variable is already bounded. The term $\text{Err}(\ell, t)$ is effectively a discretization error. In this case, we can basically take $f(\vec{x}) = \lVert(i-g_{\ell})(\vec{x})\rVert_2^2$, using notation from earlier, which will only introduce additional polynomial factors in $d$. Hence, we can also efficiently suppress the distribution error component of Equation \eqref{eqn:bounded_error} as well.

We do not provide an asymptotic estimate of the resources for suppressing the discretization error. Note that this is in principle could be done using the same techniques from Section~\ref{sec:pricing_deriv_fast_forwardable}. However, like the Heston model in Section~\ref{sec:heston-fast-forwardable}, we only have the conditional characteristic function of the pdf of the L\'evy  in closed form. Thus, we do not expect to obtain a $\mathcal{O}\left(\text{poly}(\log(Td))\right)$ qubit estimate with the techniques presented. 

Using Lemma~\ref{lem:left_rule_error} along with a simple estimate of $\mathcal{O}\left(T2^{\ell} \lVert \widetilde{X}(0)\rVert_1 (dR^2)^T\right)$ on $\sup_{[-R, R]^{\Theta(Td2^{\ell})}} \lVert \nabla g\rVert_{\infty}$, it should be apparent that $\mathcal{O}\left(\text{poly}(T, d, \log(1/\epsilon_{\text{trunc}}), \log(1/\epsilon_{\text{disc}}))\right)$ qubits suffice for an $\epsilon_{\text{disc}}$ discretization error. Since $g_{\ell}$ is linear, in the case of $f\circ g(\vec{x}) = \lVert(i-g_{\ell})(g(\vec{x}))\rVert_2^2$, we have $\sup_{[-R, R]^{\Theta(Td2^{\ell})}} \lVert \nabla (f\circ g)(\vec{x}) \rVert_{\infty} = \mathcal{O}\left(\lVert(i-g_{\ell})(g(\vec{x}))\rVert_{\infty}\lVert \nabla g \rVert_{\infty}\right)$, so a similar bound will hold. Hence we can suppress the discretization error associated with Equation \eqref{eqn:bounded_error}. Regardless, unlike the distribution error, the discretization error only contributes to the qubit count.

Hence, the above implies that we can obtain a speedup with quantum MLMC when using the Milstein scheme in the $(d,2)$ setting. Thus the following result becomes self-evident.

\begin{theorem}
\label{thm:speed_mlmc}
Suppose $(f, \widehat{X})$ is a Financial Derivative Model (Definition~\ref{def:fin_deriv_model}), where $f$ is globally-Lipschitz and $\widehat{X}$ has at most bipartite correlations. Additionally, suppose that $\sigma$ and $\mu$ have at most linear growth and are Lipschitz continuous. Then there exists a quantum algorithm for estimating the price (Definition~\ref{def:deriv_price}) to $\epsilon$-additive error, with constant probability, using $\widetilde{\mathcal{O}}\left(\textup{poly}(T,d)\frac{\sigma_0}{\epsilon}\right)$ one- and two-qubit gates and $\mathcal{O}\left(\textup{poly}(T, d, \log(\sigma_0/\epsilon))\right)$ qubits.
\end{theorem}

\section{Quantum PDE Solvers for Distribution Loading}

\label{sec:quantum_pde_solvers}

As mentioned in Section~\ref{sec:motivation}, one of the main limitations of the existing quantum derivative pricing framework is that the quantum sampling or loading of the discrete sum utilizes a number of qubits that scale at least linearly with $T$, where $T$ is the number of monitoring points for a path-dependent derivative (e.g. approach mentioned in Section \ref{eqn:discrete-sum-prep}). The quantity $T$ is typically not considered to be an asymptotic quantity, but it can still be large. Hence, savings in $T$ can have substantial practical benefits, as qubits are a precious resource. 

One class of state preparation techniques that can utilize  fewer than $\text{poly}(T)$ qubits  are quantum PDE solvers. Specifically, if a $T$-dimension state can be expressed as the solution to a PDE then, in some cases, quantum computers can prepare the state using only $\text{poly}(\log T)$ qubits. For example, it is well-known that the marginals of an It\^o SDE follow the Fokker-Planck PDE (reviewed in Section~\ref{sec:fokker_planck}). Hence the distribution can be loaded onto a quantum state using  a quantum PDE solver, which we outline how to do in Section \ref{sec:algo_for_solving_fp}.

Unfortunately, in Section~\ref{sec:challenges_pricing_pde}, we show that the quantum PDE distribution-loading procedure \emph{is not compatible with quantum algorithms for derivative pricing}, for a variety of reasons. We will discuss significant obstacles to this approach. Thus, it currently seems open as to whether one can obtain sublinear space quantum sampling from classical SDEs in a way that is compatible with (Monte Carlo based) quantum derivative pricing algorithms. 

Beyond PDE solvers, there have been a few other techniques proposed for quantumly accelerating the simulation of classical stochastic processes in either space or time. One such case was shown by Apers and Sarlette \cite{apers2019quantumfastforwardingmarkovchains}, where they showed discrete-time quantum walks can be used to approximately simulate symmetric Markov chains in sublinear time. One issue with this approach is that it incurs a subnormalization penalty as we are directly amplitude encoding probability vectors, not their squares (this issue also occurs with quantum PDE solvers, Section \ref{sec:curseofdim}). The second is that the chain needs to be symmetric, which is unlikely for financial processes. The second is by Prakash et al. \cite{prakash2024quantum} who utilized the Karhunen-Lo\`eve expansion to accelerate the simulation of GBM in space and time. However, using this subroutine within QMCI results in additional factors of $\frac{1}{\epsilon}$. Thus, all of these approaches appear to have challenges of their own.

Lastly, our proposed usage of quantum PDE solvers as a subroutine for pricing is distinct from other usages in the existing literature, and is motivated by our focus on Monte Carlo based pricing. An alternative pricing approach to QMCI is to use quantum PDE solvers to solve a PDE (e.g. the Black-Scholes or Backward Kolmogorov PDE) for the price itself, like in Ref.~\cite{miyamoto2021pricing}. One of the major issues with such approaches is the need to extract the price (encoded as an amplitude) from the quantum state. This is well known already to have at least an inverse-polynomial dependence on the desired precision, and hence comparable to classical and quantum MCI. The allure of such methods is to overcome the exponential dimension dependence present in classical PDE solvers. However, as we have shown (e.g., Section \ref{sec:speedup_heston_cir}), Monte Carlo methods for pricing common derivatives \emph{already} have a polynomial dependence on dimension. Hence, the current benefits of this approach remain unclear. Additionally some of the challenges found in this paper apply to this setting. %

\subsection{The Fokker-Planck Equation}
\label{sec:fokker_planck}

For any It\^o SDE:
\begin{equation}
    d \vec{X}(t) = \vec{\mu}(\vec{X}(t), t) dt + \boldsymbol{\sigma} (\vec{X}(t), t)d \vec{W}(t),
\end{equation}
we can formulate the PDE for the marginal probability density $p(\vec{x}, t)$ for $\vec{X}(t)$, known as the Fokker-Planck (FP) equation:
\begin{equation} \label{eqn:fokker-planck}
    \frac{\partial p(\vec{x}, t)}{\partial t} = - \nabla \cdot \left[ \vec{\mu }(\vec{x}, t) p(\vec{x}, t) \right] + \Tr \left\{ \nabla^2 \left[\boldsymbol{D}(\vec{x}, t) p(\vec{x}, t)\right] \right\},
\end{equation}
where $\nabla^2$ is the Hessian operator, and $\boldsymbol{D} = \frac{1}{2} \boldsymbol{\sigma} \boldsymbol{\sigma}^\mathsf{T}$. For a given SDE, the corresponding FP equation is the forward Kolmogorov equation. 

Below, we present the reduction from SDE to FP for two processes that we have previously discussed in this work.
\paragraph{Geometric Brownian motion (equity, constant volatility) :} 
For a $d$-dimensional Geometric Brownian motion (GBM),

\multiassetgbm*
We have
\begin{align*}
    \vec{\mu}(\vec{X}(t), t) &\to \vec{\mu} \circ \vec{X}(t), \\
    \boldsymbol{\sigma}(\vec{X}(t), t) &\to [\vec{\sigma} \circ  \vec{X}(t)]\circ \mathbf{A} ,
\end{align*}
where  $\circ$ denotes element-wise multiplication (with broadcasting) and $\mathbf{A} \in \mathbb{R}^{d\times d}$ is the Cholesky factor of the correlation matrix $\mathbf{C}$.

\paragraph{Cox–Ingersoll–Ross process (interest rate, local volatility):}
Recall the CIR process
\CIR*

For a $d$, uncorrelated CIR processes, we have
\begin{align*}
    \vec{\mu}(\vec{V}(t), t) &\to \vec{\kappa} \circ (\vec{\theta} - \vec{V}(t)), \\
    \vec{\sigma}(\vec{V}(t), t) &\to \vec{\sigma} \circ \sqrt{\vec{V}(t)},
\end{align*}
where $\vec{\kappa}, \vec{\theta}, \vec{\sigma} \in \mathbb{R}_+^{d}$ and the square-root is applied element-wise.

\subsection{Solving Fokker-Planck with Quantum Algorithms}

\label{sec:algo_for_solving_fp}

In this section, we present a quantum algorithm based on finite-difference for solving the FP equation.
We start by converting the Fokker-Planck PDE given by Equation~\eqref{eqn:fokker-planck} to a system of linear ODEs by using the finite difference approximation for the derivatives w.r.t. $\vec{x}$.
First, observe that the Equation~\eqref{eqn:fokker-planck} can be rewritten in the following conservation form (using the symmetry of $\boldsymbol{D}$)
\begin{align} \label{eqn:fokker-planck-conservation}
    \frac{\partial p(\vec{x}, t)}{\partial t} &= - \nabla \cdot \left\{ \left[ \vec{\mu }(\vec{x}, t) - \left[\nabla^{\mathsf{T}} \boldsymbol{D}(\vec{x}, t)\right]^{\mathsf{T}} \right] p(\vec{x}, t) \right\} + \nabla\cdot\left[\boldsymbol{D}(\vec{x}, t)\nabla p(\vec{x}, t)\right] \nonumber \\
    &= - \nabla \cdot \left[ \vec{s}(\vec{x}, t) p(\vec{x}, t) \right] + \nabla \cdot \left[\boldsymbol{D}(\vec{x}, t)\nabla p(\vec{x}, t)\right],
\end{align}
where we have defined $\vec{s}(\vec{x}, t) = \vec{\mu }(\vec{x}, t) - \left[\nabla^{\mathsf{T}} \boldsymbol{D}(\vec{x}, t)\right]^{\mathsf{T}}$.
We apply the finite-difference method (FDM) and discretize the $d$-dimensional $\vec{x}$ space into $n$ grid points for each dimension, and use the upwind scheme \cite{li2020large} for the convection and diffusion terms.
The result is an $n^d+1$-dimensional system of linear ODEs with the following form
\begin{equation} \label{eqn:fokker-planck-ode}
    \frac{d {\vec{p}}}{d t} = L {\vec{p}},
\end{equation}
where ${\vec{p}}$ is a $n^d$-dimensional vector carrying the values of $p(\boldsymbol{x}, t)$ at each grid point, and
\[
    L = L_C + L_D,
\]
is an $n^d \times n^d$ matrix with $I$ being the $n \times n$ identity matrix, and $L_C$ and $L_D$ denoting the discretized convection and diffusion operators respectively.
$L_C$ and $L_D$ are defined as follows
\begin{align*}
    L_C &= \frac{1}{\Delta x} \sum_{i=1}^d 
        \left [
            A^i_{+} \circ  \left( \boldsymbol{\hat{s}^-_i} \right)^\mathsf{T} - A^i_{-} \circ \left(\boldsymbol{\hat{s}^+_i}\right)^\mathsf{T}
        \right], \\
    L_D  &= 
        \frac{1}{\Delta x^2}  
            \sum_{i=1}^d \left[ 
                A^i_{-} \left( A^i_+ \circ \boldsymbol{\hat{D}}^*_{ii} \right)
                + \sum_{\substack{j=1\\ j \ne i}}^d A^i_1 \left( A^j_1 \circ \boldsymbol{\hat{D}}_{ij} \right)
        \right],
\end{align*}
where ${\vec{s}}^+_i$, ${\vec{s}}^-_i$, $\boldsymbol{{D}}_{ij}$ and $\boldsymbol{{D}}^*_{ii}$ are $n^d$-dimensional vectors containing values of $s_i \lor 0$, $- (s_i \land 0)$, $D_{ij}$ and $D^*_{ii} = D_{ii}(\vec{x} + \frac{1}{2}\Delta x \vec{e}_i, t)$, respectively, on the same grid as ${\vec{p}}$, and $\vec{e_i}$ is the unit vector in the $i$-th dimension of $\vec{x}$.
The $A^i_\star$ matrices are defined as
\[
    A^i_\star = I^{\otimes i-1} \otimes A_{\star} \otimes I^{\otimes d-i},
\]
where $\star \in \{+, -, 1\}$, and $A_+$, $A_-$, $A_1$ are the matrix representations of the finite difference stencils for the forward, backward and central differencing schemes approximating the one-dimensional first order derivative.
Specifically, under Dirichlet boundary conditions, $A_+$, $A_-$, and $A_1$ have the following forms
\begin{align*}
    A_+ &=
    \begin{pmatrix}
        0 \\
         & -1 & 1 \\
        & \ldots  & & \ldots \\
        & & & -1 & 1 \\
     & & & & 0
    \end{pmatrix}
    ,\\
    A_- &= 
    \begin{pmatrix}
        0 \\
        -1 & 1 \\
        & \ldots  & & \ldots \\
        & & -1 & 1 \\
      & & & & 0
    \end{pmatrix}
    ,\\
    A_1 &= \frac{1}{2}
    \begin{pmatrix}
        0 \\
        -1 & 0 & 1 \\
         & \ldots  & & \ldots \\
        & & -1 & 0 & 1 \\
      & & & & 0
    \end{pmatrix}
    .
\end{align*}

We can then discretize the time dimension into $m+1$ points $\boldsymbol{\hat{p}}^0, \dots, \boldsymbol{\hat{p}}^m$ with $m = T / \Delta t$, and employ a forward Euler scheme for time stepping on Equation~\eqref{eqn:fokker-planck-ode}.
The result of such time discretization is a linear system of equations
\[
    \vec{{p}}^{k} = M \vec{{p}}^{k-1},
\]
where $M = I^{\otimes d} + L \Delta t$, and $k \in [m]$.
Now, we define the matrix $B$, as 
\begin{align*}
B \coloneqq
    \begin{pmatrix}
        I^{\otimes d} \\
        -M & I^{\otimes d} \\
        & \ldots & \ldots \\
        & &  -M & I^{\otimes d}
    \end{pmatrix}.
\end{align*}
Combining equations for all time steps into one system of equations
\[
    \begin{pmatrix}
        I^{\otimes d} \\
        -M & I^{\otimes d} \\
        & \ldots & \ldots \\
        & &  -M & I^{\otimes d}
    \end{pmatrix}
    \begin{pmatrix}
        {\vec{p}}^0 \\
        {\vec{p}}^1 \\
        \vdots \\
        {\vec{p}}^m
    \end{pmatrix}
    =
    \begin{pmatrix}
        {\vec{p}}^0 \\
        0 \\
        \vdots \\
        0
    \end{pmatrix}.
\]

One can then prepare a quantum state encoding the solution to this system using quantum linear systems algorithms \cite{harrow2009quantum, costa2021optimal}. These utilize $\mathcal{O}(\kappa\log(1/\epsilon))$ queries to an oracle for the initial state and matrix, where $\kappa$ is the condition number of the stencil. The number of  additional gates and qubits is $\mathcal{O}(\text{poly}(\log(T), \log(n), d))$. The whole thing has $\mathcal{O}(\text{poly}(\log(T), \log(n), d)\cdot \kappa\log(1/\epsilon))$ gate complexity if we can prepare the initial state and access the stencil efficiently.

\subsection{Challenges in Monte Carlo Pricing with Quantum PDE solvers}
\label{sec:challenges_pricing_pde}

Here we provide three significant barriers to utilizing quantum PDE solvers as distribution loading subroutines for quantum MCI.

\subsubsection{Insufficiency of History States for Path-Dependent Derivative Pricing}

The output of the quantum PDE solver applied to the FP equation is a superposition encoding the marginals of the Markov chain at each time step, which we call a history state.

\begin{definition}[History State]
  \label{defn:history-state}
  Let $\{X(t) \mid 0 \le t \le T\}$ be a stochastic process over a space $\CX$ such that the probability density of $X(t)$ at any time $t$ is given by $p_{t} \colon \CX \to \R_{{\ge 0}}$. Under some suitable choice of discretization, let the quantum state with amplitudes proportional to $p_{t}^{\alpha/2}$ be given by $|p_{t}^{(\alpha)}\rangle$. For any finite set of times $\CT \subseteq [0,T]$, the corresponding history state is defined as
  \begin{align}
    H_{X,\CT}^{(\alpha)} = \frac{1}{\sqrt{|\CT|}} \bigoplus_{t \in \CT} |p_{t}^{\alpha}\rangle
  \end{align}
\end{definition}

The reason for only encoding marginals comes from the definition of the Fokker-Planck equation, and the need to encode the history comes from the need to unitarily-embed the FP dynamics. We present the following result regarding the inability to price path-dependent derivatives given access to only history states.

\begin{lemma}
  \label{lem:history-insufficient-discrete}
There exist stochastic processes $X_1(t),X_2(t)$ over the same "discrete" domain $\CX$ and a starting point $x_{0}$, such that the "history" states (\ref{defn:history-state}) corresponding to running $X_{1}(t)$ and $X_{2}(t)$ for time $T$, starting from $x_{0}$ are identical, and a barrier payoff function $F$ such that the price computed by integrating $F$ over the paths $\{X_1(t)|0 \le t \le T\}$ and $\{X_2(t)|0 \le t \le T\}$ differs by greater than $c$, where $c$ is a universal constant.
\end{lemma}
\begin{proof}
  We choose any discrete set $\CX$ without loss of generality and $x_{0}$ some element of $\CX$. Define an arbitrary subset $\CX^{'} \subseteq \CX$ such that $x_{0} \not\subseteq \CX^{'}$. We define a barrier payoff $F: \CX^{T+1} \to \R$ as 
  \begin{align*}
      F(\vec{x}) = \begin{cases}
          1 & \exists i \in [T+1], ~\text{s.t.}~ x_i \in \mathcal{X}'\\
          0 & \text{o.w.}
      \end{cases}
  \end{align*} Define the two stochastic processes as follows
  \begin{itemize}
    \item At $t = 0$, $X_{1}$ makes a uniformly random step on to the whole domain. At every subsequent step, $X_{1}$ deterministically does not make any move from any point of the domain (the transition matrix is the identity).
    \item $X_{2}$ makes a uniformly random step on to the whole domain at every time step.
  \end{itemize}
  It is easy to see that at $t = 0$ both stochastic processes have the marginal distribution $\delta(x - x_{0})$ and at every subsequent time step the marginal distributions for both distributions are uniform over $\CX_{0}$. However, the integral of the payoff function over the distribution of the paths differs by more than a constant if the stochastic processes are run for time $\Theta(\log |\CX|)$. To see this, we notice that for $X_{1}$, half the paths jump to points in $\CX'$ at $t = 1$ and after time $t=1$ every path is stationary. Thus the corresponding price is $1/2$. On the other hand, every path in $X_{2}$ jumps to an independent new point in $\CX$ at each time step. After $\Theta(|\log |\CX|)$ steps, it is overwhelmingly likely (say with probability $> 0.99$) that any individual path enters $\CX'$ at some point. The expected price equals the probability that the payoff for any given path is $1$, and thus the payoff corresponding to $X_{2}$ is at least $0.99$.
\end{proof}

The above lemma shows that in general, history states may be completely insufficient for the computation of integrals over paths even when given any finite number of states as a resource, or allowed any finite number of calls to a state preparation oracle for the quantum state. However, the distributions in  Lemma~\ref{lem:history-insufficient-discrete} are quite different from typical stochastic processes encountered in finance. In particular, the distribution is over a discrete support, and the parameters describing the distribution are rapidly changing with time. We next show that these properties are in fact unnecessary and there are continuous stochastic processes that exhibit similar properties for a properly chosen barrier option. These examples are based on the same essential idea as above, namely that the marginal distributions of a stochastic process do not fully capture its transition dynamics. Note that these generic lower bound arguments cannot be made for an arbitrary distribution: as we have seen, payoffs with dependence only on the marginals such as path-independent payoffs can in fact be evaluated using only history states.

\begin{lemma}
There exist stochastic processes $\vec{X}_{1}, \vec{X}_{2}$ on the domain $\R^{d}$ described by stochastic differential equations $d \vec{X}_{j} = \vec{\mu}(\vec{X}_{j}, t) dt + \boldsymbol{\sigma} (\vec{X}_{j}, t) d \vec{W}_t$ for $j = {1,2}$, such that the drift and volatility coefficients are all Lipschitz continuous, and the corresponding "history" states (\ref{defn:history-state}) have a trace distance of $\le \epsilon = 1/2^{\Omega(d)}$. Additionally, there is a time-dependent barrier payoff function $F$ such that the price computed by integrating $F$ over the paths $\{\vec{X}_1(t)|0 \le t \le T\}$ and $\{\vec{X}_2(t)|0 \le t \le T\}$ differs by greater than $c$, where $c$ is a universal constant.
\end{lemma}
\begin{proof}
Let $\alpha,b,\sigma$ be real  parameters such that $2\alpha b \ge \sigma^2$. Define the stochastic processes $X_1,X_2$ as follows:
\begin{align}
    dX_1(t) &= \alpha(b - X(t)) \,dt + \sigma \sqrt{X(t)} dW_t \\
    dX_2(t) &= 4\alpha(b - X(t)) \,dt + 2\sigma \sqrt{X(t)} dW_t \\
\end{align}
Note that both these processes are well known instances of the "square root diffusion", a commonly used model for interest rates. It is a well-known fact~\cite[Section 3.4.1]{glasserman2004monte} that the stationary distribution of both processes is identical to $\frac{\sigma^2}{4\alpha} \chi_{\frac{4b\alpha}{\sigma^2}}$, where $\chi_{\frac{4b\alpha}{\sigma^2}}$ is a central $\chi$-squared random variable with $\frac{4b\alpha}{\sigma^2}$ degrees of freedom. Let both processes have their initial point sampled from this distribution. It is clear from stationarity that the history states in this setting are identical. On the other hand, we can define a derivative whose payoff on a realization $(x_1,\dots,x_T)$ is given by $\max(K,\hat{\sigma}(x_1,\dots,x_T))$ where $\hat{\sigma}$ is the empirical covariance of the path. The empirical covariances of $X_1,X_2$ are concentrated around values whose difference is bounded below as $T \to \infty$. Setting $K$ to be between these two values, we can ensure that the expected payoff value differs by more than a constant.
\end{proof}

\subsubsection{Overhead of Integration: Curse of Dimensionality}
\label{sec:curseofdim}
The approach in Section~\ref{sec:algo_for_solving_fp} solves a PDE whose solution is $\ell_1$ normalized. However, the state output by the quantum algorithm must be $\ell_2$ normalized. This issue is also present with the quantum-walk algorithms \cite{apers2019quantumfastforwardingmarkovchains}. More specifically, the quantum PDE solver for the Fokker-Planck prepares
\begin{align*}
&H_{X,\CT}^{(2)} = \frac{1}{\sqrt{|\CT|}} \bigoplus_{t \in \CT} |p_{t}^{2}\rangle, \\
&|p_{t}^{2}\rangle = \sum_{\vec{x}} \frac{p_t(\vec{x})}{\lVert p_t\rVert_2}|\vec{x}\rangle.
\end{align*}
We can extract the distribution corresponding to the final time point $T$ via amplitude amplification and a cost of $\mathcal{O}\left(\sqrt{\lvert \mathcal{T}\rvert}\right)$. This means in terms of time, we would only obtain a sublinear complexity, not exponentially reduced. However, we would still use $\mathcal{O}\left(\log(\lvert \mathcal{T}\rvert)\right)$ space.

Recall that quantum MCI works with states of the form 
\begin{align*}
|p_{t}^{1}\rangle = \sum_{\vec{x}} \sqrt{p_t(\vec{x})}|\vec{x}\rangle,
\end{align*}
and is not compatible with the $|p^{2}_t\rangle$ states. The fix to this is to utilize quantum inner product estimation (QIPE) \cite[Lemma 4.2]{kerenidis2018}. Specifically, if we want to price a path-independent derivative with payoff $f$, then we apply QIPE to
\begin{align*}
&|p_{T}^{2}\rangle = \sum_{\vec{x}} \frac{p_T(\vec{x})}{\lVert p_T\rVert_2}|\vec{x}\rangle\\
& |f\rangle = \sum_{\vec{x}} \frac{f(\vec{x})}{\lVert f \rVert_2}|\vec{x}\rangle,
\end{align*}
which utilizes $\mathcal{O}\left(\frac{\lVert p_T\rVert_2\lVert f\rVert_2}{\epsilon}\right)$ queries to unitaries for preparing these states. While a quadratic speedup over a classical estimator, this is substantially worse than classical MCI, i.e. the quantity $\lVert p_T\rVert_2\lVert f\rVert_2$ is likely to be exponential in the dimension, unless there is some perfect cancellation. Note that even preparing the state $|f\rangle$ could be computationally expensive, i.e. with black-box state prep it costs $\Omega(\lVert f \rVert_2)$. Note that QMCI (1) does not need to prepare the $|f\rangle$ state and (2) has a dependence on $f$ and $p_T$ that is more like $p_T \cdot f$, which can be substantially smaller than QIPE's.

Alternatively, one could transform  the FP PDE into a nonlinear one in order to produce a qsample, i.e. $|p^{1}\rangle$. However, quantum nonlinear PDE solvers have stringent conditions under which they work efficiently \cite{liu2021efficient}.%

\subsubsection{Insufficient Runtime with Quantum PDE Solvers}
The matrices that arise as discretizations of PDEs are typically sparse and have efficiently row-computable entries, which makes them amenable to quantum linear-algebra algorithms. Unfortunately, quantum algorithms, like classical iterative methods, operate on the spectrum of the system matrix and thus have a polynomial dependence on the condition number. It is not clear how to sufficiently bound this condition number for the task of quantum Monte Carlo integration. This is because any additional inverse-polynomial dependence in the error can destroy the quantum speedup. Preconditioning may be one way to elleviate this issue \cite{rendon2025exponentialimprovementasianoption}. In addition, quantum PDE solvers can suffer a substantial slowdown when the PDE infinitesimal generator is not a normal operator \cite{An2022_theoryqode}. %

An alternative to FDM that is applicable to certain classes of PDEs are the pseudo-spectral methods \cite{Childs2021highprecision}. A pseudospectral method approximates the solution to the PDEs using polynomial approximations, which results in a system of ordinary differential equations.  The Chebyshev spectral method is particularly efficient for quantum due to quantum Fourier transform. In the case of FP, once the time discretiation has been applied, under certain conditions the result is a system of Elliptic PDEs. Unfortunately, there are are additional spectral constraints \cite{Childs2021highprecision}, such as diagonal dominance, that are required for the quantum algorithm to be efficient. This condition appears to be too strong to assume in general for financial applications.

\section*{Acknowledgments}
D.H., S.C., and Y.S. thank Guneykan Ozgul for helpful technical discussions, and their colleagues at the Global Technology Applied Research center of JPMorganChase for their support and participation in early stages of ideation.  A.W.H. and J.L. were funded by NSF grant PHY-2325080.

\section*{Disclaimer}
This paper was prepared for informational purposes by the Global Technology Applied Research center of JPMorgan Chase \& Co. This paper is not a product of the Research Department of JPMorgan Chase \& Co. or its affiliates. Neither JPMorgan Chase \& Co. nor any of its affiliates makes any explicit or implied representation or warranty and none of them accept any liability in connection with this paper, including, without limitation, with respect to the completeness, accuracy, or reliability of the information contained herein and the potential legal, compliance, tax, or accounting effects thereof. This document is not intended as investment research or investment advice, or as a recommendation, offer, or solicitation for the purchase or sale of any security, financial instrument, financial product or service, or to be used in any way for evaluating the merits of participating in any transaction.

\appendix

\section{Additional Proofs for Section \ref{sec:reduced_qubit_estimates}}

Recall the multi-asset GBM model: 
\multiassetgbm*

We can simulate a path process $\mathbf{S}^{(\Delta)}$ as follows. First,  draw $\mathbf{Z} = (\vec{Z}_0, \dots, \vec{Z}_{T-1}) \in \mathbb{R}^{d\times T}$ all independently from $\mathcal{N}(0, \mathbb{I}_{d\times d})$. Let $\mathbf{L} \in \mathbb{R}^{d\times d}$ be Cholesky factor of the correlation matrix $\rho$. Then compute 
\begin{align*}
    \vec{S}_{k} = \vec{S}_0\exp\left(k\Delta \cdot \vec{\mu} + \sum_{r=0}^{k-1}\sqrt{\Delta}\vec{\sigma}\circ\mathbf{L}\vec{Z}_r\right) =g_t(\vec{S}_{k-1}, \vec{Z}_{k-1}) = g\left(\vec{Z}_0, \dots, \vec{Z}_{k-1}\right),
\end{align*}
with $\mathbf{S}^{(\Delta)}(k) = \vec{S}_k$.

\subsection{Proof of Theorem \ref{thm:gaussian-relative}}
\label{subsec:proof-thm-gauss-rela}
\begin{lemma}
\label{lem:gbm_truncation}
Let $f$ be a payoff function that is upper bounded by a $B$-Lipschitz linear function in all components of $g(\mathbf{Z})$. Then it suffices to truncate each standard Gaussian to an $\ell_{\infty}$ ball of radius $\mathcal{O}\left(\sqrt{d}T^{3/2}\ln(BTd/\epsilon_{\text{trunc}})\right)$ to achieve an $\epsilon_{\text{trunc}}$ truncation error.
\end{lemma}
\begin{proof}

The assumption on  the payoff implies that 
\begin{align*}
\mathbb{E}[f\circ g\left(\mathbf{Z}\right)] =
    \int_{\mathbb{R}^{d \times T}} f \circ g(\mathbf{Z}) d\mathbf{Z} \leq B\sum_{t, k} \int_{\mathbb{R}^{d \times T}} S_{t,k} d\mathbf{Z} ,
\end{align*}
where $S_{t,k}$ is the price of the $k$-th asset at the $t$-th time step. Thus we can simply restrict to determining truncation bounds for $\int S_{t,k} d\vec{\mathbf{Z}}$ and scale the error down accordingly. %

Consider truncating $\mathbf{Z}$ to $[-R, R]^{dt}$, then the error for $S_{t,k}$ is
\begin{align}
&\frac{\exp(t\mu_k)}{\CN^{dT}}\int_{z \notin [-R,R]^{dt}} \exp(\sum_{r=0}^{t-1} \sigma_k[\mathbf{L}\vec{Z}_r]_k) \exp{-\lVert \mathbf{Z}\rVert_{F}^2 /2}  \,d\mathbf{Z} \\
&=\frac{\exp(t\mu_k + t\sigma_k^2)}{\CN^{dt}}\int_{z \notin [-R, R]^{dT}} \exp(-\frac{1}{2}\sum_{r=0}^{t-1}\left(\vec{Z}_r - \sqrt{2}\sigma_k\mathbf{L}_{k,\cdot}\right)^2) d\mathbf{Z} \\
&\leq \frac{\exp(t\mu_k + t\sigma_k^2)}{\CN^{dt}}\int_{\lVert z\rVert_2 \geq R}\exp(-\frac{1}{2}\sum_{r=0}^{t-1}\left(\vec{Z}_r - \sqrt{2}\sigma_k\mathbf{L}_{k,\cdot}\right)^2) d\mathbf{Z},
\end{align}
where the first equality follows from completing-the-square and the second inequality follows from the $\ell_2$ norm dominating  the $\ell_{\infty}$ norm.  Thus the above corresponds to the tail of a $dt$ dimensional Gaussian with mean $\sqrt{2}\sigma_k\mathbf{L}_{k,\cdot}$. Standard Gaussian concentration gives that for each component we can truncate to
\begin{align*}
     \sqrt{2}\sigma_k + \mathcal{O}\left(\sqrt{dt}\ln(1/\epsilon)\right)
\end{align*}
for an $\mathcal{O}\left(\epsilon\right)$ error. If we scale $\epsilon$ and maximize over $k \in [d]$, we get
\begin{align*}
    R = \sqrt{2}\max_k\sigma_k + \mathcal{O}\left(\sqrt{d}T^{3/2}\ln(BTd/\epsilon)\right).
\end{align*}
\end{proof}

\gbmRelativeError*
\begin{proof}

Let $\mathcal{M}$ be a uniform grid in $[-R, R]^{td}$ with $R/N$ spacing in each dimension. Suppose $m = d\cdot t$. The previous lemma gives a sufficient bound on $R$ to achieve $\epsilon$ truncation error. Using Lemma~\ref{lem:left_rule_error},  the error from the multi-dimensional left-point rule is bounded by
\begin{align}
\label{eqn:quadrature_error}
\sum_{\vec{x} \in \mathcal{M}} \sup_{\vec{y} \in \mathcal{C}_{\vec{x}}}\lVert p(\vec{y})\nabla f(\vec{y})\rVert_{\infty} \left(\frac{2R}{N}\right)^{m} \cdot \frac{dR}{N} + \sum_{\vec{x} \in \mathcal{M}} \sup_{\vec{y} \in \mathcal{C}_{\mathbf{x}}}\lVert f(\vec{y})\nabla p(\vec{y})\rVert_{\infty} \left(\frac{2R}{N}\right)^{m} \cdot \frac{dR}{N}.
\end{align}

Since the payoff is assumed to be piecewise linear in all assets across all time points, with maximum Lipschitz constant $B$, we can  analyze each linear component separately. Thus without loss of generality, we can take $g(\vec{z}) = \exp(\mu t + \sigma \vec{\alpha}\cdot \vec{z})$.  With $\vec{z} \in \mathbb{R}^m$ being a linearization of $\mathbf{Z}$ from the previous proof, and $\vec{\alpha}$ is chosen such that $\vec{\alpha}\cdot \vec{z} = \sum_{r=0}^{t-1} [\mathbf{L}\vec{Z}_r]_k$ for some arbitrary $k \in [d]$.

Then, we obviously have 
\begin{align*}
p(\vec{z}) = (2\pi)^{-m/2}\exp(-\lVert \vec{z}\rVert_2^2/2).    
\end{align*}
 Then we have that 
\begin{align*}
&\lVert \nabla g(\vec{y}) \rVert_{\infty} \leq \sigma g(\vec{y})\\
&\lVert \nabla p(\vec{y}) \rVert_{\infty} \leq R p(\vec{y}).
\end{align*}

Let  $\vec{w}(\vec{x}) \in \mathcal{C}_{\vec{x}}$ denote the maximizer in each grid cell $\mathcal{C}_{\vec{x}}$. If we can replace each $\vec{w}(\vec{x})$ with the minimum in $\mathcal{B}_{\infty}(\vec{x}, \frac{R}{N})$, then, by standard Riemann integration, the error is bounded by 
\begin{align*}
    \mathcal{O}\left(\lvert \int_{[-R, R]^d} f\circ g(\vec{x})p(\vec{x})d\vec{x} \frac{mR}{N}\rvert\right),
\end{align*} giving us the desired relative-error approximation. We will actually end up with something close but slightly different:
\begin{align*}
    \mathcal{O}\left(\lvert \int_{[-R, R]^d} g(\vec{x})p(\vec{x})d\vec{x} \frac{mR}{N}\rvert\right) = \mathcal{O}\left(\mathbb{E}[g(\mathbf{Z})]\frac{mR}{N}\right).
\end{align*}

Let $\vec{r}(\vec{x})$ denote  the minimizer of $g(\vec{y})p(\vec{y})$ over $\mathcal{C}_{\vec{x}}$. Then using the known form of $g$ and $p$, the relative change is 

\begin{align*}
e^{\sigma\vec{\alpha} \cdot [\vec{w}(\vec{x}) - \vec{r}(\vec{x})] - [\lVert \vec{w}(\vec{x})\rVert_2^2/2 - \lVert \vec{r}(\vec{x})\rVert_2^2/2 ]} &\leq e^{\sigma m\lVert \vec{w}(\vec{x}) - \vec{r}(\vec{x})\rVert_2 + \frac{1}{2}\lvert [\lVert \vec{w}(\vec{x})\rVert_2^2 - \lVert \vec{r}(\vec{x})\rVert_2^2]\rvert} \\
&\leq e^{\sigma m^{3/2}\frac{R}{N} + m\frac{R^2}{N^2}} ,
\end{align*}
which we can make $\mathcal{O}\left(1\right)$ by taking $N = \mathcal{O}\left(\max(\sigma,1)m^{3/2}R\right)$. Hence $\frac{\sup_{\vec{y}\in\mathcal{C}_{\vec{x}}}g(\vec{y})p(\vec{y})}{\inf_{\vec{y}\in\mathcal{C}_{\vec{x}}}g(\vec{y})p(\vec{y})}$ can be made uniformly $\mathcal{O}(1)$.

Hence the discretization error is bounded by
\begin{align*}
\mathcal{O}\left(\sum_{\vec{x}\in\mathcal{M}}\frac{\sup_{\vec{y}\in\mathcal{C}_{\vec{x}}}g(\vec{y})p(\vec{y})}{\inf_{\vec{y}\in\mathcal{C}_{\vec{x}}}g(\vec{y})p(\vec{y})}\max(\sigma, 1) R \inf_{\vec{y}\in\mathcal{C}_{\vec{x}}} g(\vec{y})p(\vec{y})\cdot\left(\frac{2R}{N}\right)^m \frac{dR}{N}\right)=\mathcal{O}\left(\max(\sigma, 1) R \mathbb{E}[g(\mathbf{Z})]\cdot\frac{dR}{N}\right).
\end{align*}

From Lemma~\ref{lem:gbm_truncation}, we have it suffices to take $R=\mathcal{O}\left(\sqrt{d}T^{3/2}\ln(BTd/\epsilon_{\text{trunc}})\right)$. Thus we can take $\log(N)=\mathcal{O}\left(\log\left((1+\sigma)dT/\epsilon_{\text{disc}}\right)\right)$ to get that 
\begin{align*}
    \lvert \int_{[-R, R]^d} f\circ g(\vec{x})p(\vec{x})d\vec{x} - \sum_{\vec{x}\in\mathcal{M}} f\circ g(\vec{x}) p(\vec{x})\left(\frac{2R}{N}\right)^d\rvert = \mathcal{O}\left(\mathbb{E}[g(\mathbf{Z})]\epsilon_{\text{trunc}}\right).
\end{align*}

\end{proof}

\section{Additional Proofs for Section \ref{sec:subroutines_for_distribution_loading}}

\subsection{Proof of Lemma \ref{lem:poly_approx_qet}}
\label{lem:proofLemQET}
\polyApproxStatePrep*
\begin{proof}

Without loss of generality, we can assume $\mathcal{Z}_p < 1$.

We will denote $h$ to be our assumed $d_{\delta}$ degree,  $\delta$-uniform polynomial approximation \eqref{eqn:shft_sqrt_p}, and $|\widetilde{\Psi}\rangle$ will be the state
\begin{align}
\label{eqn:h_state}
|\Psi_{h} \rangle = \frac{1}{\mathcal{Z}_h} \sum_{k=0}^{N-1}\frac{\sqrt{b-a}}{\sqrt{N}}h\left(\sin\left(\frac{k}{N}\right)\right)|k\rangle.
\end{align}

The above is the result of applying QET with polynomial $h$ to the block-encoding \eqref{eqn:sin_block_encoding}, with appropriately redefined normalization constant. Additionally to ensure the polynomial is valid we will choose $\delta$ and scale $h^2$, by say $\frac{1}{2}$ since $\lVert p \rVert_{\infty} \leq 1$, to ensure  $\lVert h \rVert_{\infty} \leq 1$ for QET viability. This will not impact the overall result. After boosting the block-encoding, the scaling factor will vanish and only adds a constant factor overhead to the query complexity.

We can bound the TVD in the computational basis between $|\Psi_h\rangle$ and $|\widehat{\Psi}_p\rangle$ (state in \eqref{eqn:one-d-disc-q-samp}) as
\begin{align*}
\frac{(b-a)}{2N}\sum_{k=0}^{N-1}\frac{\lvert\mathcal{Z}_{{p}}^2 h^2(\sin(k/N)) - \mathcal{Z}_{h}^2{p}(x_k)\rvert}{\mathcal{Z}_{h}^2\mathcal{Z}_{p}^2}\leq \frac{3(b-a)\delta}{2\mathcal{Z}_{p}^2} + \frac{b-a}{2\mathcal{Z}_{p}^2}\lvert 1 - (\mathcal{Z}_{p}^2/\mathcal{Z}_{h}^2) \rvert,
\end{align*}
where we used Lemma~\ref{lem:poly_approx_lem}.

Thus for $\delta = o\left(\frac{b-a}{\mathcal{Z}_p^2}\right)$,

\begin{align*}
\lvert 1 - (\mathcal{Z}_{p}^2/\mathcal{Z}_{h}^2) \rvert &= \lvert 1 - \frac{\frac{(b-a)}{N}\sum_k {p}(x_k)}{\frac{(b-a)}{N}\sum_k h^2(\sin(k/N))} \rvert  \\ 
&\leq \lvert 1- \frac{\mathcal{Z}_p^2}{\mathcal{Z}_p^2 + 3\delta(b-a)} \rvert \\
&= \mathcal{O}\left(\frac{\delta(b-a)}{\mathcal{Z}_p^2}\right).
\end{align*}

Hence TVD is bounded by $ \mathcal{O}\left(\frac{(b-a)^2\delta}{\mathcal{Z}_p^4}\right)$,
so it suffices to take $\delta = \mathcal{O}\left(\frac{\epsilon\mathcal{Z}_p^4}{(b-a)^2}\right)$ for $\epsilon$ TVD. The amplification cost from Theorem~\ref{thm:amplitude_amp} is 
 \begin{align*}
1/\mathcal{Z}_h &\leq \frac{1}{\sqrt{\mathcal{Z}_p^2 -3 (b-a)\delta}} \\
&=\frac{1}{\sqrt{\mathcal{Z}_p^2 - 3\epsilon \mathcal{Z}_p^4/(b-a)}} \\
&\leq \mathcal{Z}_p^{-1}\frac{1}{\sqrt{1 - 3\epsilon}}\\
&=\mathcal{O}\left(\mathcal{Z}_p^{-1}\right).
 \end{align*}
The block-encoding in Equation \eqref{eqn:sin_block_encoding} requires $\log(N)$ basic gates.
\end{proof}

\subsection{Proof of Lemma \ref{lem:black-box-state-prep}}

The following is the standard Grover black-box state preparation \cite{grover2000synthesis, Sanders_2019}. We include a proof of the result with our notation for completeness. Additionally, we have switched to a weaker metric, i.e. TVD of measurement distributions in the computational basis. This turns out to be sufficient for our purposes.

\begin{lemma}[Black-box State Preparation]
\label{lem:black-box-state-prep}
Consider $N$ uniform grid points, $\{x_k\}_{k=0}^{N-1}$, over $[a, b]$, and let
 $\mathcal{Z}_p^2 := \frac{b-a}{N}\sum_{k=0}^{N-1}p(x_k)$. Let $p : [a, b] \rightarrow \mathbb{C}$. Suppose we have a $\delta$-accurate quantum binary oracle $O_p : |x\rangle|0\rangle \rightarrow |x\rangle|\widetilde{p}(x)\rangle$ for $p$, i.e. $\lVert p - \widetilde{p} \rVert_{\infty} < \delta$, and $\max_{[a, b]} \lvert \widetilde{p} \rvert \leq \Lambda$ provided. If $\delta <\epsilon^2 \cdot \frac{\mathcal{Z}_p^8}{\Lambda^3(b-a)^4}$, then, we can prepare a $\log_2(N)$-qubit quantum state $|\widetilde{\Psi}\rangle$ with a measurement distribution in the computational basis that is at most $\epsilon$ in TVD from the state in \eqref{eqn:one-d-disc-q-samp} using 
\begin{align*}
\mathcal{O}\left(\frac{1}{\mathcal{Z}_p}\right)
\end{align*}
queries to $O_p$ and $\mathcal{O}\left(\frac{\log^3(\Lambda/\delta)}{\mathcal{Z}_p}\right)$ additional one- and two-qubit gates.
\end{lemma}
\begin{proof}
Note that we can prepare a state like \eqref{eqn:h_state} in the following way:
\begin{align*}
    \frac{1}{\sqrt{N}}\sum_{k=0}^{N-1}|k\rangle|0\rangle|0\rangle|0\rangle &\rightarrow_{(16\Lambda)^{-1}O_p}  \frac{1}{\sqrt{N}}\sum_{k=0}^{N-1}|k\rangle|(16\Lambda)^{-1}\widetilde{{p}}(x_k)\rangle|0\rangle|0\rangle \\ &\rightarrow_{\text{sqrt}}  \frac{1}{\sqrt{N}}\sum_{k=0}^{N-1}|k\rangle|{\sqrt{(16\Lambda)^{-1}\widetilde{{p}}(x_k)}}\rangle|0\rangle|0\rangle \\
    &\rightarrow_{\sin^{-1}}  \frac{1}{\sqrt{N}}\sum_{k=0}^{N-1}|k\rangle|{\sqrt{(16\Lambda)^{-1}\widetilde{{p}}(x_k)}}\rangle|\widetilde{\sin^{-1}}({\sqrt{(16\Lambda)^{-1}\widetilde{{p}}(x_k)}}\rangle|0\rangle.
\end{align*}
The division by $16\Lambda$ requires $\mathcal{O}(\log_2(\Lambda))$ bit shifts. The square-root requires $\mathcal{O}\left(\log^2(\Lambda/\delta)\right)$ gates and the square-root error is below our precision, so can be ignored.

Then we apply a bank of controlled $R_y$ rotations, each  controlled on a single qubit of the $\sin^{-1}$ register and applied to the last qubit. After uncomputing, we apply exact amplitude amplification (Theorem \ref{thm:amplitude_amp}) to get the state:
\begin{align*}
|\Psi_{O}\rangle = \frac{1}{\widetilde{\mathcal{Z}_O}} \sum_{k=0}^{N-1}\frac{\sqrt{b-a}}{\sqrt{N}} a_k|k\rangle,
\end{align*}
where $a_k = \sin(\widetilde{\sin^{-1}}(\sqrt{4\Lambda^{-1}\widetilde{p}(x_k)}))$.

Note that by assumption $\lVert p - \widetilde{p} \rVert_{\infty} < \delta$, so
\begin{align*}
\lvert  \sqrt{(16\Lambda)^{-1}\widetilde{p}(x_k)} - \sqrt{(16\Lambda)^{-1}p(x_k)}\rvert < \sqrt{\frac{\delta}{\Lambda}},
\end{align*}
where we at least want $\sqrt{\frac{\delta}{\Lambda}} < \frac{1}{4}$. This is so that $\lVert \sqrt{(16\Lambda)^{-1}\widetilde{p}} \rVert_{\infty} < \frac{1}{2}$ and $\lVert \sqrt{(16\Lambda)^{-1}{p}} \rVert_{\infty} < \frac{1}{2}$.

Since we are taking $\sin^{-1}$ of $\mathcal{O}(\log(\Lambda/\delta))$ bit numbers. Thus it suffices for the $\sin^{-1}$ polynomial approximation error to be $\mathcal{O}(\sqrt{\frac{\delta}{\Lambda}})$, so the gate complexity for the $\sin^{-1}$ is $\mathcal{O}\left(\log^3(\Lambda/\delta)\right)$. Note there is an error from the square-root due to Newton's method \eqref{eqn:sqrt_gate_compl}, but the $\sin^{-1}$ is the dominating error and gate cost.

By $\lvert \sin(x) - \sin(y)| \leq \lvert x -y \rvert$ and $\forall x, y \in [-\frac{1}{2}, \frac{1}{2}], \lvert \sin^{-1}(x) - \sin^{-1}(y) \rvert \leq \frac{3}{2}\lvert x - y \rvert$:
\begin{align*}
 \lvert a_k - \sqrt{(16\Lambda)^{-1} \widetilde{p}(x_k)}\rvert &\leq \lvert \widetilde{\sin^{-1}}(\sqrt{(16\Lambda)^{-1}\widetilde{p}(x_k)}) - \sin^{-1}\left(\sqrt{(16\Lambda)^{-1} {p}(x_k)}\right) \rvert \\
 &\leq \lvert \widetilde{\sin^{-1}}(\sqrt{(16\Lambda)^{-1}\widetilde{p}(x_k)}) - {\sin^{-1}}(\sqrt{(16\Lambda)^{-1}\widetilde{p}(x_k)})\rvert \\ &+ \lvert {\sin^{-1}}(\sqrt{(16\Lambda)^{-1}\widetilde{p}(x_k)}) - {\sin^{-1}}(\sqrt{(16\Lambda)^{-1}{p}(x_k)})\rvert\\
 &\leq \mathcal{O}\left(\sqrt{\frac{\delta}{\Lambda}}\right) + \frac{3}{2}\lvert \sqrt{\Lambda^{-1}\widetilde{p}(x_k)} - \sqrt{\Lambda^{-1}{p}(x_k)}\rvert\\
 &= \mathcal{O}\left(\sqrt{\frac{\delta}{\Lambda}}\right).
\end{align*}

Also by our assumptions, the above implies:
\begin{align*}
    \lvert a_k - \sqrt{(16\Lambda)^{-1}p(x_k)}\rvert = \mathcal{O}\left( \sqrt{\frac{\delta}{\Lambda}}\right) \implies .\lvert 16 \Lambda a_k^2 - {p(x_k)}\rvert = \mathcal{O}\left( \sqrt{\Lambda \delta}\right).
\end{align*}

We will re-express the state $|\Psi_O\rangle$ with a different normalization:
\begin{align*}
|\Psi_{O}\rangle = \frac{1}{{\mathcal{Z}_O}} \sum_{k=0}^{N-1}\frac{\sqrt{b-a}}{\sqrt{N}} \cdot 4\sqrt{\Lambda} a_k|k\rangle.
\end{align*}

Hence, the TVD in the computational basis is bounded by
\begin{align*} \frac{(b-a)}{2N}\sum_{k=0}^{N-1}\frac{\lvert\mathcal{Z}_{{p}}^2 \cdot 16\Lambda a_k^2 - \mathcal{Z}_{O}^2p(x_k)\rvert}{\mathcal{Z}_{O}^2\mathcal{Z}_{p}^2}
&= \frac{(b-a)}{2N}\sum_{k=0}^{N-1}\left(16\Lambda a_k^2\frac{\lvert \mathcal{Z}_p^2 - \mathcal{Z}^2_{O}\rvert}{\mathcal{Z}_O^2 \mathcal{Z}_p^2} + \frac{\lvert 16 \Lambda a_k^2 - p(x_k)\rvert}{\mathcal{Z}_p^2}\right)
\\
&=  \mathcal{O}\left(\frac{(b-a)\Lambda
}{\mathcal{Z}_{p}^2}\lvert 1 - (\mathcal{Z}_{p}^2/\mathcal{Z}_{O}^2) \rvert + \frac{{(b-a)\sqrt{\Lambda \delta}}}{\mathcal{Z}_{p}^2}\right).
\end{align*}

For sufficiently small $\delta$,
\begin{align*}
\lvert 1 - (\mathcal{Z}_{p}^2/\mathcal{Z}_{O}^2) \rvert &= \lvert 1 - \frac{\frac{(b-a)}{N}\sum_k {p}(x_k)}{\frac{(b-a)}{N}\sum_k 16\Lambda a_k^2} \rvert  \\ 
&\leq \lvert 1- \frac{\mathcal{Z}_p^2}{\mathcal{Z}_p^2 + (b-a)\sqrt{\Lambda\delta}} \rvert \\
&= \mathcal{O}\left(\frac{(b-a)\sqrt{\delta\Lambda}}{\mathcal{Z}_p^2}\right).
\end{align*}

Thus, the TVD is bounded by

\begin{align*}
\mathcal{O}\left(\frac{\Lambda(b-a)^2\sqrt{\delta\Lambda}}{\mathcal{Z}_p^4}\right),
\end{align*}

so $\delta = \frac{\epsilon^2\mathcal{Z}_p^8}{\Lambda^3(b-a)^4}$ suffices for $\mathcal{O}(\epsilon)$ error. The amplification cost from Theorem~\ref{thm:amplitude_amp} is 
 \begin{align*}
1/\mathcal{Z}_h &\leq \frac{1}{\sqrt{\mathcal{Z}_p^2 -3 (b-a)\delta}} \\
&=\frac{1}{\sqrt{\mathcal{Z}_p^2 - \frac{3\epsilon^2 \mathcal{Z}_p^8}{\Lambda^3(b-a)^3}}} \\
&\leq \mathcal{Z}_p^{-1}\frac{1}{\sqrt{1 - 3\epsilon^2}}\\
&=\mathcal{O}\left(\mathcal{Z}_p^{-1}\right),
 \end{align*}
 where we use that $\mathcal{Z}_p^2/(b-a) \leq \Lambda$.
The block-encoding in Equation \eqref{eqn:sin_block_encoding} requires $\log(N)$ basic gates.
\end{proof}

We also inlude a result counting the resources for loading the payoff. This uses effectively the same techniques as the previous lemma.
\begin{lemma}[Payoff Loading]
\label{lem:payoff_loading}
Suppose we have a quantum circuit $O_{f}$ on $\mathcal{O}\left(\log(N)\right)$ output qubits that uses  $\mathcal{N}_f$ gates  that computes $O_{f} |x\rangle|0\rangle \mapsto |x\rangle|f(x)\rangle$ for some sequence of arithmetic functions $f : \mathcal{M} \rightarrow \mathbb{R}_{+}$ over a $\mathcal{O}(N^{-1})$-precise grid (i.e. precision of the output register, holding $f(x)$). Then we can perform the operation 
\begin{align*}
    |x\rangle|0\rangle \mapsto {\widetilde{f}(x)}|x\rangle|0\rangle + \sqrt{1 -  \widetilde{f}^2(x)}|x\rangle|1\rangle,
\end{align*}
where $\lVert \Lambda\widetilde{f}^2 - f\rVert_{\infty} = \mathcal{O}(N^{-1})$ using $\mathcal{O}\left(\mathcal{N}_f+\log^3\left(\Lambda N\right)\right)$
one- and two-qubit gates.
\end{lemma}
\begin{proof}
    Similar to Lemma \ref{lem:black-box-state-prep} we utilize coherent arithmetic for $a_k\sin^{-1}$, $\sqrt{\cdot}$, and a bank of $R_y$ rotations to compute  $a_k=\sin\left(\widetilde{\sin}^{-1}(\sqrt{f(x_k)/\Lambda})\right)$ onto an amplitude for $x_k \in \mathcal{M}$. As in the proof of the previous lemma, the main gate complexity comes from approximate $\sin^{-1}$ using $\mathcal{O}\left(\log^3\left(\Lambda N\right)\right)$ one- and two qubit gates.
\end{proof}

\subsection{Proof of Lemma \ref{lem:poly-approx-int-cir}}
\label{subsec:proof-lem-int-cir-poly}
The characteristic function of the process $\int_{t}^{t+2\Delta}V(s)ds$ (Integral of CIR), where $V(s)$ is a CIR process, is the following function \cite{broadie2006exact}:
\begin{align}
\Phi(a) &= \mathbb{E}[\exp(ia\int_{t}^{t+2\Delta} V(s) ds) | V(t), V({t+2\Delta})]\\
&=\frac{\gamma_{a}\sinh(\kappa \Delta)}{\kappa \sinh(\gamma_{a}\Delta)}\exp\left(\frac{V
(t) + V({t+2\Delta})}{\sigma^2}\cdot \left(\frac{\kappa}{\tanh(\kappa\Delta)} - \frac{\gamma_{a}}{\tanh(\gamma_{a}\Delta)}\right)\right)\\
&\cdot \frac{I_{\xi}\left(\frac{\sqrt{V({t})V({t+2\Delta})}}{\sigma^2}\frac{2\gamma_{a}}{\sinh(\gamma_{a}\Delta)}\right)}{I_{\xi}\left(\frac{\sqrt{V({t})V({t+2\Delta})}}{\sigma^2}\frac{2\kappa}{\sinh(\kappa\Delta)}\right)},
\end{align}
where $\gamma_a = \sqrt{\kappa^2 - 2\sigma^2ia}$

We will now state a few  properties of the terms involved in the conditional characteristic function $\Phi(a)$. Let $\Delta\gamma_{a} = x + iy = \pm \Delta(\kappa^4 + 4\sigma^4a^2)^{1/4}e^{i\theta/2}$, so
\begin{align}
\label{eqn:realpart}
&x = \text{Re}(z(a)) = \left[\frac{\sqrt{\kappa^4 + 4\sigma^4a^2}}{2} +  \frac{\kappa^2}{2}\right]^{1/2} \\
\label{eqn:imgpart}
&y = \text{Im}(z(a))= -\left[\frac{\sqrt{\kappa^4 + 4\sigma^4a^2}}{2} -  \frac{\kappa^2}{2}\right]^{1/2}.
\end{align}
Also, clearly the real part $x$ upper bounds the imaginary part $y$ in magnitude.

\begin{lemma}
\label{lem:gamma_sin_bounds}
We have that 
\begin{align}
 \lvert\frac{\gamma_{a}}{\sinh(\gamma_{a}\Delta)} \rvert \leq \frac{\sqrt{2}}{\Delta}
\end{align}
\end{lemma}
\begin{proof}
\begin{align}
\lvert \frac{\gamma_{a}\Delta}{\sinh(\gamma_{a}\Delta)} \rvert^{2} &= \frac{x^2 + y^2}{\sinh^2(x) + \sin(y)^2}  \\&\leq \frac{(\kappa^4 + 4\sigma^4a^2)^{1/2}}{\sinh^2((\kappa^4 + 4\sigma^4a^2)^{1/4}\cos(\theta/2))}\\
&= \frac{1}{\cos^2(\theta/2)}\left(\frac{(\kappa^4 + 4\sigma^4a^2)^{1/4}\cos(\theta/2)}{\sinh((\kappa^4 + 4\sigma^4a^2)^{1/4}\cos(\theta/2))}\right)^2\\
&\leq  \frac{1}{\cos^2(\theta/2)}.
\end{align}

\begin{align}
    \cos(\theta/2)^2 = \frac{1}{2} +  \frac{\kappa^2}{2\sqrt{\kappa^4 + 4\sigma^4a^2}} \geq \frac{1}{2} \implies \forall a,  \lvert \frac{\gamma_{a}}{\sinh(\gamma_{a}\Delta)} \rvert \leq \frac{\sqrt{2}}{\Delta}.
\end{align}
\end{proof}

\begin{lemma}
\label{lem:real-part-bound}
We have that
$\textup{Re}\left(\frac{\Delta \gamma_{a}}{\tanh(\gamma_a\Delta)}\right)$ is nonegative and is asymptotically $\Theta(\sqrt{a})$.
\end{lemma}
\begin{proof}
Note that by direct computation:
\begin{align}
\label{eqn:realpart_tangent}
&\text{Re}\left(\frac{\Delta \gamma_{a}}{\tanh(\gamma_a\Delta)}\right) = \frac{x\tanh(x)\text{sec}^2(y) + y\tan(y)\text{sech}^2(x)}{\text{sech}^2(x) + \text{sec}^2(y)}.
\end{align}

Thus,
\begin{align*}
 \text{Re}\left(\frac{\Delta \gamma_{a}}{\tanh(\gamma_a\Delta)}\right)  \geq \frac{y\tanh(y)\text{sec}^2(y)}{1 + \tan^2(y)}  \geq 0.
\end{align*}

Also, asymptotically $\gamma_{a}$ approaches the $-\frac{\pi}{4}$ radians line so we have $ \text{Re}\left(\frac{\Delta \gamma_{a}}{\tanh(\gamma_a\Delta)}\right) = \Theta(\sqrt{a})$.
\end{proof}

\polyApproxIntCir*

\begin{proof}
The condition pdf $f$ of $\int_{t}^{t+2\Delta}V(s) ds$ can then be computed via inverse Fourier transform as:
\begin{align}
\label{eqn:fourier_transf}
    f(x) = \frac{1}{\pi}\int_{0}^{\infty}\text{Re}[e^{-iax}\Phi(a)da].
\end{align}
If we only include the $a$-dependent quantities, we get the following:

\begin{align}
   \Phi(a) =  C_3' \frac{\gamma_a}{\sinh(\gamma_a)}e^{-C_1\frac{\gamma_a}{\tanh(\gamma_a)}}I_{\xi}\left(C_2 \frac{\gamma_a}{\sinh(\gamma_a)}\right),
\end{align}
where $C_1, C_2, C_3, C_3'$ are 
\begin{align*}
&C_1 = \frac{V(t) + V(t+2\Delta)}{\sigma^2} \\
&C_2 = \frac{\sqrt{V(t)V(t+2\Delta)}}{\sigma^2}\\
& C_3 = \frac{e^{C_1\frac{\kappa}{\tanh(\kappa\Delta)}}\gamma_a \sinh(\kappa\Delta)}{\kappa\sinh(\gamma_a\Delta)}\\
&C_3 ' =\frac{C_3}{I_\xi\left(C_2\frac{2\kappa }{\sinh(\kappa\Delta)}\right)}.
\end{align*}
We will want to keep track of terms that can grow with $V(t)$ and ignore dependencies on non-asymptotic, model parameters like $\kappa, \Delta, \sigma$. Also, we will make use of some simple and useful properties of $\gamma_{a}$ mentioned previously.

The first step is to determine a truncated approximation for the improper integral. We have that
\begin{align}
    \lvert \int_{m}^{\infty} \text{Re}[e^{-iax}\Phi(a)da]\rvert \leq \int_{m}^{\infty}\lvert \Phi(a)\rvert da.
\end{align}
Note that 
\begin{align*}
\frac{(x/2)^{\nu}}{\Gamma(\nu+1)} < I_{\nu}(x) < \frac{(x/2)^{\nu}e^x}{\Gamma(\nu+1)},
\end{align*}
so by Lemma \ref{lem:gamma_sin_bounds},
\begin{align*}
\lvert \frac{I_{\xi}\left(C_2 \frac{z(a)}{\sinh{z(a)}}\right)}{I_\xi\left(C_2\frac{2\kappa }{\sinh(\kappa\Delta)}\right)} \leq  \frac{I_{\xi}\left(C_2 \lvert \frac{z(a)}{\sinh{z(a)}}\rvert\right)
}{I_\xi\left(C_2\frac{2\kappa }{\sinh(\kappa\Delta)}\right)} \leq \left(\frac{\sinh(\kappa\Delta)}{\Delta \kappa}\right)^{\xi} e^{\sqrt{2}C_2/\Delta}.
\end{align*}

Thus,
\begin{align*}
 \lvert \int_{m}^{\infty} \text{Re}[e^{-iax}\Phi(a)da]\rvert &\leq C_3\left(\frac{\sinh(\kappa\Delta)}{\Delta \kappa}\right)^{\xi} e^{\sqrt{2}C_2/\Delta} \int_{m}^{\infty}e^{-C_1\text{Re}(\frac{\gamma_a}{\tanh(\gamma_a)})}  \\
 &=\mathcal{O}\left(e^{4(V(t) + V(t+2\Delta))}\int_m^{\infty} e^{-[V(t) + V(t+2\Delta)]\sqrt{a}}da\right)\\
 &=\mathcal{O}\left(e^{4(V(t) + V(t+2\Delta))}\int_{\sqrt{m}}^{\infty} xe^{-[V(t) + V(t+2\Delta)]x}dx\right) \\
 &=\mathcal{O}\left(\frac{e^{4(V(t) + V(t+2\Delta))}}{V(t) + V(t+2\Delta)} \sqrt{m}e^{-[V(t) + V(t+2\Delta)]\sqrt{m}}\right),
\end{align*}
so to ensure a truncation error $\epsilon_{\text{imp}}$, we can take $m  = \mathcal{O}(\log^2(1/\epsilon_{\text{imp}}))$.

We can thus restrict to approximating just $\int_{0}^{m} \text{Re}[e^{-iax}\Phi(a)da]$.

For our setting $z = \gamma_{a}$ whose real part upper bounds the imaginary in magnitude, and so $z \in \mathcal{R}$ from Lemma~\ref{lem:poly_approx_lem}. Also $\lvert z\rvert$ is $\Theta(\sqrt{a}) = \mathcal{O}(\sqrt{m})$ asymptotically from Lemma \ref{lem:real-part-bound}.
Thus the conditions of Lemma~\ref{lem:poly_approx_lem} are satisfied, and we can approximate both $\frac{z}{\sinh(z)}$ and $\frac{z}{\tanh(z)}$ to $\epsilon$ additive error using $\mathcal{O}(\log(\sqrt{m} /\epsilon))$ degree polynomials. Also $m$ is $\mathcal{O}\left(\log^2(1/\epsilon_{\text{imp}})\right)$ due to the truncation, giving a degree of $\mathcal{O}(\log(\log(1/\epsilon_{\text{imp}})/\epsilon))$.

For the modified Bessel function we have  the globally-convergent power series
\begin{align}
\label{eqn:mod_bessel}
I_{\alpha}(z) =\sum_{m=0}^{\infty}\frac{(z/2)^{2m+\alpha}}{m!\Gamma(m+\alpha + 1)}, z \in \mathbb{C},
\end{align}
and so if $\lvert z \rvert \leq M$, then the truncation index scales as $\mathcal{O}\left(\log\left(M/\epsilon\right)\right)$. So we take $M = \mathcal{O}\left(C_2\right)$, giving a degree scaling of $\mathcal{O}\left(\log(C_2C_3'/\epsilon)\right) = \mathcal{O}\left([V(t) + V(t+2\Delta)]\log(V(t) + V(t+2\Delta)/\epsilon)\right)$. This is also roughly the same for the exponential term $e^{-C_1\frac{\gamma_a}{\tanh(\gamma_a)}}$, except for an additional factor of $\sqrt{m} = \log(1/\epsilon_{\text{imp}})$, which goes under a log anyways. Hence the same degree scaling holds.

 Thus we have the following expression for the error from polynomial approximation via triangle inequality:
 \begin{align}
&\lvert (\frac{z(a)}{\sinh(z(a))} \pm \epsilon)\left(e^{-C_1\frac{z(a)}{\tanh(z(a))}\pm \epsilon} \pm \epsilon\right)\left(I_{\alpha}\left(C_2 \frac{z(a)}{\sinh{z(a)}} \pm \epsilon \right) \pm \epsilon \right) \nonumber\\& - \frac{z(a)}{\sinh(z(a))}e^{-C_1\frac{z(a)}{\tanh(z(a))}}I_{\alpha}\left(C_3 \frac{z(a)}{\sinh{z(a)}}\right)\rvert \nonumber\\
\label{eqn:poly_error_bound}
&\leq \epsilon \lvert \left(e^{-C_1\frac{z(a)}{\tanh(z(a))}\pm \epsilon} \pm \epsilon\right)\left(I_{\alpha}\left(C_2 \frac{z(a)}{\sinh{z(a)}} \pm \epsilon \right) \pm \epsilon \right)\rvert \nonumber\\
&+\lvert \frac{z(a)}{\sinh(z(a))}\epsilon\left(I_{\alpha}\left(C_2 \frac{z(a)}{\sinh{z(a)}} \pm \epsilon \right) \pm \epsilon \right)\rvert \nonumber\\
&+\lvert \frac{z(a)}{\sinh(z(a))}e^{-C_1\frac{z(a)}{\tanh(z(a))}}\left(I_{\alpha}\left(C_2 \frac{z(a)}{\sinh{z(a)}} \pm \epsilon \right) \pm \epsilon \right)\rvert \lvert 1- e^{\pm \epsilon}\rvert \nonumber\\
&+\lvert (\frac{z(a)}{\sinh(z(a))} )\left(e^{-C_1\frac{z(a)}{\tanh(z(a))}}\right)\epsilon \rvert \nonumber\\
&+\lvert \frac{z(a)}{\sinh(z(a))}\left(e^{-C_1\frac{z(a)}{\tanh(z(a))}}\right)\left(I_{\alpha}\left(C_2 \frac{z(a)}{\sinh{z(a)}} \pm \epsilon \right) - I_{\alpha}\left(C_2 \frac{z(a)}{\sinh{z(a)}} \right) \right)\rvert.
 \end{align}

Recall $\lvert \frac{z(a)}{\sinh(z(a))}\rvert \leq \frac{\sqrt{2}}{\Delta}$. Note  that $I_{\alpha}'(z) = \frac{1}{2}\left(I_{\alpha-1}(z) + I_{\alpha+1}(z)\right)$ and so

\begin{align*}
   \lvert I_{\alpha}\left(C_2 \frac{z(a)}{\sinh{z(a)}} \pm \epsilon \right) - I_{\alpha}\left(C_2 \frac{z(a)}{\sinh{z(a)}} \right) \rvert &\leq \max_{\lvert z - C_2 \frac{z(a)}{\sinh{z(a)}}\rvert \leq \epsilon}  \lvert I'_{\alpha}(z) \rvert \epsilon\\
   &\leq I_{\alpha +1}\left(\frac{\sqrt{2}C_2}{\Delta} + 1\right) \epsilon\\
   & = \mathcal{O}(\epsilon),
\end{align*}

and recall $e^{-C_1\text{Re}(\frac{z(a)}{\tanh(z(a))})} < 1$. So Equation \eqref{eqn:poly_error_bound} is $\mathcal{O}(\epsilon)$.

Thus $\Phi(a)$ can be $\epsilon$-uniformly approximated by a ${\mathcal{O}}\left((\max_t V(t))^2\log^2(\max_t V(t)/ \epsilon)\right)$ degree polynomial on a compact domain.  The squaring comes from approximating $e^{C_1\frac{\gamma_a}{\tanh(\gamma_a)}}$ by a composition of polynomials.

We then need to integrate the terms in the series, which will fit to the form of the lower incomplete Gamma function. Recall the lower incomplete Gamma function:
\begin{align}
    \Gamma_{m}(s + 1) :=  \int_{0}^m t^se^{-t} dt,
\end{align}
which can be expressed as the the following power series valid for $m \in \mathbb{C}$ and $s$ not a non-positive integer:

\begin{align}
\Gamma_{m}(s + 1) = m^{s+1}\sum_{k=0}\frac{(-m)^k}{k!(s+k)},
\end{align}
which has a truncation index of $\mathcal{O}\left(\log(\lvert m \rvert/\epsilon)\right)$ for $\epsilon$ additive error. We will also need to account for the various scaling factors $C_1, C_2, C_3'$. However, one will note that the degree asymptotics is still no more than anything we have previously estimated.

\end{proof}

\section{Additional Proofs for Section \ref{sec:cir_fast_forwardable}}
\subsection{Proof of Theorem \ref{thm:cir_qsample_resources}}
\label{subsec:proof-cir-load}
\cirDiscreteSumLoad*
\begin{proof}

There are two kinds of primitive distributions one-dimensional $\chi^2$ with $\eta-1$ degrees of freedom $p_{\chi}$ and standard Gaussian $p_{\mathcal{N}}$. 

We will be using fixed-point arithmetic. The arithmetic will be done on a grid over $[-b, b]^{2T}$ with $N$ grid points per dimension. %

From Corollary~\ref{cor:standard_gauss_qsvt}, we can load a quantum state that $\epsilon$ approximates the measurement distribution of ~\ref{eqn:unnorm-one-d-disc-q-samp} in the computational basis with $N = \Omega\left(b/\epsilon_{\textup{distr}}\right)$ grid points, where $p$ is  standard Gaussian over $[-b, b]$ using
\begin{align*}
\mathcal{O}\left(b\log(N)\log(b/\epsilon_{\text{distr}})\right)
\end{align*}
one- and two-qubit gates. Corollary~\ref{cor:central_chi_square_qsvt} provides an analogous result for the central $\chi^2$
\begin{align*}
\mathcal{O}\left((b+r)\log(N)\log(b/\epsilon_{\text{distr}})\right),
\end{align*}
with $N = \Omega\left(br/\epsilon_{\textup{distr}}\right)$.

Note we will have $2T$ primitive distributions in total, and hence by Equation \eqref{eqn:distribution_loading_error_bound}, we take
\begin{align*}
\epsilon_{\text{distr}} \mapsto \frac{\epsilon_{\text{distr}}}{2T \max_{\vec{x} \in [-b, b]^T \times [0, b]^{T}} \lvert f\circ g(\vec{x})\rvert}.
\end{align*}

By Lemma~\ref{lem:sqrt_dif_bound} and Lipschitzness of $f$, we have 
\begin{align*}
    \max_{\vec{x} \in [-b, b]^T \times [0, b]^{T}} \lvert f\circ g(\vec{x})\rvert = \mathcal{O}\left(BT^2b^2\right),
\end{align*}

so we scale all $\epsilon_{\text{distr}}$ by 
\begin{align*}
    \epsilon_{\text{distr}} \mapsto \mathcal{O}\left(\frac{\epsilon_{\text{distr}}}{T\cdot BT^2b^2}\right). 
\end{align*}

Thus the total gate count, after scaling the error, for the $2T$ Gaussians and $\chi^2$ is
\begin{align*}
&\mathcal{O}\left(T(b+r)\log(N)\log(BTb/\epsilon_{\text{distr}}) \right),
\end{align*}
with $N = \Omega\left(BT^3b^3r/\epsilon_{\text{distr}}\right)$.

We then need to coherently compute the recursion
\begin{align*}
    |x\rangle|z \rangle|y\rangle|0\rangle \mapsto |x\rangle|z \rangle|y\rangle|c(y + (z + \sqrt{\beta x})^2)\rangle,
\end{align*}
$c$ and $\beta$ are fixed.

Over $T$ steps this requires $2T$ additions, $3T$ multiplications/scalings and $T$ square-roots. The precision to which we have each $z$ and $y$ is at most $\mathcal{O}\left(\log(N)\right)$. We can for simplicity suppose that the initial point $x_0$ is also to this precision.

Due to the recursive nature of the path construction for the CIR process, there is a potential for the error to grow. Let the iteration at the $t$-th step be
\begin{align}
\label{eqn:x_t_recursion}
    x_t = c(y_t + (z_t + \sqrt{\beta x_{t-1}})^2),
\end{align}
where all operations are done to $n$ bits of fixed-point arithmetic. The only approximate operation is the inverse square-root computed via Newton's method (See Section~\ref{sec:arith_for_square_root}), which we denote by $\epsilon_{\text{NM}}$. Consider the following example
\begin{align*}
&\widetilde{x}_1 = c(y_1 + (z_1 + \sqrt{\beta x_{0}} \pm \epsilon_{\text{NM}})^2) =  x_1 + 2c(z_1+\sqrt{\beta x_0})\epsilon_{\text{NM}} +  c\epsilon_{\text{NM}}^2.\\
&\widetilde{x}_2 = c(y_2 + (z_2 + \sqrt{\beta \widetilde{x}_{1}} \pm \epsilon_{\text{NM}})^2) \\ &\leq c(y_2 + (z_2 + \sqrt{\beta x_1} + \sqrt{2c(z_1+\sqrt{\beta x_0})\epsilon_{\text{NM}}} + \sqrt{c}\epsilon_{\text{NM}} \pm \epsilon_{\text{NM}})^2) \\&= x_2 + (2c(z_2 + \sqrt{\beta x_1})\sqrt{2c(z_1 +\sqrt{\beta x_0})\epsilon_{\text{NM}}} + 2c(\sqrt{c}+1)(z_2+\sqrt{\beta x_1})\epsilon_{\text{NM}} + o(\epsilon_{\text{NM}}).
\end{align*}

In the next step the term $(2c(z_2 + \sqrt{\beta x_1})\sqrt{2c(z_1 +\sqrt{\beta x_0})\epsilon_{\text{NM}}}$ appears in the upper bound, implying a lose of precision. The $\epsilon_{\text{NM}}$ will need to be chosen such that the value under the square-root is $<1$, hence it can potentially grow to be significant (i.e. towards a constant). Here, we use a loose upper bound  of $\mathcal{O}(b^T\epsilon_{\text{NM}}^{2^{-T}})$ ($b$ from Lemma \ref{lem:sqrt_diff_truncation}) after $t$ steps. Thus, we need to consider $\epsilon_{\text{NM}} \mapsto (\epsilon_{\text{NM}}/b^T)^{2^T}$.

Note that in the region of quadratic convergence of Newton's method, there is a $\mathcal{O}(\log\log(1/\epsilon_{\text{NM}}))$ iteration complexity. Hence the above combined with \eqref{eqn:sqrt_gate_compl} gives a gate cost of $\mathcal{O}\left(T^2\log(N)^2\log(T\log(b/\epsilon_{\textup{NM}}))\right)$ for $\mathcal{O}(T)$ square-roots. We only need to ensure that the $\epsilon_{\text{NM}}$ from Newton's method is below $\mathcal{O}(N^{-1})$.

Next let us determine how many qubits are required for the output register, i.e. to hold the valeue $f\circ g_t$. Note that in \eqref{eqn:x_t_recursion}, each time $z_t$ is squared, unless it is the final time point, it gets square-rooted in the next step. Additionally, each $x_{t-1}$ is square-rooted before squaring. Hence, ignoring the Newton error, the required number of bits to hold the output of the recursion does not grow with time. So without loss of generality, we can suppose $g_t$ requires $\mathcal{O}(\log_2(N))$ bits to represent. The number of lower order bits is $\mathcal{O}(\log_2(N/b))$. Note that $f$ can expand the output by at most a factor of $\mathcal{O}(BT)$. Hence $\mathcal{O}(\log(BTN))$ bits also suffices for the output register. %

The $\mathcal{O}(T)$ other additions and multiplications/scalings cost together $\mathcal{O}(T\log(N)^2)$ gates. Along with the oracle for the payoff,  we can perform the operation:
\begin{align*}
    |\vec{z}\rangle|\vec{y}\rangle|0\rangle \rightarrow |\vec{z}\rangle|\vec{y}\rangle|f\circ g(\vec{z}, \vec{y})\rangle
\end{align*}
using 
\begin{align*}
     \mathcal{O}\left( T^2 \log^2(N)\log(T) + \mathcal{N}_f \right)
\end{align*}
gates. 

The rest now follows Lemma \ref{lem:payoff_loading}, with $\Lambda = \mathcal{O}(BT^2b^2)$. This gives  $\mathcal{O}\left(\log^3(\Lambda N)\right)$ additional gates to perform the rotation onto an amplitude.

\end{proof}

\subsection{Proof of Theorem \ref{thm:sqrt_discretization}}
\label{subsec:proof-thm-cir-discr}
The recursion on $x_t$ can be upper bounded over time using the following simple estimate.
\begin{lemma}
    \label{lem:sqrt_dif_bound}
        Let $X_t$ follow a CIR process with $\gamma + \gamma^2 < 1$. Then any point $x_t$ on a realized path satisfies
        \begin{align*}
            x_t \leq x_0 + c(1+\gamma)\left(\sum_{k=0}^{t} y_k + z_k^2\right).
        \end{align*}
    \end{lemma}
\begin{proof}

Recall that $\gamma := \sqrt{c\beta} < 1$ and let $v_t := x_t/c$ then 
\begin{align*}
&v_t = (y_{t} + (z_{t} + \gamma\sqrt{v_{t-1}})^2),
\end{align*}
where from Jensen's
\begin{align*}
&v_t \leq  y_t + (1+\gamma)z_t^2 + (\gamma+\gamma^2)v_{t-1},
\end{align*}
an inductive argument shows that 
\begin{align*}
v_t &\leq (\gamma+\gamma^2)^{t}v_0 + \sum_{k=1}^{t}(\gamma+\gamma^2)^{t-k}(y_k + (1+\gamma)z_k^2)\\
&\leq \frac{x_0}{c} + \left(\sum_{k=0}^{t} y_k + (1+\gamma)z_k^2\right),
\end{align*}
where we used H\"older's inequality and $\gamma + \gamma^2 < 1$.
\end{proof}

We will truncate the densities for the the standard Gaussian and chi-square to $[-a, a]$ and $[b_L ,b_U]$, respectively. The following result determines the scaling of the endpoints in terms of  the truncation error $\epsilon_{\text{trunc}}$.

\begin{lemma}
\label{lem:sqrt_diff_truncation}
Suppose the payoff $f$ is upper bounded by a $B$-Lipschitz linear function in the path process $\widehat{V}$, then
    \begin{align*}
        &b_{L} = \mathcal{O}\left(\frac{\epsilon}{BT^2r}\right)\\
        &b_{U}  = \mathcal{O}\left(r + \ln(BT/\epsilon)\right)\\
        &a = \mathcal{O}\left(\sqrt{\ln(BT/\epsilon)}\right)
    \end{align*}
suffice for $\epsilon$ truncation error.
\end{lemma}

\begin{proof}
    Let $S:=([b_L, b_U]^{T}\times [-a, a]^{T})^{c}$, then via the previous lemma and i.i.d. of increments:

    \begin{align*}
     \int_{S}\lvert f\circ g(\vec{y}, \vec{z}, x_0)\rvert p(\vec{z})p(\vec{y})d\vec{y}d\vec{z} &\leq B\sum_{t}\int_{S}g_t p(\vec{z})p(\vec{y})d\vec{y}d\vec{z}\\ &\leq c(1+\gamma)BT^2\left(x_0+\mathbb{E}[y]\int_{\lvert z\rvert \geq a}z^2 p_{\mathcal{N}}(z) + \mathbb{E}[z^2]\int_{y \geq b}y p_{\chi}(y)\right)\\
     &\leq c(1+\gamma)BT^2r\left(x_0+\int_{\lvert z\rvert \geq a}z^2 p_{\mathcal{N}}(z) +\int_{y \geq b}y p_{\chi}(y)\right).
    \end{align*}
    
    For $a > 1$,      
    \begin{align*}
        \int_{\lvert z\rvert \geq a}z^2 p_{\mathcal{N}}(z) = \frac{2}{\sqrt{2\pi}}\int_{ z\geq a}z^2e^{-z^2/2}dz \leq 2ae^{-a^2/2},
    \end{align*}
    Consider $a =  \mathcal{O}(\sqrt{\ln(1/\epsilon)})$ suffices for an error that is $\leq \frac{\epsilon}{2}$.

Recall that for the Gamma function $s\Gamma(s) = \Gamma(s+1)$, so
\begin{align*}
\int_{ y \geq b} y p_{\chi}(y) &= 2\int_{ y \geq b} \frac{\left(\frac{y}{2}\right)^{r/2}e^{-y/2}}{2\Gamma(r/2)}\\
&=2r\int_{ y \geq b} \frac{\left(\frac{y}{2}\right)^{(r+2)/2-1}e^{-y/2}}{2\Gamma(\frac{r+2}{2})}\\
&=2r\mathbb{P}[Y \geq b],
\end{align*}
where $Y \sim \chi^2_{r+2}$. Note that by Chernoff bounding with $\lambda = \frac{1}{4}$
\begin{align*}
\mathbb{P}[Y \geq b] \leq e^{-\lambda b}(1-2\lambda)^{-(r+2)/2} \leq e^{-\frac{b}{4} + \frac{(r+2)}{2}\ln(2)},
\end{align*} 
so $ b_U = \mathcal{O}\left(r + \ln(1/\epsilon)\right)$ suffices.

Recall that we assume at least $r\geq 2$ and $b_L \ll 1$, then
\begin{align*}
\int_{\leq b_{L}} yp_{\chi_r}(y)dy &\leq   \int_{\leq b_{L}} yp_{\chi_2}(y)dy \\
&\leq\int_{\leq b_{L}} ye^{-y/2}dy\\
&\leq b_L
\end{align*}
so $b_L = \mathcal{O}\left(\frac{\epsilon}{BT^2r}\right)$ suffices.

\end{proof}

\cirDiscreteError*
\begin{proof}

By the piecewise linear assumption,  we can upper bound the partial derivatives by the partial derivatives of
\begin{align*}
[\sum_{t=0}^{T-1} Bx_t]p_{\mathcal{N}}(\vec{z})p_{\chi}(\vec{y}) = [\sum_{t=0}^{T-1} Bg_t(\vec{y}, \vec{z}, x_0)]p_{\mathcal{N}}(\vec{z})p_{\chi}(\vec{y}).
\end{align*}

 We need to compute the various partial derivatives: $\partial_{y_k}, \partial_{z_k}$. %
We will bound the derivatives by iterating on  $x_{t} =  c(y_{t-1} + (z_{t-1} + \sqrt{\beta x_{t-1}})^2) = g_t(\vec{y}, \vec{z}, x_0)$. Recall that $\gamma < 1$ and from the recurrence for $x_t$: $\sqrt{\beta x_t} \geq \gamma\sqrt{y_{t-1}}$.

The chain-rule recursion (for $y$ or $z$) is:
\begin{align*}
(x_t)' =c\beta\left(\frac{z_{t-1}}{\sqrt{\beta x_{t-1}}} + 1\right)\cdot (x_{t-1})' \leq \left(\frac{z_{t-1}}{\sqrt{y_{t-2}}} + 1\right)\cdot (x_{t-1})'
\end{align*}

By our notation $x_{j+1}$ is determined by $z_{j}, y_k$ and $x_{j}$. Let $T$ be arbitrary and $T > j+1$. Hence,
\begin{align}
\label{eqn:chain_rule_rec}
|\partial_{x_{j+1}}g_{T-1}(\vec{y}, \vec{z}, x_0)|\prod_{k=0}^{T-2}p(z_k)\prod_{k=0}^{T-2}p(y_k)&= |\partial_{x_{j+1}}g_{T-1}(\vec{y}, \vec{z}, x_0)|\prod_{k=j+1}^{T-2}p(z_k)\prod_{k=j}^{T-3}p(y_k) \cdot[p(z_j)\prod_{k=0}^{j-1}p(z_k), p(y_k)] \nonumber\\ &\leq \prod_{k=j+1}^{T-2}\left[\left(\frac{z_{k}}{\sqrt{y_{k-1}}} + 1\right)p(z_{k})p(y_{k-1})\right] \cdot [p(z_j)\prod_{k=0}^{j-1}p(z_k), p(y_k)] \nonumber\\&\leq \prod_{k=j+1}^{T-2}\frac{z_{k}y_{k-1}^{r/2-3/2} + y_{k-1}^{r/2-1}}{\sqrt{2\pi}\Gamma(r/2)2^{r/2}}e^{-\frac{y_{k-1} +z_{k}^2}{2}} \cdot [p(z_j)\prod_{k=0}^{j-1}p(z_k), p(y_k)].
\end{align}

Recall that $\Gamma(r/2) \geq (r/2e)^{r/2-1}$, for  $r/2  \geq 1$. Consider some $\ell, 1 \leq \ell \leq r/2$, then $y^{r/2-\ell}e^{-y/2}$ is maximized at $y = r-2\ell$, so
\begin{align*}
\frac{y^{r/2-\ell}e^{-y/2}}{\Gamma(r/2)2^{r/2}} &\leq \frac{(r-2\ell)^{r/2-\ell}e^{-r/2 + \ell}}{\Gamma(r/2)2^{r/2}}\\
&\leq \frac{(r-2\ell)^{r/2-\ell}e^{\ell-1}}{2r^{r/2-1}} \\
&\leq \frac{1}{2}(1-2\ell/r)^{r/2-\ell}e^{\ell-1}.
\end{align*}

 Thus, for $r \geq 4$,
\begin{align*}
    \frac{z_{k}y_{k-1}^{r/2-3/2} + y_{k-1}^{r/2-1}}{\sqrt{2\pi}\Gamma(r/2)2^{r/2}}e^{-\frac{y_{k-1} +z_{k}^2}{2}} \leq \frac{(1-3/r)^{r/2-3/2} + (1-2/r)^{r/2-1}}{2\sqrt{2\pi}} \leq \frac{1}{2}.
\end{align*}

Thus 
\begin{align*}
    |\partial_{x_{j+1}}g_{T-1}(\vec{y}, \vec{z}, x_0)|\prod_{k=0}^{T-2}p(z_k)\prod_{k=0}^{T-2}p(y_k) \leq 2^{-(T-j-3)} [p(z_j)\prod_{k=0}^{j-1}p(z_k), p(y_k)],
\end{align*}

and so
\begin{align*}
    p_{\mathcal{N}}(\vec{z})p_{\chi}(\vec{y}) \lVert\nabla g_{T-1}\rVert_{\infty} \leq \sum_{j=0}^{T-2} 2^{-(T-j-3)}(\partial_{y_j}x_{j+1} + \partial_{z_j}x_{j+1})[p(z_j)\prod_{k=0}^{j-1}p(z_k), p(y_k)].
\end{align*}

Also simply,
\begin{align*}
&\lvert\partial_{y}p_{\chi_r}(y)\rvert = \lvert \frac{[(r/2-1)y^{r/2-2} - (1/2)y^{r/2-1}]e^{-y/2}}{2^{r/2}\Gamma(r/2)}\rvert  \leq \frac{1}{2}p_{\chi_{r-2}}(y_k) + \frac{1}{2}p_{\chi_{r}}(y)\\
&\partial_{z}p_{\mathcal{N}}(z) \leq \lvert z \rvert p_{\mathcal{N}}(z),
\end{align*}

\begin{align*}
     g_{T-1}\lVert \nabla(p_{\mathcal{N}}(\vec{z})p_{\chi}(\vec{y})) \rVert_{\infty} \leq  g_{T-1} \sum_{k=0}^{T-2}p(y_{\neq k}, z_{\neq k})\left[\frac{1}{2}\left(p_{\chi_{r-2}}(y_k) + p_{\chi_{r}}(y_k)\right)p_{\mathcal{N}}(z_k) +  \lvert z_k\rvert p_{\chi_{r}}(y_k)p_{\mathcal{N}}(z_k)\right]
\end{align*}

Suppose we consider the left-endpoint rule in the region $(\vec{y}, \vec{z}) \in [R_L, R]^{T}\times [-\frac{R}{2}, \frac{R}{2}]^{T}$, where $R = 2\max(a, b_U) = \mathcal{O}\left(r + \log(BT/\epsilon_{\text{trunc}})\right)$ with $a, b_U$ from Lemma~\ref{lem:sqrt_diff_truncation}. We take the grid spacing to be $\frac{R}{N}$ uniformly and denote the grid $\mathcal{M}$.

From Lemma~\ref{lem:left_rule_error} on the left-endpoint rule (for a single $x_T$), we get that the discretization error is bounded by
\begin{align*}
\sum_{j=0}^{T-2}\left(\sum_{(\vec{y}, \vec{z}) \in \mathcal{M}} \mathcal{R}(j, \vec{w}(\vec{y}), \vec{w}(\vec{z})) \prod_{k=1}^{2T}\Delta_k\right) \frac{2TR}{N},
\end{align*}
where
\begin{align*}
\mathcal{R}(j, \vec{y}, \vec{z}) &:= 2^{-(T-j-3)}\lVert \nabla_{(y_j, z_j)} x_{j+1}\rVert_{1} p(y_{\leq (j-1)}, z_{\leq (j-1)}) \\ &+ g_{T-1}\cdot p(y_{\neq j}, z_{\neq j})\left[\frac{1}{2}\left(p_{\chi_{r-2}}(y_j) + p_{\chi_{r}}(y_j)\right)p_{\mathcal{N}}(z_j) +  \lvert z_j\rvert p_{\chi_{r}}(y_j)p_{\mathcal{N}}(z_j)\right],
\end{align*}
and $\vec{w}$ maps $\vec{y}$ and $\vec{z}$ to the maximizer of $R(j, \cdot, \cdot)$ in the cell with $(\vec{y}, \vec{z})$ as its lowest corner.

Now if for $N \geq N_0$ ($N_0$ independent of dimension) the volumes with respect to all of the truncated, one-dimensional pdfs present in the above expression are bounded by a constant, then we can upper each inner sum by the max of the term multiplied by each pdf times some constant.  Without loss of generality, we can just look at sums of the form
\begin{align}
\label{eqn:mass_discr}
\sum_{(\vec{y}, \vec{z}) \in \mathcal{M}'} p(\vec{y}, \vec{z})\prod_{k=1}^{2T}\Delta_k = \left(\sum_{\vec{z} \in \mathcal{M}'
_{\mathcal{N}}}  \frac{e^{-\lVert \vec{z}\rVert_{2}^2/2}}{(2\pi)^{(T-2)/2}} \prod_{k=0}^{T-2}\Delta_k\right)\left(\sum_{\vec{y} \in \mathcal{M}'
_{\chi}}  \frac{\left(\prod_{j=0}^{T-2} y_j\right)^{r/2-1}e^{-\lVert \vec{y}\rVert_{1}/2}}{[\Gamma(r/2)2^{r/2}]^{T-2}} \prod_{k=0}^{T-2}\Delta_k\right)
\end{align}
where $\mathcal{M}'$ is a modified version of the grid $\mathcal{M}$ with each grid point $(\vec{y}, \vec{z})$ being replaced with $(\vec{w}(\vec{y}), \vec{w}(\vec{z}))$. Also $\mathcal{M}_{\chi}'$ and $\mathcal{M}_{\mathcal{N}}'$ denote the grid split over the cartesian product of $\vec{y}$ and $\vec{z}$. 

By a similar analysis for GBM in Theorem~\ref{thm:gaussian-relative}, the Gaussian sum can be made $\mathcal{O}\left(1\right)$ for $N = \mathcal{O}\left(\max(\sigma,1)T^{3/2}R\right)$. 

Suppose we shift the grid $\mathcal{M}'_{\chi}$ such that each grid point $\vec{y}$ is replaced with the minimizer of $p_{\chi}$ over $\mathcal{B}_{\infty}(\vec{y}, \Delta)$. On this new grid, the sum in the right factor is upper bounded by one. We then need to compute  the error term from shifting the grid, which is done by looking at the relative change in the $\chi^2$ pdf between the two grid points. 

The relative change in the $\chi^2$ pdf when going from $\vec{y}$ to some $\vec{x} \in \mathcal{C}_{\vec{y}} = \mathcal{B}_{\infty}(\vec{y}, \Delta)$ is  bounded by (using $b_L$ from Lemma~\ref{lem:sqrt_diff_truncation}, so $\forall j, x_j = \Omega(\frac{\epsilon_{\text{trunc}}}{BT^2r})$)
\begin{align*}
\lvert\left(\prod_{j=0}^{T-2} \frac{y_j}{x_j}\right)^{r/2-1}e^{-(\lVert \vec{y}\rVert_{1} - \lVert \vec{x}\rVert_{1}) /2}\rvert &\leq \lvert\left(\prod_{j=0}^{T-2}(1 + \Delta/x_j)\right)^{r/2-1}e^{-(\lVert \vec{y}\rVert_{1} - \lVert \vec{x}\rVert_{1}) /2}\rvert\\ 
&\leq \left(1+\frac{BT^2r\Delta}{\epsilon_{\text{trunc}}}\right)^{Tr}e^{T\Delta} \\
&=\mathcal{O}(e^{BT^3r^2\Delta/\epsilon_{\text{trunc}}}).
\end{align*}
Thus the above is $\mathcal{O}\left(1\right)$ with $\Delta = \mathcal{O}(\frac{\epsilon_{\text{trunc}}}{Br^2T^3})$, recall the number of bits is log in $\frac{R}{\Delta}$. Thus, the right sum in \eqref{eqn:mass_discr}involving the $\chi^2$ pdf is bounded by a constant independent of $T$ for $N = \Omega\left(\frac{BRr^2T^3}{\epsilon_{\text{trunc}}}\right)$.

Thus for $N = \Omega\left(BT^3R^{2}r^2/\epsilon_{\text{trunc}}\right)$, where $R$ is the maximum of the Gaussian and $\chi^2$ truncation endpoints $a, b$, we have that Equation \eqref{eqn:mass_discr} is bounded by a constant independent of $T$. Then, if we combined the previous observation with Jensen's inequality we get that \begin{align*}
\sum_{j=0}^{T-2}\left(\sum_{(\vec{y}, \vec{z}) \in \mathcal{M}} \mathcal{R}(j, \vec{w}(\vec{y}), \vec{w}(\vec{z})) \prod_{k=1}^{2T}\Delta_k\right) = \mathcal{O}\left(\text{poly}(T, R)\right).
\end{align*}

The asymptotics on the number of bits follows.
\end{proof}

\section{Additional Proofs for Section \ref{sec:heston-fast-forwardable}}

\label{sec:heston-additional-proofs}

\subsection{Heston Moment Explosions : Proof of \cite[Proposition 3.1]{andersen2007efficient}}

\begin{proposition}[Riccati equation] \label{prop:riccati}
Consider the Riccati ODE
\begin{align*}
    \frac{d}{dx} y(x) = a y^2(x) + b y(x) + c
\end{align*}
with initial condition $y(0) = 0$.
The solution to this ODE is given by
\begin{align*}
    y = \frac{1}{2a} \left( \beta \tan(\frac{1}{2}\beta (x-C)) - b \right),
\end{align*}
where
\begin{align*}
    \beta = \sqrt{4 a c - b^2}, \quad C = -\frac{2}{\beta} \arctan(\frac{b}{\beta}).
\end{align*}
\end{proposition}
\begin{proof}
Separating $x$ and $y$ and integrating on both sides, we have
\begin{align*}
    x &= \int \frac{dy}{a y^2 + b y + c} = \frac{2}{\sqrt{4 a c - b^2}} \arctan(\frac{2 a y + b}{\sqrt{4 a c - b^2}}) + C\\
    &= \frac{2}{\beta} \arctan(\frac{2 a y + b}{\beta}) + C.
\end{align*}
Invert the equation above gives
\begin{align*}
    y = \frac{1}{2 a} \left( \beta \tan(\frac{1}{2}\beta (x-C)) - b \right).
\end{align*}
Applying the initial condition $y(0) = 0$ gives the value of $C$.
\end{proof}

\begin{proposition}[{\cite[Proposition 3.1]{andersen2007moment}}] %
Consider a stochastic process $S_t$ following the Heston SDEs given by \Cref{defn:heston-single}.
Consider $\omega > 1$, then the $\omega$-th moment of $S_t$, i.e. $\mathbf{E}[S^\omega_t]$, is finite for all $t \in [0, T^*)$ and infinite for $t \ge T^*$, where
\begin{align*}
    T^* =
    \begin{cases}
        \infty, & b^2 - 4ac \ge 0, b < 0;\\
        \frac{1}{\gamma} \log(\frac{b + \gamma}{b - \gamma}), & b^2 - 4ac \ge 0, b > 0;\\
        \frac{1}{\beta} \left( \pi - 2 \arctan(\frac{c}{\beta}) \right), & b^2 - 4ac < 0;
    \end{cases}
\end{align*}
where $a= \frac{\sigma^2}{2}, b = \rho\sigma\omega - \kappa, c=\frac{\omega^2 - \omega}{2}$, $\beta := \sqrt{4ac - b^2}$, and $\gamma := -i \beta$.
\end{proposition}
\begin{proof}
Define $Z_t := \log S_t$. Then, we have the pair of coupled SDEs:
\begin{align*}
&dZ_t = (\mu - \frac{V_t}{2})dt + \sqrt{V(t)} dW_Z, \\
&dV_t =  \kappa(\theta - V_tdt + \sigma \sqrt{V(t)}dW_V \\
&dW_Z \cdot dW_V = \rho dt,
\end{align*}
leading to 
\begin{align*}
&\mathbf{X}_t = \begin{pmatrix}
Z_t \\ 
V_t
\end{pmatrix}.
\end{align*}
\begin{align}
\label{eqn:coupled_sde}
& d\mathbf{X}_t = \boldsymbol{\mu}(\mathbf{X}_t)dt + \boldsymbol{\sigma}(\mathbf{X}_t)d\mathbf{W}.
\end{align}

An application of the Kolmogorov backward equation to \eqref{eqn:coupled_sde} shows that with $u \leq t$ and 
\begin{align*}
&\mathbb{E}[g(\mathbf{X}_{t}) |(Z_u, V_u) = (z, v) ] = \mathbb{E}[e^{\omega Z_{t}} | (Z_u, V_u) = (z, v)] = f(t, z, v)\\
&\mathbb{E}[S^\omega_t] = f(0, Z(0), V(0)),
\end{align*}
we have that $f(u,z,v)$ satisfies the following PDE
\begin{align}
\label{eqn:kbe}
    \partial_u f + (\mu - \frac{v}{2})\partial_zf + \kappa(\theta - v)\partial_v f + \rho \sigma v\partial_{zv}f + \frac{\sigma^2v}{2}\partial_{v}^2f + \frac{v}{2}\partial_{z}^2f = 0
\end{align}
subject to the final condition $f(t,z,v) = e^{\omega z}$.
The solution to the PDE is given in the form of
\begin{align*}
    f(u,z,v) = e^{\omega z} e^{A(t-u) + v B(t-u)},
\end{align*}
where $A(0) = B(0) = 0$.  We can then isolate $A$ and $B$ from Equation \eqref{eqn:kbe} to get
\begin{align*}
&\frac{d}{d\tau}A(\tau) = \kappa\theta B(\tau) + \mu\omega\\
&\frac{d}{d\tau}B(\tau) = a B^2(\tau) +b B(\tau) +  c,
\end{align*}
where  $a= \frac{\sigma^2}{2}, b = \rho\sigma\omega - \kappa, c=\frac{\omega^2 - \omega}{2}$ and $\tau := t - u$.
Applying \Cref{prop:riccati} with $\beta = \sqrt{4ac-b^2}$ gives the solution for $B(\tau)$
\begin{align*}
    B(\tau) = \frac{1}{2 a} \left(\beta \tan(\frac{\beta}{2} (\tau - C)) - b \right)
\end{align*}
where
\begin{align*}
    C = -\frac{2}{\beta} \arctan(\frac{b}{\beta}).
\end{align*}
Now we analyze $B(\tau)$ depending on the signs of $b^2 - 4 ac$ and $b$.

\paragraph{Case (3) $b^2 - 4ac < 0$.}
In this case, $\beta$ is real, and $B(\tau)$ is finite if and only if the $\tan$ term is finite, or equivalently $\frac{\beta}{2}(\tau - C) < \frac{\pi}{2}$.
Substitute in $C$ we have the range of $\tau$ for $B(\tau)$ being finite
\begin{align*}
    \tau \in \left[0, \frac{1}{\beta}\left( \pi - 2 \arctan(\frac{b}{\beta}) \right) \right).
\end{align*}

\paragraph{Cases (1) \& (2) $b^2 - 4ac \ge 0$.}
In these cases, $\beta$ is imaginary and $\gamma = -i \beta \ge 0$ is real.
Rewriting $B(\tau)$ in terms of $\gamma$ we obtain
\begin{align*}
    B(\tau) &= \frac{1}{2 a} \left(i \gamma \tan(i\frac{\gamma}{2} (\tau - C)) - b \right) \\
    &= -\frac{1}{2 a} \left(\gamma \tanh(\frac{\gamma}{2} (\tau - C)) + b \right),
\end{align*}
where we have used the identity $\tanh(x) = -i\tan(i x)$.
Similarly, we can rewrite $C$ as
\begin{align*}
    C = -\frac{2}{i \gamma} \arctan(\frac{b}{i \gamma}) 
    = \frac{2}{\gamma} \arctanh\left(\frac{b}{\gamma}\right) = \frac{1}{\gamma} \log( \frac{b + \gamma}{\gamma - b}),
\end{align*}
where we have used the identity $i\arctan(-i x) = \arctanh(x) = \frac{1}{2} \log(\frac{1+x}{1-x})$.
Since $a, c > 0$, we have $\abs{b} > \gamma$, and consequently
\begin{align*}
    \frac{b + \gamma}{b - \gamma} = \frac{1 + \gamma/b}{1 - \gamma/b} > 0.
\end{align*}
Therefore $\log(\frac{b + \gamma}{b - \gamma})$ is real and we can write $C$ as
\begin{align*}
    C = \frac{1}{\gamma} \left( \log( \frac{b + \gamma}{b - \gamma}) + i \pi \right).
\end{align*}
Substitute $C$ into $B(\tau)$ we obtain
\begin{align*}
    B(\tau) &= -\frac{1}{2 a} \left(\gamma \tanh(\frac{\gamma}{2} \tau - \frac{1}{2} \log( \frac{b + \gamma}{b - \gamma}) - i \frac{\pi}{2} ) + b \right) \\
    &= -\frac{1}{2 a} \left(\gamma \coth(\frac{\gamma}{2} \tau - \frac{1}{2} \log( \frac{b + \gamma}{b - \gamma}) ) + b \right).
\end{align*}
Observe that $\coth(x)$ has a singularity at $0$ and will be finite for all $x \neq 0$.
We will use this fact to determine the range of $\tau$ for $B(\tau)$ to be finite.
If $b < 0$, then $\frac{b + \gamma}{b - \gamma} < 1$, therefore $\log( \frac{b + \gamma}{b - \gamma}) < 0$ and hence $\frac{\gamma}{2} \tau - \frac{1}{2} \log( \frac{b + \gamma}{b - \gamma}) > 0$ for all $\tau \ge 0$, which means that the $\coth$ term will be finite. 
Consequently, $B(\tau)$ will be finite for all $\tau \ge 0$.

On the other hand, when $b > 0$, we have $\log( \frac{b + \gamma}{b - \gamma}) > 0$, and hence the argument inside $\coth$ will cross zero at $\tau = \frac{1}{\gamma} \log( \frac{b + \gamma}{b - \gamma})$. 
This means that $B(\tau)$ will be finite for all $\tau \in \left[0, \frac{1}{\gamma} \log( \frac{b + \gamma}{b - \gamma})\right)$ and becomes infinite for all larger $\tau$.
\end{proof}

\begin{remark}
For the case of $b^2 - 4ac < 0$, \cite[Proposition 3.1]{andersen2007moment} expressed the result as
\begin{align*}
    T^* = \frac{2}{\beta}\left( \pi 1_{\{b < 0\}} + \arctan(\frac{b}{\beta}) \right).
\end{align*}
We note that this is equivalent to our result due to the following identity
\begin{align*}
    \arctan(x) =
    \begin{cases}
        \frac{\pi}{2} - \arctan(\frac{1}{x}), & x > 0; \\
        -\frac{\pi}{2} - \arctan(\frac{1}{x}), & x < 0.
    \end{cases}
\end{align*}
\end{remark}

\begin{remark}[extension of \Cref{prop:heston_finite_moments} to $\omega=1$]
In the case where $\omega = 1$, we have $c = 0$, $b^2 - 4ac = b^2 \ge 0$ and $\gamma = \abs{b}$. 
Following the same argument as in the proof of $\Cref{prop:heston_finite_moments}$ for $b^2 - 4ac \ge 0$, we have
\begin{align*}
    C = \frac{2}{\gamma} \arctanh \left( \frac{b}{\gamma} \right) = \frac{2}{\gamma} \arctanh \left( \mathrm{sign}(b) \right) = \mathrm{sign}(b) \infty.
\end{align*}
Therefore $B(\tau) \equiv 0$ for all $\tau \ge 0$, and hence the first moment of $S(t)$ always exists.
\end{remark}

\subsection{Proof of Theorem \ref{thm:heston_loading_resource}}
\hestDiscrSumLoading*
\begin{proof}

The arithmetic will be done on a grid over $[-b, b]^{4Td}$ with $N$ grid points per dimension. 

By Lemma~\ref{lem:sqrt_dif_bound} Lipschitzness of $f$, and Equation \eqref{eqn:heston-path-construction} we have 
\begin{align}
\label{eqn:max_payoff}
    \max_{\vec{x} \in [-b, b]^{4Td}} \lvert f\circ h(\vec{y}, \vec{z}, \vec{x}, \vec{w})\rvert = \mathcal{O}\left(BTde^{dT^2b^{3/2}}\right).
\end{align}
By Equation \eqref{eqn:distribution_loading_error_bound}, we take
\begin{align*}
\epsilon_{\text{distr}} \mapsto \frac{\epsilon_{\text{distr}}}{4Td \max_{\vec{x} \in [-b, b]^{4Td}} \lvert f\circ h(\vec{y}, \vec{z}, \vec{x}, \vec{w})\rvert}.
\end{align*}

Hence all $\epsilon_{\text{distr}}$ need to be scaled down by this, which we will handle latter.

Note that the fast-forwarding scheme for Heston requires that we first load a CIR processes. Specifically, we need to load $d$ states of the form
\begin{align*}
    \sum_{(k, j) \in \mathcal{M}}\sqrt{p_{\mathcal{N}}(\vec{z}_k)p_{\chi}(\vec{y}_j)}|\vec{z}_k\rangle|\vec{y}_j\rangle|g_1(\vec{z}_k, \vec{y}_j)\rangle\cdots|g_{T}(\vec{z}_k, \vec{y}_j)\rangle. 
\end{align*}
This cost follows from the proof of  Theorem~\ref{thm:cir_qsample_resources} and the scaling of $\epsilon_{\textup{distr}}$ presented above.

Hence, we get
\begin{align*}
    \mathcal{O}\left(dT^3b^{3/2}\cdot Td(b+r)\log(N)\log(BTd/\epsilon_{\textup{distr}}) + dT^2\log^2(N)\log(T)\right),
\end{align*}
where $N = \Omega\left(Td\cdot BTdre^{dT^2b^{3/2}}/\epsilon_{\text{distr}}\right)$.

We then need to load $Td$ more standard Gaussians over $N$ grid points each, which can be done with a gate cost of 
\begin{align*}
\mathcal{O}\left(dT^3b^{3/2}\cdot Tdb\log(N)\log(BTd/\epsilon_{\textup{distr}}) \right).
\end{align*}

Then to load the integral over CIR, which we need to do $dT$ times, we can apply Corollary~\ref{cor:int_of_cir_loading_cost}. Note that by our choice of truncation and Lemma~\ref{lem:sqrt_dif_bound}, we have $\max_t V(t) = \mathcal{O}(bT)$. Hence the gate cost to load the integral over CIR is 
\begin{align*}
&\mathcal{O}\left(Td(bT)^4\log^{4}\left(bT N\right)\right).
\end{align*}

The next step is to compute \eqref{eqn:log-increment-heston} into a register. This requires $\mathcal{O}\left(T\right)$ additions/subtractions,  multiplications/scalings and square-roots. There also an additional inner product computation for $d$-dimensional vectors. However, unlike CIR, there is not additional repeated square-rooting, beyond the CIR loading step. 
This operations will be applied to roughly $\mathcal{O}\left(\log(N)\right)$ numbers, indicating a $\mathcal{O}\left(Td^2\log^2(N)\right)$
total gate cost for the basic arithmetic operations.

Then, we need to sum and exponentiate the results of the previous step. Since we are exponentiating, by \eqref{eqn:max_payoff}, the number of output bits can be $\mathcal{O}(dT^2b^{3/2})$. Hence we can perform $dT$ exponentials applied to the sum of the $U_k(t)$, using 
\begin{align*}
\mathcal{O}\left(dT \cdot (dT^2b^{3/2})^2 \log(TdN)\right).
\end{align*}

The last steps are to compute the payoff and rotate the result onto an amplitude. The rest now follows Lemma \ref{lem:payoff_loading}, with $\Lambda = \mathcal{O}(BTde^{dT^2b^{3/2}})$. This gives  $\mathcal{O}\left(\log^3(\Lambda N)\right) = \mathcal{O}\left(T^6d^3b^{9/2}\right)$ additional gates to perform the rotation onto an amplitude.

\end{proof}

\subsection{Proof of Lemma \ref{lem:hest_truncation_error}}
\label{subsec:proofOfLem61}

We have the following tail bound, which will be useful in bounding the truncation error. For conciseness, we will also interchange the notation $V(s)$ with $V_s$.
\begin{lemma}
\label{lem:v_int_tail_bound}
Let $V(s)$  be a CIR process with initial condition $v_0$, parameters $\kappa, \sigma, \theta$, and step-size $\Delta \geq 1$. In addition suppose the process satisfies the Feller condition s.t. a.s. $V \succ \mathbf{0}$. 
Then
\begin{align*}
   \mathbb{P}[\int_{r}^{r+2\Delta}V(s) ds \geq x | V_r, V_{r+2\Delta}] &\leq \left(\frac{\Delta \sinh(\kappa\Delta)}{2\kappa}\right)^{\xi + 1}\exp(\left[\frac{\kappa}{\tanh(\kappa\Delta)}\right]\frac{V_r + V_{r+2\Delta}}{\sigma^2})\exp(-\frac{\kappa^2}{2\sigma^2}x).
\end{align*}
where $\xi$ is  the Feller gap.
\end{lemma}
\begin{proof}
 Let $g(a) = \sqrt{\kappa^2 - 2a\sigma^2}$.  Then the conditional MGF is \cite{broadie2006exact} :
\begin{align*}
&\mathbb{E}[\exp(a\int_{r}^{r+2\Delta} V({s}) ds) | V(r), V({r+2\Delta})]\\
&=\frac{g(a)\sinh(\kappa \Delta)}{\kappa \sinh(g(a)\Delta)}\exp\left(\frac{V_r + V_{r+2\Delta}}{\sigma^2}\cdot \left(\frac{\kappa}{\tanh(\kappa\Delta)} - \frac{g(a)}{\tanh(g(a)\Delta)}\right)\right)\frac{I_{\xi}\left(\frac{\sqrt{V_{r}V_{r+2\Delta}}}{\sigma^2}\frac{2g(a)}{\sinh(g(a)\Delta)}\right)}{I_{\xi}\left(\frac{\sqrt{V_{r}V_{r+2\Delta}}}{\sigma^2}\frac{2\kappa}{\sinh(\kappa\Delta)}\right)}.
\end{align*}

The above is upper bounded by 
\begin{align*}
&\mathbb{E}[\exp(a\int_{r}^{r+2\Delta} V(s) ds) | V(r), V({r+2\Delta})]\leq \frac{\sinh(\kappa \Delta)}{\kappa }\exp\left(\frac{V_r + V_{r+2\Delta}}{\sigma^2}\left(\frac{\kappa}{\tanh(\kappa\Delta)}- \frac{g(a)}{\tanh(g(a)\Delta)}\right)\right)\\&\cdot\frac{I_{\xi}\left(2\frac{\sqrt{V_{r}V_{r+2\Delta}}}{\Delta\sigma^2}\right)}{I_{\xi}\left(\frac{\sqrt{V_{r}V_{r+2\Delta}}}{\sigma^2}\frac{2\kappa}{\sinh(\kappa\Delta)}\right)}
\end{align*}

Note that 
\begin{align*}
\frac{(x/2)^{\nu}}{\Gamma(\nu+1)} < I_{\nu}(x) < \frac{(x/2)^{\nu}e^x}{\Gamma(\nu+1)},
\end{align*}

so 
\begin{align*}
    \frac{I_{\nu}(x)}{I_{\nu}(y)} < e^x(x/y)^{\nu}.
\end{align*}

Consider $a = \frac{\kappa^2}{2\sigma^2}$ , then $g(a) = 0$, so
\begin{align*}
&\mathbb{E}[\exp(a\int_{r}^{r+2\Delta} V(s) ds) | V(r), V({r+2\Delta})] \\
&\leq (\frac{\Delta\sinh(\kappa \Delta)}{2\kappa})^{\xi + 1}\exp\left(\left[\frac{\kappa}{\tanh{\kappa\Delta}} - \Delta\right]\left(\frac{V_r + V_{r+2\Delta}}{\sigma^2}\right) + \frac{2\sqrt{V_rV_{r+2\Delta}}}{\Delta\sigma^2}\right)\\
&\leq (\frac{\Delta\sinh(\kappa \Delta)}{2\kappa})^{\xi + 1}\exp\left(\left[\frac{\kappa}{\tanh{\kappa\Delta}}\right]\left(\frac{V_r + V_{r+2\Delta}}{\sigma^2}\right)\right).
\end{align*}

Chernoff bounding gives:

\begin{align*}
   \mathbb{P}[\int_{r}^{r+2\Delta}V(s) ds \geq x | V_r, V_{r+2\Delta}] &\leq \left(\Delta\frac{\sinh(\kappa\Delta)}{\kappa}\right)^{\xi + 1}\exp(\left[\frac{\kappa}{\tanh(\kappa\Delta)}\right]\frac{V_r + V_{r+2\Delta}}{\sigma^2})\exp(-\frac{\kappa^2}{2\sigma^2}x).
\end{align*}

\end{proof}

\hestTruncLemma*
\begin{proof}
Since we are assuming the payoff  $f$ is upper bounded by a $B$-Lipschitz linear function in all assets across all time points:
\begin{align*}
\lvert f(\widehat{S}) \rvert \leq B\sum_{k,t}\widehat{S}_k(t),
\end{align*}
which implies that we can without loss of generality look at the truncation bounds for a single $\widehat{S}_k(t)$. Thus, we will now drop the $k$ subscript on $S$. 

Due to the Feller condition, we can assume $\widehat{V}(t) \succeq 0$ and thus $\vec{X}(t) \succeq  0$. Note then $\vec{z}, \vec{y} \in \mathbb{R}^{T}, \vec{w} \in \mathbb{R}^{(T-1)d}, \vec{x} \in \mathbb{R}^{T-1}$. We use $\vec{w}_t \in \mathbb{R}^d$ to denote the entries from $td$ to $(t+1)d$.  Consider
\begin{align*}
&\int s(t) p_{\mathcal{N}}(\vec{z})p_{\chi}(\vec{y})p_{X}(\vec{x} | \vec{z}, \vec{y})p_{\mathcal{N}}(\vec{w}) \\ 
&\propto  \int e^{\frac{\rho}{\sigma}(g_{t+1}(\vec{y}, \vec{z}))+ \left(\frac{2\kappa\Delta\rho}{\sigma} -\Delta\right)\lVert \vec{x}\rVert_1 + \sum_t\sqrt{2\Delta(1-\rho^2)x_t}\cdot (\mathbf{A}_{k, \star} \vec{w}_{t})} p_{\mathcal{N}}(\vec{z})p_{\chi}(\vec{y})p_{X}(\vec{x} | \vec{z}, \vec{y})p_{\mathcal{N}}(\vec{w})\\
&= \int e^{\frac{\rho}{\sigma}v_{t+1}+ \left(\frac{2\kappa\Delta\rho}{\sigma} -\Delta\right)\lVert\vec{x}\rVert_1 + \sum_t\sqrt{2\Delta(1-\rho^2)x_t}\cdot (\mathbf{A}_{k, \star}  \vec{w}_{t})} p(\vec{x} | v_{t+1}, v_0) p(v_{t+1} | v_0)p_{\mathcal{N}}(\vec{w}).
\end{align*}
The equality follows by performing a change of variables using 
\begin{align*}
(z_t, y_t) \rightarrow (z_t, c(y_t + (z_t + \sqrt{\beta g_{t-1}})^2)) =: (z_t, v_t),    
\end{align*}
where $(\vec{z}, \vec{y}) \rightarrow (\vec{y}, \vec{v})$ can be seen to be invertible. More specifically, 
\begin{align*}
p_{z, y}(\vec{z}, \vec{y})d\vec{z}d\vec{y} &= p_{z, y}(\vec{z}, \vec{y}(\vec{z}, \vec{v}))\frac{\partial(\vec{z}, \vec{y})}{\partial(\vec{z}, \vec{v})}d\vec{z}d\vec{v} = p_{z, v}(\vec{z}, \vec{v})d\vec{z}d\vec{v}.\\
\end{align*}
Then we marginalize out the $\vec{z}$ and $v_{1}, \dots, v_{t}$.
We use proportionality since some constant factors are dropped for now.

Let us start with the innermost Gaussian integral over $\vec{w}$. Recall that $\vec{w} \in \mathbb{R}^{(T-1)d}$, where the $w_{m,r}$ will denote the $m$-th time step  for the $r$-th asset. Also, $\vec{w}_m$  denotes the vector of $m$-th time step increments for all assets.

Suppose we integrate over $\lVert \vec{w}\rVert_2 \geq \alpha_2$:
\begin{align*}
&\int_{\lVert \vec{w}\rVert_{2} \geq \alpha_2}p(\vec{w}
)\exp(\sum_{m=0}^{t-2}\sqrt{2\Delta(1-\rho^2)x_{m}} \cdot (\mathbf{A}_{k, \star}\cdot \vec{w}_{m})) \nonumber\\&= \int_{\lVert \vec{w}\rVert_{2} \geq \alpha_2}\mathcal{C}\exp\left(\sum_{m=0}^{t-2}\sum_{r=0}^{d-1}\left(-\frac{1}{2}w_{m,r}^2+\sqrt{2\Delta(1-\rho^2)x_{m}} (A_{k, r}\cdot w_{m,r})\right)\right)\nonumber\\
&=\int_{\lVert \vec{w}\rVert_{2} \geq \alpha_2}\mathcal{C}\exp\left(\sum_{m=0}^{t-2}\sum_{r=0}^{d-1}-\frac{1}{2}\left(w_{m,r} - \sqrt{2\Delta (1-\rho^2)x_{m}}A_{k, r}\right)^2\right)\exp\left(\sum_{m=0}^{t-2}\sum_{r=0}^{d-1}(A_{k, r})^2[\Delta(1-\rho^2)] x_{m}\right) \nonumber\\
&=\int_{\lVert \vec{w}\rVert_{2} \geq \alpha_2}\mathcal{C}\exp\left(\sum_{m=0}^{t-2}\sum_{r=0}^{d-1}-\frac{1}{2}\left(w_{m,r} - \sqrt{2\Delta (1-\rho^2)x_{m}}A_{k, r}\right)^2\right)\exp\left([\Delta(1-\rho^2)] \lVert \vec{x}\rVert_1\right)\\
&=\int_{\lVert \vec{w}\rVert_{2} \geq \alpha_2}\mathcal{C}\exp\left(-\frac{1}{2}\lVert\vec{w}-\vec{\mu}\rVert_{2}^2\right)\exp\left([\Delta(1-\rho^2)] \lVert \vec{x}\rVert_1\right),
\end{align*}
where we used that since $\mathbf{C} = \mathbf{A}^{\mathsf{T}}\mathbf{A}$ is a correlation matrix.  Also $\mu_{k,r} = \sqrt{2\Delta (1-\rho^2)x_{m}}A_{k, r}$. From Gaussian concentration bounds we get that
\begin{align*}
\alpha_2 = \mathcal{O}\left(\lVert \vec{x}\rVert_1 + \sqrt{dT}\log(1/\epsilon)\right)
\end{align*}
suffices for an $\mathcal{O}\left(\epsilon\right)$ error. From the above, we also have the following upper bound on the integral over the whole domain:
\begin{align}
\label{eqn:total_bound}
\int_{\lVert \vec{w}\rVert_{2} \geq 0}p_{\mathcal{N}}(\vec{w}
)\exp(\sum_{m=0}^{t-2}\sqrt{2\Delta(1-\rho^2)x_{m}} \cdot (\mathbf{A}_{k, \star}\cdot \vec{w}_{m})) \leq \exp\left(\Delta(1-\rho^2) \lVert \vec{x}\rVert_1\right).
\end{align}

Next we bound the integral w.r.t. $\vec{x}$: 
\begin{align*}
\int_{\lVert \vec{x} \rVert_1 \geq \alpha_3} e^{\left(\frac{2\kappa\Delta\rho}{\sigma} -\rho^2\Delta\right)\lVert\vec{x}\rVert_1} p(\vec{x} |v_{t+1}, v_0) d\vec{x}.
\end{align*}
Recall $\vec{x} \succ 0$, so given the invertible map :
$\vec{x} \rightarrow (x_0, x_1, \dots, x_{t-2}, \lVert \vec{x}\rVert_1) =: (x_0, x_1, \dots, x_{t-1}, x)$, we perform a change of variables and marginalize out $x_0, x_1, \dots, x_{t-2}$:

\begin{align*}
\int_{x \geq \alpha_3} e^{\left(\frac{2\kappa\Delta\rho}{\sigma} -\rho^2\Delta\right)x} p(x |v_{t+1}, v_0) dx =\mathcal{O}\left(e^{\left(\frac{2\kappa\Delta\rho}{\sigma} -\rho^2\Delta\right)\alpha_3}\mathbb{P}[x \geq \alpha_3 | v_{t+1}, v_0]\right), 
\end{align*}
where $X$ is clearly distributed as $\int_0^{t+1} V(2\Delta s)ds$. The equality follows from Lemma~\ref{lem:tail_trick_lem}. We can then apply Lemma~\ref{lem:v_int_tail_bound} via:
\begin{align*}
    \mathbb{P}[ \int_0^{t+1} V(2\Delta s)ds \geq x | \widehat{V}(t+1), \widehat{V}(0)] =     \mathbb{P}[ \int_0^{2\Delta(t+1)} V(s')ds' \geq 2\Delta x | V(2\Delta(t+1)), V(0)].
\end{align*}
We have that
\begin{align*}
    \mathbb{P}[x \geq \alpha_3 | v_{t+1}, v_0] &\leq \left((t+1)\Delta\frac{\sinh(\kappa(t+1)\Delta)}{2\kappa}\right)^{\xi+1}\exp\left(\left[\frac{\kappa}{\tanh(\kappa(t+1)\Delta)}\right]\frac{v_{t+1} + v_0}{\sigma^2} - \frac{\kappa^2\Delta}{\sigma^2}\alpha_3\right) \\ &\leq \left(\frac{(t+1)\Delta}{2\kappa}\right)^{\xi+1}\exp\left(\kappa(t+1)(\xi+1)\Delta + \left[\frac{\kappa}{\tanh(\kappa(t+1)\Delta)}\right]\frac{v_{t+1} + v_0}{\sigma^2} - \frac{\kappa^2\Delta}{\sigma^2}\alpha_3\right).
\end{align*}
Hence Lemma~\ref{lem:v_int_tail_bound} is valid if \begin{align*}
&\kappa^2 > 2\kappa\sigma\rho - \rho^2\sigma^2 \geq 0.
\end{align*}

Thus
\begin{align*}
&\int_{x \geq \alpha_3} e^{\left(\frac{2\kappa\Delta\rho}{\sigma} -\rho^2\Delta\right)x} p(x |v_{t+1}, v_0) dx \\&=\mathcal{O}\bigg(\left(\frac{(t+1)\Delta}{2\kappa}\right)^{\xi+1}\exp\left(\Delta(t+1)(\xi+1)+\left[\frac{\kappa}{\tanh(\kappa(t+1)\Delta)}\right]\frac{v_{t+1} + v_0}{\sigma^2}\right)\\&\cdot \exp\left(\left[- \frac{\kappa^2\Delta}{\sigma^2} + \frac{2\kappa\Delta\rho}{\sigma} -\rho^2\Delta\right]\alpha_3\right)\bigg).
\end{align*}

From a calculation below, it will be apparent that the integral of the above w.r.t. $v_{t+1}$ will $\mathcal{O}(1)$, so we will be left with
\begin{align*}
&\int_{\mathbb{R}_+}\int_{x \geq \alpha_3} e^{\left(\frac{2\kappa\Delta\rho}{\sigma} -\rho^2\Delta\right)x} p(x |v_{t+1}, v_0) dx dv_{t+1} =
\mathcal{O}\left(\left(\frac{t\Delta}{\kappa}\right)^{\xi}e^{\Delta\xi t -\left[ \frac{\kappa^2\Delta}{\sigma^2} - \frac{2\kappa\Delta\rho}{\sigma} +\rho^2\Delta\right]\alpha_3}\right)
\end{align*}
so 
\begin{align*}
    \alpha_3 = \frac{T\xi\sigma^2\ln(T\Delta/\epsilon\kappa)}{(\kappa^2 - 2\kappa\sigma\rho + \rho\sigma^2)}
\end{align*}
suffices for $\mathcal{O}(\epsilon)$ error. %

Next we bound the integral w.r.t. $\vec{v}$ (dropping some constant factors depending on the initial condition).  We also undo the change of variables to bring back the dependence on $\vec{y}$ and $\vec{z}$.
\begin{align*}
&\int_{ \lVert \vec{z}\rVert_{2}^2 + \lVert \vec{y}\rVert_1 \geq \alpha_1}p_{\chi}(\vec{y})p(\vec{z})\exp(\frac{\kappa/\tanh(\Delta(t+1)\kappa) + \rho\sigma}{\sigma^2}g_{t+1}(\vec{z}, \vec{y}))\\
&\leq \int_{ \lVert \vec{z}\rVert_{2}^2 + \lVert \vec{y}\rVert_1 \geq \alpha_1}p_{\chi}(\vec{y})p_{\mathcal{N}}(\vec{z})\exp((1+\gamma)c\frac{\kappa/\tanh(\Delta(t+1)\kappa) + \rho\sigma}{\sigma^2}(\lVert \vec{z} \rVert_2^2 + \lVert \vec{y}\rVert_1))\\
&\leq \int_{ \mathcal{Y} \geq \alpha_1}p(\mathcal{Y})\exp((1+\gamma)c\frac{\kappa/\tanh(\Delta(t+1)\kappa) + \rho\sigma}{\sigma^2}\mathcal{Y})
\end{align*}
where we  used  Lemma~\ref{lem:sqrt_dif_bound} to bound $g_{t+1}$.
Lastly, recall that the components of $\vec{Y}$ are i.i.d. $\chi^2_{\eta-1}$ and $\vec{z}$ are i.i.d. standard Gaussian. Thus $\lVert \vec{y}\rVert_{2}^2 + \lVert \vec{y}\rVert_1$ is $\chi^2_{t\eta}$ distributed leading to the change of variables to $\mathcal{Y} \sim \chi^2_{t\eta}$. Specifically, $(y_1, \dots y_{t+1}, \vec{z}) \rightarrow (\lVert \vec{z}\rVert_2^2 + \lVert \vec{y}\rVert_1, y_2, \dots y_{t}, \vec{z})$, which is invertible since $\vec{y} \succeq 0$.

Note that by Chernoff bounding and using Lemma~\ref{lem:chi-square-shift-lem} for  $\lambda = \frac{1}{4}$:
\begin{align*}
\mathbb{P}[\mathcal{Y} \geq b] \leq e^{-\lambda b}(1-2\lambda)^{-r/2} \leq e^{-\frac{b}{4} + \frac{r}{2}\ln(2)}.
\end{align*}
Define $m = c(1+\gamma)\frac{\kappa/\tanh(\Delta(t+1)\kappa) + \rho\sigma}{\sigma^2}$, clearly we need $m < 1/2$. Then let $r = t\eta$
\begin{align*}
\int_{ \mathcal{Y} \geq \alpha_1}p(\mathcal{Y})\exp(m\mathcal{Y}) &= (1-2m)^{-r/2}\mathbb{P}[\mathcal{Y} \geq (1-2m)\alpha_1]\\
&\leq (1 - 2m)^{-r/2}e^{-\frac{(1-2m)\alpha_1}{4} + \frac{r}{2}\ln(2)},
\end{align*}
implies we can select $\alpha_1  = \mathcal{O}\left(r\ln(1/\epsilon)\right) = \mathcal{O}\left(T\ln(1/\epsilon)\right)$.  

To ensure the overall error is $\epsilon$ we apply a union bound and  need to scale down the $\epsilon$'s by $\mathcal{O}(Td)$.
Hence
\begin{align}
\label{eqn:alpha_one_three}
&\alpha_1, \alpha_3 = \mathcal{O}(T^2\ln(Td/\epsilon))\\
\label{eqn:alpha_two}
&\alpha_2 = \mathcal{O}(T^{3/2}\sqrt{d}\log(Td/\epsilon)).
\end{align}
There is an additional factor of $T$ from the $e^{T\mu}$ prefactor in the definition of $\widehat{S}$, which requires scaling down $\epsilon$.
Thus it suffices to truncate $\vec{w} \in \mathbb{R}^{(T-1)d}, \vec{y},\vec{z} \in \mathbb{R}^{T},  \vec{x} \in \mathbb{R}^{T-1}$ to $\ell_{\infty}$ balls of size $\mathcal{O}\left((T^{3/2}\sqrt{d} + T^2)\log(Td/\epsilon_{\text{trunc}})\right)$. 

Note that

\begin{align*}
c\frac{\frac{\kappa}{\tanh(\kappa\Delta)} + \rho\sigma}{\sigma^2} &= \frac{1-e^{-2\kappa\Delta}}{4\kappa}\cdot \left(\frac{\kappa}{\tanh(\kappa\Delta)}+\rho\sigma\right) \\
&=  \frac{1 + e
^{-2\kappa\Delta}}{4} + \frac{\rho\sigma}{\kappa} \frac{1 - e
^{-2\kappa\Delta}}{4},
\end{align*}

which we need to be less than $\frac{1}{2(1+\gamma)} = \frac{1}{2(1+e^{-\kappa\Delta/2})}$.

\end{proof}

\subsection{Proof of Theorem \ref{thm:hest_discretization}}
\label{subsec:hestDiscrProof}
We start by analyzing the error incurred from leaving out an interval around $0$ for $\vec{x}_{r}$. The reason for doing this is that for the discretization error, the upper bound we have on the partial derivative is upper bounded by $\frac{1}{\alpha}$ for $\vec{x}_r \geq \alpha$. Since the number of qubits depends logarithmically on this upper bound, the poor dependence on the cutoff point is not a problem.

\begin{lemma}
\label{lem:trunc_around_zero_hest}
Suppose the conditions of Lemma~\ref{lem:hest_truncation_error} are satisfied. If we truncate all $X$ to an  $\ell_{\infty}$ ball of size \begin{align*}
\mathcal{O}\left((\epsilon_{\text{trunc}}/Td)^{T^3d+1}\right)
\end{align*} around zero, then the truncation error is only $\mathcal{O}\left(\epsilon\right)$.
\end{lemma}
\begin{proof}
 Note that the probability that the continuous-time process $V_s$ hits the set $[0, \alpha]$ within the time interval $[0, \Delta T]$ is known given the Feller condition ($\xi > 0$) is satisfied \cite{gikhman2011short}. Specifically, we get
\begin{align}
\label{eqn:feller_bound}
\mathbb{P}[\min_{s \in [0, T-1]} V_{s} \leq \alpha | V_0 = v_0] \leq (v_0e^{\Delta T}\alpha)^{\xi},
\end{align}
where  we clearly need $V_0 > \alpha$. Note that this inequality clarifies the role of the Feller condition, it identifies a phase transition at which $\mathbb{P}[\min_{s \in [0, T-1]} V_{s} \leq \alpha] \rightarrow 0$ as $\alpha \rightarrow 0$. %

From the previous Lemma, we can consider integrating  $\vec{w}, \vec{y}, \vec{z}, \vec{x}$ over a bounded domain. The following will bound the additional error that comes from leaving out a small ball of $\vec{x}$ near zero. The triangle inequality bounds the overall truncation error.

Let $\alpha_4 \leq \frac{1}{2}$ and $m = \frac{\kappa/\tanh(\kappa) + \rho\sigma}{\sigma^2}$ and some arbitrary $t > s$. Let $x_{[a, b]} := \sum_{k=a}^{b} x_{k}$. The following drops some irrelevant constant factors. From the previous proof and Equation \eqref{eqn:total_bound} we know that
\begin{align*}
&\int_{x_s \leq \alpha_4} s(t) p_{\mathcal{N}}(\vec{z})p_{\chi}(\vec{y})p_{X}(\vec{x} | \vec{z}, \vec{y})p_{\mathcal{N}}(\vec{w})  \leq \int_{x_s \leq \alpha_4} p(v_{s}, v_{s+1}, v_{t+1} | v_0)p(\vec{x}| v_0, v_{s}, v_{s+1}, v_{t+1})e^{ \frac{\rho}{\sigma}v_{t+1}+ \Delta\left(\frac{2\kappa\rho}{\sigma} -\rho^2\right)\sum_t x_t}\\
&=\int p(v_{s}, v_{s+1}, v_{t+1} | v_0)p(x_{[0,s-1]}, x_{[s+1, t]}| v_0, v_{s}, v_{s+1}, v_{t+1})e^{\frac{\rho}{\sigma}v_{t+1}+\Delta \left(\frac{2\kappa\rho}{\sigma} -\rho^2\right)\sum_{k\neq s} x_k}\int_{x_s \leq \alpha_{4}} p(x_s | v_{s}, v_{s+1})e^{\Delta\left(\frac{2\kappa\rho}{\sigma} -\rho^2\right)x_s}\\
&\leq \int p(v_{s}, v_{s+1}, v_{t+1} | v_0)p(x_{[0,s-1]}, x_{[s+1, t]}| v_0, v_{s}, v_{s+1}, v_{t+1})e^{ \frac{\rho}{\sigma}v_{t+1}+ \Delta\left(\frac{2\kappa\rho}{\sigma} -\rho^2\right)(x_{[0, s-1]}+x_{[s+1, t]})}\mathbb{P}[x_s \leq \alpha_4 | v_s, v_{s+1}]\\
&\leq \int p(v_{s}, v_{s+1}, v_{t+1}|v_0)e^{\frac{\kappa/\tanh(\Delta\kappa)}{\sigma^2}(v_s + v_{s+1}) + mv_{t+1}}\mathbb{P}[x_s \leq \alpha_4 | v_s, v_{s+1}]\\
&\leq\int p(v_{s}, v_{s+1} | v_0)e^{\frac{\kappa/\tanh(\Delta\kappa)}{\sigma^2}(v_s + v_{s+1})}\mathbb{P}[\min_{r\in[s,s+1]}v_r \leq \alpha_4 | v_s, v_{s+1}]\int p(v_{t+1} | v_{s+1}) e^{m v_{t+1}}.
\end{align*}
The fourth line follows from $x_s = \int_{s}^{s+1}v_r dr \geq \min_{r \in [s, s+1]} v_{r}$. Note that some irrelevant constant factors have been dropped.

We want to marginalize out the $v_{s}, v_{s+1},$ and $v_{t+1}$ and can use the following trick.  Recall that $V(t)$ conditioned on $V(s)$ is noncentral $\chi^2$ distributed:
\begin{align*}
\frac{V(t)}{c} &= (Y(t) + (Z(t) + \gamma\sqrt{\frac{V(s)}{c}})^2)\\
&\implies V(t) \leq c(1+\gamma)(Y(t) + Z(t)^2) + \gamma(1+\gamma)V(s),
\end{align*}
by Jensen's inequality and where $Y$ is $\chi^2_{\xi}$ and $Z$ standard Gaussian.

For arbitrary $v_s, v_t$ and $b$, then
\begin{align*}
\int p(v_t | v_s)e^{b v_t} dv_t &=  \int p(y, z) e^{c(y + (z + \sqrt{\beta v_s})^2)}\\
&\leq e^{(\gamma + \gamma^2)v_s} \int p(y, z) e^{b c(1+\gamma) (y + z^2)}\\
&= e^{(\gamma+\gamma^2)v_s} \int p(y + z^2) e^{b c(1+\gamma) (y + z^2)},
\end{align*}
where $\mathcal{Y} := Y + Z^2$ is $\chi^2_{\eta}$ distributed. We need  $b' = bc(1+\gamma) < 1/2$.

Thus from Lemma~\ref{lem:chi-square-shift-lem}
\begin{align*}
\int p(v_t | v_s)e^{b v_t} &\leq e^{(\gamma+\gamma^2)v_s}(1 - 2b')^{-\eta/2}\int p((1-2b')\mathcal{Y}) \\ &\leq e^{(\gamma+\gamma^2)v_s}(1 - 2b')^{-\eta/2}.
\end{align*}

If we apply this to the expression above
\begin{align*}
\int p(v_{t+1} | v_{s+1}) e^{m v_{t+1}} \leq e^{(\gamma+\gamma^2)v_{s+1}}(1-2c(1+\gamma)m)^{-\eta/2},
\end{align*}
where $c(1+\gamma)m < \frac{1}{2}$ is implied by the fourth condition in Lemma~\ref{lem:hest_truncation_error}.

Then we apply H\"older's inequality to get
\begin{align*}
&\int p(v_{s+1} | v_{s}) e^{\left(\frac{\kappa/\tanh(\Delta \kappa)}{\sigma^2} + (\gamma+\gamma^2)\right)v_{s+1}}\mathbb{P}[\min_{r\in[s,s+1]}v_r \leq \alpha_4 | v_s, v_{s+1}] dv_{s+1} \\
&\leq \text{ess sup}_{p(\cdot | v_s)} e^{\left(\frac{\kappa/\tanh(\Delta \kappa)}{\sigma^2} + (\gamma+\gamma^2)\right)v_{s+1}} \cdot \mathbb{P}[\min_{r\in[s,s+1]}v_r \leq \alpha_4 | v_s ],
\end{align*}
and again
\begin{align*}
&\int p(v_s | v_0) \text{ess sup}_{p(\cdot | v_s)} e^{\left(\frac{\kappa/\tanh(\Delta \kappa)}{\sigma^2} + (\gamma+\gamma^2)\right)v_{s+1}}e^{\frac{\kappa/\tanh(\Delta \kappa)}{\sigma^2}v_s} \mathbb{P}[\min_{r\in[s,s+1]}v_r \leq \alpha_4 | v_s ] \\
&\leq \text{ess sup}_{p(\cdot | v_0)} e^{\frac{\kappa/\tanh(\Delta \kappa)}{\sigma^2}(v_s + v_{s+1}) + (\gamma+\gamma^2)v_{s+1}} \mathbb{P}[\min_{r\in[s,s+1]}v_r \leq \alpha_4 | v_0 ].
\end{align*}

Recall, we can assume $v_{s}, v_{s+1}$ are bounded (we assume the upper truncation has already been applied). Hence, the overall bound is 
\begin{align*}
&\int_{x_s \leq \alpha_4} s(t) p_{\mathcal{N}}(\vec{z})p_{\chi}(\vec{y})p_{X}(\vec{x} | \vec{z}, \vec{y})p_{\mathcal{N}}(\vec{w}) \\ &\leq \left(1 - 2c(1+\gamma)m\right)^{-\eta/2} \text{ess sup}_{p(\cdot | v_0)} e^{\frac{\kappa/\tanh(\Delta \kappa)}{\sigma^2}(v_s + v_{s+1}) + (\gamma+\gamma^2)v_{s+1}} \mathbb{P}[\min_{r\in[s,s+1]}v_r \leq \alpha_4 | v_0 ]\\
&\leq \left(1 - 2c(1+\gamma)m\right)^{-\eta/2} \text{ess sup}_{p(\cdot | v_0)} e^{\frac{\kappa/\tanh(\Delta \kappa)}{\sigma^2}(v_s + v_{s+1}) + (\gamma+\gamma^2)v_{s+1}} \mathbb{P}[\min_{r\in[0,t+1]}v_r \leq \alpha_4 | v_0 ]\\
&\leq \left(1 - 2c(1+\gamma)m\right)^{-\eta/2} \text{ess sup}_{p(\cdot | v_0)} e^{\frac{\kappa/\tanh(\Delta \kappa)}{\sigma^2}(v_s + v_{s+1}) + (\gamma+\gamma^2)v_{s+1}}(v_0e^{\Delta (t+1)}\alpha_4)^{\xi}\\
&= \mathcal{O}\left((Td/\epsilon_{\text{trunc}})^{T^3d}\alpha_4\right),
\end{align*}
which implies $\alpha_4 = \mathcal{O}\left((\epsilon_{\text{trunc}}/Td)^{T^3d+1}\right)$
suffices. We use that $Z, Y$ are truncated to $\ell_{\infty}$ balls of size $\mathcal{O}(T^2d\log(Td/\epsilon_{\text{trunc}}))$ and hence $V(t)$ to $\mathcal{O}(T^3d\log(Td/\epsilon_{\text{trunc}}))$. 
\end{proof}

\hesDiscreteError*
\begin{proof}
Let $\mathcal{O}^{*}$ be $\mathcal{O}$ but ignore factors that are polynomial in $T$ or $d$ or $\log$ in the inverse error. Since  the number of bits is logarithmic in the derivatives, we only care about factors that are exponential in $T$ or $d$ or polynomial in inverse error.

We continue to use the notation from Section~\ref{sec:heston-fast-forwardable}. Hence we have four discrete-time processes: $\vec{Z}, \vec{Y}$ are length $T$, and $\vec{X}, \vec{W}$ are length $T-1$. The path process $\widehat{S}$ is length $T$ and $\widehat{V}$ is length $T+1$. The function $h_t$ computes the $t$-th point in $\widehat{S}$ from the vector of path increments. 

Since the payoff $f$ is piecewise linear with slope at most $B$, we can without loss of generality, apply the discretization analysis to the function $B\sum_{t, k} h_t^{(k)}$. Hence, we simply denote $h^{(k)}_t$ the function that produces the $k$ asset at the $t$ time step for arbitrary $t$ and $k$. Since $k$ is arbitrary, for conciseness, we will drop the $k$ superscript on $h$. Specifically,

\begin{align*}
&u_t(\vec{z}, \vec{y}, \vec{x}, \vec{w}) =  2\Delta\mu t -2\kappa\theta\Delta\frac{\rho}{\sigma} +
\frac{\rho}{\sigma}(g_{t+1}(\vec{y}, \vec{z})  - g_{t}(\vec{y}, \vec{z}))+ \left(\frac{2\kappa\Delta\rho}{\sigma} -\Delta\right)x_t + \sqrt{2\Delta(1-\rho^2)x_t}\cdot (\mathbf{C}_{k, *} \vec{w})\\
&h_t(\vec{z}, \vec{y}, \vec{x}, \vec{w}) = \widehat{S}_k(0)\exp\left(\sum_{r=0}^{t-1}u_t(\vec{z}, \vec{y}, \vec{x}, \vec{w})\right)\\
&\widehat{h}_t(\vec{z}, \vec{y}, \vec{x}, \vec{w}) = h_t(\vec{z}, \vec{y}, \vec{x}, \vec{w})p_{\mathcal{N}}(\vec{z})p_{\chi}(\vec{y})p_{X}(\vec{x} | \vec{z}, \vec{y})p_{\mathcal{N}}(\vec{w}) 
\end{align*}
The desired error will then be scaled down by $\Theta(BTd)$ to ensure the overall error for $B\sum_{t, k} h_t^{(k)}$ is sufficiently small.

The truncated price integral is over a $m:=2dT + 2d(T-1)$-dimensional hyper-rectangle: 
\begin{align*}
(\vec{z},  \vec{y},\vec{x}, \vec{w}) \in [-\alpha_{1}, \alpha_{1}]^{dT} \times [0, \alpha_{1}]^{dT} \times [\alpha_4, \alpha_3]^{d(T-1)} \times [-\alpha_2, \alpha_2]^{d(T-1)} =: \widetilde{\mathcal{M}}.
\end{align*}
Let $R = \max(2\alpha_1, \alpha_3-\alpha_4, 2\alpha_2)$. We let $\mathcal{M}$ denote a uniform grid over $\widetilde{\mathcal{M}}$ with spacing $\frac{R}{N}$ per dimension.

The discretization error from Lemma~\ref{lem:left_rule_error} is 
\begin{align}
\lvert \mathcal{R}_N(\widehat{h}) - \int_{\mathcal{M}} \widehat{h}(\vec{r}) d\vec{r}\rvert  &\leq  \left[\sum_{\vec{r} \in {\mathcal{M}}} \sup_{\vec{j} \in C_{\vec{r}}}\lVert \nabla \widehat{h}(\vec{j})\rVert_{\infty}\left(\frac{R}{N}\right)^m\right] \frac{8dTR}{N} \\
\label{eqn:heston_quadrature_bound}
&\leq \sup_{\vec{j} \in {\mathcal{M}}}\lVert \nabla \widehat{h}(\vec{j})\rVert_{\infty} \frac{8dTR^{m+1}}{N}.
\end{align}
Note that we choose to utilize a uniform bound on $\nabla h$ over the whole box, unlike the more refined estimates we used for CIR and GBM. This is because we do not know the pdf of $\vec{X}$ in closed formed, which makes it challenging to achieve the relative-error bounds we got for GBM. This leads to a bit scaling that grows with dimension, which is likely overly-pessimistic.  Due to this bound, we will only need to track factors that are exponential in $d$ or $T$.

Next, we need to bound the partial derivatives $h_t$ with respect to $\vec{z}, \vec{y}, \vec{x}, \vec{w}$. Since we only care about asymptotics, we will ignore constant factors.
\begin{enumerate}
    \item $y_v$ : $\partial_{y_{v}}\widehat{h}_t \lesssim h_t\left[\sum_{r=v}^{t-1}p_{\mathcal{N}}(\vec{z})p_{\chi}(\vec{y})\partial_{y_v}u_{r}
    + \partial_{y_v}p_{\chi}(\vec{y}) + \partial_{y_v}p_X(\vec{x} | \vec{z}, \vec{y})\right]$
    \item $z_v$ : $\partial_{z_{v}}\widehat{h}_t \lesssim h_t\left[\sum_{r=v}^{t-1}p_{\mathcal{N}}(\vec{z})p_{\chi}(\vec{y})\partial_{z_v}u_{r}
    + \partial_{z_v}p_{\mathcal{N}}(\vec{z}) + \partial_{z_v}p_X(\vec{x} | \vec{z}, \vec{y})\right]$
    \item $x_v$ : $\partial_{x_{v}}\widehat{h}_t\lesssim h_t\left[
   \partial_{x_v}u_{v}
     + \partial_{x_v}p_{X}(\vec{x} | \vec{z}, \vec{y})\right]$
    \item $w_v$ : $\partial_{w_{v}}\widehat{h}_t\lesssim h_t\left[
   \partial_{w_v}u_{v}
     + \partial_{w_v}p_{\mathcal{N}}(\vec{w})\right],$
\end{enumerate}
where $\lesssim$ here is used to indicate asymptotic comparison.

Recall that 
\begin{align}
\label{eqn:char_func_hest}
\Phi(a)&:=\frac{\gamma_{a}\sinh(\kappa \Delta)}{\kappa \sinh(\gamma_{a}\Delta)}\exp\left(\frac{v_t + v_{t+2\Delta}}{\sigma^2}\cdot \left(\frac{\kappa}{\tanh(\kappa\Delta)} - \frac{\gamma_{a}}{\tanh(\gamma_{a}\Delta)}\right)\right)\frac{I_{\xi}\left(\frac{\sqrt{v_{t}v_{t+1}}}{\sigma^2}\frac{2\gamma_{a}}{\sinh(\gamma_{a}\Delta)}\right)}{I_{\xi}\left(\frac{\sqrt{v_{t}v_{t+1}}}{\sigma^2}\frac{2\kappa}{\sinh(\kappa\Delta)}\right)},
\end{align}
and so
\begin{align}
\label{eqn:pdf_of_w}
    p_{X}(x_v | \vec{z}, \vec{y})= \frac{1}{2\pi}\int_{0}^{M}(e^{-iax_v}\Phi(a) + e^{iax_v}\Phi^{*}(a))da.
\end{align}

First we start by bounding $\lvert \partial_{s_v}p_X( \cdot | \vec{z}, \vec{y})\rvert$ for $s_v = y_v$ or $z_v$. Recall that $v_t$ can be computed from $\vec{z}$ and $\vec{y}$ through $g_t$. 

Consider 
\begin{align*}
&\lvert \frac{d}{dq}\frac{I_{\xi}\left(\frac{q}{\sigma^2}\frac{2\gamma_{a}}{\sinh(\gamma_{a}\Delta)}\right)}{I_{\xi}\left(\frac{q}{\sigma^2}\frac{2\kappa}{\sinh(\kappa\Delta)}\right)}\rvert\\
&=\lvert \left[I_{\xi}\left(\frac{q}{\sigma^2}\frac{2\kappa}{\sinh(\kappa\Delta)}\right)\right]^{-2}\frac{2\gamma_a}{\sigma^2\sinh(\gamma_a\Delta)}I_{\xi}'\left(\frac{q}{\sigma^2}\frac{2\gamma_{a}}{\sinh(\gamma_{a}\Delta)}\right)I_{\xi}\left(\frac{q}{\sigma^2}\frac{2\kappa}{\sinh(\kappa\Delta)}\right) \nonumber\\
& - \left[I_{\xi}\left(\frac{q}{\sigma^2}\frac{2\kappa}{\sinh(\kappa\Delta)}\right)\right]^{-2}\frac{2\kappa}{\sigma^2\sinh(\kappa\Delta)}I_{\xi}\left(\frac{q}{\sigma^2}\frac{2\gamma_{a}}{\sinh(\gamma_{a}\Delta)}\right)I_{\xi}'\left(\frac{q}{\sigma^2}\frac{2\kappa}{\sinh(\kappa\Delta)}\right)\rvert \nonumber\\
& = \mathcal{O}^{*}\left( e^{\frac{8q}{\Delta\sigma^2}}\right),
\end{align*}
where the last equality follows from
\begin{align}
\label{eqn:useful_bounds_1}
&I_{\xi}'(z) = \frac{1}{2}\left(I_{\xi-1}(z) + I_{\xi+1}(z)\right) \\
\label{eqn:useful_bounds_2}
&\frac{(x/2)^{\xi}}{\Gamma(\xi+1)} < I_{\xi}(x) < \frac{(x/2)^{\xi}e^x}{\Gamma(\xi+1)}\\
\label{eqn:useful_bounds_3}
&\lvert \frac{\gamma_a}{\sinh(\gamma_a\Delta)}\rvert \leq \frac{2}{\Delta}.
\end{align}

We have that $q = \sqrt{v_{t+1}v_{t}}$, so we apply chain rule for $s_v = y_v$ or $s_v = z_v$:
\begin{align*}
&\frac{v_{t+1}\partial_{s_v}v_{t} + v_t\partial_{s_v}v_{t+1}}{2\sqrt{v_tv_{t+1}}}p_{\mathcal{N}}(\vec{z})p_{\chi}(\vec{y}) \\ &\leq \frac{\partial_{s_v}v_t}{2\sqrt{v_t}}\sqrt{\alpha_1}p_{\mathcal{N}}(\vec{z})p_{\chi}(\vec{y}) + \frac{\partial_{s_v}v_{t+1}}{2\sqrt{v_{t+1}}}\sqrt{\alpha_1}p_{\mathcal{N}}(\vec{z})p_{\chi}(\vec{y}),
\end{align*}
where $\alpha_1$ is our truncation bound on $\vec{v}$ from the proof of Lemma~\ref{lem:hest_truncation_error}. 
Since $\sqrt{v_{t}} =\sqrt{g_{t}} \geq \sqrt{y_{t-1}}$ we get that
\begin{align*}
\frac{\partial_{s_{j}}v_{t}}{\sqrt{v_{t}}}p_{\mathcal{N}}(\vec{z})p_{\chi}(\vec{y}) &= \partial_{s_j} v_j\frac{\partial_{v_{j}}v_{t}}{\sqrt{v_{t}}}p_{\mathcal{N}}(\vec{z})p_{\chi}(\vec{y})\\
&\leq \partial_{s_j} v_j\frac{\partial_{v_{j}}v_{t}}{\sqrt{y_{t-1}}}p_{\mathcal{N}}(\vec{z})p_{\chi}(\vec{y})\\
&\leq \partial_{s_j} v_j\frac{1}{\sqrt{y_{t-1}}}\prod_{k=j}^{t-1}\frac{z_{k}y_{k-1}^{r/2-3/2} + y_{k-1}^{r/2-1}}{\sqrt{2\pi}\Gamma(r/2)2^{r/2}}e^{-\frac{y_{k-1} +z_{k}^2}{2}}\\
&\leq\partial_{s_j} v_j\frac{y_{t-1}^{r/2-3/2}}{\sqrt{2\pi}\Gamma(r/2)2^{r/2}}e^{-\frac{y_{k-1}}{2}}\prod_{k=j}^{t-1}\frac{z_{k}y_{k-1}^{r/2-3/2} + y_{k-1}^{r/2-1}}{\sqrt{2\pi}\Gamma(r/2)2^{r/2}}e^{-\frac{y_{k-1} +z_{k}^2}{2}}\\
&=\mathcal{O}^{*}(1),
\end{align*}
which follows from Equation \eqref{eqn:chain_rule_rec}. Hence $\frac{\partial q}{\partial s_{v}} = \mathcal{O}^{*}(1)$.

The above then implies that
\begin{align*}
\lvert \partial_{s_v} p_{X} \rvert &\leq M \cdot \max_{a\in [0, M]}\lvert  \partial_{s_v}\Phi(a)\rvert \\&= \mathcal{O}^{*}\left(\exp\left(\frac{v_t + v_{t+1}}{\sigma^2}\left(\frac{\kappa}{\tanh(\kappa\Delta)} - \text{Re}[\frac{\gamma_{a}}{\tanh(\gamma_{a}\Delta)}]\right) + \frac{8}{\Delta\sigma^2} \sqrt{v_tv_{t+1}}\right)\right)\\
&=\mathcal{O}^{*}\left(\exp\left(\frac{8\alpha_1(\kappa+1)}{\Delta\sigma^2} \right)\right).
\end{align*}

Thus combining the above results, we get that $\lvert \partial_{s_v}p_X(\vec{x} | \vec{z}, \vec{y})\rvert = \mathcal{O}^{*}\left(\exp\left(\frac{8\alpha_1(\kappa+1)}{\Delta\sigma^2} \right)\right)$.

Next, we bound $\lvert \partial_{x_v}p_X(x_v | \vec{z}, \vec{y})\rvert$. From Equation \eqref{eqn:pdf_of_w}, we get that $\lvert \partial_{x_v}p_X(x_v | \vec{z}, \vec{y})\rvert \leq M \cdot \max_{a\in [0, M]} \Phi(a).$  Thus, due to form of $\Phi(a)$, $\lvert \partial_{x_v}p_X(x_v | \vec{z}, \vec{y})\rvert$ is $\mathcal{O}^*$ of the same quantity as $\partial_{s_v}p_X(\vec{x} | \vec{z}, \vec{y})$.

From the proof of Theorem~\ref{thm:sqrt_discretization}, we know that $\partial_{y_v}p_{\chi}(\vec{y}) = \mathcal{O}(1)$, and $\partial_{z}p_{\mathcal{N}}(\vec{z}) = \mathcal{O}(1)$. 
Next, we bound \begin{align*}
    p_{\mathcal{N}}(\vec{z})p_{\chi}(\vec{y})\partial_{s_v}u_{r},
\end{align*}
for $s_v$ equal to $y_v, z_v, w_v$ or $x_v$. Note that only the first term in Equation \eqref{eqn:v_part_of_u_incr} depends on  $y_v, z_v$, which is just a linear combination of $p_{\mathcal{N}}(\vec{z})p_{\chi}(\vec{y})\partial_{s_v}v_r$ terms. From the proof  of Theorem~\ref{thm:sqrt_discretization}, we know that these terms are $\mathcal{O}^{*}(1)$. For $\partial_{w_v}u_v$, we have a bound of \begin{align*}
    A_{k, v}\cdot\sqrt{\Delta(1-\rho^2)x_v} \leq A_{k, v}\cdot\sqrt{\Delta(1-\rho^2)\alpha_{3}} = \mathcal{O}^{*}(1),
\end{align*}
and for $\partial_{x_v}u_v$ we have a bound of 
\begin{align*}
A_{k, \cdot }\vec{w}\cdot\frac{\sqrt{\Delta(1-\rho^2)}}{\sqrt{x_v}} = \mathcal{O}^{*}(1/\alpha_4) =\mathcal{O}^{*}((Td/\epsilon_{\text{trunc}})^{T^3d+1}),
\end{align*}
using the value from Lemma \ref{lem:trunc_around_zero_hest}.

Note that the additional exponential factor in $\nabla h$ is $\exp\left(\sum_{r=0}^{t-1}u_r\right) = \mathcal{O}^{\star}\left(e^{T(\alpha_1 + \alpha_{3} + \alpha_2\sqrt{\alpha_{3}})}\right)$. If we include this with our bounds on the various components of $\nabla \ln(h)$, we get from \eqref{eqn:alpha_one_three} and \eqref{eqn:alpha_two}:
\begin{align*}
    \max_{\mathcal{M}}\lVert \nabla h\rVert_{\infty} &= \mathcal{O}^{\star}\left(e^{T(\alpha_1 + \alpha_{3} + \alpha_2\sqrt{\alpha_{3}})}(Td/\epsilon_{\text{trunc}})^{T^3d+1}\right)\\
    & =\mathcal{O}^{\star}\left(e^{T^{5/2}d^{1/2}\log^{3/2}(Td/\epsilon_{\text{trunc}})}(Td/\epsilon_{\text{trunc}})^{T^3d+1}\right)
\end{align*}
where summing over all assets at each time step does not change $\mathcal{O}^{*}$ complexity. Thus combining this bound with Equation \eqref{eqn:heston_quadrature_bound} gives 
\begin{align*}
&\mathcal{O}^{\star}\left(e^{T^{5/2}d^{1/2}\log^{3/2}(Td/\epsilon_{\text{trunc}})}(Td/\epsilon_{\text{trunc}})^{T^3d}R^{m}\right) \\ &= \mathcal{O}^{\star}\left(e^{T^{5/2}d^{1/2}\log^{3/2}(Td/\epsilon_{\text{trunc}})}(Td/\epsilon_{\text{trunc}})^{T^3d}\left[(T^2+T^{3/2}\sqrt{d})\log(1/\epsilon_{\text{trunc}})\right]^{4Td}\right) .
\end{align*}

So it suffices to take

\begin{align*}
\log_2(N) &= \mathcal{O}\left(T^3d\log^{3/2}(Td/\epsilon_{\text{trunc}})\log(Td/\epsilon_{\text{disc}})\right),
\end{align*}
and so
\begin{align*}
\mathcal{O}\left(T^4d^2\log^{3/2}(Td/\epsilon_{\text{trunc}})\log(Td/\epsilon_{\text{disc}})\right)
\end{align*}
(qu)bits in total.

\end{proof}

\section{Additional Proofs for Section \ref{sec:approximate_sde_simulation} \&  \ref{sec:multi_level_monte_carlo}}

\subsection{Proof of Theorem~\ref{thm:gen_quantum_mlmc}}

\begin{theorem}
Let $P$ denote a random variable, and let $P_l (l=0,1,\dots,L)$ denote a sequence of random variables such that $P_l$ approximates $P$ at level $l$.
Let $\hat{Y}_l$ denote the unbiased estimator for $P_l - P_{l-1}$ constructed from $N_l$ samples of $P_l - P_{l-1}$, where we define $P_{-1} \equiv 0$.
Let $V_l$ and $C_l$ be the variance and computational complexity of $\hat{Y}_l$ respectively.
If there exists positive constants $\alpha$, $\beta$, $\gamma$, $\delta$ and $c_1$, $c_2$, $c_3$ such that
\begin{align}
    &\abs{\mathbb{E}\left[ P_l - P \right]} \le c_1 h_l^\alpha, \\
    &V[\hat{Y}_l] \le c_2 N_l^{-\delta} h_l^\beta, \\
    &C_l \le c_3 N_l h_l^{-\gamma},
\end{align}
where $h_l = M^{-l}T$ and $M > 1$ is an integer,
then for any $\epsilon < 1$, there is an algorithm that estimates $\mathbb{E}[P]$ up to a mean-squared error $\epsilon^2$ with a computational complexity $C$ bounded by
\begin{equation}
\begin{cases}
O\left( \epsilon^{-\frac{2}{\delta}} \left( \log \epsilon^{-1} \right)^{\frac{1}{\delta}} \right) + O \left( 
\epsilon^{-\frac{\gamma}{\alpha}} \right), & \gamma < \frac{\beta}{\delta}, \\
O\left( \epsilon^{-\frac{2}{\delta}} \left( \log \epsilon^{-1} \right)^{\frac{1}{\delta}+1} \right) + O \left( 
\epsilon^{-\frac{\gamma}{\alpha}} \right), & \gamma = \frac{\beta}{\delta}, \\
O\left( \epsilon^{-\frac{2}{\delta}-\frac{\gamma - \frac{\beta}{\delta}}{\alpha}} \left( \log \epsilon^{-1} \right)^{\frac{1}{\delta}} \right) = O\left(  \epsilon^{-\frac{\gamma}{\alpha}-\frac{1}{\delta}\left(2 - \frac{\beta}{\alpha}\right)} \left( \log \epsilon^{-1} \right)^{\frac{1}{\delta}} \right), & \gamma > \frac{\beta}{\delta}.
\end{cases}
\end{equation}
\end{theorem}
\begin{proof}
Let $\hat{Y} = \sum_{l=0}^{L} \hat{Y}_l$, and choose $L$ to be
\begin{equation}
\label{eq:mlmc-num-levels-bound}
    L = \left\lceil \frac{\log \left(\sqrt{2} c_1 T^\alpha \epsilon^{-1} \right) }{\alpha \log M} \right\rceil,
\end{equation}
so that
\begin{equation}
\label{eq:mlmc-hl-bound}
    \frac{1}{\sqrt{2}}M^{-\alpha} \epsilon < c_1 h_L^\alpha \le \frac{1}{\sqrt{2}}\epsilon,
\end{equation}
we will have
\begin{align}
    \left( \mathbb{E}[\hat{Y}] - \mathbb{E}[P]\right)^2  
    &= \left( \sum_{l=0}^L \mathbb{E}[\hat{Y}_l] - \mathbb{E}[P]\right)^2 \\
    &= \left( \sum_{l=0}^L \mathbb{E}[P_l - P_{l-1}] - \mathbb{E}[P]\right)^2 \\
    &= \left( \mathbb{E}[P_L] - \mathbb{E}[P]\right)^2 \\
    &\le \left(c_1 h_L^\alpha\right)^2 \\
    &= \frac{1}{2}\epsilon^2.
\end{align}

If we also choose $N_l$ to be
\begin{equation}
    N_l = \left \lceil \left( \frac{1}{2 c_2 (L+1)} \epsilon^2 h_l^{-\beta}\right)^{-\frac{1}{\delta}} \right \rceil,
\end{equation}
then we will have
\begin{equation}
    \mathbb{V}[\hat{Y}] = \sum_{l=1}^L c_2 N_l^{-\delta} h_l^\beta \le \frac{1}{2} \epsilon^2,
\end{equation}
and consequently,
\begin{equation}
    \mathbb{E}\left[\left(\hat{Y} - \mathbb{E}[P]\right)^2\right] = \mathbb{V}[\hat{Y}] + \left( \mathbb{E}[\hat{Y}] - \mathbb{E}[P]\right)^2 \le \epsilon^2,
\end{equation}
and hence we can use $\hat{Y}$ as the estimator that satisfies the requirements in the theorem.

And the total computational complexity for $\hat{Y}$ is
\begin{align}
    C &= \sum_{l=0}^L C_l \le c_3 \sum_{l=0}^L N_l h_l^{-\gamma} \\
    &\le c_3 \sum_{l=0}^L \left[\left( \frac{1}{2 c_2 (L+1)} \epsilon^2 h_l^{-\beta}\right)^{-\frac{1}{\delta}} + 1 \right] h_l^{-\gamma} \\
    &= c_3 [2 c_2 (L+1)]^\frac{1}{\delta} \epsilon^{-\frac{2}{\delta}} \sum_{l=0}^L h_l^{-\left(\gamma - \frac{\beta}{\delta}\right)} + c_3 \sum_{l=0}^L h_l^{-\gamma}. \label{eq:mlmc-cost-raw}
\end{align}
From \Cref{eq:mlmc-hl-bound} we have
\begin{equation}
    h_L^{\alpha} > \frac{1}{\sqrt{2}c_1} M^{-\alpha} \epsilon,
\end{equation}
and hence
\begin{equation}
    \sum_{l=0}^{L} h_l^{-d} = h_L^{-d} \sum_{l=0}^L M^{-dl} < \frac{M^d}{M^d - 1} h_L^{-d} < \frac{M^{2d}}{M^d-1}\left(\frac{\epsilon}{\sqrt{2}c_1}\right)^{-\frac{d}{\alpha}}, \qquad \forall \, d > 0.
\end{equation}
Therefore, the second term in \Cref{eq:mlmc-cost-raw} is bounded by
\begin{equation}
    c_3 \sum_{l=0}^L h_l^{-\gamma} < c_3 \frac{M^{2\gamma}}{M^\gamma-1}\left(\frac{\epsilon}{\sqrt{2}c_1}\right)^{-\frac{\gamma}{\alpha}} = O\left(\epsilon^{-\frac{\gamma}{\alpha}}\right).
\end{equation}

Now we prove the bounds on the first term in the computational complexity for $\hat{Y}$ (\Cref{eq:mlmc-cost-raw}) for different cases.

\paragraph{(a) If $\gamma < \frac{\beta}{\delta}$, } 
\begin{equation}
    \sum_{l=0}^L h_l^{-\left(\gamma - \frac{\beta}{\delta}\right)} < \frac{1}{1 - h_l^{\left(\frac{\beta}{\delta} - \gamma\right)}} = \frac{M^{\left(\frac{\beta}{\delta} - \gamma\right)}}{M^{\left(\frac{\beta}{\delta} - \gamma\right)}-1}T^{\left(\frac{\beta}{\delta} - \gamma\right)} = O(1).
\end{equation}
In addition, from \Cref{eq:mlmc-num-levels-bound} we have $L = O(\log \epsilon^{-1})$.
Therefore
\begin{equation}
    c_3 [2 c_2 (L+1)]^\frac{1}{\delta} \epsilon^{-\frac{2}{\delta}} \sum_{l=0}^L h_l^{-\left(\gamma - \frac{\beta}{\delta}\right)} = O\left( \epsilon^{-\frac{2}{\delta}} \left( \log \epsilon^{-1} \right)^{\frac{1}{\delta}}  \right),
\end{equation}
and hence
\begin{equation}
    C = O\left( \epsilon^{-\frac{2}{\delta}} \left( \log \epsilon^{-1} \right)^{\frac{1}{\delta}} \right) + O \left( 
\epsilon^{-\frac{\gamma}{\alpha}} \right).
\end{equation}

\paragraph{(b) If $\gamma = \frac{\beta}{\delta}$, } 
\begin{equation}
    \sum_{l=0}^L h_l^{-\left(\gamma - \frac{\beta}{\delta}\right)} = L+1 = O\left( \log \epsilon^{-1} \right).
\end{equation}
Therefore, similar as the previous case,
\begin{equation}
    C = O\left( \epsilon^{-\frac{2}{\delta}} \left( \log \epsilon^{-1} \right)^{\frac{1}{\delta}+1} \right) + O \left( 
\epsilon^{-\frac{\gamma}{\alpha}} \right).
\end{equation}

\paragraph{(c) If $\gamma > \frac{\beta}{\delta}$, } 
from \Cref{eq:mlmc-hl-bound} we have,
\begin{equation}
    \sum_{l=0}^L h_l^{-\left(\gamma - \frac{\beta}{\delta}\right)} < \frac{M^{2\left(\gamma-\frac{\beta}{\delta}\right)}}{M^{\gamma-\frac{\beta}{\delta}}-1}\left(\frac{\epsilon}{\sqrt{2}c_2}\right)^{-\frac{\gamma-\frac{\beta}{\delta}}{\alpha}} = O\left(\epsilon^{-\frac{\gamma-\frac{\beta}{\delta}}{\alpha}}\right).
\end{equation}
Therefore
\begin{align}
    C &= O\left(\epsilon^{-\frac{2}{\delta}-\frac{\gamma-\frac{\beta}{\delta}}{\alpha}} \left( \log \epsilon^{-1} \right)^{\frac{1}{\delta}} \right) + O \left( 
\epsilon^{-\frac{\gamma}{\alpha}} \right) \\
    &= O\left(\epsilon^{-\frac{2}{\delta}-\frac{\gamma-\frac{\beta}{\delta}}{\alpha}} \left( \log \epsilon^{-1} \right)^{\frac{1}{\delta}} \right).
\end{align}
The second equality above is due to $\beta < 2 \alpha$.
\end{proof}

\subsection{Derivation of Equation \eqref{eqn:mlmc-variance}}
\label{sec:mlmc-variance_proof}

Recall the notation from Section~\ref{sec:mlmc_error_analysis}. We start by providing the bound on the variance for the truncated and discretized MLMC estimator. Let $p_{\ell} : \mathbb{R}^{d\times 2^{\ell}T} \rightarrow \mathbb{R}$ denote the density for $\widetilde{Y}_{\ell}$  and $p$ the density for $\widehat{Y}$.

\begin{align*}
\textup{Var}(f(\mathcal{Y}_{\ell}) - f(\mathcal{Y}_{\ell-1})) &\leq B^2T \sum_{t=0}^{T-1}\mathbb{E}[\lVert \mathcal{Y}_{\ell}(t2^{\ell}) - \mathcal{Y}_{\ell-1}(t2^{\ell})\rVert_{2}^2]\\
&=\mathcal{O}\bigg(B^2T \bigg(\sum_{t=0}^{T-1}\lvert \mathbb{E}[\lVert \mathcal{Y}_{\ell}(t2^{\ell}) - g_{\ell}(\mathcal{Y}_{\ell}(t2^{\ell}))\rVert_{2}^2]- \int_{[-R, R]^{d\times 2^{\ell}T}}\lVert \vec{x}_{t2^{\ell}} - [g_{\ell}(\vec{x})]_{t2^{\ell}}\rVert_{2}^2p_{\ell}(\vec{x})\rvert \\ &+ \int_{[-R, R]^{d\times 2^{\ell}T}}\lVert \vec{x}_{t2^{\ell}} - [g(\vec{x})]_{t2^{\ell}}\rVert_{2}^2p_{\ell}(\vec{x})\rvert\bigg)\bigg)\\
&=\mathcal{O}\bigg(B^2T\sum_{t=0}^{T-1}\bigg(\lvert \mathbb{E}[\lVert \mathcal{Y}_{\ell}(t2^{\ell}) - g_{\ell}(\mathcal{Y}_{\ell}(t2^{\ell}))\rVert_{2}^2]- \int_{[-R, R]^{d\times 2^{\ell}T}}\lVert \vec{x}_{t2^{\ell}} - [g_{\ell}(\vec{x})]_{t2^{\ell}}\rVert_{2}^2p_{\ell}(\vec{x})\rvert \\ &+ \int_{[-R, R]^{d} \times \mathbb{R}^{d}} \lVert \vec{x} - \vec{y}\rVert_2^2 \widetilde{p}_{\ell}(\vec{x})p(y)\bigg)\bigg)\\
&=\mathcal{O}\bigg(B^2T\sum_{t=0}^{T-1}\bigg(\lvert \mathbb{E}\lVert \mathcal{Y}_{\ell}(t2^{\ell})- g_{\ell}(\mathcal{Y}_{\ell}(t2^{\ell}))\rVert_{2}^2- \int_{[-R, R]^{d\times 2^{\ell}T }}\lVert \vec{x} - g_{\ell}(\vec{x})\rVert_{2}^2p_{\ell}(\vec{x})\rvert \\ &+ \mathbb{E}[\lVert \widetilde{Y}_{\ell}(t2^{\ell}) - \widehat{Y}(t2^{\ell})\rVert_2^2]\bigg)\bigg).
\end{align*}
The last equality follows from the positivity of the integrand.  Note that $\mathbb{E}[\lVert \mathcal{Y}_{\ell}(t2^{\ell}) - g_{\ell}(\mathcal{Y}_{\ell}(t2^{\ell}))\rVert_{2}^2]$ can be viewed as a discretized and renormalized approximation to the \emph{truncated} expectation of the function $\lVert (i - g_{\ell})(\vec{x})\rVert_2^2$ applied to a random variable with density $p_{\ell}(\vec{x})$, where $i(\vec{x}) = \vec{x}$ is the identity function. It is clear that $g_{\ell}$ is $1$-Lipschitz, where the only non-Lipschitz component comes from the $\ell_2$ norm.
Hence the second term corresponding to $\textup{Err}(\ell, t)$ from the main text is just discretization error.

\subsection{Proof of Lemma~\ref{lem:milstein_loader_truncation}}

\milTruncErr*
\begin{proof}

Since the function $g$ in Equation \eqref{eqn:integral_2n_milstein_sampling} corresponds to $(n,2)$-Milstein sampling, we can without loss of generality assume each component of the $d$-dimensional SDE (Equation \eqref{eqn:ito-process}) is of the form:
\begin{align}
\label{eqn:two-correlated-sde}
    dX_i(t) = \mu_i(\vec{X}(t)) dt + \sigma_{ii}(\vec{X}(t))dW_i + \sigma_{ij}(\vec{X}(t))dW_j, 
\end{align}
for some unique pair of $(i,j)$. This captures that each $\vec{X}_i$ can only be correlated with two Brownian motions $dW_i$ and $dW_j$ for some $j \neq i$. Note, however, the drift $\mu_i$ and diffusion $\vec{\sigma}_{i, \cdot}$ can be functions of the entire process $\vec{X}$.

The general Milstein scheme is 
\begin{align*}
 \widetilde{X}_i(t+h) &= \widetilde{X}_i(t) + \mu_i(\widetilde{X}(t))h + \sum_{k=1}^{d}\sigma_{ik}(\widetilde{X}(t))\Delta W_k + \sum_{k, j=1}^d \mathcal{L}^j\sigma_{ik}(\widetilde{X}(t))(\Delta W_{j}\Delta W_{k}  + A_{(j,k)}),
\end{align*}
where \begin{align*}
    \mathcal{L}^{k} := \sum_{i=1}^{d}\frac{1}{2}\sigma_{ik}\frac{\partial}{\partial x_i}, k= 1, \dots, d.
\end{align*}

We can use the form of \eqref{eqn:two-correlated-sde} to reduce this to
\begin{align*}
 \widetilde{X}_i(t+h) &= \widetilde{X}_i(t) + \mu_i(\widetilde{X}(t))h + \sigma_{ii}(\widetilde{X}(t))\Delta W_i  + \sigma_{ij}(\widetilde{X}(t))\Delta W_j \\&+\mathcal{L}^i\sigma_{ii}(\widetilde{X}(t))(\Delta W_{i}\Delta W_{i}  + A_{(i,i)}) + \mathcal{L}^j\sigma_{ii}(\widetilde{X}(t))(\Delta W_{j}\Delta W_{i}  + A_{(j,i)})\\& + \mathcal{L}^i\sigma_{ij}(\widetilde{X}(t))(\Delta W_{i}\Delta W_{j}  + A_{(i,j)})  + \mathcal{L}^j\sigma_{ij}(\widetilde{X}(t))(\Delta W_{j}\Delta W_{j}  + A_{(j,j)}),
\end{align*}
giving that

\begin{align*}
\widetilde{X}_i(t+h)&= \widetilde{X}_i(t) + \mu_i(\widetilde{X}(t))h + \sigma_{ii}(\widetilde{X}(t))\Delta W_i  + \sigma_{ij}(\widetilde{X}(t))\Delta W_j  \\
 &+\mathcal{L}^i\sigma_{ii}(\widetilde{X}(t))(\Delta W_{i}\Delta W_{i}) + \mathcal{L}^j\sigma_{ij}(\widetilde{X}(t))(\Delta W_{j}\Delta W_{j}) \\
 &+[\mathcal{L}^i\sigma_{ij}(\widetilde{X}(t)  - \mathcal{L}^j\sigma_{ij}(\widetilde{X}(t))][\Delta W_i \Delta W_j + A_{(i,j)}].
\end{align*}

Our hypothesis implies that $\mu_i$ and $\sigma_{ik}$ are at most linear functions, and so 

\begin{align*}
&\mu_i(\vec{X}(t)) \leq \alpha\sum_{k=1}^{d}\lvert {X}_k(t)\rvert \\
&\sigma_{ik}(\vec{X}(t)) \leq \beta \sum_{k=1}^{d}\lvert{X}_k(t)\rvert \\
&\mathcal{L}^j\sigma_{ik}(\vec{X}(t))\leq \beta^2 d\sum_{k=1}^{d}\lvert{X}_k(t)\rvert
\end{align*}
for constants $\alpha$ and $\beta$. Hence we can write
\begin{align*}
 \lvert \widetilde{X}_i(T-1) \rvert &\leq \lvert \widetilde{X}_i(T-2)\rvert +  \alpha \sum_{k=1}^d\lvert \widetilde{X}_k(T-2)\rvert h + \beta \sum_{k=1}^{d}\lvert\widetilde{X}_k(T-2)\rvert[\lvert\Delta W_i(T-1)\rvert + \lvert \Delta W_j(T-1)\rvert] \\
&+\beta^2 d\sum_{k=1}^{d}\lvert\widetilde{X}_k(T-2)\rvert[ (\Delta W_i(T-1))^2 + (\Delta W_j(T-1))^2] + 2\beta^2 d \sum_{k=1}^{d}\lvert\widetilde{X}_k(T-2)\rvert[\lvert \Delta W_i(T-1)\Delta W_j(T-1)\rvert \\&+ \lvert A_{(i,j)}(T-1)\rvert ]\\
&\leq \sum_{k=1}^{d}\lvert\widetilde{X}_k(T-2)\rvert\bigg(\alpha h + \beta\lvert \Delta W_i(T-1) \rvert + \beta \lvert \Delta W_j(T-1) \rvert \\ &+ \beta^2 d(\Delta W_i(T-1))^2 + \beta d(\Delta W_j(T-1))^2 + \beta^2 d\lvert \Delta W_i(T-1)\Delta W_j(T-1)\rvert  +\beta^2 d \lvert A_{(i,j)}(T-1)\rvert\bigg),
\end{align*}
where $\widetilde{X}(T-2)$ is independent of $\Delta W_i(T-1), \Delta W_j(T-1), A_{(i,j)}(T-1)$.

The integration problem is 
\begin{align*}
    \int_{ (\mathbb{R}^{d} \times \mathbb{R}^{\lceil d/2 \rceil})^{\times T}}  f\circ g(\vec{z}, \vec{a}) d\vec{z}d\vec{a}.
\end{align*}
We need to bound the integral outside of the hypercube $\mathcal{M} = [-R, R]^{T(d + \lceil d/2\rceil)}$ in terms of $R$. By our assumptions on $f$, we can focus on bounding
\begin{align*}
    \int_{\mathcal{M}^{c}} \sum_{i=1}^{d}\lvert g_{i,T-1}(\vec{z}, \vec{a})\rvert d\vec{z}d\vec{a},
\end{align*}
where $g_{i, T-1}$ constructs $\widetilde{X}_i(T-1)$ for $i$. To bound this, we need to unroll the recursion in $g_{i, T-1}$. For $k \leq d$, we will use $k'$ to denote the unqiue index  $\neq k$ such that $b_{k,k'} \neq 0$. This leads to

\begin{align*}
\lvert g_{i,T-1}(\vec{z}, \vec{a})\rvert  &\leq \sum_{k=1}^{d} \lvert g_{k,T-2}(\vec{z}, \vec{a}) \rvert \bigg(\alpha h + \beta\lvert w_{i, T-1} \rvert + \beta \lvert w_{j,T-1}\rvert + \beta^2 d(w_{i, T-1})^2 + \beta^2 d(w_{i, T-1})^2 + \beta^2 d \lvert w_{i, T-1} w_{j,T-1}\rvert \\ &+ \beta^2 d \lvert a_{i,j,T-1}\rvert\bigg)\\
&\leq \left(\alpha h + \beta\lvert w_{i, T-1} \rvert + \beta \lvert w_{j,T-1}\rvert + \beta^2 d(w_{i, T-1})^2 + \beta^2 d(w_{i, T-1})^2 + \beta^2 d \lvert w_{i, T-1} w_{j,T-1}\rvert + \beta^2 d \lvert a_{i,j,T-1}\rvert\right)\\
&\cdot \sum_{k=1}^{d}\left[\left(\alpha h + \beta\lvert w_{k, T-1} \rvert + \beta \lvert w_{k',T-2}\rvert + \beta^2 d(w_{k, T-2})^2 + \beta^2 d(w_{k, T-2})^2 + \beta^2 d \lvert w_{k, T-2} w_{k',T-2}\rvert +\beta^2 d \lvert a_{k,k',T-2}\rvert\right)\right]\\&\cdot\sum_{k=1}^{d} \lvert g_{k,T-3}(\vec{z}, \vec{a}) \rvert.
\end{align*}

Observing the pattern, we see that
\begin{align*}
\lvert g_{i,T-1}(\vec{z}, \vec{a})\rvert &\leq \lVert \widetilde{X}(0)\rVert_1\left(\alpha h + \beta\lvert w_{i, T-1} \rvert + \beta \lvert w_{j,T-1}\rvert + \beta^2 d(w_{i, T-1})^2 + \beta^2 d(w_{i, T-1})^2 + \beta^2 d \lvert a_{i,j,T-1}\rvert\right)\\ &\cdot \prod_{t=0}^{T-2}\left(\sum_{k=1}^{d}\alpha h + \beta\lvert w_{k, t} \rvert + \beta \lvert w_{k',t}\rvert + \beta^2 d(w_{k, t})^2 + \beta^2 d \lvert w_{k, t} w_{k',t}\rvert+\beta^2 d(w_{k, t})^2 + \beta^2 d \lvert a_{k,k',t}\rvert\right),
\end{align*}
so
\begin{align*}
&\sum_{i=1}^{d}\lvert g_{i,T-1}(\vec{z}, \vec{a})\rvert \\ &\leq  \lVert \widetilde{X}(0)\rVert_1\prod_{t=0}^{T-1}\left(\sum_{k=1}^{d}\alpha h + \beta \sqrt{h}\lvert w_{k, t} \rvert + \beta \sqrt{h}\lvert w_{k',t}\rvert + \beta^2 h d(w_{k, t})^2 + \beta^2 h d(w_{k, t})^2 + \beta^2 d \lvert w_{k, t} w_{k',t}\rvert+\beta^2 d h \lvert a_{k,k',t}\rvert\right),
\end{align*}
where we have implicitly rescaled the $w$ and $a$, pulling out the step-size $h$. Then, integrating and using the i.i.d. property over time
\begin{align*}
&\int_{\mathcal{M}^c}\sum_{i=1}^{d}\lvert g_{i,T-1}(\vec{z}, \vec{a})\rvert \\ &\leq   \lVert \widetilde{X}(0)\rVert_1 d^T\left(\int_{\mathcal{M}^c} \alpha h + \beta \sqrt{h}\lvert w_1 \rvert + \beta \sqrt{h}\lvert w_2\rvert + \beta h d(w_1)^2 + \beta h d(w_2)^2 + \beta d \lvert w_{1} w_{2}\rvert + \beta d h \lvert a_{1,2}\rvert \right)^{T} \\
&\leq \lVert \widetilde{X}(0)\rVert_1 d^T\left(\int_{\mathcal{M}^c} \left[\alpha h + 2\beta \sqrt{h}\lvert w \rvert + 2\beta h d(w)^2] + \beta d h \lvert a_{1,2}\rvert 
 \right]+\beta d h\left(\int_{\mathcal{M}^c} \lvert w\rvert\right)^2 \right)^{T} \\
 &\leq \lVert \widetilde{X}(0)\rVert_1 Td^T\left(\int \left[\alpha h + 2\beta \sqrt{h}\lvert w \rvert + 2\beta h d(w)^2] + \beta^2 d h \lvert a_{1,2}\rvert 
 \right]+\beta^2 d h\left(\int_{\mathcal{M}^c} \lvert w\rvert\right)^2 \right)^{T-1}\\& \cdot \int  \left[\alpha h + 2\beta \sqrt{h}\lvert w \rvert + 2\beta^2 h d(w)^2] + \beta^2 d h \lvert a_{1,2}\rvert 
 \right]+\beta^2 d h\left(\int \lvert w\rvert\right)^2\\
 &\leq \lVert \widetilde{X}(0)\rVert_1 Td^T\left(\alpha h + 2\beta \sqrt{h} + 2\beta^2 h d + \beta^2 d h C
 +\beta d h \right)^{T-1}\\ &\cdot \left[\int_{\mathcal{M}^c}  \left[\alpha h + 2\beta \sqrt{h}\lvert w \rvert + 2\beta^2 h d(w)^2] + \beta^2 d h \lvert a_{1,2}\rvert 
 \right]+\beta^2 d h\left(\int_{\mathcal{M}^c} \lvert w\rvert\right)^2\right]
\end{align*}
where the second-to-last inequality follows from the union bound and $C$ is the universal constant $C = \int_{\mathbb{R}}\lvert a_{(1,2)}\rvert  da_{(1,2)}$.
Hence the problem has be reduced to analyzing the truncation error of a triple integral.

Hence there are only two kinds of integrals we need to worry about. One is the simple Gaussian \begin{align*}
    \int_{\lvert w \rvert \geq c} \lvert w \rvert^2 p_{\mathcal{N}}(w)dw,
\end{align*}
where it suffices to take $c = \mathcal{O}\left(\sqrt{\log(1/\epsilon)}\right)$ to make $\epsilon$

The second one is for the L\'evy . We know that the marginal moment generating function of the Levy  is \cite{levy1951wiener, gaines1994random}: $\mathbb{E}[e^{xA}] = \frac{1}{\cos(x)}$. Thus by Chernoff, $\mathbb{P}[\lvert a \rvert \geq c] \leq 2e^{-2\pi c}$

Following, effectively, the same arguments in the proof of Lemma~\ref{lem:parts_lemma}, one can obtain that
\begin{align*}
\int_{\lvert a \rvert \geq c}^\infty \lvert a \rvert p(a) da 
&= c\mathbb{P}[\lvert a \rvert \geq c] + \int_{c}^{\infty} x\mathbb{P}[\lvert a \rvert \geq x] dx \\
&\leq c\mathbb{P}[\lvert a \rvert \geq c] + \int_{c}^{\infty} 2x e^{-2\pi x}dx \\
&=\mathcal{O}\left(c e^{-2\pi c}\right),
\end{align*}
so $c = \mathcal{O}\left(\log(1/\epsilon)\right)$ suffices.

Hence,
\begin{align*}
\int_{\mathcal{M}^c} \lvert f\circ g(\vec{z}, \vec{a}) \rvert d\vec{z}d\vec{a} \leq B\lVert \widehat{X}(0)\rVert_1 T^2d^T\left(\alpha h + 2\beta \sqrt{h} + 2\beta h d + \beta d h C
 +\beta d h \right)^{T-1}\left[\left(\alpha h + 2\beta\sqrt{h} + 4\beta hd\right)\epsilon\right],
\end{align*}

so we can truncate the $\vec{w}$ and $\vec{a}$ to $\ell_{\infty}$ balls of size $\mathcal{O}\left(T\log(\lVert B\widehat{X}(0)\rVert_1Td/\epsilon)\right)$ for an overall error of $\mathcal{O}(\epsilon)$.

\end{proof}

\section{Additional Technical Lemmas}
\begin{lemma}
\label{lem:poly_approx_bound_lem}
    Let $h : \mathcal{X} \rightarrow \mathbb{R}, p : \mathcal{X} \rightarrow \mathbb{R}_+$ be two continuous functions over compact set $\mathcal{X} \subset \mathbb{R}$, and $\epsilon_{\textup{poly}} := \lVert h - \sqrt{p}\rVert_{\infty}$. For $\epsilon_{\text{poly}} < 1$, we have    $$\lVert p - h^2 \rVert_{\infty} \leq 3\epsilon_{\textup{poly}},$$.
\end{lemma}
\begin{proof}
For $\epsilon_{\text{poly}} < 1$, we have 
\begin{align*}
&\lVert p - h^2 \rVert_{\infty} \leq \lVert \sqrt{p} - h \rVert_{\infty}\lVert \sqrt{p} + h \rVert_{\infty}\\  &\leq \lVert \sqrt{p} - h\rVert_{\infty}\left[2\lVert\sqrt{p}\rVert_{\infty} + \lVert \sqrt{p} - h \rVert_{\infty}\right] \leq 3\epsilon_{\text{poly}}.
\end{align*}
\end{proof}

\begin{lemma}
\label{lem:poly_approx_lem}
Suppose $\mathcal{R} = \{ z \in \mathbb{C} : \arg(z) \in [-\frac{\pi}{4}, \frac{\pi}{4}] \cup [\frac{3\pi}{4}, \frac{5\pi}{4}]$\}, then the functions $\frac{z}{\sinh(z)}$ and $\frac{z}{\tanh(z)}$ can be uniformly approximate to $\epsilon$ additive error on $\mathcal{D}(r) \cap \mathcal{R}$ by an $\mathcal{O}(\log(r/\epsilon))$ degree polynomial, where $\mathcal{D}(r)$ is the closed disc of radius $r$ around zero in $\mathbb{C}$.
\end{lemma}
\begin{proof}

It is known that we have the following series expansions for $\frac{z}{\sinh(z)}$ and $\frac{z}{\tanh(z)}$:
\begin{align*}
&\frac{z}{\sinh(z)} =  (-1)^{\text{sgn}(\text{Re}(z))}2ze^{(-1)^{\text{sgn}(\text{Re}(z))+1}z}\sum_{k=0}^{\infty}e^{(-1)^{\text{sgn}(\text{Re}(z))+1}2zk}\\
&\frac{z}{\tanh(z)}  = (-1)^{\text{sgn}(\text{Re}(z))}z(1+e^{(-1)^{\text{sgn}(\text{Re}(z))+1}2z})\sum_{k=0}^{\infty}e^{(-1)^{\text{sgn}(\text{Re}(z))+1}2zk}\\
&\frac{z}{\sinh(z)} = \sum_{n=0}^{\infty}\frac{2(1-2^{2n-1})B_{2n}z^{2n}}{(2n)!}, \lvert z \rvert \in (0, \pi)\\
&\frac{z}{\tanh(z)} =\sum_{n=0}^{\infty}\frac{2^{2n}B_{2n}z^{2n}}{(2n)!}, \lvert z \rvert \in (0, \pi),
\end{align*}
where $B_{2n}$ are the Bernoulli numbers.

Let $\epsilon$ denote the desired polynomial approximation error.
For the first two series expansions for $z\csch(z)$ and $z\coth(z)$ we have that the truncation error when keeping $k-1$ terms is bounded by  
\begin{align*}
&\frac{2\lvert z \rvert e^{-\lvert\text{Re}(z)\rvert}}{\lvert 1-e^{-2\lvert\text{Re}(z)\rvert}\rvert}e^{-2k\lvert\text{Re}(z)\rvert} = \frac{\lvert z \rvert e^{-2k\lvert\text{Re}(z)\rvert}}{\lvert \sinh(\text{Re}(z))\rvert} \leq e^{-2k\lvert\text{Re}(z)\rvert},\\
&\frac{\lvert z(1+e^{-2\lvert\text{Re}(z)\rvert})\rvert}{\lvert1-e^{-2\lvert\text{Re}(z)\rvert}\rvert}e^{-2k\lvert\text{Re}(z)\rvert} = \frac{ze^{-2k\lvert\text{Re}(z)\rvert}}{\lvert \tanh(\text{Re}(z))\rvert} \leq z e^{-2k\lvert\text{Re}(z)\rvert},
\end{align*}
respectively. Hence since $z \in \mathcal{D}(r)$ if suffices to take $k = \mathcal{O}\left(\frac{1}{\lvert \text{Re}(z)\rvert}\log(r/\epsilon)\right)$ for an $\epsilon$ truncation error. If $\lvert \text{Re}(z)\rvert \geq 1$, then we can just take $k = \mathcal{O}(\log(r/\epsilon))$. 

For the other case, since $z \in \mathcal{R}$, $\lvert \text{Re}(z)\rvert < 1 \implies  \lvert z\rvert < \sqrt{2} < \pi$. Thus, we are within the radius of convergence of the second pair of series. Also, $k = \mathcal{O}\left(\log(1/\epsilon)\right)$ suffices for those series.
\end{proof}

\begin{lemma}
\label{lem:left_rule_error}
Consider an integral $\int_{\mathcal{X}} f(\vec{x}) d\vec{x}$, where $\mathcal{X} = [a_L^{(1)}, a_U^{(1)}] \times \cdots \times [a_L^{(d)}, a_U^{(d)}]$ is a rectangular region in $\mathbb{R}^d$, and  $\Delta_k = \frac{a_U^{(k)} - a_L^{(k)}}{N}$. Let $\mathcal{R}_N(f)$ is the left-endpoint Riemann sum which takes $N$ uniform grid points in each dimension, and let $\mathcal{M} \subset \mathcal{X}$ define the grid. Let $\mathcal{C}_{\vec{x}}$ denote the cell associated with grid point $\vec{x}$ in $\mathcal{M}$, i.e. $\mathcal{C}_{\vec{x}} = \vec{x} + \prod_k [0, \Delta_k]$. Then 
\begin{align*}
&\lvert \mathcal{R}_N(f) - \int_{\mathcal{X}} f(\vec{x}) d\vec{x}\rvert  \leq  \left[\sum_{\vec{x} \in \mathcal{M}} \sup_{\vec{y} \in C_{\vec{x}}}\lVert \nabla f(\vec{y})\rVert_{\infty}\prod_{k=1}^{d} \Delta_{k}\right] \frac{\textup{diam}(\mathcal{X})}{N}.
\end{align*}
\end{lemma}
\begin{proof}
Suppose $\mathcal{C}_{\vec{x}} = [x_1,  x_1 + \Delta_1] \times \cdots \times [x_d, x_d + \Delta_d]$. Then,
\begin{align*}
\lvert \mathcal{R}_N(f) - \int_{\mathcal{X}} f(\vec{x}) d\vec{x}\rvert &\leq \sum_{\vec{x} \in \mathcal{M}}\int_{\mathcal{C}_{\vec{x}}} \lvert f(\vec{x}) - f(\vec{y})\rvert d\vec{y} \\
&\leq\sum_{\vec{x} \in \mathcal{M}}\sup_{\vec{y} \in \mathcal{C}_{\vec{x}}}\lVert \nabla f(\vec{y})\rVert_{\infty}\int_{\mathcal{C}_{\vec{x}}} \lVert \vec{x} - \vec{y} \rVert_{1} d\vec{y}.
\end{align*}

Note that
\begin{align*}
\int_{\mathcal{C}_{\vec{x}}} \lVert \vec{x} - \vec{y} \rVert_{1} d\vec{y} = \sum_{k=1}^{d}\int_{\mathcal{C}_{\vec{x}}} \lvert x_k - y_k \rvert d\vec{y} = \prod_{k=1}^{d}\Delta_k \sum_{k=1}^{d} \frac{\Delta_k}{2} \leq \frac{\textup{diam}(\mathcal{X})}{N}\prod_{k=1}^{d} \Delta_{k}.
\end{align*}
So the result follows.
\end{proof}

\begin{lemma}
\label{lem:parts_lemma}
Let ${X}$ be a random variable with pdf $p(x)$. Suppose $\beta > 0$, and that $e^{\beta y}\mathbb{P}[X \geq y]$ is integrable on $[0, \infty)$. Then for any $\alpha \in [0, \infty)$
\begin{align*}
\int_{\alpha}^{\infty} e^{\beta x} p(x) dx = e^{\beta \alpha}\mathbb{P}[ X \geq \alpha] +  \int_{\alpha}^{\infty}\beta e^{\beta y}\mathbb{P}[ X \geq y]dy.
\end{align*}
\end{lemma}
\begin{proof}
\begin{align*}
\int_{\alpha}^{\infty} e^{\beta x} p(x) dx &= \int_{-\infty}^{\infty} \mathbb{1}_{x\geq \alpha} p(x)\int_{-\infty}^{x}\beta e^{\beta y}dydx\\
&=\int_{-\infty}^{\infty}\beta e^{\beta y}\int_{y}^{\infty}  \mathbb{1}_{x\geq \alpha} p(x) dx dy\\
&=\int_{-\infty}^{\infty}\beta e^{\beta y}\mathbb{P}[ X \geq \max(\alpha, y)]\\
&=\int_{-\infty}^{\alpha}\beta e^{\beta y}\mathbb{P}[ X \geq \alpha]dy +  \int_{\alpha}^{\infty}\beta e^{\beta y}\mathbb{P}[ X \geq y]dy\\
&=e^{\beta \alpha}\mathbb{P}[ X \geq \alpha] +  \int_{\alpha}^{\infty}\beta e^{\beta y}\mathbb{P}[ X \geq y]dy
\end{align*}
\end{proof}

\begin{lemma}
\label{lem:tail_trick_lem}
Let ${X}$ be a random variable with pdf $p(x)$. Suppose $\beta > 0$, and that $e^{\beta y}\mathbb{P}[X \geq y] = \mathcal{O}\left(e^{-(\theta - \beta)y}\right)$ with $\theta > \beta$, then
\begin{align*}
\int_{\alpha}^{\infty} e^{\beta x}p(x) dx = \mathcal{O}\left(e^{-(\theta - \beta)\alpha}\right). 
\end{align*}
\end{lemma}
\begin{proof}

If we have $e^{\beta y}\mathbb{P}[X \geq y] = \mathcal{O}\left(e^{-(\theta - \beta)y}\right)$ for $\theta > \beta$. Then the above is upper bounded by 
\begin{align*}
e^{-(\theta-\beta)\alpha} +  \int_{\alpha}^{\infty}\beta e^{-(\theta-\beta)y}dy = (1 + \frac{\beta}{\theta-\beta})e^{-(\theta-\beta)\alpha} = \mathcal{O}\left(e^{-(\theta - \beta)\alpha}\right).
\end{align*}
\end{proof}

\begin{lemma}
\label{lem:chi-square-shift-lem}
Let $\mathcal{Y}$ be $\chi^2_r$ distributed, then for $m < \frac{1}{2}$. Then
\begin{align*}
\int_{\alpha}^{\infty} \exp(my)p(y)dy = (1-2m)^{-r/2}\mathbb{P}[\mathcal{Y} \geq (1-2m)\alpha].
\end{align*}
\end{lemma}
\begin{proof}
We will apply a change of variables $y = \frac{z}{1-2m}$.
\begin{align*}
\int_{\alpha}^{\infty} \exp(my)p(y)dy &= \int_{\alpha}^{\infty}\frac{y^{r/2 -1}e^{-y(1/2-m)}}{2^{r/2}\Gamma(r/2)}dy\\
&=(1-2m)^{-r/2}\int_{(1-2m)\alpha}^{\infty}\frac{z^{r/2 -1}e^{-z/2}}{2^{r/2}\Gamma(r/2)}dz \\ &= (1-2m)^{-r/2}\mathbb{P}[\mathcal{Y} \geq (1-2m)\alpha].
\end{align*}
\end{proof}

\section{Review of Fixed-Point Coherent Arithmetic}
\label{sec:coherent-arith}
We consider fixed-point quantum arithmetic, i.e. approximating real numbers by $n$ bits:
\begin{align*}
    x := \underbrace{x_{n-1}\cdots x_{n-p}}_{p}. \underbrace{x_{n-p-1}\cdot x_{0}}_{n-p}.
\end{align*}

The idea is to implement reversible versions of arithmetic operations:
\begin{align*}
    |x\rangle|0\rangle \rightarrow |x\rangle|f(x)\rangle,
\end{align*}
for arithmetic function $f$.

\subsection{Arithmetic for Square-root}
\label{sec:arith_for_square_root}
The paper proposes implementing the inverse-square root function via Newton's method,i.e.
\begin{align*}
    x_{n+1} = x_n(1.5 - \frac{ax_n^2}{2}),
\end{align*}
 so that $x_n \underset{n\rightarrow \infty}{\rightarrow} \frac{1}{\sqrt{a}}$. They mention a sufficient initial guess is $2^{\lfloor - \frac{\lfloor \log_2 a\rfloor}{2}\rfloor}$. 
 Note within a sufficient region of convergence Netwton's method has quadratic convergence. The number of Toffolis required to implement $m$ Newton steps is 
 \begin{align*}
     T_{\text{invsqrt}}= n^2\left(\frac{15}{2}m + 3\right) + 15npm +n\left(\frac{23}{2}m + 5\right) - 15p^2m + 15pm - 2m.
 \end{align*}

The square-root function can be calculated by $x \cdot \frac{1}{\sqrt{x}}$. The Toffoli count of multiplying two $n$-bit numbers is \begin{align*}
    T_{\text{mul}}(n,p) = \frac{3}{2}n^2 + 3np + \frac{3}{2}n - 3p^2 + 3p,
\end{align*}
where $n = \mathcal{O}(p)$ gives
\begin{align}
    T_{\text{mul}} = \mathcal{O}\left(n^2\right),
\end{align}
and 

\begin{align*}
    T_{\text{add}} = \mathcal{O}\left(n\right)
\end{align*}

Note that the number of qubits used by the reversible implementation of Newton's method grows linearly with the number of iterations. However, since Newton's method (with good initialization) has $\mathcal{O}\left(\log\log(1/\epsilon)\right)$ convergence. 

Assuming we are in a region of quadratic convergence with $\epsilon < 2^{-p}$. If $n = \mathcal{O}(p)$, we have that
\begin{align}
\label{eqn:sqrt_gate_compl}
    T_{\text{sqrt}} = \mathcal{O}\left(n^2\log\log(1/\epsilon)\right).
\end{align}

\subsection{Arithmetic for Polynomials}

According to \cite{häner2018optimizingquantumcircuitsarithmetic} we can implement a $d$ degree polynomial in $x$ 
\begin{align*}
    P(x) = \sum_{k=0}a_kx^{k},
\end{align*}
using 
\begin{align*}
    T_{\text{poly}}(n,d,p) = \frac{3}{2}n^2d + 3npd + \frac{7}{2}nd - 3p^2d + 3pd - d
\end{align*}
Toffoli gates. Taking $n = \mathcal{O}(p)$, we have in general
\begin{align}
\label{eqn:polynomial_gate_cost}
    T_{\text{poly}}= \mathcal{O}\left(n^2 d\right)
\end{align}
gates.

For $\sin^{-1}$ we have for $\lvert x\rvert \leq 1$:
\begin{align*}
    \sum_{k=0}^{d}\frac{(2n)!}{2^{2n}(n!)^2}\frac{x^{2n+1}}{2n+1},
\end{align*}
and again we can take $d =\mathcal{O}\left(\log(1/\epsilon)\right)$ for $\epsilon$ additive error over $[-1, 1]$.

Note that if we are working with $n + p$ bits, then we just need $\epsilon < 2^{-p}$. Hence 

\begin{align*}
    T_{\sin^{-1}} = T_{\text{poly}}(n,p,p) = \frac{3}{2}n^2p + 3np^2 + \frac{7}{2}np - 3p^3 + 3p^2 - p,
\end{align*}
assuming $ n = \mathcal{O}(p)$. We have
\begin{align}
\label{eqn:sin_gate_cost}
    T_{\sin^{-1}} = \mathcal{O}\left(n^2\log(1/\epsilon)\right).
\end{align}

Due to potential numerical instabilities near $x=\pm 1$, \cite{häner2018optimizingquantumcircuitsarithmetic} suggest changing coordinates for larger values of $x$ and implementing $\sin^{-1}$ via
\begin{align*}
    \sin^{-1}\left(\sqrt{1-x^2}\right).
\end{align*}

\printbibliography

\end{document}